%% file: mainfull.tex
\documentclass[runningheads]{llncs}
\input{macros}

\input{styles}

\begin{document}

\title{A Program Logic for Fresh Name Generation}

\author{Harold Pancho Eliott \inst{1}\and Martin Berger \inst{1,2}
\authorrunning{HPG Eliott and M Berger}
\institute{Department of Informatics, University of Sussex, Brighton, UK.
\and
Turing Core, Huawei 2012 Labs, London, UK.}
}
\maketitle        
\input{abstract}

\input{introduction}
\input{language}

\input{logic}

\input{model}
\input{formulae_properties}

\input{axioms}
\input{rules}

\input{reasoning}
\input{conclusion}

\bibliographystyle{splncs04}
\bibliography{semantic}

\input{appendix/allappendices}

\end{document}

%% file: macros.tex
\newcommand{\NI}{\noindent}

\newcounter{line}
\newcommand{\infer}[2]{\frac{\displaystyle{ #1 }}{\displaystyle{ #2 }}}
\newcommand{\ZEROPREMISERULE}[1]{\infer{-}{#1}}
\newcommand{\ONEPREMISERULE}[2]{\infer{#1}{#2}}
\newcommand{\TWOPREMISERULE}[3]{\infer{#1 \quad #2}{#3}}
\newcommand{\THREEPREMISERULE}[4]{\infer{#1 \quad #2 \quad #3}{#4}}

\newcommand{\FIVEPREMISERULE}[6]{\infer{#1 \quad #2 \quad #3 \quad #4 \quad #5}{#6}}

\newcommand{\RULENAME}[1]{{\sc{#1}}}
\newcommand{\SMALLRULENAME}[1]{\textsc{\tiny #1}}

\newcommand{\SUBST}[2]{[#1/#2]}
\newcommand{\PSUBST}[2]{\PROGRAM{ [ }#1 \PROGRAM{/} #2\PROGRAM{]}} 
\newcommand{\LSUBST}[2]{\PROGRAMLOGIC{ [ }#1 \PROGRAMLOGIC{/} #2\PROGRAMLOGIC{]}} 
\newcommand{\LSUBSTLTC}[3]{\PROGRAMLOGIC{ [ }#1 \PROGRAMLOGIC{/} #2\PROGRAMLOGIC{]}_{#3}} 

\newenvironment{FIGURE}{\raggedbottom\begin{figure}\rule{\linewidth}{.5pt}}{\rule{\linewidth}{.5pt}\end{figure}\raggedbottom}

\newcommand{\REF}[1]{\mathsf{Ref}(#1)}
\newcommand{\HIDDEN}[1]{}
\newcommand{\PROOFFINISHED}[1]
{ #1}

\newcommand{\FRESH}[2]{#1 \# #2}
\newcommand{\FORALL}[2]{\forall #1 \in (#2).}
\newcommand{\EXISTS}[2]{\exists #1 \in (#2).}

\newcommand{\DONE}{}
\newcommand{\HIDE}[1]{}

\newcommand{\FAD}[1]{\forall #1.}

\newcommand{\TCV}{\delta}
\newcommand{\PLUSTC}{\!+\!}
\newcommand{\PLUSV}{\!+\!}
\newcommand{\PLUSG}{\!+\!}
\newcommand{\GAMMA}{\mathrm{\mathit I}\!\!\Gamma}

\newcommand{\CONSTRUCT}[2]{#1 \triangleright #2}

\newcommand{\LTCtoSTC}{\downarrow_{-TC}} 
\newcommand{\REMOVETCVfrom}{\backslash_{-TCV}} 

\newcommand{\EMPH}[1]{\emph{#1}}

\newcommand{\NUC}{$\nu$-calculus}
\newenvironment{GRAMMAR}{\[\begin{array}{lcl}}{\end{array}\]}
\newenvironment{RULES}{\[\begin{array}{c}}{\end{array}\]}
\newcommand{\VERTICAL}{\  \mid\hspace{-3.0pt}\mid \ }
\newcommand{\UNIT}{\mathsf{Unit}}
\newcommand{\IFTHEN}[2]{\mathsf{if}\; #1\;\mathsf{then}\; #2}
\newcommand{\IFTHENELSE}[3]{\IFTHEN{#1}{#2}\;\mathsf{else}\;#3}
\newcommand{\TC}{\mathbb{TC}}
\newcommand{\NAME}{\mathsf{Nm}}

\newcommand{\BOOL}{\mathsf{Bool}}
\newcommand{\INT}{\mathsf{Int}}
\newcommand{\TYBASE}{\mathsf{\alpha_b}}
\newcommand{\RED}{\rightarrow}
\newcommand{\FS}{\rightarrow}
\newcommand{\TRUE}{\PROGRAM{true}}
\newcommand{\FALSE}{\PROGRAM{false}}
\usepackage{ifthen}
\newboolean{showcomments}
\setboolean{showcomments}{true}
\ifthenelse{\boolean{showcomments}}
{\newcommand{\mynote}[2]{
		\fbox{\bfseries\sffamily\scriptsize#1}
		{\small$\blacktriangleright$\textsf{\emph{#2}}$\blacktriangleleft$}
	}
}
{\newcommand{\mynote}[2]{}
}

\newcommand{\RBOX}[1]{\parbox[t]{3cm}{\raggedleft #1 }}

\newcommand{\GENSYM}{\mathsf{gensym}}
\newcommand{\EXT}[1]{\mathsf{Ext}{(#1)}}

\newcommand{\TCTYPES}[2]{#1 \Vdash {#2}}

\newcommand{\FORMULATYPES}[2]{#1 \Vdash {#2}}
\newcommand{\JUDGEMENTTYPES}[2]{\FORMULATYPES{#1}{#2}}
\newcommand{\MODELTYPES}[2]{#1 \Vdash #2}

\newcommand{\FV}[1]{\mathsf{fv}(#1)}

\newcommand{\FTCV}[1]{\mathsf{ftcv}(#1)}
\newcommand{\AN}[1]{\textsf{\aa}(#1)}

\newcommand{\ONEEVAL}[4]{#1 \bullet #2=#3 \{#4\}}

\newcommand{\DEFEQ}{\stackrel{\text{\EMPH{def}}}{=}}
\newcommand{\CAL}[1]{\mathcal{#1}}
\newcommand{\MMM}{\xi}
\newcommand{\EEE}[1]{\CAL{E}[#1]}

\newcommand{\RAWPAIR}[1]{\langle#1\rangle}
\newcommand{\PAIR}[2]{\RAWPAIR{#1, #2}}
\newcommand{\EQA}[2]{#1 = #2}

\newcommand{\PROJ}[3]{\pi^{#1}_{#2}(#3)}

\newcommand{\LET}[3]{\PROGRAM{let}\ #1 = #2\ \PROGRAM{in}\ #3}

\newcommand{\EXTSINGLE}{\preccurlyeq}
\newcommand{\EXTSTAR}{\preccurlyeq^{\star}}

\newcommand{\LTCDERIVEDVALUE}[4]{#1 \stackrel{[#2,\ #3]}{\rightsquigarrow} #4}

\newcommand{\CONGEXTSINGLE}{\sim^{\EXTSINGLE}}
\newcommand{\CONGEXTSTAR}{\sim^{\EXTSTAR}}

\newcommand{\DOM}[1]{\mathsf{dom}(#1)}
\newcommand{\TYPES}[3]{#1 \vdash #2 : #3}

\newcommand{\EXPRESSIONTYPES}[3]{#1 \Vdash #2 : #3}
\newcommand{\CONG}{\cong}

\newcommand{\CONGCONTEXT}[2]{\CONG_{#1}^{#2}}

\newcommand{\MINUS}{{\mbox{\bf\small -}}}
\newcommand{\LOGIC}[1]{\mathsf{#1}}
\newcommand{\PROGRAM}[1]{\mathsf{#1}}

\newcommand{\SEM}[2]{\lbrack\!\lbrack #1 \rbrack\!\rbrack_{#2}}
\newcommand{\TRANSLATE}[1]{\langle\!\langle #1 \rangle\!\rangle}

\newcommand{\ZEROPREMISERULENAMEDRIGHT}[2]{\ZEROPREMISERULE{#1}\,\SMALLRULENAME{#2}}
\newcommand{\ONEPREMISERULENAMEDRIGHT}[3]{\ONEPREMISERULE{#1}{#2}\,\SMALLRULENAME{#3}}
\newcommand{\TWOPREMISERULENAMEDRIGHT}[4]{\TWOPREMISERULE{#1}{#2}{#3}\,\SMALLRULENAME{#4}}
\newcommand{\THREEPREMISERULENAMEDRIGHT}[5]{\THREEPREMISERULE{#1}{#2}{#3}{#4}\,\SMALLRULENAME{#5}}

\newcommand{\FIVEPREMISERULENAMEDRIGHT}[7]{\FIVEPREMISERULE{#1}{#2}{#3}{#4}{#5}{#6}\,\SMALLRULENAME{#7}}
\newcommand{\NOVSPACEPARAGRAPH}[1]{\NI\textbf{\EMPH{#1}.}}
\newcommand{\PARAGRAPH}[1]{\vspace{2mm}\NOVSPACEPARAGRAPH{#1}}
\newcommand{\TRUTH}{\LOGIC{T}}
\newcommand{\FALSITY}{\LOGIC{F}}
\newcommand{\IMPLIES}{\rightarrow}
\newcommand{\VEC}[1]{\tilde{#1}}

\newcommand{\AND}{\land}
\newcommand{\ASSUMETERMINATION}{}

\newcommand{\EVALFORMULA}[4]{#1 \bullet #2 = #3 \{#4\}}
\newcommand{\EVALFORMULASHORT}[3]{#1 \bullet #2 = #3}

\newcommand{\OR}{\lor}

\newcommand{\IFF}{\leftrightarrow}
 \newcommand{\GREYBOX}[1]{{\colorbox{darkgray!10}{$#1$}}}
\newcommand{\CONVERGES}{\Downarrow}
\newcommand{\CONV}{\CONVERGES}

 \newcommand{\nextLine}{\\[1mm] \hline \\[-3mm]}

\newenvironment{NDERIVATION}[1]{\setcounter{line}{#1}\[\begin{array}{ll}}{\end{array}\]}

\newcommand{\NLINESKELETON}[2]{\theline &\quad  #1\ \quad\hfill \text{\EMPH{#2}}\addtocounter{line}{1}}
\newcommand{\NLINE}[2]{\NLINESKELETON{#1}{#2}\nextLine}
\newcommand{\NLASTLINE}[2]{\NLINESKELETON{#1}{#2}}
\newcommand{\NPLINESKELETON}[3]{\theline &\quad  \!\!\begin{array}[t]{l} #1 \end{array} \hfill \!\!\parbox[t]{#2}{\raggedleft \EMPH{#3}}\addtocounter{line}{1}}
\newcommand{\NPLINE}[3]{\NPLINESKELETON{#1}{#2}{#3}\nextLine}

\newcommand{\ASSERT}[4]{\{#1\}\; #2 :_{#3} \{#4\}}

\newcommand{\METALOGIC}[1]{{\boldsymbol{\color{cyan} #1}}}

\newcommand{\MIMPLIES}{\METALOGIC{\IMPLIES}}
\newcommand{\MIMPLIEDBY}{\METALOGIC{\leftarrow}}
\newcommand{\MIFF}{\METALOGIC{\IFF}}
\newcommand{\MAND}{\METALOGIC{\AND}}
\newcommand{\MOR}{\METALOGIC{\OR}}
\newcommand{\Mforall}{\METALOGIC{\forall}}
\newcommand{\Mexists}{\METALOGIC{\exists}}
\newcommand{\Mequiv}{\METALOGIC{\equiv}}

\newcommand{\PROGRAMLOGIC}[1]{ #1}
\newcommand{\PIMPLIES}{\PROGRAMLOGIC{\IMPLIES}}
\newcommand{\PIMPLIEDBY}{\PROGRAMLOGIC{\leftarrow}}
\newcommand{\PIFF}{\PROGRAMLOGIC{\IFF}}
\newcommand{\PAND}{\PROGRAMLOGIC{\AND}}
\newcommand{\POR}{\PROGRAMLOGIC{\OR}}

\newcommand{\THINWRT}[1]{\textsf{ thin w.r.t } #1 }

\newcommand{\REMOVEVARIABLE}{ \backslash }
\newcommand{\EXTINDEP}{\textsf{Ext-Ind}}
\newcommand{\SYNEXTINDEP}{\textsf{SYN-EXT-IND}}

%% file: styles.tex
 \usepackage{hyperref}
 \hypersetup{
     colorlinks=true, 
     linktoc=all,     
     linkcolor=black,  
 }   
 \usepackage{cite}
 \hypersetup{
     citecolor=black
 }  
\usepackage{amsmath}
\usepackage{amssymb}
\usepackage{xcolor}
\usepackage{appendix}

%% file: abstract.tex
\begin{abstract}
  We present a program logic for Pitts and Stark's \NUC, an
  extension of the call-by-value simply-typed $\lambda$-calculus with a
  mechanism for the generation of fresh names. Names can be compared
  for equality and inequality, producing programs with subtle observable
  properties.  Hidden names produced by interactions between
  generation and abstraction are captured logically with a
  second-order quantifier over type contexts.
  We illustrate usage of the logic through
  reasoning about well-known difficult cases from the literature.
\end{abstract}

%% file: introduction.tex
\section{Introduction}

Naming is a long-standing problem in computer science. Most
programming languages can define naming constructs, which, when
called, yield a fresh name.  The $\pi$-calculus \cite{MilnerR:calmp1}
made naming and the $\nu$-operator, a constructor for name creation, a
first-class construct, leading to a flurry of research,
e.g.~\cite{PittsAM:newaas,FernandezM:nomrews,HondaK:elestriptswr,PittsAM:nomsetnasics,PittsAM:nomlfo-jv,UrbanC:nomteciih}. Initially
it was unclear if the $\pi$-calculus approach had purchase beyond
process calculi. Pitts and Stark \cite{PittsAM:obsproohoftdclnown} as
well as Odersky \cite{OderskyM:funtheoln} added the $\nu$-operator to
the simply-typed $\lambda$-calculus (STLC from now on), and showed
that the subtleties of naming are already present in the interplay
between higher-order functions and fresh name generation.  This raises
the question of how compositionally to reason about programs that can
generate fresh names?  There are program logics for ML-like languages
that can generate fresh references, such as
\cite{YHB07:local:full,DreyerD:impohigosacfolrr},
but, to the best of our knowledge, always in the context of languages
with other expressive features such as aliasing, mutable higher-order
state or pointer arithmetic, leading to complex logics, where the
contribution of fresh name generation to the difficulties of reasoning
is not apparent. This is problematic because, while the type $\NAME$
carries the same information as $\REF{\UNIT}$ in ML, we are often
interested in reasoning about languages that combine fresh name
generation with other features, such as meta-programming
\cite{BergerM:modhomgmp}. Can we study reasoning about fresh names in
as simple a programming language as possible?

\begin{quote}
\textbf{Research question}. Is there a Hoare-style program logic for the \NUC,
conservatively extending program logics for the STLC in a natural
manner, that allows for compositional reasoning about fresh name
generation?
\end{quote}

\NI The present paper gives an affirmative answer to the research
question, and presents the first program logic for the \NUC.

\PARAGRAPH{Informal explanation}
By the \NUC\ we mean the STLC with a type $\NAME$ of names, a
constructor $\GENSYM$ of type $\UNIT \FS \NAME$ and a destructor, in
form of equality and inequality on names ($\GENSYM$ and $\nu$ are essentially
identical, but the former is more widely used). Immediately we realise that
the \NUC\ loses extensionality, as $\GENSYM () = \GENSYM ()$ evaluates to $\FALSE$.
While the loss of extensionality is expected in a stateful language,
the \NUC\ does not have state, at least not in a conventional sense.

A first difficulty is expressing freshness in logic.  What does it
mean for a name $x$ to be fresh? A first idea might be to say that $x$
is guaranteed to be distinct from all existing names. We cannot simply
say
\[
\ASSERT{\TRUTH}{\GENSYM()}{u}{\forall x. u \neq x}
\]
since we must prevent $\forall x. u \neq x$ being instantiated to $u
\neq u$. We want to say something like:
\begin{align}\label{no_future}
\forall x. \ASSERT{\TRUTH}{\GENSYM()}{u}{u \neq x}
\end{align}
Unfortunately we cannot quantify over Hoare triples.  A second problem
is that (\ref{no_future}) is not strong enough, in the sense that
$\GENSYM$ does not just create names that are fresh w.r.t.~existing
names, but also w.r.t.~all future calls to $\GENSYM$.  We 
introduce a new quantifier to deal with both problems at the same
time.  A third difficulty is that fresh names can be exported or
remain hidden, with observable consequences.  Consider:
\begin{align}\label{example_canonical}
  \LET{x}{\GENSYM()}{\lambda y. x=y}
\end{align}
of type $\NAME \FS \BOOL$. It receives a name $y$ as argument and
compares it with fresh name $x$. Since $x$ is never exported to the
outside, no context can ever supply this fresh $x$ to
$\lambda y. x=y$. Hence (\ref{example_canonical}) must be contextually
indistinguishable from $\lambda y.\FALSE$. Operationally, this is
straightforward. But how can we prove this compositionally?
Note that this is not a property of $\lambda y. x=y$, but it is also
not a consequence of $x$'s being freshly generated, for $x$ is also
fresh in this program:
\begin{align}\label{example_canonical_pair}  
  \LET{x}{\GENSYM()}{\PAIR{x}{\lambda y. x=y}}
\end{align}
But in (\ref{example_canonical_pair}), $\lambda y. x=y$ can return
$\TRUE$, for example if we use (\ref{example_canonical_pair}) in this
context:
\[
   \LET{p}{[\cdot]}{(\pi_2\ p)(\pi_1\ p)}
\]

In program logics like \cite{HY04PPDP}, the specification of any
abstraction $\lambda y.M$ will be a universally quantified formula
$\forall y.A$.  With fresh names, instantiation of quantification is a core
difficulty.  Recall that in first-order logic, $\forall y.A$
always implies $A\SUBST{e}{y}$, for all ambient expressions
$e$. Clearly, in the case of (\ref{example_canonical}) we cannot
conclude to $A\SUBST{x}{y}$ from $\forall y.A$, because $x$ is, in
some sense, not available for instantiation. In contrast, in the case
of (\ref{example_canonical_pair}) we can infer $A\SUBST{x}{y}$.  Hence
we need to answer the question how to express logically the inability
to instantiate a universal quantifier with a fresh and hidden name
like $x$ in (\ref{example_canonical}).
We introduce a novel restricted quantifier,
limiting the values based on a type context,
and a new quantifier over type contexts to 
extend the reach of restricted quantifiers.

%% file: language.tex
\section{Programming Language}

Our programming language is essentially the \NUC \ of
\cite{PittsAM:obsproohoftdclnown}, with small additions in particular
pairs, included for the sake of convenience.  We assume a countably
infinite set of variables, ranged over by $x, y, ...$ and a countably
infinite set, disjoint from variables, of names, ranged over by $r,
...$. 
Constants ranged over by $c$ are  Booleans $\TRUE$, $\FALSE$, and Unit $()$.
For simplicity we also call our language \NUC. It is given by
the following grammar, where $\alpha$ ranges over types, $\Gamma$ over
\EMPH{standard type contexts} (STC), $V$ over values and $M$ over
programs.  (Additions over the STLC highlighted.)
\begin{GRAMMAR}
	\begin{array}{lclllllll}
		\alpha
		&\ ::= \ &
		\UNIT \VERTICAL \BOOL \VERTICAL \GREYBOX{\NAME} \VERTICAL \alpha \FS \alpha \VERTICAL \alpha \times \alpha
		&\qquad&
		\Gamma
		&::=&
		\emptyset \VERTICAL \Gamma, x:\alpha
        \end{array}        
	\\\\
	\begin{array}{lclllllll}        
		V
		&\ ::=\ &
		\GREYBOX{r} \VERTICAL \GREYBOX{\GENSYM } \VERTICAL x \VERTICAL c \VERTICAL \lambda x.M
		\VERTICAL \PAIR{V}{V}
		\\[0.5mm]
		M
		& ::= &
		V \VERTICAL MM  \VERTICAL \LET{x}{M}{M} \VERTICAL M = M 
		\\
		& \VERTICAL & \IFTHENELSE{M}{M}{M} \VERTICAL \PAIR{M}{M} \VERTICAL \PROJ{}{i}{M}  
	\end{array}
\end{GRAMMAR}
\NI \EMPH{Free variables} in $M$, written $\FV{M}$ are defined as
usual.  $M$ is \EMPH{closed} if $\FV{M} = \emptyset$.  
There are no binders for names so the set
$\AN{M}$ of \EMPH{all names} in $M$, is given by the obvious
rules, including $\AN{r} = \{r\}$, $\AN{MN} = \AN{M} \cup \AN{N}$. 
If $\AN{M}= \emptyset$ then $M$ is \EMPH{compile-time} syntax.
The $\nu n. M$ constructor from the \NUC \ \cite{PittsAM:obsproohoftdclnown} is equivalent to
$\LET{n}{\GENSYM()}{M}$ as $\GENSYM()$ generates fresh names. 
The typing judgements is $\TYPES{\Gamma}{M}{\alpha}$, 
with the STC $\Gamma$  being an unordered mapping from variables to types.
Typing rules are standard \cite{PierceBC:typsysfpl} with the following extensions:
$\TYPES{\Gamma}{\GENSYM}{\UNIT \FS \NAME}$ and $\TYPES{\Gamma}{r}{\NAME}$.

The operational semantics of our \NUC \ is straightforward and the same as 
\cite{PittsAM:obsproohoftdclnown}.  A \EMPH{configuration of type
  $\alpha$} is a pair $(G, M)$ where $M$ is a closed term of
type $\alpha$, and $G$ a finite set of previously \underline{G}enerated names such that $\AN{M}
\subseteq G$.
The standard call-by-value reduction relation, $\RED$,  
has the following key rules.
\begin{equation*}
  \begin{array}{rclcr}
		(G, \ (\lambda x. M) V) & \quad\RED\quad & (G, \ M \PSUBST{V}{x}) 
		\\
		(G, \ \GENSYM()) & \RED & (G \cup \{n\}, \  n)  &\quad& (n \notin G)
		\\
		(G \cup\{ n\}, \ n = n) & \RED & (G \cup\{ n\}, \ \TRUE)
		\\
		(G \cup\{ n_1,n_2\}, \ n_1 = n_2) & \RED & (G \cup\{ n_1,n_2\}, \ \FALSE) && (n_1 \neq n_2)
  \end{array}
\end{equation*}
\[
		(G, \ M) \RED (G', \ N) \quad\text{implies}\quad (G, \ \EEE{M})  \RED (G', \ \EEE{N})  
\]
Here $M\PSUBST{V}{x}$ is the usual capture-avoiding substitution, and
$\EEE{\cdot}$ ranges over the usual reduction contexts of the
STLC.
Finally,  $\CONV$ is short for $\RED^{*}$.

%% file: logic.tex
\section{Logical Language}	\label{def:logical_language_syntax}
This section defines the syntax of the logic.  As is customary for
program logics, ours is an extension of first order logic with
equality (alongside axioms for arithmetic).  \EMPH{Expressions},
ranged over by $e$,$e'$,..., \EMPH{formulae}, ranged over by $A$, $B$,
$C$,~... and \EMPH{Logical Type Contexts} (LTCs), ranged over by
$\GAMMA$, $\GAMMA'$, $\GAMMA_i$, ... , are given by the grammar below.
(Extensions over \cite{HY04PPDP} highlighted.)
\begin{GRAMMAR}
		e
		&\quad::=\quad& x^{\alpha}
		\VERTICAL c  \VERTICAL \PAIR{e}{e} \VERTICAL \pi_i(e) 
		\\[0.5mm]
		\GREYBOX{\GAMMA}
		&::=&
		\GREYBOX{\emptyset} 
		\VERTICAL 
		\GREYBOX{\GAMMA \PLUSV x: \alpha}
		\VERTICAL
		\GREYBOX{\GAMMA \PLUSTC \TCV:\TC}
		\\[0.5mm]
		A
		&::=&
		e = e
		\VERTICAL
		\neg A
		\VERTICAL
		A \AND A
		\VERTICAL
		\EVALFORMULA{e}{e}{x^{\alpha}}{A} 
		\VERTICAL
		\GREYBOX{\FORALL{x^{\alpha}}{\GAMMA} A}
		\VERTICAL 
		\GREYBOX{\FAD{\TCV} A }
\end{GRAMMAR}
Expressions, $e$, are standard, where constants, $c$, range over
Booleans and $()$, but do \EMPH{not} include names or
$\GENSYM$ as constants.  Equality, negation and conjunction are
standard. Evaluation formulae $\ONEEVAL{e}{e'}{m}{A}$ internalise
triples \cite{HY04PPDP} and express that if the program denoted
by $e$ is executed with argument denoted by $e'$, then the result,
denoted by $m$, satisfies $A$. Since the \NUC\ has no recursion, all
applications terminate and we do not distinguish partial from
total correctness.
We write $\EVALFORMULASHORT{e_1}{e_2}{e_3}$ as shorthand for $\EVALFORMULA{e_1}{e_2}{m}{m=e_3}$.

Given variables represent values,
ensuring hidden names cannot be revealed in an unsafe manner requires the idea that a value is \EMPH{derived} from an LTC if a name free term uses the variables in the LTC to evaluate to said value.  
Specifically define a name as reachable from said LTC if it can be derived from it, and hidden otherwise.

Freshness is not an absolute notion. Instead, a name is fresh with
respect to something, in this case names generated in the past, and
future of the computation.  Formulae refer to names by variables, and
variables are tracked in the STC. Freshness is now defined in two
steps: (1) First we characterise freshness of a name w.r.t. the
current STC, 
meaning the name cannot be derived from the variables in the STC.
Then, (2) we define freshness w.r.t. all
future extension of the current STC, details in Sec.~\ref{model}.
The modal operator is used in \cite{YHB07:local:full} in order to express ``for all future extensions'', but
we found modalities inconvenient, since they don't allow us to name
extensions. 
We introduce a new quantifier $\FORALL{x^{\alpha}}{\GAMMA} A$ instead, 
where $\GAMMA$ ranges over LTCs from which $x$ can be derived. 
To make this precise, we need LTCs (explained next), a generalisation of STCs.

\PARAGRAPH{LTCs} Like STCs, LTCs
map variables to types, and are needed for typing expressions,
formulae and triples (introduced in Sec. \ref{triples}), LTCs
generalise STCs in two ways: they are \EMPH{ordered}, and they don't
just contain program variables, but also \EMPH{type context
  variables} (TCVs), ranged over by $\TCV$. TCVs are always mapped
to the new type $\TC$, short for \EMPH{type context}.  The ordering in
LTCs is essential because $\GAMMA \PLUSTC \TCV:\TC$ implies $\TCV$
represents an \EMPH{extension} of the LTC $\GAMMA$.

\PARAGRAPH{Restricted universal quantification}
The meaning of $\FORALL{x^{\alpha}}{\GAMMA} A$ is intuitively simple:
$A$ must be true for all $x$ that range only over values of type
$\alpha$, derived from $\GAMMA$ that do \EMPH{not} reveal hidden
names. For example if the model contained the name $r$ but only as
$\lambda y.y = r$, then $r$ was hidden and whatever $x$ in
$\FORALL{x^{\alpha}}{\GAMMA} A$ ranged over, it must not reveal $r$.
Formalising this requirement is subtle.

\PARAGRAPH{Quantification over LTCs}
Below we formalise the axiomatic semantics of $\GENSYM$ by saying that
the result of each call to this function is fresh w.r.t.~all future
extensions of the present state (with the present state
being included). The purpose of $\FAD{\TCV} A$ is to allow us to
do so: $\FAD{\TCV} A$ implies for
all future states derived from the current state (included), when the LTC for
that state is assigned to the TCV $\TCV$, then $A$ holds.

\PARAGRAPH{A convenient shorthand, the freshness predicate}
We express freshness of the name $x$ relative to the LTC $\GAMMA$ as $
\FORALL{z}{\GAMMA} x \neq z.  $ and, as this predicate is used
pervasively, abbreviate it to $\FRESH{x}{\GAMMA}$.  Intuitively,
$\FRESH{x}{\GAMMA}$, a variant of a similar predicate in
\cite{YHB07:local:full}, states that the name denoted by $x$ is not
derivable, directly or indirectly, from the LTC $\GAMMA$.

\PARAGRAPH{Typing of expressions, formulae and triples} We continue
with setting up definitions that allow us to type expressions,
formulae and triples.  The ordered union of $\GAMMA$ and $\GAMMA'$
with $\DOM{\GAMMA} \cap \DOM{\GAMMA'} = \emptyset$ is written $\GAMMA
\PLUSG \GAMMA'$, and should be understood as: every variable from
$\DOM{\GAMMA}$ comes before every variable from $\DOM{\GAMMA'}$.
Other abbreviations include $\EXISTS{x^{\alpha}}{\GAMMA} A \ \DEFEQ
\ \neg \FORALL{x^{\alpha}}{\GAMMA} \neg A$, and where $\alpha$ is
obvious $(\GAMMA \PLUSV y) \ \DEFEQ \ (\GAMMA \PLUSV y:\alpha)$
(respectively $(\GAMMA \PLUSTC \TCV) \ \DEFEQ \ (\GAMMA \PLUSTC
\TCV:\TC)$).  For simplicity, where not explicitly required,
$\GAMMA\PLUSTC \TCV$ is written $\TCV$.  Functions on LTCs are defined
as expected including mapping variables, $\GAMMA(x)$, and TCVs,
$\GAMMA(\TCV)$; obtaining the domain, $\DOM{\GAMMA}$; ordered removal
of a variable, $\GAMMA \REMOVEVARIABLE x$; ordered removal of all TCV,
$\GAMMA \REMOVETCVfrom$; and removal of TCV to produce a STC,
$\GAMMA\LTCtoSTC$.  We define free variables of LTC, $\FV{\GAMMA} \DEFEQ
\DOM{\GAMMA \LTCtoSTC} \DEFEQ \DOM{\GAMMA \REMOVETCVfrom}$, then free
variables of formulae defined as expected, with the addition of
$\FV{\FRESH{x}{\GAMMA}} \DEFEQ \FV{\GAMMA} \cup \{x\}$,
$\FV{\FORALL{x}{\GAMMA} A} \DEFEQ (\FV{A} \backslash \{x\}) \cup
\FV{\GAMMA}$, and $\FV{\FAD{\TCV} A}  \DEFEQ \FV{A}$.  Similarly
$\FTCV{\GAMMA}$ and $\FTCV{A}$ define all TCV occurring in $\GAMMA$
and unbound by $\FAD{\TCV}$ in $A$ respectively, calling $\GAMMA$
\EMPH{TCV-free} if $\FTCV{\GAMMA} \DEFEQ \emptyset$.  The typing judgement for 
LTCs, written $\TCTYPES{\GAMMA}{\GAMMA'}$, checks that $\GAMMA'$ is an
`ordered subset' of $\GAMMA$.  Type checks on expressions, formulae
and triples use LTC as the base, written
$\EXPRESSIONTYPES{\GAMMA}{e}{\alpha}$, $\FORMULATYPES{\GAMMA}{A}$ and
$\JUDGEMENTTYPES{\GAMMA}{\ASSERT{A}{M}{u}{B}}$ respectively.  Fig.~\ref{figure_typing_formulae} gives the rules defining the typing
judgements. From now on we adhere to the following convention:
\EMPH{All expressions, formulae and triples are typed}, and we will
mostly omit being explicit about typing.

\PARAGRAPH{Advanced substitutions}
Reasoning with quantifiers requires quantifier instantiation. This
is subtle with $\FAD{\TCV} A$, and we need to
define two substitutions, $A\LSUBSTLTC{e}{x}{\GAMMA}$ (substitutes
expressions for variables) and $A\LSUBSTLTC{\GAMMA_0}{\TCV}{\GAMMA}$
(LTCs substituted for TCVs).  First extend the definition
\EMPH{$e$ is free for $x^{\alpha}$ in $A$} in \cite{Mendelson}, to 
ensure if $e$ contains destructors i.e.~$\pi_i( \ )$ or
$\EQA{}{}$, then all free occurrences of $x$ in any LTC $\GAMMA_0$ in
$A$ must imply $\EXPRESSIONTYPES{\GAMMA_0}{e}{\alpha}$.  Below, we
assume the standard substitution $e\SUBST{e'}{x}$ of expressions for
variables in expressions, simple details omitted.

We define $A\LSUBSTLTC{e}{x}{\GAMMA}$, \EMPH{logical substitution of
$e$ for $x$ in A in the context of $\GAMMA$}, 
if $e$ is free for $x$ in $A$ and $e$ is typed by $\GAMMA$,
by the following clauses (simple cases omitted) and the auxiliary operation on LTCs below.
We often write $A\LSUBSTLTC{e}{x}{}$ for $A\LSUBSTLTC{e}{x}{\GAMMA}$.
\input{figure_typing_formulae}
\begin{equation*}\label{def:sec:logical_substitution}
\begin{array}{l}
	\begin{array}{rcll}                
    \\
    (\ONEEVAL{e_1}{e_2}{m}{A})\LSUBSTLTC{e}{x}{\GAMMA}
    & \ \DEFEQ \ &
    \ONEEVAL{e_1\SUBST{e}{x}}{e_2\SUBST{e}{x}}{m}{A\LSUBSTLTC{e}{x}{\GAMMA \PLUSV m}}
    &
    (x\neq m, \ m \notin \FV{\GAMMA})
    \\
    (\FRESH{y}{\GAMMA'})\LSUBSTLTC{e}{x}{\GAMMA}
    & \ \DEFEQ \ &
	\FRESH{y\SUBST{e}{x}}{(\GAMMA'\LSUBSTLTC{e}{x}{\GAMMA})}
    \\
    (\FORALL{m}{\GAMMA'} A)\LSUBSTLTC{e}{x}{\GAMMA} 
    & \ \DEFEQ \ &
    \FORALL{m}{\GAMMA'\LSUBSTLTC{e}{x}{\GAMMA}} (A\LSUBSTLTC{e}{x}{\GAMMA \PLUSV m}) 
    &
    (x \neq m, \ m \notin \FV{\GAMMA})
    \\
    (\FAD{\TCV} A)\LSUBSTLTC{e}{x}{\GAMMA}
	& \ \DEFEQ \ &
    \FAD{\TCV} (A\LSUBSTLTC{e}{x}{\GAMMA \PLUSTC \TCV})
	\end{array}
	\\
	\GAMMA'\LSUBSTLTC{e}{x}{\GAMMA}  \DEFEQ
	\begin{cases}
		\GAMMA_e' \text{ s.t. } 
		\DOM{\GAMMA_e'} = \FV{e} \cup \DOM{\GAMMA'\REMOVEVARIABLE x}, 
		\TCTYPES{\GAMMA}{\GAMMA_e'} 
		& x \in \DOM{\GAMMA'}
		\\
		\GAMMA'  & x \notin \DOM{\GAMMA'}
	\end{cases}
\end{array}
\end{equation*}

\EMPH{Type context substitution}
$A\LSUBSTLTC{\GAMMA_0}{\TCV}{\GAMMA}$ instantiates $\TCV$ with
$\GAMMA_0$ in $A$, similar to classical substitution.
We often write $\LSUBST{\GAMMA_0}{\TCV}$ for $\LSUBSTLTC{\GAMMA_0}{\TCV}{\GAMMA}$ as $\GAMMA$ is used for ordering and is obvious. 
As above, the omitted cases are straightforward and the auxiliary operation on LTCs is included.
\begin{equation*}\label{def:LSUBST_LTC}
\begin{array}{l}
  \begin{array}{rclr}
			(\FRESH{x}{\GAMMA'})\LSUBSTLTC{\GAMMA_0}{\TCV}{\GAMMA} 
			& \ \DEFEQ \ &
			\FRESH{x}{(\GAMMA'\LSUBSTLTC{\GAMMA_0}{\TCV}{\GAMMA})}
			\\
			(\ONEEVAL{e_1}{e_2}{m}{A})\LSUBSTLTC{\GAMMA_0}{\TCV}{\GAMMA} 
			& \ \DEFEQ \ &
			\ONEEVAL{e_1}{e_2}{m}{A\LSUBSTLTC{\GAMMA_0}{\TCV}{\GAMMA \PLUSV m}}
			&
			(m \notin \DOM{\GAMMA_0})
			\\
			(\FORALL{x}{\GAMMA'} A)\LSUBSTLTC{\GAMMA_0}{\TCV}{\GAMMA} 
			& \ \DEFEQ \ &
			\FORALL{x}{\GAMMA'\LSUBSTLTC{\GAMMA_0}{\TCV}{\GAMMA}} (A\LSUBSTLTC{\GAMMA_0}{\TCV}{\GAMMA \PLUSV x})
			&
			(x \notin \DOM{\GAMMA_0})
			\\
			(\FAD{\TCV'} A)\LSUBSTLTC{\GAMMA_0}{\TCV}{\GAMMA}
            & \ \DEFEQ \ &
			\begin{cases}
				(\FAD{\TCV'} A\LSUBSTLTC{\GAMMA_0}{\TCV}{\GAMMA \PLUSTC \TCV'})	
				& 
				\TCV \neq \TCV' 
				\\
				\FAD{\TCV} A &\text{otherwise}
			\end{cases}
			& 
			\begin{array}{r}
				(\TCV' \notin \DOM{\GAMMA_0})
				\\ \
			\end{array}
  \end{array}
\\
\GAMMA'\LSUBSTLTC{\GAMMA_0}{\TCV}{\GAMMA}  
\ \DEFEQ \
\begin{cases}
	\GAMMA_1 \text{ s.t. } 
	\DOM{\GAMMA_1} = \DOM{\GAMMA_0, \GAMMA'}, 
	\TCTYPES{\GAMMA}{\GAMMA_1} 
	& \TCV \in \DOM{\GAMMA'}
	\\
	\GAMMA'  & \TCV \notin \DOM{\GAMMA'}
\end{cases}
\end{array}
\end{equation*}

%% file: figure_typing_formulae.tex
\begin{FIGURE}
  \begin{RULES}
	\ONEPREMISERULE
	{
	b \in \{\TRUE, \FALSE\}
	}
	{
	\EXPRESSIONTYPES{\GAMMA}{b}{\BOOL}
	}
	\quad 
	\ZEROPREMISERULE
	{
	\EXPRESSIONTYPES{\GAMMA}{()}{\UNIT}
	}
	\quad 
	\ONEPREMISERULE
	{
	\GAMMA(x) = \alpha
	}
	{
	\EXPRESSIONTYPES{\GAMMA}{x}{\alpha}
	}
	\quad
	\TWOPREMISERULE
	{
	\EXPRESSIONTYPES{\GAMMA}{e}{\alpha}
	}
	{
	\EXPRESSIONTYPES{\GAMMA}{e'}{\beta}
	}
	{
	\EXPRESSIONTYPES{\GAMMA}{\PAIR{e}{e'}}{\alpha \times \beta}
	}
	\quad
	\ONEPREMISERULE
	{
	\EXPRESSIONTYPES{\GAMMA}{e}{\alpha_1 \times \alpha_2}
	}
	{
	\EXPRESSIONTYPES{\GAMMA}{\pi_i(e)}{\alpha_i}
	}
	\\\\
	\ZEROPREMISERULE
	{
		\TCTYPES{\GAMMA}{\emptyset}
	}
	\quad
	\ONEPREMISERULE
	{
		\TCTYPES{\GAMMA}{\GAMMA_0}
	}
	{
		\TCTYPES{\GAMMA \PLUSV x:\alpha}{\GAMMA_0 \PLUSV x:\alpha}
	}
	\quad
	\ONEPREMISERULE
	{
		\TCTYPES{\GAMMA}{\GAMMA_0}
	}
	{
		\TCTYPES{\GAMMA \PLUSTC \TCV}{\GAMMA_0 \PLUSTC \TCV}
	}
	\quad
	\ONEPREMISERULE
	{
		\TCTYPES{\GAMMA}{\GAMMA_0}
	}
	{
		\TCTYPES{\GAMMA \PLUSG \GAMMA'}{\GAMMA_0}
	}
	\\\\
	\TWOPREMISERULE
	{
		\EXPRESSIONTYPES{\GAMMA}{e_1}{\alpha}
	}
	{
		\EXPRESSIONTYPES{\GAMMA}{e_2}{\alpha}
	}
	{
	\FORMULATYPES{\GAMMA}{e_1 = e_2}
	}
	\quad
	\TWOPREMISERULE
	{
	\FORMULATYPES{\GAMMA}{A_1}
	}
	{
	\FORMULATYPES{\GAMMA}{A_2}
	}
	{
	\FORMULATYPES{\GAMMA}{A_1 \PAND A_2}
	}
	\quad
	\ONEPREMISERULE
	{
	\FORMULATYPES{\GAMMA}{A}
	}
	{
	\FORMULATYPES{\GAMMA}{\neg A}
	}
	\quad
	\ONEPREMISERULE
	{
	\FORMULATYPES{\GAMMA\PLUSTC \TCV:\TC }{A}
	}
	{
	\FORMULATYPES{\GAMMA}{\FAD{\TCV} A}
	}
	\\\\
	\THREEPREMISERULE
	{
	\EXPRESSIONTYPES{\GAMMA}{e}{\alpha \FS \beta}
	}
	{
	\EXPRESSIONTYPES{\GAMMA}{e'}{\alpha}
	}
	{
	\FORMULATYPES{\GAMMA \PLUSV x : \beta}{A}
	}
	{
	\FORMULATYPES{\GAMMA}{\ONEEVAL{e}{e'}{x^{\beta}}{A}}
	}
	\quad
	\TWOPREMISERULE
	{
	\EXPRESSIONTYPES{\GAMMA}{x}{\NAME}
	}
	{
	\TCTYPES{\GAMMA}{\GAMMA'}
	}        
	{
	\FORMULATYPES{\GAMMA}{\FRESH{x}{\GAMMA'}}
	}
	\\\\
	\TWOPREMISERULE
	{
	\TCTYPES{\GAMMA}{\GAMMA'}
	}
	{
	\FORMULATYPES{\GAMMA \PLUSV x : \alpha}{A}
	}
	{
	\FORMULATYPES{\GAMMA}{\FORALL{x^{\alpha}}{\GAMMA'} A}
	}
	\quad
	\THREEPREMISERULE
	{
		\FORMULATYPES{\GAMMA}{A}
	}
	{
		\TYPES{\GAMMA\LTCtoSTC}{M}{\alpha}
	}
	{
		\FORMULATYPES{\GAMMA \PLUSV m :\alpha}{B}
	}
	{
		\JUDGEMENTTYPES{\GAMMA}{\ASSERT{A}{M}{m}{B}}
	}
  \end{RULES}
  \vspace{-4mm}
  \caption{Typing rules for LTCs, expressions, formulae and triples
    (see Sec. \ref{triples}).  Simple cases omitted. $M$ in the last
    rule is compile-time syntax.}\label{figure_typing_formulae}
\label{def:typing_triples}
\end{FIGURE}

%% file: model.tex
\section{Model}\label{model}

	We define a \EMPH{model} $\MMM$ as a finite (possibly empty) map
	from variables and TCV to closed values and TCV-free LTCs respectively.

	\[\MMM \ ::= \ \emptyset \VERTICAL \MMM \cdot x:V \VERTICAL \MMM \cdot \TCV:\GAMMA'\]

	Standard actions on models $\MMM$ are defined as expected and include:
	variable mappings to values, $\MMM(x)$, or TCV mapping to LTC, $\MMM(\TCV)$;
	removal of variable $x$ as $\MMM \REMOVEVARIABLE x$ (with $(\MMM \cdot \TCV:\GAMMA_1)\REMOVEVARIABLE x = (\MMM \REMOVEVARIABLE x) \cdot \TCV: (\GAMMA_1 \REMOVEVARIABLE x)$);
	removal of TCV $\TCV$ as $\MMM\REMOVEVARIABLE \TCV$;
	removal of all TCVs as $\MMM \REMOVETCVfrom$;
	and defining all names in $\MMM$ as $\AN{\MMM}$ noting that $\AN{\GAMMA}= \emptyset$.

	A model  $\MMM$ is typed by a LTC $\GAMMA$ written
	$\MMM^{\GAMMA}$, if $\MODELTYPES{\GAMMA}{\MMM}$ as defined below,
	were $\GAMMA_d = \GAMMA_d \REMOVETCVfrom$ formalises that
	$\GAMMA_d$ is TCV-free.
\[
		\ZEROPREMISERULE
		{
		\MODELTYPES{\emptyset}{\emptyset}
		}
		\qquad
		\TWOPREMISERULE
		{
			\MODELTYPES{\GAMMA}{\MMM}
		}{
			\TYPES{\emptyset}{V}{\alpha}
		}{
			\MODELTYPES{\GAMMA \PLUSV x:\alpha}{\MMM \cdot x:V}
		}
		\qquad
		\THREEPREMISERULE
		{
			\MODELTYPES{\GAMMA}{\MMM}
		}{
			\TCTYPES{\GAMMA}{\GAMMA_d}
		}{
			\GAMMA_d = \GAMMA_d \REMOVETCVfrom
		}{
			\MODELTYPES{\GAMMA \PLUSTC \TCV}{\MMM \cdot \TCV:\GAMMA_d}
		}
\]

	The \EMPH{closure} of a term $M$ by a model $\MMM$, written $M \MMM $ is defined as standard with the additions, $\GENSYM\MMM \DEFEQ \GENSYM$ and $r\MMM \DEFEQ r$.
	Noting that $M \MMM \REMOVETCVfrom = M \MMM =  M \MMM\cdot \TCV:\GAMMA'$ holds for all $\TCV$ and $\GAMMA'$ as $\TYPES{\GAMMA\LTCtoSTC}{M}{\alpha}$.

	The \EMPH{interpretation of expression} $e$ in a model
        $\MMM^{\GAMMA}$, written $\SEM{e}{\MMM}$, is standard,
        e.g.~$\SEM{c}{\MMM} \DEFEQ c$, ~$\SEM{x}{\MMM} \DEFEQ \MMM(x)$, $\SEM{\PAIR{e}{e'}}{\MMM} \DEFEQ
        \PAIR{\SEM{e}{\MMM}}{\SEM{e'}{\MMM}}$, etc.

	The \EMPH{interpretation of LTCs} $\GAMMA_0$ in a model $\MMM^{\GAMMA}$, written
	$\SEM{\GAMMA_0}{\MMM}$, outputs a STC. It is assumed $\TCTYPES{\GAMMA}{}$ the LTC in the following definition:
	\[
		\SEM{\emptyset}{\MMM} \DEFEQ \emptyset
		\qquad
		\SEM{\GAMMA_0 \PLUSV x:\alpha}{\MMM} \DEFEQ  \SEM{\GAMMA_0}{\MMM}, x:\alpha
		\qquad
		\SEM{\GAMMA_0 \PLUSTC \TCV:\TC}{\MMM} \DEFEQ  \SEM{\GAMMA_0}{\MMM} \cup \SEM{\MMM(\TCV)}{\MMM}
	\]

Write $\LTCDERIVEDVALUE{M}{\GAMMA}{\MMM}{V}$ as the \EMPH{derivation of a value}
$V$ from term $M$ which  is typed by the LTC $\GAMMA$ 
and closed and evaluated in a model $\MMM$.
This ensures names are derived from actual reachable values in $\MMM$ as if they
were programs closed by the model, hence not revealing hidden names from $\MMM$.
$\LTCDERIVEDVALUE{M}{\GAMMA}{\MMM}{V}$
holds exactly when:
\begin{itemize}
	\item	$	\AN{M}=\emptyset$
	
	\item	$	\TYPES{\SEM{\GAMMA}{\MMM}}{M}{\alpha} $
	
	\item	$	(\AN{\MMM}, M\MMM) \CONV (\AN{\MMM} \cup G', \ V)$
\end{itemize}

Model extensions aim to capture the fact that models represent real states of execution, by stating a model is only constructed by evaluating terms derivable from the model.

  A model $\MMM'$ is a \EMPH{single step model extension} to another
  model $\MMM^{\GAMMA}$, written $\MMM \EXTSINGLE \MMM'$, if the single new
  value in $\MMM'$ is derived from $\MMM$ or the mapped LTC is $\GAMMA$ with TCVs removed. Formally
  $
  \MMM^{\GAMMA} \EXTSINGLE \MMM'
  $
  holds if either of the following hold:
  \begin{itemize}
  \item There is $M_y^{\alpha}$ such that
    $\LTCDERIVEDVALUE{M_y}{\GAMMA}{\MMM}{V_y}$ and  $\MMM'^{\GAMMA \PLUSV y:\alpha} =  \MMM \cdot y:V_y$.

  \item $\MMM'^{\GAMMA \PLUSTC \TCV} = \MMM
    \cdot \TCV: \GAMMA\REMOVETCVfrom$ for some $\TCV$.

  \end{itemize}
  We write $\EXTSTAR$ for the transitive, reflexive closure of
  $\EXTSINGLE$. If $\xi \EXTSTAR \xi'$ we say $\xi'$ is an
  \EMPH{extension} of $\xi$ and $\xi$ is a \EMPH{contraction} of
  $\xi'$.

  A model $\MMM$ is \EMPH{constructed by $\GAMMA$}, written
  $\CONSTRUCT{\GAMMA}{\MMM}$, if any TCV represents a model extension.
  Formally we define $\CONSTRUCT{\GAMMA}{\MMM}$ by the following rules:
  \begin{RULES}
    \ZEROPREMISERULE{
      \CONSTRUCT{\emptyset}{\emptyset}
    }
    \quad
    \TWOPREMISERULE{
      \CONSTRUCT{\GAMMA}{\MMM}
    }{
      \text{exists}\ M^{\alpha}. \LTCDERIVEDVALUE{M}{\GAMMA}{\MMM}{V}
    }{
      \CONSTRUCT{\GAMMA \PLUSV x:\alpha}{\MMM \cdot x:V}
      }
    \quad
    \THREEPREMISERULE{
      \CONSTRUCT{\GAMMA}{\MMM_0}
    }{
      \MMM_0 \EXTSTAR \MMM^{\GAMMA_2}
    }{
      \GAMMA_1 = \GAMMA_2 \REMOVETCVfrom
    }{
      \CONSTRUCT{\GAMMA \PLUSTC \TCV}{\MMM \cdot \TCV:\GAMMA_1}
    }
  \end{RULES}
  A model $\MMM^{\GAMMA}$ is  \EMPH{well constructed}
  if there exists an LTC, $\GAMMA'$, such that
  $\CONSTRUCT{\GAMMA'}{\MMM}$, noting that
  $\TCTYPES{\GAMMA}{\GAMMA'}$.

Model extensions and well constructed models represent models
derivable by \NUC\ programs, ensuring names cannot be revealed by later
programs.  Consider the basic model: $y:\lambda
a. \IFTHENELSE{a=r_1}{r_2}{r_3}$, \ if $r_1$ could be added to the
model, this clearly reveals access to $r_2$ otherwise $r_2$ is hidden.
Hence the assumption that all models are well constructed from here onwards.

\EMPH{Contextual equivalence} of two terms requires them to be contextually indistinguishable in all variable-closing single holed contexts of Boolean type in any valid configuration, as is standard \cite{BentonN:mechbisftnc, stark:namhof}.
When $M_1$ and $M_2$ are closed terms of type $\alpha$ and $\AN{M_1}\cup \AN{M_2} \subseteq G$, we write 
$M_1 \CONGCONTEXT{\alpha}{G} M_2$  to be equivalent to 
$G, \emptyset \vdash M_1 \equiv M_2 :\alpha$ from \cite{BentonN:mechbisftnc}.

\subsection{Semantics}
The \EMPH{satisfaction relation} for formula $A$ in a well constructed model $\MMM^{\GAMMA}$, written $\MMM \models A$, assumes $\FORMULATYPES{\GAMMA}{A}$, and is defined as follows:
	\begin{itemize}
		\item $\MMM \models e = e'$ \ if \ $\SEM{e}{\MMM}  \CONGCONTEXT{\alpha}{\AN{\MMM}} \SEM{e'}{\MMM}$.

		\item $\MMM \models \neg A$ \ if \ $\MMM \not \models A$.

		\item $\MMM \models A \PAND B$ \ if \ $\MMM \models A$ and $\MMM \models B$.

		\item $\MMM \models \ONEEVAL{e}{e'}{m}{A}$
		\ if \
		$\LTCDERIVEDVALUE{\SEM{e}{\MMM}\SEM{e'}{\MMM}}{\emptyset}{\MMM}{V}$ and  $\MMM \cdot m:V \models A  $
		\item $\MMM \models \FORALL{ x^{\alpha}}{\GAMMA'} A$ \
		if for all  $M. \
		\LTCDERIVEDVALUE{M}{\GAMMA'}{\MMM}{V}$
		implies $\MMM \cdot x:V \models A
		$

		\item $\MMM \models \FAD{\TCV} A$ if forall
                  $\MMM'^{\GAMMA'}.  \MMM \EXTSTAR \MMM'$ implies
                  $\MMM' \cdot \TCV : (\GAMMA' \REMOVETCVfrom) \models
                  A$

 		\item $\MMM \models \FRESH{x}{\GAMMA_0} $ \ if \
 		there is no  $M_x$ such that $\LTCDERIVEDVALUE{M_x}{\GAMMA_0}{\MMM}{\SEM{x}{\MMM}}$
                  
	\end{itemize}

%% file: formulae_properties.tex
In first-order logic, if a formula is satisfied by a model, then it is
also satisfied by extensions of that model, and vice-versa (as long as
all free variables of the formula remain in the model). This can no
longer be taken for granted in our logic. Consider the formula $
\FAD{\TCV} \EXISTS{z}{\TCV}(\FRESH{z}{\GAMMA} \PAND
\neg\FRESH{z}{\TCV}) $ Validity of this formula depends on how many
names exist in the ambient model: it may become invalid under
contracting the model.  Fortunately, such formulae are rarely needed
when reasoning about programs.  In order to simplify our
soundness proofs we will therefore restrict some of our axioms and rules to formulae that are stable
under model extension and contractions.  Sometimes we need a weaker
property, where formulae preserve their validity when a variable is
removed from a model. Both concepts are defined semantically next.

We define formula $A$ as \EMPH{model extensions independent}, short $\EXTINDEP$,
if for all $\GAMMA, \MMM^{\GAMMA}, \MMM'$ such
that $\FORMULATYPES{\GAMMA}{A}$ and $ \MMM \EXTSTAR \MMM'$ we
have: $ \MMM \models A \ \text{iff}\ \MMM' \models A $.

We define formula $A$ as \EMPH{thin} w.r.t.~$x$, written $A
\THINWRT{x^{\alpha}}$, if for all $\GAMMA$ such that
$\FORMULATYPES{\GAMMA\REMOVEVARIABLE x}{A}$ and $x^{\alpha} \in \DOM{\GAMMA}$ we
have for all well constructed models $\MMM^{\GAMMA}$ and $\MMM \REMOVEVARIABLE x$ that: 
$ \MMM \models A
\ \text{implies}\ \MMM \REMOVEVARIABLE x \models A.  $

%% file: axioms.tex
\section{Axioms}\label{axioms}
\label{sec:axioms}

Axioms and axiom schemas are similar in intention to those of the logic for the STLC, but expressed within the constraints of our logic.
Axiom schemas are indexed by the LTC that types them and the explicit types where noted.
We introduce the interesting axioms (schemas) and those used in Sec. \ref{reasoning}.

Equality axioms are standard where $(eq1)$ allows for substitution. 
Most axioms for universal quantification over LTCs $(u1)$-$(u5)$ are inspired by those of first order logic. 
The exceptions are $(u2)$ which allows for the reduction of LTCs and $(u5)$ which holds only on $\NAME$-free types.
Axioms for existential quantification over LTCs $(ex1)$-$(ex3)$ are new aside from $(ex1)$ which is the dual of $(u1)$.
Axiom $(ex2)$ introduces existential quantification from evaluation formulae that produce a fixed result.
Reducing $\GAMMA$ in $\EXISTS{x}{\GAMMA} A$ is possible via $(ex3)$ for a specific structure. 
We use base types $\TYBASE \ ::= \ \UNIT \VERTICAL \BOOL \VERTICAL \TYBASE \times \TYBASE$ as core lambda calculus types excluding functions.
Freshness axioms $(f1)$-$(f2)$ show instances LTCs can be extended, whereas $(f3)$-$(f4)$ reduce the LTC.
Axiom $(f1)$ holds due to $f$ being derived from $\GAMMA \PLUSV x$, and the rest are trivial.

\begin{RULES}
	\begin{array}{l}
		\begin{array}{lrclr}
			(eq1)
			&
			\GAMMA \Vdash 
			A(x) \PAND x= e 
			& \PIFF &
			A(x)\LSUBSTLTC{e}{x}{\GAMMA}
			\\
			(u1) \quad
			&
			\GAMMA \Vdash 
			\FORALL{x^{\alpha}}{\GAMMA_0} A & \quad \PIMPLIES \quad & A \LSUBSTLTC{e}{x}{\GAMMA}
			&
			\TYPES{\GAMMA_0}{e}{\alpha} 
			\\
			(u2)
			&
			\FORALL{x}{\GAMMA_0 \PLUSG \GAMMA_1}A & \PIMPLIES & (\FORALL{x}{\GAMMA_0}A ) \PAND (\FORALL{x}{\GAMMA_1} A )
			\\
			(u3)
			&	
			A^{-x} & \PIFF & \FORALL{x}{\GAMMA_0} A^{-x}
			&
			A-\EXTINDEP 
			\\
			(u4)
			&
			\FORALL{x}{\GAMMA_0}(A \PAND B) & \PIFF & (\FORALL{x}{\GAMMA_0}A ) \PAND (\FORALL{x}{\GAMMA_0} B )
			\\
			(u5)
			&
			\FORALL{x^{\alpha}}{\GAMMA_0} A & \PIFF & \FORALL{x^{\alpha}}{\emptyset}  A
			&
			\hspace{-1cm}
			\text{$\alpha$ is $\NAME$-free}
		\end{array}
		\\\\
		\begin{array}{lrclr}
			(ex1)
			&
			\GAMMA \Vdash 
			A\LSUBSTLTC{e}{x}{\GAMMA} & \quad \PIMPLIES \quad  & \EXISTS{x'}{\GAMMA_0} A 
			&
			\hspace{-0.3cm}
			\TCTYPES{\GAMMA}{\GAMMA_0} \text{ and } \TYPES{\GAMMA_0}{e}{\alpha}
			\\
			(ex2)
			&
			\GAMMA \PLUSV x \PLUSTC \GAMMA_0 \Vdash 
			\ONEEVAL{a}{b}{c}{c=x}
			& \PIMPLIES &
			\EXISTS{x'}{\GAMMA_0} x=x'
			&
			\{ a,b \} \subseteq \DOM{\GAMMA_0}
			\\
			\multicolumn{5}{l}{
				(ex3) \quad
				\GAMMA \PLUSV x \Vdash 
				\FORALL{y}{\emptyset}\EXISTS{z^{\NAME}}{\GAMMA_0 \PLUSV y} x=z
				\quad \PIMPLIES \quad
				\EXISTS{z}{\GAMMA_0} x=z
			}
		\end{array}
		\\\\
		\begin{array}{lrclr}
			(f1) 
			\quad
			&
			\GAMMA\PLUSV  x \PLUSV f :\alpha \FS \TYBASE \Vdash  \FRESH{x}{\GAMMA} & \quad \PIMPLIES \quad & \FRESH{x}{\GAMMA \PLUSV f: \alpha \FS \TYBASE}
			\\
			(f2) 
			&
			\GAMMA\Vdash  
			\FRESH{x}{\GAMMA_0} \PAND \FORALL{y^{\alpha}}{\GAMMA_0}A 
			& \PIFF &
			\FORALL{y^{\alpha}}{\GAMMA_0}(\FRESH{x}{(\GAMMA_0 \PLUSV y)} \PAND A) 
			\hspace{-1cm}
			\\
			(f3)
			& 
			\FRESH{x}{\GAMMA_0} 
			& \PIMPLIES &
			x \neq e
			&
			\TYPES{\GAMMA_0}{e}{\NAME} 
			\\
			(f4) 
			&
			\FRESH{x}{(\GAMMA_0 \PLUSG \GAMMA_1)} 
			& \PIMPLIES &
			\FRESH{x}{\GAMMA_0} \PAND \FRESH{x}{ \GAMMA_1}
		\end{array}
	\end{array}
\end{RULES}

Axioms for quantification over LTCs are also similar to those for the classical universal quantifier except (utc2) which extends the restricted quantifier to any future LTC which can only mean adding fresh names.
\begin{RULES}
	\begin{array}{lr c lr}
	(utc1) 
	\quad
	&
	\GAMMA \Vdash  
	\FAD{\TCV} A & \quad \PIMPLIES \quad &  A \LSUBSTLTC{\GAMMA}{\TCV}{\GAMMA}
	\\
	(utc2) 
	& 
	\GAMMA \Vdash 
	\FORALL{x^{\NAME}}{\GAMMA} A^{-\TCV}
	& \PIFF &  
	\FAD{\TCV} \FORALL{x^{\NAME}}{\GAMMA \PLUSTC \TCV}  A
	&
	\quad
	A-\EXTINDEP
	\\
	(utc3) 
	&
	A^{-\TCV} 
	& \PIFF &  
	\FAD{\TCV} A^{-\TCV} 
	&
	A-\EXTINDEP
	\\
	(utc4) 
	&
	\FAD{\TCV} (A \PAND B) 
	&\PIFF &
	(\FAD{\TCV} A) \PAND (\FAD{\TCV} B)
	\end{array}
\end{RULES}

Axioms for the evaluation formulae are similar to those of \cite{BergerM:prologfhgrtmp}. The interaction between evaluation formulae and the new constructors are shown. 
All STLC values are included in the variables of $\NAME$-free type, and 
if we let $\mathsf{Ext}(e_2, e_2)$ stand for $\FORALL{x}{\emptyset} \ONEEVAL{e_1}{x}{m_1}{\ONEEVAL{e_2}{x}{m_2}{m_1 = m_2}}$
then (ext) maintains extensionality in this logic for the STLC terms.
Typing restrictions require $m\notin \FV{A}$ in $(e1)$ and $\FV{e_1,e_2,m} \cap \DOM{\GAMMA \PLUSV x} = \emptyset$ in $(e2)$.

\begin{RULES}
	\begin{array}{lr c lr}
		(e1)
		& 
		\ONEEVAL{e_1}{e_2}{m}{A \PAND B} 
		& \quad \PIFF \quad & 
		(A \PAND \ONEEVAL{e_1}{e_2}{m}{B})
		& 
		A-\EXTINDEP
		\\
		(e2)
		&
		\ONEEVAL{e_1}{e_2}{m}{\FORALL{x}{\GAMMA}A} 
		& \PIFF & 
		\FORALL{x}{\GAMMA} \ONEEVAL{e_1}{e_2}{m}{A})
		\\
		(e3) 
		&
		\ONEEVAL{e_1}{e_2}{m^{\TYBASE}}{\FAD{\TCV} A} 
		& \PIFF &
		\FAD{\TCV} \ONEEVAL{e_1}{e_2}{m^{\TYBASE}}{A} 
		&
		A-\EXTINDEP
		\\
		(ext)
		&
		\mathsf{Ext}(e_1, e_2)
		& \PIFF & 
		e_1 =^{\alpha_1 \FS \alpha_2} e_2
		&
		\hspace{-1cm}
		\alpha_1 \FS \alpha_2 \text{ is $\NAME$-free}
	
	\end{array}
\end{RULES}

%% file: rules.tex
\section{Rules}\label{rules}

\label{triples}
Our logic uses standard \EMPH{triples} $ \ASSERT{A}{M}{m}{B} $ where
in this logic, the program $M$ is restricted to compile-time syntax.
Triples are typed by the rule in Fig.~\ref{figure_typing_formulae}.
Semantics of triples is standard: if the \EMPH{pre-condition} $A$
holds and the value derived from $M$ is assigned to the \EMPH{anchor}
$m$ then the \EMPH{post-condition} $B$ holds. In detail: let
$\MMM^{\GAMMA}$ be a model.
\[
\MMM^{\GAMMA} \models \ASSERT{A}{M}{m}{B}
\qquad
\text{iff}
\qquad
\MMM \models A 
\ \text{implies}\ 
(\LTCDERIVEDVALUE{M}{\GAMMA}{\MMM}{V}
\ \text{and}\ \MMM \cdot m:V \models B )
\]
The triple is \EMPH{valid},  written $\models \ASSERT{A}{M}{m}{B}$,
if for all $\GAMMA$ and $\MMM_0^{\GAMMA_0}$ we have
\[
\JUDGEMENTTYPES{\GAMMA}{\ASSERT{A}{M}{m}{B}}
\ \text{and}
\ \CONSTRUCT{\GAMMA}{\MMM_0}\ \text{together
	imply:}\ \MMM_0 \models \ASSERT{A}{M}{m}{B}
\]

From here on we will assume all models are well constructed, noting
that the construction of models is the essence of $\FAD{\TCV}$ as it
allows for all possible future names generated. Variables occurring in
$\DOM{\MMM_0}-\DOM{\GAMMA}$ may never occur directly in the triple,
but their mapped values will have an effect.

The rules of inference can be found in Fig.~\ref{figure_rules_language}  and Fig.~\ref{figure_rules_structural}.  
We write $\vdash \ASSERT{A}{M}{m}{B}$ to indicate that $\ASSERT{A}{M}{m}{B}$ can be derived from these rules.
Our rules are similar to those of \cite{HY04PPDP} for vanilla
$\lambda$-calculi, but suitably adapted to the effectful nature of the
\NUC.  All rules are typed. The typing of rules follows the
corresponding typing of the programs occurring in the triples, but
with additions to account for auxiliary variables.  We have two
substantially new rules: \RULENAME{[Gensym]} and
\RULENAME{[Let]}. The former lets us reason about fresh name
creation by $\GENSYM$, the latter about the
$\LET{x}{M}{N}$. Operationally, $\LET{x}{M}{N}$ is often just an
abbreviation for $(\lambda x.N)M$, but we have been unable to derive
\RULENAME{[Let]} using the remaining rules and axioms.
Any syntactic proof of \RULENAME{[Let]} requires \RULENAME{[Lam]} and \RULENAME{[App]}, 
which requires the postcondition: $C \THINWRT{p}$ for $p$ the anchor of the \RULENAME{[Lam]} rule.
We have not been able to prove this thinness for all models of the relevant type.

In comparison with \cite{HY04PPDP, YHB07:local:full}, 
the primary difference with our rules is our substitution. 
Our changes to substitution only affects \RULENAME{[Eq]} and \RULENAME{[Proj($i$)]} 
which are reduced in strength by the new definition of substitution 
as more constraints are placed on the formulae to ensure correct substitution occurs. 
All other rules remain equally strong.
The other difference from \cite{HY04PPDP} is the need for thinness to replace the standard `free from', which is discussed above.
Removal of variables via thinness is required in the proof of soundness, 
for example \RULENAME{[App]}, which produces $u$ from $m$ and $n$, 
hence $\EXTINDEP$ \ is insufficient given the order of 
$\GAMMA \PLUSV m \PLUSV n \PLUSV u \Vdash C$, i.e.~$u$ introduced after $m$ and $n$.
We explain the novelty in the rules in more detail below.

\input{figure_rules_language}

In \RULENAME{[Gensym]}, $\ONEEVAL{u}{()}{m}{\FRESH{m}{\TCV}}$ 
indicates the name produced by $u()$ and stored at $m$ is not derivable from the  LTC $\TCV$.
If there were no quantification over LTCs prior to the evaluation we could only say $m$ is fresh from the current typing context,
however we want to say that even if there is a future typing context with new names 
and we evaluate $u()$, this will still produce a fresh name.
Hence we introduce the $\FAD{\TCV}$ 
to quantify over all future LTCs (and hence all future names).
Elsewhere in reasoning it is key that the post-condition of \RULENAME{[Gensym]} is \EXTINDEP \ and hence holds in all 
extending and contracting models (assuming the anchor for $\GENSYM$ is present),
reinforcing the re-applicability of $\GENSYM$ in any context.

Rules for $\lambda$-abstraction in previous logics for
lambda-calculi \cite{HY04PPDP, YHB07:local:full} universally quantify over all possible
arguments.  Our corresponding \RULENAME{[Lam]} rule refines this and
quantifies over current or future values that do not contain
hidden names.
Comparing the two LTCs typing the assumption and conclusion implies 
$\GAMMA \PLUSTC \TCV \PLUSV x$ extends $\GAMMA$ to any possible extension assigned to $\TCV$,
and extends to $x$ a value derived from $\GAMMA \PLUSTC \TCV$.
Hence the typing implies precisely what is conveyed in the post-condition of the conclusion: `$\FAD{\TCV}\FORALL{x}{\TCV}$'.
Constraints on $\TCV$ and $x$ are introduced by $B$,
and $A-\EXTINDEP$ implies $A$ still holds in all extensions of $\GAMMA$ including $\GAMMA \PLUSTC \TCV \PLUSV x$.
The rest is trivial when we consider $((\lambda x. M)x)\MMM \CONGCONTEXT{\alpha}{\AN{\MMM}} M\MMM$.

The STLC's \RULENAME{[Let]} rule introduces $x$ in the post-condition by means of an `$\exists x. C$'.
This fails here as $x$ may be unreachable, hence not derivable from any extending or contracting LTC.
The requirement that $C\THINWRT{x}$ ensures $x$ is not critical to $C$ so can either be derived from the current LTC or is hidden.
Thinness ensures no reference to the variable $m$ is somehow hidden under quantification over LTCs.
	
The \RULENAME{[Invar]} rule is standard with the constraint that $C-\EXTINDEP$ \ to ensure $C$ holds in the extension where $m$ has been assigned.
The \RULENAME{[LetFresh]} rule is commonly used and hence included for convenience, but it is entirely derivable from the other rules.

\input{figure_rules_structural}

\begin{theorem}
All axioms and rules are sound.
\end{theorem}

\begin{theorem}
  The logic for the \NUC\ is a conservative extension of the logic
  \cite{HY04PPDP} for the STLC.
\end{theorem}
All proofs can be found in the first author's forthcoming dissertation \cite{EliottHP:Thesis_draft}.

%% file: figure_rules_language.tex
\begin{FIGURE}
  \begin{RULES}
	\ZEROPREMISERULENAMEDRIGHT
        {
		\ASSERT{A\LSUBST{x}{m}}{x}{m}{A}
        }{[Var]}
	\quad
    \ZEROPREMISERULENAMEDRIGHT
    {
    	\ASSERT{\TRUTH}{\GENSYM}{u}{\FAD{\TCV} \ONEEVAL{u}{()}{m}{\FRESH{m}{\TCV}}}
    }
    {[Gensym]}         
	\\\\
	\ZEROPREMISERULENAMEDRIGHT
        {
		\ASSERT{A\LSUBST{\LOGIC{c}}{m}}{\PROGRAM{c}}{m}{A}
        }{[Const]}
    \quad
     \TWOPREMISERULENAMEDRIGHT
     {
        	\ASSERT{A}{M}{m}{B}
     }
     {
        	\ASSERT{B}{N}{n}{C
        		\LSUBST{\EQA{m}{n}}{u}}
     }
     {
        	\ASSERT{A}{M = N}{u}{C}
     }{[Eq]}   
		\\\\
		\TWOPREMISERULENAMEDRIGHT
		{
			A-\EXTINDEP
		}
		{
			\GAMMA  \PLUSTC \TCV\PLUSV  x:\alpha  \vdash 
			\ASSERT{A^{\MINUS x} \AND B}{M}{m}{C}
		}
		{
			\GAMMA \vdash 
			\ASSERT
			{A}
			{\lambda x^{\alpha}. M}{u}
			{\FAD{\TCV} 
				\FORALL{x^{\alpha}}{\TCV} (B \IMPLIES \ONEEVAL{u}{x}{m}{C})}
		}{[Lam]}
		\\\\
		\TWOPREMISERULENAMEDRIGHT
		{
			\ASSERT{A}{M}{m}{B}
		}
		{
			\ASSERT{B}{N}{n}{\ONEEVAL{m}{n}{u}{C}}
		}
		{
			\ASSERT{A}{MN}{u}{C}
		}{[App]}
		\\\\
		\FIVEPREMISERULENAMEDRIGHT
		{
			\ASSERT{A}{M}{m}{B}
		}
		{
			\ASSERT{B\LSUBST{b_i}{m}}{N_i}{u}{C}
		}
		{
			b_1 = \TRUE
		}
		{
			b_2 = \FALSE
		}
		{
			i = 1, 2
		}
		{
			\ASSERT{A}{\IFTHENELSE{M}{N_1}{N_2}}{u}{C}
		}{[If]}
		\\\\
		\TWOPREMISERULENAMEDRIGHT
		{
				\ASSERT{A}{M}{m}{B}
		}
		{
				\ASSERT{B}{N}{n}{C \LSUBST{\PAIR{m}{n}}{u}}
		}
		{
				\ASSERT{A}{\PAIR{M}{N}}{u}{C}
		}{[Pair]}
		\quad
		\ONEPREMISERULENAMEDRIGHT
		{
				\ASSERT{A}{M}{m}{C
					\LSUBST{\pi_i(m)}{u}}
		}
		{
			\ASSERT{A}{\pi_i(M)}{u}{C}
		}{[Proj($i$)]}
		\\\\
		\TWOPREMISERULENAMEDRIGHT
		{
			\ASSERT{A}{M}{m}{B}
		}
		{
			\ASSERT{B}{N}{u}{C}
		}
		{
			\ASSERT{A}{\LET{m^{\alpha}}{M}{N}}{u}{C}
		}{[Let]}
  \end{RULES}
  \vspace{-4mm}  

  \caption{ Rules for the core language,
    cf.~\cite{BergerM:prologfhgrtmp,YHB07:local:full,HY04PPDP}.  We
    require $C \THINWRT{m}$ in \RULENAME{[Proj($i$), Let]} and
    $C\THINWRT{m,n}$ in \RULENAME{[Eq, App, Pair]}.  We omit
    LTCs where not essential.  }\label{figure_rules_language}
\end{FIGURE}

%% file: figure_rules_structural.tex
\begin{FIGURE}
  \begin{RULES}
    \THREEPREMISERULENAMEDRIGHT
        {
          A \IMPLIES A'
        }
        {
          \ASSERT{A'}{M}{m}{B'}
        }
        {
          B' \IMPLIES B
        }
        {
          \ASSERT{A}{M}{m}{B}
        }{[Conseq]}
        \qquad
        \TWOPREMISERULENAMEDRIGHT
        {
        	C-\EXTINDEP
        }
        {
                \ASSERT{A}{M}{m}{B}
        }
        {
			\ASSERT{A \AND C}{M}{m}{B \AND C}
        }{[Invar]}
    	\\\\
    	\TWOPREMISERULENAMEDRIGHT
    	{
    		A-\EXTINDEP
    	}
    	{
    		\GAMMA \PLUSV m:\NAME  \vdash \ASSERT{A \AND \FRESH{m}{\GAMMA}}{M}{u}{C}
    	}
    	{
    		\GAMMA \vdash \ASSERT{A}{\LET{m}{\GENSYM()}{M}}{u}{C}
    	}{[LetFresh]}
  \end{RULES}
  \vspace{-4mm}  
  \caption{Key structural rules \RULENAME{[Conseq]} and \RULENAME{[Invar]} and for convenience the derived \RULENAME{[LetFresh]} rule where $C\THINWRT{m}$ is required.}\label{figure_rules_structural}
\end{FIGURE}

%% file: reasoning.tex
\section{Reasoning Examples}
\label{reasoning}

\PARAGRAPH{Example 1}
We reason about the core construct $\GENSYM()$ in an LTC $\GAMMA$.
In Line 2, $(utc1)$ instantiates the postcondition to $\ONEEVAL{b}{()}{a}{\FRESH{a}{\GAMMA \PLUSV b}}$ and $(f4)$ removes the $b$ from the LTC to ensure the postcondition satisfies the $\THINWRT{b}$ requirement in \RULENAME{[App]}.
\begin{NDERIVATION}{1}
	\NLINE{\GAMMA \Vdash\ASSERT{\TRUTH}{\GENSYM}{b}{ \FAD{\TCV} \ONEEVAL{b}{()}{a}{\FRESH{a}{\TCV}} }}{\RULENAME{[Gensym]}}
	\NLINE{\GAMMA \Vdash \ASSERT{\TRUTH}{\GENSYM}{b}{ \ONEEVAL{b}{()}{a}{\FRESH{a}{\GAMMA}} }}{\RULENAME{[Conseq]}, (utc1), (f4), 1}
	\NLINE{\GAMMA \PLUSV b \Vdash \ASSERT{ \ONEEVAL{b}{()}{a}{\FRESH{a}{\GAMMA}}}{()}{c}{ \ONEEVAL{b}{c}{a}{\FRESH{a}{\GAMMA} }}}{\RULENAME{[Const]}}
	\NLASTLINE{\GAMMA \Vdash \ASSERT{\TRUTH}{\GENSYM ()}{a}{\FRESH{a}{\GAMMA}}}{\RULENAME{[App]}, 2, 3}
\end{NDERIVATION}

\PARAGRAPH{Example 2}
We reason about the comparison of two fresh names, clearly returning $\FALSE$, by applying Example 1 in the relevant LTCs.
\begin{NDERIVATION}{1}
	\NLINE{\GAMMA \Vdash \ASSERT{\TRUTH}{\GENSYM()}{a}{\FRESH{a}{\GAMMA}}}{See Example 1}
	\NLINE{\GAMMA \PLUSV a \Vdash \ASSERT{\TRUTH}{\GENSYM()}{b}{\FRESH{b}{\GAMMA \PLUSV a}}}{See Example 1}
	\NLINE{\GAMMA \PLUSV a \Vdash \ASSERT{\FRESH{a}{\GAMMA}}{\GENSYM()}{b}{a \neq b}}{\RULENAME{[Conseq]}, (f3), 2}
	\NLASTLINE{\GAMMA \Vdash \ASSERT{\TRUTH}{\GENSYM()=\GENSYM()}{u}{u = \FALSE}}{\RULENAME{[Eq]}, 3}
\end{NDERIVATION}

\PARAGRAPH{Example 3}
Placing name generation inside an abstraction halts the production of fresh names until the function is applied. When $y$ is of type $\UNIT$ then this specification is identical to that of $\GENSYM$.
\begin{NDERIVATION}{1}
	\NLINE{\GAMMA \PLUSTC \TCV \PLUSV y \Vdash \ASSERT{\TRUTH}{\GENSYM()}{m}{\FRESH{m}{\GAMMA \PLUSTC \TCV \PLUSV y}}}{See Example 1}
	\NLASTLINE{\GAMMA \Vdash \ASSERT{\TRUTH}{\lambda y. \GENSYM()}{u}{\FAD{\TCV} \FORALL{y}{\TCV} \ONEEVAL{u}{y}{m}{\FRESH{m}{\GAMMA \PLUSTC \TCV \PLUSV y}}}}{\RULENAME{[Lam]}, 2}
\end{NDERIVATION}

\PARAGRAPH{Example 4} Generating a name outside an abstraction and returning that same name in the function is often compared to Example 3 \cite{BentonN:mechbisftnc, stark:namhof}.
We reason as follows: letting $A_4(p) \DEFEQ \FAD{\TCV} \FORALL{y}{\TCV} \ONEEVAL{u}{y}{m}{\FRESH{m}{\GAMMA} \PAND p=m}$.
\begin{NDERIVATION}{1}
	\NLINE{\ASSERT{\FRESH{x}{\GAMMA}}{x}{m}{\FRESH{m}{\GAMMA} \PAND x=m}}{\RULENAME{[Var]}}
	\NLINE{\ASSERT{\FRESH{x}{\GAMMA}}{\lambda y. x}{u}{A_4(x)}}{\RULENAME{[Lam]}, 1}
	\NLINE{\ASSERT{\FRESH{x}{\GAMMA}}{\lambda y. x}{u}{\EXISTS{x'}{u} A_4(x')}}{\RULENAME{[Conseq]}, 2}
	\NLASTLINE{\GAMMA \Vdash \ASSERT{\TRUTH}{\LET{x}{\GENSYM()}{\lambda y. x}}{u}{\EXISTS{x'}{u}  A_4(x')}}{\RULENAME{[LetFresh]}, 3}
\end{NDERIVATION}
Proof of line 3 above, essentially proves $x$ is derivable from $u$:
\begin{NDERIVATION}{5} 
	\NLINE{A_4(x) \PAND \FORALL{y}{\emptyset} \ONEEVAL{u}{y}{m}{x=m}}{(utc1), (u2), FOL}
	\NLINE{A_4(x) \PAND \EXISTS{x'}{u} x=x'}{(ex2), (ex3)}
	\NLINE{ \EXISTS{x'}{u} (A_4(x) \PAND x=x')}{(u3), (u4)}
	\NLASTLINE{\EXISTS{x'}{u} A_4(x')}{(eq1)}
\end{NDERIVATION}

\PARAGRAPH{Example 5} In order to demonstrate the subtlety of hidden names, the Introduction used
Program (\ref{example_canonical}), which was
$
  M \DEFEQ \LET{x}{\GENSYM()}{\lambda y. x=y}
$. 
We now use our logic to reason about $M$.

\begin{NDERIVATION}{1}
	\NLINE{\GAMMA \PLUSV x \PLUSTC \TCV \PLUSV y \Vdash 
		\ASSERT
		{\TRUTH}
		{x = y}{m}
		{m = (\EQA{x}{y})}}
	{\RULENAME{[Eq]}}
	\NLINE{\GAMMA \PLUSV x  \Vdash 
		\ASSERT
		{\TRUTH}
		{\lambda y.x = y}{u}
		{\FAD{\TCV}\FORALL{y}{\TCV} \EVALFORMULASHORT{u}{y}{(\EQA{x}{y})}}
	}{\RULENAME{[Lam]}, 1}
	\NLINE{\GAMMA \PLUSV x  \Vdash 
		\ASSERT
		{\FRESH{x}{\GAMMA}}
		{\lambda y.x = y}{u}
		{\FRESH{x}{\GAMMA} \PAND  \FAD{\TCV}\FORALL{y}{\TCV}\EVALFORMULASHORT{u}{y}{(\EQA{x}{y})}}
	}{\RULENAME{[Invar]}, 2}
	\NLINE{\GAMMA \PLUSV x  \Vdash 
		\ASSERT
		{\FRESH{x}{\GAMMA}}
		{\lambda y.x = y}{u}
		{\FORALL{y}{\GAMMA \PLUSV u} \EVALFORMULASHORT{u}{y}{\FALSE}}
	}{\RULENAME{[Conseq]}, 3}
	\NLINE{\GAMMA \Vdash 
		\ASSERT
		{\TRUTH}
		{M}{u}
		{\FORALL{y}{\GAMMA \PLUSV u} \EVALFORMULASHORT{u}{y}{\FALSE}}
	}{ \RULENAME{[LetFresh]}}
	\NLASTLINE{\GAMMA \Vdash 
		\ASSERT
		{\TRUTH}
		{M}{u}
		{\FAD{\TCV}\FORALL{y^{\NAME}}{\TCV} \EVALFORMULASHORT{u}{y}{\FALSE}}
	}{(utc2)}
\end{NDERIVATION}
To prove line 4 above we apply the axioms as follows:
\PROOFFINISHED{
	\begin{NDERIVATION}{7}
		\NLINE{
			\GAMMA \PLUSV x \PLUSV u  \Vdash 
			\FRESH{x}{\GAMMA} \PAND \FAD{\TCV}\FORALL{y}{\TCV} \EVALFORMULASHORT{u}{y}{(\EQA{x}{y})}	
		}{}
		\NLINE{
			\FRESH{x}{\GAMMA} \PAND \FORALL{y}{\GAMMA \PLUSV x \PLUSV u} \EVALFORMULASHORT{u}{y}{(\EQA{x}{y})}	
		}{(utc1)}
		\NLINE{
				\FRESH{x}{\GAMMA \PLUSV u} \PAND \FORALL{y}{\GAMMA \PLUSV x \PLUSV u} \ \EVALFORMULASHORT{u}{y}{(\EQA{x}{y})}
		}{(f1)}
		\NLINE{
			\FRESH{x}{\GAMMA \PLUSV u} \PAND \FORALL{y}{\GAMMA \PLUSV u} \ \EVALFORMULASHORT{u}{y}{(\EQA{x}{y})}
		}{(u2)}
		\NLINE{
			\FORALL{y}{\GAMMA \PLUSV u} \ \FRESH{x}{\GAMMA \PLUSV u \PLUSV y} \PAND \EVALFORMULASHORT{u}{y}{(\EQA{x}{y})}
		}{(f2)}
		\NLINE{
			\FORALL{y}{\GAMMA \PLUSV u} \ x \neq y \PAND \EVALFORMULASHORT{u}{y}{(\EQA{x}{y})}
		}{(f3)}
		\NLASTLINE{
			\FORALL{y}{\GAMMA \PLUSV u} \EVALFORMULASHORT{u}{y}{\FALSE}
		}{(e1)}
	\end{NDERIVATION}
}

\PARAGRAPH{Example 6} To demonstrate the release of a hidden variable using Program \eqref{example_canonical_pair}, which was 
$
M \DEFEQ \LET{x}{\GENSYM()}{\PAIR{x}{\lambda y. x=y}}
$,
we reason as follows, with $A_6(p,q) \DEFEQ \FRESH{p}{\GAMMA} \PAND \FAD{\TCV}\FORALL{y}{\TCV} \EVALFORMULASHORT{q}{y}{(\EQA{p}{y})}$:
\begin{NDERIVATION}{1}
	\NLINE{
		\ASSERT
		{\FRESH{x}{\GAMMA}}
		{x}{b}
		{x=b \PAND \FRESH{x}{\GAMMA}}
	}{\RULENAME{[Var]}}
	\NLINE{
		\ASSERT
		{\TRUTH}
		{\lambda y.x = y}{c}
		{\FAD{\TCV}\FORALL{y}{\TCV} \EVALFORMULASHORT{c}{y}{(\EQA{x}{y})}}
	}{See Example 4, lines 1-2}
	\NLINE{
		\ASSERT
		{x=b \PAND \FRESH{x}{\GAMMA}}
		{\lambda y.x = y}{c}
		{
			x=\pi_1(\PAIR{b}{c}) \PAND C(x,c)
		}
	}{
			\RULENAME{[Conseq]}, \RULENAME{[Invar]}, 2
	}
	\NLINE{
		\ASSERT
		{x=b \PAND \FRESH{x}{\GAMMA}}
		{\lambda y.x = y}{c}
		{A_6(\pi_1(a), \pi_2(a)) \LSUBST{\PAIR{b}{c}}{a}}
	}{\RULENAME{[Conseq]}, (eq1)}
	\NLINE{
		\ASSERT
		{\FRESH{x}{\GAMMA}}
		{\PAIR{x}{\lambda y.x = y}}{a}
		{A_6(\pi_1(a), \pi_2(a))}
	}{\RULENAME{[Pair]}, 1, 4}
	\NLASTLINE{
		\GAMMA \Vdash 
		\ASSERT
		{\TRUTH}
		{M}{a}
		{A_6(\pi_1(a), \pi_2(a))}
	}{
		\RULENAME{[LetFresh]}, 5
	}
\end{NDERIVATION}

%% file: conclusion.tex
\section{Conclusion}

We have presented the first program logic for the \NUC, a variant of the STLC
with names as first class values. Our logic is a conservative
extension with two new universal quantifiers of the logic in
\cite{HY04PPDP} for the STLC. We provide axioms and proof rules for
the logic, prove their soundness, and show its expressive power by
reasoning about well-known difficult examples from the literature.

We are currently unable to reason about this example from \cite{PittsAM:obsproohoftdclnown}:
\[
\begin{array}{l}
  \LET{F}{(\LET{x,y}{\GENSYM()}{\lambda f^{\NAME \FS \NAME}. fx=fy})}{}
  \\
  \LET{G}{\lambda v^{\NAME}. F(\lambda u^{\NAME}. v=u) }{FG}
\end{array}
\]
Is this because our logic is too inexpressive, or did we simply fail
to find the right proof?  Another open question is whether our logic's
approach to freshness is independent of the \NUC's lack of integers
and recursion, or not?  For both questions we conjecture the former,
and leave them as future work.

%% file: appendix/allappendices.tex
\appendix

\section*{\LARGE Appendix}
\vspace{0.3cm}
In the proofs that follow, we will use $\METALOGIC{blue}$ coloured to represent the ``meta'' logic constructors, i.e. $\Mforall$ meaning ``for all'' (in meta language), ...

We prove soundness of the logic by proving that all the rules and axioms are sound, i.e.  $\vdash \ASSERT{A}{M}{u}{B}$ implies $\vDash \ASSERT{A}{M}{u}{B}$

The proof of soundness of the logic requires us to prove all axioms and rules sound, this is done in App.~\ref{appendix_soundness_axioms} and App.~\ref{appendix_soundness_rules}.

We give a syntactic characterisation of $\EXTINDEP$ and thinness and prove they hold as such in App.~\ref{appendix_EXTINDEP} and App.~\ref{appendix_thinness}.

Conservativity from the STLC is introduced formally and proven in App.~\ref{appendix_conservativitiy}.

\PARAGRAPH{Structure of appendix}
\begin{itemize}
	\item[App.~\ref{appendix_figure_typing_programs}] Typing rules of \NUC \ program.
	
	\item[App.~\ref{appendix_lemmas}] Lemmas used in following proofs
	
	\item[App.~\ref{appendix_soundness_axioms}] Proof of soundness of axioms
	
	\item[App.~\ref{appendix_soundness_rules}] Proof of soundness of rules
	
	\item[App.~\ref{appendix_EXTINDEP}] Define Syntactic \EXTINDEP \  and proof that syntactic \EXTINDEP \ implies \EXTINDEP
	
	\item[App.~\ref{appendix_thinness}] Define Syntactic thinness and proof that syntactic thinness implies thinness
	
	\item[App.~\ref{appendix_conservativitiy}] Definition and proof of conservativity
\end{itemize}

\input{appendix/appendix_typing_programs}

\input{appendix/appendix_soundness_lemmas}

\newpage
\input{appendix/appendix_soundness_axioms}

\newpage
\input{appendix/appendix_soundness_rules}

\newpage
\input{appendix/appendix_EXTINDEP}
\newpage
\input{appendix/appendix_Thinness}

\newpage
\input{appendix/conservativity}

%% file: appendix/appendix_typing_programs.tex
\section{Typing Rules For The \NUC}
\label{appendix_figure_typing_programs}
The typing rules for the \NUC\ are included for completeness in Fig.~\ref{figure_typing_programs}
\begin{FIGURE}

\begin{RULES}
	\ONEPREMISERULE
	{
		c \in \{ \TRUE, \FALSE \}
	}
	{
		\TYPES{\Gamma}{c}{\BOOL}
	}
	\quad
	\ZEROPREMISERULE
	{
		\TYPES{\Gamma}{()}{\UNIT}
	}
	\quad
	\ONEPREMISERULE
	{
		\Gamma(x) = \alpha
	}
	{
		\TYPES{\Gamma}{x}{\alpha}
	}
	\quad
	\ZEROPREMISERULE
	{
		\TYPES{\Gamma}{\GENSYM}{\UNIT \FS \NAME}
	}
	\qquad
	\ZEROPREMISERULE
	{
		\TYPES{\Gamma}{r}{\NAME}
	}
	\\\\
	\ONEPREMISERULE
	{
		\TYPES{\Gamma, x:\alpha_1}{M}{\alpha_2}
	}
	{
		\TYPES{\Gamma}{\lambda x^{\alpha_1}. M}{\alpha_1 \FS \alpha_2}
	}
	\qquad
	\TWOPREMISERULE
	{
		\TYPES{\Gamma}{M}{\alpha_1 \FS \alpha_2}
	}
	{
		\TYPES{\Gamma}{N}{\alpha_1}
	}
	{
		\TYPES{\Gamma}{MN}{\alpha_2}
	}
	\qquad
	\TWOPREMISERULE
	{
		\TYPES{\Gamma}{M}{\alpha_1}
	}
	{
		\TYPES{\Gamma, x:\alpha_1}{N}{\alpha_2}
	}
	{
		\TYPES{\Gamma}{\LET{x^{\alpha_1}}{M}{N}}{\alpha_2}
	}
	\\\\
	\THREEPREMISERULE
	{
		\TYPES{\Gamma}{M}{\BOOL}
	}
	{
		\TYPES{\Gamma}{N_1}{\alpha}
	}
	{
		\TYPES{\Gamma}{N_2}{\alpha}
	}
	{
		\TYPES{\Gamma}{\IFTHENELSE{M}{N_1}{N_2}}{\alpha}
	}
	\qquad
	\TWOPREMISERULE
	{
		\TYPES{\Gamma}{M}{\alpha_1}
	}
	{
		\TYPES{\Gamma}{N}{\alpha_2}
	}
	{
		\TYPES{\Gamma}{\PAIR{M}{N}}{\alpha_1 \times \alpha_2}
	}
	\qquad
	\ONEPREMISERULE
	{
		\TYPES{\Gamma}{M}{\alpha_1 \times \alpha_2}
	}
	{
		\TYPES{\Gamma}{\pi_i(M)}{\alpha_i}
	}
\end{RULES}
\caption{
	Typing rules for the \NUC \ programming language.
}
\label{figure_typing_programs}
\end{FIGURE}

%% file: appendix/appendix_soundness_lemmas.tex
\section{Lemmas Used For The Following Proofs}
\label{lemmas} \label{appendix_lemmas}

First we introduce some definitions that are either from the literature or derived from other definitions.

\PARAGRAPH{Contextual Equivalence}
\begin{equation}
	\begin{array}[t]{c}
		M_1 \CONGCONTEXT{\alpha}{G} M_2
		\\ \MIFF \\
		\Mforall C[\cdot]^{\alpha}, b^{\BOOL}.
		\left(
		\begin{array}{c} 
			\TYPES{\emptyset}{C[\cdot]^{\alpha}}{\BOOL}
			\\
			\MAND \ \AN{C[\cdot]} \subseteq G
		\end{array}
		\right)
		\MIMPLIES \ 
		\left(
		\begin{array}{c}
			(\AN{\MMM} \cup G, \ C[M_1] ) \CONV (G_1', \ b) 
			\\ \MIFF \\ 
			(\AN{\MMM} \cup G, \ C[M_2] ) \CONV (G_2', \ b) 
		\end{array}
		\right)
	\end{array}
\end{equation}

\PARAGRAPH{Derived Semantics of formulae}
\begin{itemize}
	\item $\MMM^{\GAMMA} \models \EXISTS{ x^{\alpha}}{\GAMMA'} A$
	if
	there exists 
	$ M^{\alpha}. \
	\LTCDERIVEDVALUE{M}{\GAMMA}{\MMM}{V}$
	and  
	$\MMM' \cdot x:V \models A $
	
	\item $\MMM \models A^{-x}(x) \LSUBSTLTC{e}{x}{\GAMMA}$ \ if \ $x \notin \DOM{\MMM} \MAND \LTCDERIVEDVALUE{e}{\GAMMA}{\MMM}{V} \ \MAND \ \MMM \cdot x:V \models A(x)$
	
	\item $\MMM \models A^{-x}(x) \LSUBSTLTC{e}{x}{\GAMMA}$ \ if \ $x \in \DOM{\MMM} \MAND \Mforall x'. x' \notin \DOM{\MMM} \MIMPLIES \MMM \models A(x')\LSUBSTLTC{e}{x'}{\GAMMA}$
	
	\item $\MMM \models A^{-\TCV}(\TCV)
	\LSUBSTLTC{\GAMMA_0}{\TCV}{\GAMMA}$ if $\TCV \notin \DOM{\MMM}$
	then
	$\SEM{\GAMMA_0}{\MMM} = \GAMMA_1$
	and
	$\MMM \cdot \TCV:\GAMMA_1 \models A(\TCV)$
	
	\item $\MMM \models A^{-\TCV}(\TCV)
	\LSUBSTLTC{\GAMMA_0}{\TCV}{\GAMMA}$  if  $\TCV \in \DOM{\MMM}$
	and for all $\TCV'
	\notin \DOM{\MMM}.$
	$\MMM \models A(\TCV')\LSUBSTLTC{\GAMMA_0}{\TCV'}{\GAMMA}$
\end{itemize}

\PARAGRAPH{Lemmas}
Now we introduce lemmas used to simplify the later proofs in App.~\ref{appendix_soundness_axioms}, App.~\ref{appendix_soundness_rules}, App.~\ref{appendix_EXTINDEP}, App.~\ref{appendix_thinness}

\begin{lemma}[\DONE Adding/removing unused names maintains evaluation]
	\label{lem:adding/remove_unused_names_maintains_evaluation}
	\[
	\Mforall M, V, G', G_m.  \
	\AN{V} \cap G' = \emptyset \ \MIMPLIES \
	\left(
	\begin{array}{c}
		(\AN{M},G', \ M) \CONV (\AN{M},G',G_m, \ V) \\ \MIFF \\  (\AN{M}, \ M) \CONV (\AN{M},G_m, \ V)
	\end{array}
	\right)
	\]
	\begin{proof}
		Clearly holds as $M$ does not require any name in $G'$ if $\AN{V} \cap G' = \emptyset$, noting that $G\cap G_m = \emptyset$ and  $G'\cap G_m = \emptyset$ and hence if this does not hold, the fresh names in $V$ can be renamed so that it does hold.
		\\
		Note: If $\AN{V} \cap G' \neq \emptyset$ then this fails as it would imply previously generated names that are not accessible can be reached.
	\end{proof}
\end{lemma}

\begin{lemma}[Congruence is unaffected by adding/removing excess names]
	\label{lem:adding/removing_names_to_congruence_makes_no_difference}
	\[
	\Mforall G, G'. \
	\AN{M,N} \cap \AN{G'} = \emptyset
	\ \MIMPLIES \
	(M \CONGCONTEXT{\alpha}{G} N \MIFF M \CONGCONTEXT{\alpha}{G \cup G'} N) 
	\]
	\begin{proof}
		Clearly holds by Sem. $\CONGCONTEXT{\alpha}{G \cup G'}$ where $\AN{M,N} \subseteq G$. 
	\end{proof}
\end{lemma}

In the following lemmas we use $\TYBASE$ as the base types as defined in Sec.~\ref{sec:axioms} i.e. $\TYBASE::= \UNIT \VERTICAL \BOOL \VERTICAL \TYBASE  \times \TYBASE$. Note $\TYBASE$ is also always a $\NAME$-free type.

\begin{lemma}[\DONE Base type values are name free]
	\label{lem:BaseValuesAreNameFree}
		\[
			\Mforall V\text{-value}. \ \TYPES{\emptyset}{V}{\TYBASE} \ \MIMPLIES \ \AN{V}= \emptyset
		\]
	\begin{proof}
		By induction on the structure of type $\TYBASE$, so clearly holds.
		\begin{itemize}
			\item $\alpha = \UNIT$ implies $V= \UNIT$ so this clearly holds.
			\item $\alpha = \BOOL$ implies $V= \TRUE$ or $V= \FALSE$ so this clearly holds.
			\item $\alpha = \TYBASE_1 \times \TYBASE_2$ then $V= \PAIR{V_1}{V_2}$ and by IH assuming 
			\\
			$\Mforall V_1\text{-value}. \TYPES{\emptyset}{V_1}{\TYBASE_1} \ \MIMPLIES \ \AN{V_1}= \emptyset$
			\\
			$\Mforall V_2\text{-value}. \TYPES{\emptyset}{V_2}{\TYBASE_2} \ \MIMPLIES \ \AN{V_2}= \emptyset$
			\\
			then clearly $\AN{\PAIR{V_1}{V_2}}=\emptyset$
		\end{itemize}
	\end{proof}
\end{lemma}

\begin{lemma}[\DONE Base type values can be derived equally from any LTC]
	\label{lem:BaseValues_are_equally_derivable_from_any_LTC}
	\[
	\Mforall \GAMMA, \MMM^{\GAMMA}, \GAMMA_1, V^{\TYBASE}. \ \TCTYPES{\GAMMA}{\GAMMA_1} \ \MIMPLIES \ 
	(\Mexists M_0^{\TYBASE}. \ \LTCDERIVEDVALUE{M_0}{\emptyset}{\MMM}{V}
	\ \MIFF \
	\Mexists M_1^{\TYBASE}. \ \LTCDERIVEDVALUE{M_1}{\GAMMA_1}{\MMM}{V})
	\]
	\begin{proof}
		Holds as $M_0 \equiv V \equiv M_1$ always holds as $\AN{V}=\emptyset \ \MAND \ \TYPES{\emptyset}{V}{\TYBASE}$.
	\end{proof}
	\HIDDEN{
		\begin{proof}
			Any value of type $\TYBASE$ contains no names (Lemma \ref{lem:BaseValuesAreNameFree}) hence all values of type $\TYBASE$ are also terms such that $\AN{V}=\emptyset$ and $\TYPES{\emptyset}{V}{\TYBASE}$.
			\\
			$\MIMPLIES$: clearly holds as $\emptyset \subseteq \GAMMA_1$.
			\\
			$\MIMPLIEDBY$: holds as if there is a term that is derived by $\GAMMA_1$ then the value itself is always derivable from $\emptyset$ so clearly holds.
		\end{proof}
	}
\end{lemma}

\begin{lemma}[\DONE Base type values don't extend reach]
	\label{lem:Base_Values_Dont_Extend_Reach}
	\[
	\Mforall \GAMMA_0. \Mforall \MMM^{\GAMMA_0}.  \
	\TCTYPES{\GAMMA_0}{\GAMMA \PLUSV y:\TYBASE}
	\ \MIMPLIES \
	(\Mexists M^{\NAME}. \LTCDERIVEDVALUE{M}{\GAMMA}{\MMM}{V} 
	\ \MIFF \ 
	\Mexists M^{\NAME}. \LTCDERIVEDVALUE{M}{\GAMMA \PLUSV y:\TYBASE}{\MMM}{V})
	\]
	\begin{proof}
		Clearly holds from Lemmas \ref{lem:BaseValuesAreNameFree} and \ref{lem:BaseValues_are_equally_derivable_from_any_LTC} so no names are added to the LTC.
	\end{proof}
\end{lemma}

\begin{lemma}[\DONE Base type values can be added and removed from the $\TCV$ mappings without any harm]
	\label{lem:Base_Values_can_be_added/removed_from_TCV:TC}
	\[
	\Mforall \GAMMA. \Mforall \MMM^{\GAMMA}. \ \MMM \cdot x:V_x^{\TYBASE} \cdot \TCV:\GAMMA \REMOVETCVfrom \models A \ \MIFF \ \MMM \cdot x:V_x^{\TYBASE} \cdot \TCV:(\GAMMA \PLUSV x:\TYBASE) \REMOVETCVfrom \models A
	\]
	
	
	\begin{proof}
			The only two occurrences of $\TCV$  possible in assertions are in $\GAMMA_0$ in $\FORALL{x}{\GAMMA_0}$ and $\FRESH{x}{\GAMMA_0}$:
			\\
			If $A$ contains $\FORALL{x}{\GAMMA_0} A'$ and $\GAMMA_0$ contains $\TCV$ then it is unaffected by the addition (or removal) of $x:\TYBASE$ as $\LTCDERIVEDVALUE{M}{\GAMMA_0}{\MMM}{V} \MIFF \LTCDERIVEDVALUE{M}{\GAMMA_0 \PLUSV y:\TYBASE}{\MMM}{V}$ from Lemma \ref{lem:Base_Values_Dont_Extend_Reach}
			\\
			If $A$ contains $\FRESH{x}{\GAMMA_0}$ and $\GAMMA_0$ contains $\TCV$ then it is unaffected by the addition (or removal) of $x:\TYBASE$ as $\LTCDERIVEDVALUE{M}{\GAMMA_0}{\MMM}{V} \MIFF \LTCDERIVEDVALUE{M}{\GAMMA_0 \PLUSV y:\TYBASE}{\MMM}{V}$ from Lemma \ref{lem:Base_Values_Dont_Extend_Reach}
			hence the Lemma holds.
	\end{proof}	
\end{lemma}

\begin{lemma}[\DONE Base type values can be added/removed and maintain extensions]
	\label{lem:Base_types_added/removed_maintain_extension}
	\[
	\Mforall \GAMMA^{-m},\MMM^{\GAMMA}, M^{\TYBASE}, V_m. \
	\LTCDERIVEDVALUE{M}{\GAMMA}{\MMM}{V_m}
	\ \MIMPLIES \
	\ \Mforall \GAMMA', \MMM'^{\GAMMA'}. \ \MMM \EXTSTAR \MMM' \ \MIFF \ \MMM \cdot m:V_m \EXTSTAR \MMM' \cdot m:V_m
	\]
	\begin{proof}
		$\MIMPLIES$: from Lemma \ref{lem:Gamma_derived_terms_maintain_extension_when_added} ($\AN{V_m}= \emptyset$). 
		\\
		$\MIMPLIEDBY$:
		This holds as $V_m:\TYBASE$ means the variable $m$ can always be replaced by a value $V_m$ which by Lemma \ref{lem:BaseValuesAreNameFree}, $\AN{V_m}=\emptyset$ and by being a value, $\TYPES{\emptyset}{V_m}{\TYBASE}$, so their substitution is not problematic.
	\end{proof}
\end{lemma}

\begin{lemma}[\DONE Semantics of LTC is equal in model extensions]
		\label{lem:extensions_give_same_semantics_for_type_contexts}
		\[
		\Mforall \GAMMA_1, \GAMMA_2, \MMM_1^{\GAMMA_1}, \MMM_{2}^{\GAMMA_2}.  \
		(\MMM_1 \EXTSTAR \MMM_{2} 
		\ \MAND \ \TCTYPES{\GAMMA_1}{\GAMMA_0}
		) 
		\ \MIMPLIES \ \SEM{\GAMMA_0}{\MMM_1} \equiv \SEM{\GAMMA_0}{\MMM_{2}}
		\]
\begin{proof}
	Prove by induction on the structure of $\GAMMA_0$:
	Induction on $\GAMMA_0$
	\begin{itemize}
		\item[$\GAMMA_0  \equiv \emptyset$:] $\SEM{\emptyset}{\MMM_1} \equiv \SEM{\emptyset}{\MMM_{2}}$
		
		\item[$\GAMMA_0 \equiv \GAMMA_0' \PLUSV x:\alpha$:]
		By IH $\SEM{\GAMMA_0'}{\MMM_1} \equiv \SEM{\GAMMA_0'}{\MMM_{2}}$ and also by Def $\EXTSTAR$, $\GAMMA_1(x) \equiv \GAMMA_{2}(x)$.
		
		\item[$\GAMMA_0 \equiv \GAMMA_0' \PLUSTC \TCV:\TC$:]
		By IH $\SEM{\GAMMA_0'}{\MMM_1} \equiv \SEM{\GAMMA_0'}{\MMM_{2}}$ and also by Def $\EXTSTAR$, $\GAMMA_1(\TCV) \equiv \GAMMA_{2}(\TCV)$
		and $\MMM_1(\TCV) \equiv \MMM_{2}(\TCV)$ gives us that this statement holds.
		
		\item[$\GAMMA_0 \equiv \GAMMA_0' \PLUSG \GAMMA_0''$:] by IH on both $\GAMMA_0'$ and $\GAMMA_0''$ this clearly holds.
	\end{itemize}
\end{proof}

\end{lemma}

\begin{lemma}[\DONE Model extensions close terms equally]
	\label{lem:extensions_close_terms_equally}
	\[
	\Mforall \GAMMA_1, \MMM_1^{\GAMMA_1}, \GAMMA_2, \MMM_{2}^{\GAMMA_2}. 
	(\MMM_1 \EXTSTAR \MMM_{2} 
	\ \MAND \ \TYPES{\SEM{\GAMMA_1}{\MMM_1} }{M}{\alpha}
	) 
	\ \MIMPLIES \ M\MMM_1 \equiv M\MMM_{2}
	\]
	\begin{proof}
		Prove by induction on the structure of $M$ using the definition of term closure:
		\begin{itemize}
			\item[$M  \equiv x$:] 
			$x\in \GAMMA_1$ implies (Def $\EXTSTAR$) $x\MMM_1 \equiv \MMM_1(x) \equiv \MMM_{2}(x)\equiv x\MMM_{12}$
			\\
			$x\notin \GAMMA_1$ implies $x$ is introduced by a $\lambda$ hence should be left untouched.
			
			\item[$M \equiv c$:]
			$c\MMM_1 \equiv c \equiv c\MMM_{2}$ 
			
			\item[$M \equiv \GENSYM$:]
			$\GENSYM\MMM_1 \equiv \GENSYM \equiv \GENSYM\MMM_{2}$ 
			
			\item[$M \equiv \lambda x. M'$:]
			$x\in \GAMMA_1$ implies $\alpha$-renaming can be used to substitute $x$ for some other unused variable, 
			\\
			i.e. $\lambda x. M \equiv \lambda z. (M\PSUBST{z}{x})$ hence the same proof as below holds.
			\\
			$x\notin \GAMMA_1$ implies $\MMM_1^{-x} \equiv \MMM_1 \EXTSTAR \MMM_{2} \equiv \MMM_{2}^{-x}$
			\\
			and IH implies $ M' \MMM_1^{-x} \equiv M' \MMM_{2}^{-x}$
			\\
			hence
			$(\lambda x. M')\MMM_1 \equiv \lambda x. (M' \MMM_1^{-x}) \equiv \lambda x. (M' \MMM_{2}^{-x}) \equiv (\lambda x.M')\MMM_{2}$
			\item[$M \equiv N_1N_2$:] by IH on both $N_1$ and $N_2$ this clearly holds.
			\item[$M \equiv \LET{x}{N_1}{N_2}$] holds by induction on both $N_1$ and $N_2$ similar to the above two cases with $x$ $\alpha$-renamed if required.
			 
			\item[$M \equiv \RAWPAIR{N_1 , N_2}$:] by IH on both $N_1$ and $N_2$  this clearly holds.
			
			\item[$M \equiv \pi_i(M')$:] by IH on $M'$ this clearly holds.
			
			\item[$M \equiv \IFTHENELSE{M'}{N_1}{N_2}$:] by IH on both $M'$, $N_1$ and $N_2$ this clearly holds.
			
		\end{itemize}
	\end{proof}
\end{lemma}

\begin{lemma}[\DONE LTC derived values are unaffected by the addition/removal of a TCV to model]
	\label{lem:LTC_derived_values_unaffected_by_TCV_addition/removal}
	\[
	\Mforall \GAMMA, \MMM^{\GAMMA}, \GAMMA_0, M. \
	\TCTYPES{\GAMMA}{\GAMMA_0}
	\ \MIMPLIES \ 
	(
	\LTCDERIVEDVALUE{M}{\GAMMA_0}{\MMM}{V}
	\ \MIFF \
	\LTCDERIVEDVALUE{M}{\GAMMA_0}{\MMM \cdot \TCV:\GAMMA\REMOVETCVfrom}{V}
	)
	\]
	\begin{proof}
	This holds simply by Def $\LTCDERIVEDVALUE{}{}{}{}$ and the facts that $\TCV$ cannot appear in $\GAMMA_0$ and $\AN{\MMM} \equiv \AN{\MMM \cdot \TCV:\GAMMA\REMOVETCVfrom}$ and $M\MMM \equiv M (\MMM \cdot \TCV:\GAMMA\REMOVETCVfrom) $ and the possibly derivable terms are equivalent.
	\end{proof}
\end{lemma}

\begin{lemma}[\DONE Evaluation under model extensions are equivalent]
	\label{lem:eval_under_extensions_are_equivalent}
	\[
	\begin{array}[t]{l}
		\Mforall \GAMMA_1, \MMM_1^{\GAMMA_1}, \GAMMA_2, \MMM_2^{\GAMMA_2}, \GAMMA_0, M, V. \
		\\ 
		\qquad 
		(
		\MMM_1 \EXTSTAR \MMM_2 
		\ \MAND \
		\TCTYPES{\GAMMA_1}{\GAMMA_0}
		\ \MAND \
		\AN{V} \cap \AN{\MMM_2} \subseteq \AN{\MMM_1}
		)
		\\
		\qquad 
		\MIMPLIES \
		(\LTCDERIVEDVALUE{M}{\GAMMA_0}{\MMM_1}{V}
		\ \MIFF \
		\LTCDERIVEDVALUE{M}{\GAMMA_0}{\MMM_2}{V})
	\end{array}
	\]
	
	\begin{proof}
		Essentially $M\MMM_1 \equiv M \MMM_2$ (Lemma \ref{lem:extensions_close_terms_equally}) and Lemma \ref{lem:adding/remove_unused_names_maintains_evaluation} prove this in both directions of the $\MIFF$.
		\begin{NDERIVATION}{1}
			\NLINE{\text{Assume: } \GAMMA_1, \GAMMA_2, \MMM_1^{\GAMMA_1}, \MMM_2^{\GAMMA_2}, \GAMMA_0 \text{ s.t}
			}{$\Mforall \GAMMA_1, \MMM_1^{\GAMMA_1}, \GAMMA_2, \MMM_2^{\GAMMA_2}, \GAMMA_0. $}
			\NLINE{
				\begin{array}{l}
					\MMM_1 \EXTSTAR \MMM_2 
					\ \MAND \ 
					\TCTYPES{\GAMMA_1}{\GAMMA_0}
					\ \MAND \ 
					\LTCDERIVEDVALUE{M}{\GAMMA_0}{\MMM_1}{V}
					\ \MAND \
					\AN{V} \cap \AN{\MMM_2} \subseteq \AN{\MMM_1}
				\end{array}
			}{Assume}
			\NLINE{
				\begin{array}[t]{ll}
					\MMM_1 \EXTSTAR \MMM_2 
					\ \MAND \ 
					\TCTYPES{\GAMMA_1}{\GAMMA_0}
					\\
					\MAND \ 
					\AN{M}= \emptyset 
					\ \MAND \
					\TYPES{\SEM{\GAMMA_0}{\MMM_1}}{M}{\alpha}
					\\
					\MAND \
					(\AN{\MMM_1}, M\MMM_1) \CONV (\AN{\MMM_1}, G', \ V)
					&
					\MIFF \
					(\AN{\MMM_1}, M\MMM_2) \CONV (\AN{\MMM_1}, G', \ V)
				\end{array}
			}{Lemma \ref{lem:extensions_close_terms_equally} $ \MIMPLIES M \MMM_1 \equiv M \MMM_2$}
			\NLINE{
				\begin{array}[t]{ll}
					\MMM_1 \EXTSTAR \MMM_2 
					\ \MAND \ 
					\TCTYPES{\GAMMA_1}{\GAMMA_0}
					\\
					\MAND \ 
					\AN{M}= \emptyset 
					\ \MAND \
					\TYPES{\SEM{\GAMMA_0}{\MMM_1}}{M}{\alpha}
					\\
					\MAND \
					(\AN{\MMM_1}, M\MMM_2) \CONV (\AN{\MMM_1}, G', \ V)
					&
					\MIFF \
					(\AN{\MMM_2}, M\MMM_2) \CONV (\AN{\MMM_2}, G', \ V)
				\end{array}
			}{
				\parbox[t]{4cm}{\raggedleft
					$\TYPES{\SEM{\GAMMA_0}{\MMM_1}}{M}{\alpha}$
					\\
					$ \MIMPLIES \AN{M \MMM_1} \subseteq \AN{\MMM_1} \subseteq \AN{\MMM_2}$ 
					\\
					Lemma \ref{lem:adding/remove_unused_names_maintains_evaluation} $\AN{V} \cap \AN{\MMM_2} \subseteq \AN{\MMM}$
				}
			}
			\NLASTLINE{\text{hence: }
				\begin{array}[t]{l}
					\Mforall \GAMMA_1, \MMM_1^{\GAMMA_1}, \GAMMA_2, \MMM_2^{\GAMMA_2}, \GAMMA_0, M. \
					\\ 
					\qquad 
					(\MMM_1 \EXTSTAR \MMM_2 
					\ \MAND \
					\TCTYPES{\GAMMA_1}{\GAMMA_0})
					\\
					\qquad 
					\MIMPLIES \
					(\LTCDERIVEDVALUE{M}{\GAMMA_0}{\MMM_1}{V}
					\ \MIFF \
					\LTCDERIVEDVALUE{M}{\GAMMA_0}{\MMM_2}{V})
				\end{array}
			}{Lemma \ref{lem:extensions_give_same_semantics_for_type_contexts}, $\SEM{\GAMMA_0}{\MMM_1} \equiv \SEM{\GAMMA_0}{\MMM_2} $}
		\end{NDERIVATION}
	\end{proof}
\end{lemma}

The next few lemmas will be defined for types which are $\NAME$-free, these represent the STLC types although the key is there are \NUC \ programs that are of this $\NAME$-free type but do contain names e.g. $\lambda x. (r =r)$. We show that these $\NAME$-free typed terms can be equated to a name-free term i.e. a STLC term. One of the key assumptions to make the proof simple is that the initial STLC is simple enough to not contain infinite values of a single type i.e. there are no integers and no recursion, however proving these harder extensions would be harder.
To prove a finite number of values of any particular $\NAME$-free type (up to equivalence), we first prove that it is possible to define when two functions of $\NAME$-free type are equal, this is essentially through comparing each of the finite values of the input type.

\newcommand{\EQFS}[3]{\mathsf{EQ}^{#1}(#2,#3)}
\begin{definition}
	Inductively define $\EQFS{\alpha}{M}{N}$ as the program that equates two functions of type $\alpha$ as follows:
	
	\[
	\begin{array}{rcl}
		\EQFS{\UNIT}{M}{N} & \DEFEQ & \TRUE
		\\
		\EQFS{\BOOL}{M}{N} & \DEFEQ & M=N
		\\
		\EQFS{\alpha \FS \UNIT}{M}{N} & \DEFEQ & \TRUE
		\\
		\EQFS{\alpha_1 \times \alpha_2}{M}{N} & \DEFEQ &
			\begin{array}[t]{l} 
				\IFTHENELSE{\neg \EQFS{\alpha_1}{\pi_1(M)}{\pi_1(N)}}{\FALSE}{}
				\\
				\IFTHENELSE{\neg \EQFS{\alpha_2}{\pi_2(M)}{\pi_2(N)}}{\FALSE}{\TRUE}
			\end{array}
		\\
		\EQFS{\alpha_1 \FS \alpha_2}{M}{N} & \DEFEQ &
			\begin{array}[t]{l} 
			\IFTHENELSE{\neg \EQFS{\alpha_2}{M\VEC{V}_0}{N\VEC{V}_0}}{\FALSE}{}
			\\
			\IFTHENELSE{\neg \EQFS{\alpha_2}{M\VEC{V}_1}{N\VEC{V}_1}}{\FALSE}{}
			\\
			...
			\\
			\IFTHENELSE{\neg \EQFS{\alpha_2}{M\VEC{V}_k}{N\VEC{V}_k}}{\FALSE}{\TRUE}
			\\
			\text{where $\VEC{V}$ is the set of all finite-$k$ number of values of type $\alpha_1$ (Lemma \ref{lem:STLC_finite_values})}
			\end{array}
	\end{array}
	\]
	
\end{definition}

In the following lemmas we will write $M_{\nu}$ for a standard \NUC \ term and $M_{\lambda}$ for a \NUC \ term constructed only from STLC terms (i.e. no names and no $\GENSYM$), we also ignore $\LET{x}{M}{N}$ for simplicity.
\newcommand{\NUTERM}{$\nu$-term}
\newcommand{\STLCTERM}{$\lambda$-term}

\begin{lemma}[There are finite STLC values for each type]
	\label{lem:STLC_finite_values}
	
	\[
	\Mforall \alpha. \Mexists \VEC{W}_{\lambda\text{finite}}^{\alpha}. 
	\Mforall M_{\lambda}^{\alpha}. \Mexists V^{\alpha}_{\lambda} \in \VEC{W}_{\lambda}. M_{\lambda} \CONG V_{\lambda}
	\]
	\begin{proof}
		By induction on the structure of $\alpha$ we create a complete list $\VEC{W}^{\alpha} \equiv V[\alpha]$ as follows:
		\begin{itemize}
				\item $\alpha \equiv \UNIT$ then this clearly holds as $(G, \ M) \CONV (G', \ ())$ must always hold and so $V[\UNIT] \equiv ()$ holds.
				\item $\alpha \equiv \BOOL$ then this clearly holds as $(G, \ M) \CONV (G', \ \TRUE)$ must always hold or $(G, \ M) \CONV (G', \ \FALSE)$ must always hold and so $V[\BOOL] \equiv \TRUE \VERTICAL \FALSE$ holds.
				\item $\alpha \equiv \alpha_1 \times \alpha_2$ then clearly $V[\alpha _1 \times \alpha_1] \equiv \PAIR{V[\alpha_1]}{V[\alpha_2]}$ where in the RHS $V[\alpha_1]$ represents every possible value of type $\alpha_1$, 
				hence the RHS is the list of every possible combination between $V[\alpha_1]$ and $V[\alpha_2]$.
				
				\item $\alpha \equiv \alpha_1 \FS \alpha_2$ then by induction on $\alpha_1$ we can assume there are finite values of this type i.e. let $ \VEC{W} \equiv V[\alpha_1]$, we use this to state the values of type $\alpha_1 \FS \alpha_2$ as 
				\[
				V[\alpha_1 \FS \alpha_2] 
				\equiv 
				\lambda x^{\alpha_1}.
				\begin{array}[t]{l}
					\IFTHENELSE{\EQFS{\alpha_1}{x}{\VEC{W}_0}}{V[\alpha_2]}{}
					\\
					\IFTHENELSE{\EQFS{\alpha_1}{x}{\VEC{W}_1}}{V[\alpha_2]}{}
					\\
					...
					\\
					\IFTHENELSE{\EQFS{\alpha_1}{x}{\VEC{W}_k}}{V[\alpha_2]}{V[\alpha_2]}
				\end{array}
				\]
					hence the RHS is the list of every possible combination between $\VEC{W}$, and $V[\alpha_2]$ in each instance.

			\end{itemize}
		
		The number of values grow exponentially with the size of the type but there are always finite number of values for each type.
		These values cover all possible cases by definition as no other possible inputs or outputs exist to a function of the given type.
	\end{proof}
\end{lemma}

\begin{lemma}[$\NAME$-free terms are equivalent to a name free STLC term]
	\label{lem:Nm-free_terms_have_equivalent_name_free_STLC-term}
	If $\alpha$ is $\NAME$-free.
	\[
	\Mforall M_{\nu}^{\alpha}. \Mexists N_{\lambda}. M_{\nu} \CONGCONTEXT{\alpha}{\AN{M_{\nu}}} N_{\lambda}
	\]
I.e. for each term in the \NUC \ of a type which is $\NAME$-free then there exists an equivalent term constructed only of core STLC terms.
	\begin{proof}
		It is clear that all \NUTERM \ that are of type $\UNIT$ and $\BOOL$ are equivalent to a constants of that type.
		\\
		The case for $\alpha \equiv \alpha_1 \times \alpha_2$ holds trivially by induction on $\pi_1(M_{\nu})^{\alpha_1}$ and $\pi_2(M_{\nu})^{\alpha_2}$.
		The case for $\alpha \equiv \alpha_1 \FS \alpha_2$ holds as follows:
		\begin{NDERIVATION}{1}
			\NLINE{IH(\alpha) \DEFEQ \Mforall M^{\alpha}. \Mexists N_{\lambda}. \ M \CONGCONTEXT{\alpha}{\AN{M}} N }{$IH(\alpha)$}
			\NLINE{\text{Assume: } IH(\alpha_1)  \ \MAND IH(\alpha_2)}{ for $\alpha_1$, $\alpha_2$ being $\NAME$-free}
			\NLINE{\text{Prove: } IH(\alpha_1 \FS \alpha_2)}{ for $\alpha_1 \FS \alpha_2$ being $\NAME$-free}
			\NLINE{\Mforall \alpha. \Mexists \VEC{W}_{\text{finite}}^{\alpha}. 
				\Mforall M_{\lambda}^{\alpha}. \Mexists V^{\alpha} \in \VEC{W}. M \CONG V }{$\alpha$ is $\NAME$-free Lemma \ref{lem:STLC_finite_values}}
			\NLINE{\parbox[t]{10cm}{$IH(\alpha_1)$ implies for each \NUTERM \ of type $\alpha_1$ there exists an equivalent \STLCTERM (which must mean there are finite ones of these) which we call $\VEC{W}$}}{$IH(\alpha_1)$}
			\NLINE{\Mforall M_{1\nu}^{\alpha_1}. \Mexists N_{1\lambda}^{\alpha_1}. \ M_{1\nu} \CONGCONTEXT{\alpha_1}{\AN{M_{1\nu}}} N_{1\lambda}}{$IH(\alpha_1)$}
			\NLINE{\parbox[t]{10cm}{For each value in $\VEC{W}$, then $M \VEC{W}_i$ is a term  of type $\alpha_2$ which by $IH(\alpha_2)$ implies there is an equivalent \STLCTERM \ we call $\VEC{U}_i$ i.e. $M\VEC{V}_i \CONG \VEC{U}_i$}}{$IH(\alpha_2)$}
			\NLINE{\Mforall \VEC{W}_i^{\alpha_1} \in \VEC{W}^{\alpha_1}. \ \Mforall M^{\alpha_1 \FS \alpha_2}_{\nu}. \Mexists N_{\lambda}^{\alpha_2}. \ M\VEC{W}_i \CONGCONTEXT{\alpha_2}{\AN{M}} N }{$IH(\alpha_2)$}
			\NLINE{\parbox[t]{10cm}{
				Using lines 5-8 we can build an equivalent formula to $M$ by brute force such that for each input case we have the equivalent output case i.e. 
				\\
				$N_{\lambda} \equiv \lambda x^{\alpha_1}. \begin{array}[t]{l}
					\IFTHENELSE{\EQFS{\alpha_1}{x}{\VEC{V}_0}}{\VEC{U}_0}{}
					\\
					\IFTHENELSE{\EQFS{\alpha_1}{x}{\VEC{V}_1}}{\VEC{U}_1}{}
					\\
					...
					\\
					\IFTHENELSE{\EQFS{\alpha_1}{x}{\VEC{V}_{k-1}}}{\VEC{U}_{k-1}}{\VEC{U}_k}
					\end{array}$
			}}{}
			\NLASTLINE{\parbox[t]{10cm}{
				By definition $M_{\nu} \CONGCONTEXT{\alpha1 \FS \alpha_2}{\AN{M_{\nu}}} N_{\lambda}$ as any use of $N_{\lambda}$  behaves identically to $M_{\nu}$ in any application it is used in.
			}}{
			}
		\end{NDERIVATION}
		
	\end{proof}

\end{lemma}

\input{appendix_soundness_lemmas_Name-Free}
\begin{lemma}[\DONE Functions that map to base types cannot reveal names in that function]
	\label{lem:Base-result-functions_cannot_reveal_their_names}
	\[
	\begin{array}{l}
	\Mforall M^{\NAME}, V^{(\alpha \FS \TYBASE)}, r_x^{\NAME}. (G \equiv \AN{M}\cup\AN{V}) 
	\\ 
	\left(
	\begin{array}{l}
	r_x \notin \AN{M} 
	\\ 
	\MAND \ \TYPES{f:\alpha \FS \TYBASE}{M}{\NAME}
	\\ 
	\MAND \ r_x \in \AN{V}  
	\\ 
	\MAND \ \TYPES{\emptyset}{V}{(\alpha \FS \TYBASE)}
	\end{array}
	\right) 
	\
	\MIMPLIES \ \neg (G, \ M\PSUBST{V}{f}) \CONV (G, G', \ r_x)
	\end{array}
	\]
	\begin{proof}
		Given the semantics, $\CONV \ \equiv \ \RED^{*}$.
		\\
		Given $r_x \notin \AN{M}$ then it is clear $M \neq r_x$ and given $\TYPES{\emptyset}{V}{(\alpha \FS \TYBASE)}$ then it is clear $V \neq r_x$, hence $r_x$ must be derived from terms $M$ and $V$.
		Assume there exists at least one such $M$ fro which this holds, then take the smallest one $M$, then by definition there exists an $M_k$ such that 
		\\
		$(G, \ M\PSUBST{V}{f}) \CONV (G, G', \ r_x) \ \MIFF \ (G, \ M\PSUBST{V}{f}) \CONV  (G, G', \ M_k) \RED (G, G', \ r_x)$ 
		\\
		then the following must hold:
		\begin{itemize}
			\item[$M_k \equiv \pi_1(\PAIR{r_x}{V'})$] hence $M\PSUBST{V}{f} \equiv \EEE{\pi_1(\PAIR{M_1}{M_2})} \PSUBST{V}{f}$ where $\EEE{M_2}\PSUBST{V}{f} \CONV r_x$ so this $\EEE{M_2}$ is smaller than $M$ hence contradiction.
			\item[$M_k \equiv \pi_2(\PAIR{V'}{r_x})$] hence $M\PSUBST{V}{f} \equiv \EEE{\pi_2(\PAIR{M_1}{M_2})}\PSUBST{V}{f}$ where $\EEE{M_2}\PSUBST{V}{f} \CONV r_x$ so this $\EEE{M_2}$ is smaller than $M$ hence contradiction.
			\item[$M_k \equiv \IFTHENELSE{\TRUE}{r_x}{V'}$] hence $M\PSUBST{V}{f} \equiv \EEE{\IFTHENELSE{M_b}{M_1}{M_2}}\PSUBST{V}{f}$ where $\EEE{M_1}\PSUBST{V}{f} \CONV r_x$ so this $\EEE{M_1}$ is smaller than $M$ hence contradiction.
			\item[$M_k \equiv \IFTHENELSE{\FALSE}{V'}{r_x}$] hence $M\PSUBST{V}{f} \equiv \EEE{\IFTHENELSE{M_b}{M_1}{M_2}}\PSUBST{V}{f}$ where $\EEE{M_2}\PSUBST{V}{f} \CONV r_x$ so this $\EEE{M_2}$ is smaller than $M$ hence contradiction.
			\item[$M_k \equiv (\lambda a. r_x)V'$] hence $M\PSUBST{V}{f} \equiv \EEE{(\lambda a. M_1)M_2}\PSUBST{V}{f}$ where $\EEE{M_1}\PSUBST{V}{f} \CONV r_x$ so this $C[M_1]$ is smaller than $M$ hence contradiction.
			\item[$M_k \equiv (\lambda a. a)r_x$] hence $M\PSUBST{V}{f} \equiv \EEE{(\lambda a. M_1)M_2}\PSUBST{V}{f}$ where $\EEE{M_2}\PSUBST{V}{f} \CONV r_x$ so this $C[M_2]$ is smaller than $M$ hence contradiction.
			\item[$M_k \equiv \LET{x}{V'}{M}$] Similar for $\LET{x}{V'}{M}$
			\item There are no other possible terms that reduce to $r_x$.
			\item[$M_k \equiv \GENSYM()$] fails to produce $r_x$ as $r_x \in G$
		\end{itemize}
		i.e. Assuming a smallest $M$ s.t. $M \CONV r_x$ (which is assumed to exist), for each term $M_k$ which is just one $\RED$-step away from $r_x$ then it can be proven that a smaller $M$ could be produced hence contradiction.
	\end{proof}
\end{lemma}

\begin{lemma}[Adding a name to the model but not to the context means it is fresh]
	\label{lem:Gamma_derived_name_and_not_derivable_from_model_plus_name_implies_fresh_name}
	\[
	\Mforall \GAMMA, \MMM^{\GAMMA}, r_x. \ (\Mexists M_x. \LTCDERIVEDVALUE{M_x}{\GAMMA}{\MMM}{r_x} \ \MAND \ \neg \Mexists N_x. \LTCDERIVEDVALUE{N_x}{\GAMMA }{\MMM \cdot x:r_x}{r_x}) \ \MIMPLIES \ r_x \notin \AN{\MMM}
	\]
	\begin{proof}
		Assume $r_x \in \AN{\MMM}$ then clearly there would be a direct contradiction in the assumption as $\Mexists N_x. \LTCDERIVEDVALUE{N_x}{\GAMMA }{\MMM \cdot x:r_x}{r_x}$ as $\AN{\MMM \cdot x:r_x} \equiv \AN{\MMM}$.
		Hence $r_x \notin \AN{\MMM}$.
	\end{proof}
\end{lemma}

\begin{lemma}[\DONE LTC derived terms cannot reveal old names]
	\label{lem:LTC_derived_value_cannot_reveal_old_names}
	\[ 
	\Mforall \GAMMA, \MMM^{\GAMMA}, \GAMMA_0, M_y^{\alpha_y}, s^{\NAME}. \
	\left(
	\begin{array}{l}
	\TCTYPES{\GAMMA}{\GAMMA_0} 
	\\ \MAND \
	\LTCDERIVEDVALUE{M_y}{\GAMMA_0}{\MMM}{V_y}
	\\ \MAND \  
	s \in  \AN{\MMM} 
	\end{array}
	\right)
	\ \MIMPLIES \ 
	\begin{array}{l}
	\Mexists M_1. \LTCDERIVEDVALUE{M_1}{\GAMMA_0 \PLUSV y}{\MMM \cdot y:V_y}{s} 
	\\ \MIFF \\
	\Mexists M_0. \LTCDERIVEDVALUE{M_0}{\GAMMA_0}{\MMM}{s}
	\end{array}
	\]
		Essentially $V_y$ cannot reveal any names in $\MMM$ that are not already available to $\GAMMA_0$.
		\begin{proof}
			Proof by contradiction: 
			\\
			assume some $\GAMMA, \MMM^{\GAMMA}, \GAMMA_0, M_y^{\alpha_y}, s^{\NAME}$
			with 
			$\TCTYPES{\GAMMA}{\GAMMA_0} 
			\ \MAND \
			\LTCDERIVEDVALUE{M_y}{\GAMMA_0}{\MMM}{V_y}
			\ \MAND \  
			s \in  \AN{\MMM} $
			then show
			\\
			the following fails
			$
			\neg \Mexists M_1. \LTCDERIVEDVALUE{M_1}{\GAMMA_0 \PLUSV y}{\MMM \cdot y:V_y}{s} 
			\ \MAND \ 
			\Mexists M_0. \LTCDERIVEDVALUE{M_0}{\GAMMA_0}{\MMM}{s}$
			but this creates a contradiction simply by Def $\LTCDERIVEDVALUE{}{}{}{}$.
			\\
			and the following fails
			$\Mexists s. 
			\Mexists M_1. \LTCDERIVEDVALUE{M_1}{\GAMMA_0 \PLUSV y}{\MMM\cdot y:V_y}{s}
			\ \MAND \ 
			\neg \Mexists M_0. \LTCDERIVEDVALUE{M_0}{\GAMMA_0}{\MMM}{s}$
			\\
			iff \
			$\Mexists s.
			\begin{array}[t]{l} 
			\Mexists M_1. 
			\begin{array}[t]{l}
			\TYPES{\SEM{\GAMMA_0 \PLUSV y}{\MMM \cdot y:V_y}}{M_1}{\NAME} 
			\\
			\MAND \ \AN{M_1}=\emptyset 
			\\
			\MAND \ (\AN{\MMM \cdot y:V_y}, \ M_1\MMM \cdot y:V_y) \CONV (\AN{\MMM \cdot y:V_y}, G', \ s) 
			\end{array}
			\\ 
			\MAND \ 
			\neg \Mexists M_0. 
			\begin{array}[t]{l}
			\TYPES{\SEM{\GAMMA_0}{\MMM}}{M_0}{\NAME} 
			\\
			\MAND \ \AN{M_0}=\emptyset 
			\\
			\MAND \ (\AN{\MMM}, \ M_1\MMM) \CONV (\AN{\MMM}, G', \ s) 
			\end{array}
			\end{array}
			$
			\\
			iff \
			$\Mexists s.
			\begin{array}[t]{l} 
			\Mexists M_1(y). 
			\begin{array}[t]{l}
			\TYPES{\SEM{\GAMMA_0 \PLUSV y}{\MMM \cdot y:V_y}}{M_1(y)}{\NAME} 
			\\
			\MAND \ \AN{M_1(y)}=\emptyset 
			\\
			\MAND \ (\AN{\MMM \cdot y:V_y}, \ M_1(y)\MMM \cdot y:V_y) \CONV (\AN{\MMM \cdot y:V_y}, G', \ s) 
			\end{array}
			\\ 
			\MAND \ 
			\neg \Mexists M_0. 
			\begin{array}[t]{l}
			\TYPES{\SEM{\GAMMA_0}{\MMM}}{M_0}{\NAME} 
			\\
			\MAND \ \AN{M_0}=\emptyset 
			\\
			\MAND \ (\AN{\MMM}, \ M_1\MMM) \CONV (\AN{\MMM}, G', \ s) 
			\end{array}
			\end{array}
			$
			\\
			iff \
			$\Mexists s.
			\begin{array}[t]{l} 
			\Mexists M_1(y). 
			\begin{array}[t]{l}
			\TYPES{\SEM{\GAMMA_0 \PLUSV y}{\MMM \cdot y:V_y}}{M_1(y)}{\NAME} 
			\\
			\MAND \ \AN{M_1(y)}=\emptyset 
			\\
			\MAND \ (\AN{\MMM \cdot y:V_y}, \ M_1(y)\MMM \cdot y:V_y) \CONV (\AN{\MMM \cdot y:V_y}, G', \ s) 
			\end{array}
			\\ 
			\MAND \ 
			\neg \Mexists \LET{y}{M_y}{M_1(y)}. 
			\begin{array}[t]{l}
			\TYPES{\SEM{\GAMMA_0}{\MMM}}{\LET{y}{M_y}{M_1(y)}}{\NAME} 
			\\
			\MAND \ \AN{\LET{y}{M_y}{M_1(y)}}=\emptyset 
			\\
			\MAND \ (\AN{\MMM}, \ (\LET{y}{M_y}{M_1(y)})\MMM) \CONV (\AN{\MMM}, G', \ s) 
			\end{array}
			\end{array}
			$
			\\
			hence contradiction as the name derived by $M_1(y)$ in the first evaluation, this same name can also be derived by $\LET{y}{M_y}{M_1(y)}$ in the second evaluation by the semantics of the evaluation of $\LET{y}{M_y}{M_1(y)}\MMM \RED M_1(V_y)\MMM \equiv M_1(y) (\MMM \cdot y:V_y)$. Where $M_y$ reduces to the term equal to $V_y$ except the fresh names which have no affect on the name produced by $M_0$ and $M_1$ hence contradiction.
			\\
			This works because $s \in \AN{\MMM}$ hence the name cannot be fresh (derived from $V_y$ and any fresh names in $V_y$ can be replicated by a new generation via $M_y$ and will be treated equally. 
		\end{proof}
\end{lemma}

\begin{lemma}[\DONE Extensions cannot reveal old names]
	\label{lem:extensions_cannot_reveal_old_names}
	\[ 
	\Mforall \GAMMA, \MMM^{\GAMMA}, \GAMMA', \MMM^{\GAMMA'}, s^{\NAME}. \
	\left(
	\begin{array}{l}
	\MMM \EXTSTAR \MMM'
	\\ \MAND \  
	s \in  \AN{\MMM} 
	\end{array}
	\right)
	\ \MIMPLIES \ 
	\begin{array}{l}
	\Mexists M'. \LTCDERIVEDVALUE{M'}{\GAMMA'}{\MMM'}{s} 
	\\ \MIFF \\
	\Mexists M. \LTCDERIVEDVALUE{M}{\GAMMA}{\MMM}{s}
	\end{array}
	\]
	Essentially $\MMM'$ cannot reveal any names in $\MMM$ that are not already available to $\MMM$.
	\begin{proof}
		$\leftarrow:$ clearly holds
		\\
		$\rightarrow:$ 
		Proof by induction on the structure of $\MMM'$: 
		\\
		\begin{itemize}
			\item[$\MMM' \equiv \MMM$] then this clearly holds.
			\item[$\MMM' \equiv \MMM'_0 \cdot \TCV:\GAMMA_0\REMOVETCVfrom$] then this holds as no new names are reachable by $\TCV$ and IH on $\MMM'_0$.
			\item[$\MMM' \equiv \MMM_0^{\GAMMA_0'} \cdot y:V_y$] then $\Mexists M_y. \LTCDERIVEDVALUE{M_y}{\GAMMA_0'}{\MMM_0'}{V_y}$ and using Lemma \ref{lem:LTC_derived_value_cannot_reveal_old_names} and IH on $\MMM'_0$ this clearly holds.
	\end{itemize}
	\end{proof}
\end{lemma}

In the next few lemmas we treat expressions as potential terms as each expression, $c$, $x$, $\pi_i(e)$, and $\PAIR{e}{e}$ is also a method of constructing a term and hence $\SEM{e}{\MMM} \equiv e\MMM$ (closure) implies we can treat $e$ as a term even though it is an expression, i.e. we write $e$ even though we mean ``the term constructed using the equivalent expression-constructors used to construct the expression $e$''.

\begin{lemma}[\DONE Expressions cannot create new names]
	\label{lem:expressions_cannot_create_new_names}	
	\[
	\Mforall \GAMMA, \MMM^{\GAMMA}, e. \ \EXPRESSIONTYPES{\GAMMA}{e}{\alpha} \ \MIMPLIES \ \Mexists V. \ (\AN{\MMM}, \ \SEM{e}{\MMM}) \CONV (\AN{\MMM}, \ V)
	\]
	Clearly guaranteed termination implies $\Mexists V'. (\AN{\MMM},\  \SEM{e}{\MMM}) \CONV (\AN{\MMM}, G', \ V')$, however this Lemma proves that no new names are produced in such an evaluation.
	\\
	\begin{proof}
		By induction on structure of $e \in \{c,x,\pi_i(e), \PAIR{e}{e} \}$
		\\
		let $G= \AN{\MMM}$
		\\
		\begin{itemize}
			\item $e= c$ constants  then clearly $(G, \SEM{e}{\MMM}) \Mequiv (G, c) \CONV (G, c)$
			\item $e= x$ clearly $(G, \SEM{e}{\MMM}) \Mequiv (G, \MMM(x)) \CONV (G, \MMM(x))$ as $\MMM(x)$ is a value by definition.
			\item $e= \pi_i(e')$ 
			\\
			clearly $(G, \SEM{\pi_i(e')}{\MMM}) \Mequiv (G, \pi_i(\SEM{e'}{\MMM}))$ 
			\\
			and by IH on $e'$ then  $\Mexists V'. (G, \SEM{e'}{\MMM}) \CONV (G, V')$ with $V':\alpha \times \beta$ 
			\\
			hence $V'\Mequiv \PAIR{V_1}{V_2}$
			\\
			hence $\pi_i(V_i)$ is a value and hence $(G, \SEM{\pi_i(e)}{\MMM}) \Mequiv (G, \pi_i(\SEM{e}{\MMM})) \CONV (G, V_i)$
			\item $e= \PAIR{e_1}{e_2}$
			\\
			by IH on $e_1$ and $e_2$:
			\\
			$\Mexists V_1.(G, \SEM{e_1}{\MMM}) \CONV (G, V_1)$
			\\
			and $\Mexists V_2.(G, \SEM{e_2}{\MMM}) \CONV (G, V_2)$
			\\
			hence from operational semantics: 
			\\
			$\Mexists V \ (= \PAIR{V_1}{V_2}).(G, \SEM{\PAIR{e_1}{e_2}}{\MMM}) \Mequiv (G, \PAIR{\SEM{e_1}{\MMM}}{\SEM{e_2}{\MMM}}) \CONV (G, V)$
		\end{itemize}
	\end{proof}
\end{lemma}

\begin{lemma}[\DONE Expressions are name free]
	\label{lem:expressions_are_name_free}
	\[
		\Mforall e. \ \EXPRESSIONTYPES{\GAMMA}{e}{\alpha} \ \MIMPLIES \ \AN{e}=\emptyset
	\]
	\begin{proof}
		By IH on the structure of $e$ knowing that $e \in \{c, x,\pi_i(e), \PAIR{e_1}{e_2} \}$, i.e. no built in names, names only come from closing with a model via the semantics of $\SEM{e}{\MMM}$.
	\end{proof}
\end{lemma}

\begin{lemma}[\DONE Expressions are fresh name free]
	\label{lem:expressions_results_are_fresh_name_free}
	\[
	\Mforall e. \ \EXPRESSIONTYPES{\GAMMA}{e}{\alpha} \ \MAND \ \LTCDERIVEDVALUE{e}{\GAMMA}{\MMM}{V_e} \ \MIMPLIES \ \AN{V_e}\subseteq \AN{\MMM}
	\]
	\begin{proof}
		By IH on the structure of $e$ knowing that $e \in \{c, x,\pi_i(e), \PAIR{e_1}{e_2} \}$, i.e. no built in names, names only come from closing with a model via the semantics of $\SEM{e}{\MMM}$, and hence no evaluation of $\GENSYM()$ can occur, just producing old names.
	\end{proof}	
\end{lemma}

\begin{lemma}[\DONE Expressions are congruent to their evaluation]
	\label{lem:expression_Cong_evaluation}
	\[
	\Mforall \GAMMA, \MMM^{\GAMMA}, e. \EXPRESSIONTYPES{\GAMMA}{e}{\alpha} \ \MAND \ \LTCDERIVEDVALUE{e}{\GAMMA}{\MMM}{V_e} \MIMPLIES e\MMM \CONGCONTEXT{\alpha}{\AN{\MMM}} V_e
	\]
	
	\begin{proof}
		By Lemma \ref{lem:expressions_results_are_fresh_name_free} we know all values contain names in the model.
		\\
		By induction on the structure of $e$:
		\begin{itemize}
			\item[$e\equiv c$] clearly holds.
			\item[$e\equiv x$] clearly holds as $x\MMM \equiv \MMM(x) \equiv V_e$
			\item[$e\equiv \pi_i(e')$] holds as by induction $e'$, 
			\\
			i.e. $\EXPRESSIONTYPES{\GAMMA}{e'}{\alpha_1 \times \alpha_2} \ \MAND \ \LTCDERIVEDVALUE{e'}{\GAMMA}{\MMM}{\PAIR{V_1}{V_2}} \MIMPLIES e'\MMM \CONGCONTEXT{\alpha}{\AN{\MMM}} \PAIR{V_1}{V_2}$
			\\
			hence  $\pi_i(e')\MMM \CONGCONTEXT{\alpha}{\AN{\MMM}} V_i$
			\item[$e\equiv \PAIR{e_1}{e_2}$] holds by IH on $e_1$ and $e_2$ 
			\\
			i.e. $\EXPRESSIONTYPES{\GAMMA}{e_1}{\alpha_1} \ \MAND \ \LTCDERIVEDVALUE{e_1}{\GAMMA}{\MMM}{V_1} \MIMPLIES e_1\MMM \CONGCONTEXT{\alpha}{\AN{\MMM}} V_1$
			\\
			and
			$\EXPRESSIONTYPES{\GAMMA}{e_2}{\alpha_2} \ \MAND \ \LTCDERIVEDVALUE{e_2}{\GAMMA}{\MMM}{V_2} \MIMPLIES e_2\MMM \CONGCONTEXT{\alpha}{\AN{\MMM}} V_2$
			\\
			implies $\PAIR{e_1}{e_2}\MMM \CONGCONTEXT{\alpha}{\AN{\MMM}} \PAIR{V_1}{V_2}$
		\end{itemize}
	\end{proof}
	
\end{lemma}

\begin{lemma}[\DONE Names fresh in term imply name fresh in value]
	\label{lem:Name_fresh_pre_eval_implies_name_fresh_post_eval}
	\[
	\begin{array}{l}
	\Mforall \GAMMA,  \MMM^{\GAMMA}, M, r. 
	\
	r \notin \AN{M}
	\
	\MAND
	\
	(\AN{M},r,G_0, \ M) \CONV (\AN{\MMM},r,G_0,G', \ V) 
	\ \MIMPLIES \
	r \notin \AN{V}
	\end{array}
	\]
	
	\begin{proof}
		By induction on the structure of $M$, most cases are trivial, the only non-trivial cases are:
		\\
		$M\equiv M_1 M_2$:
		\\
		clearly $M_1 \CONV V_1$ and $M_2 \CONV V_2$ with $r \notin V_1 V_2$ (by IH) hence 2 cases occur:
		\\
		$V_1 \equiv \lambda x. M_1'$ then this evaluates to $M_1'\PSUBST{V_2}{x}$ which by assumptions and Op semantics $r\notin M_1'\PSUBST{V_2}{x}$
		\\
		$V_1 \equiv \GENSYM$ and $V_2 \equiv ()$ then by Op. Sem. this generates a fresh name ($ \neq r$).
	\end{proof}
\end{lemma}

\begin{lemma}[\DONE Fresh names are underivable from model plus that name]
	\label{lem:fresh_name_underivable_from_model+name}
	\[
	\Mforall \GAMMA, \MMM^{\GAMMA}, r. \
	r \notin \AN{\MMM}
	\ \MIMPLIES \
	\neg \Mexists M^{\NAME}. \LTCDERIVEDVALUE{M}{\GAMMA}{\MMM \cdot m:r}{r}
	\]
	\begin{proof}
	This holds as $r$ must be a freshly generated name not appearing in $\MMM$ hence no term can be derived from the model $\MMM$ to generate such a name.
	\\
	\begin{NDERIVATION}{1}
		\NLINE{\text{Assume: $\GAMMA$, $\MMM^{\GAMMA}$ s.t.}}{}
		\NLINE{\parbox{11cm}{prove by contradiction: \\ i.e. assume $\neg 
			(r \notin \AN{\MMM}
			\ \MIMPLIES \
			\neg \Mexists M^{\NAME}. \LTCDERIVEDVALUE{M}{\GAMMA}{\MMM \cdot m:r}{r})$ prove contradiction}
		}{}
		\NLINE{\MIFF \ 
			(r \notin \AN{\MMM}
			\ \MAND \
			\Mexists M^{\NAME}. \LTCDERIVEDVALUE{M}{\GAMMA}{\MMM \cdot m:r}{r})
		}{FOL}
		\NLINE{\MIFF \ 
			r \notin \AN{\MMM}
			\ \MAND \
			\Mexists M^{\NAME}. 
			\begin{array}[t]{l}
				\AN{M}=\emptyset
				\
				\MAND \  
				\TYPES{\GAMMA}{M}{\alpha} 
				\\
				\MAND  (\AN{\MMM}, r, \  M(\MMM\cdot m:r)) \CONV (\AN{\MMM},r,G', \ r) 
			\end{array}
		}{Sem. $\LTCDERIVEDVALUE{}{}{}{}$}
		\NLINE{\MIFF \ 
			r \notin \AN{\MMM}
			\ \MAND \
			\Mexists M^{\NAME}. 
			\begin{array}[t]{l}
				\AN{M}=\emptyset
				\
				\MAND \  
				\TYPES{\GAMMA}{M}{\alpha} 
				\\
				\MAND  (\AN{\MMM}, r, \  M\MMM) \CONV (\AN{\MMM},r,G', \ r) 
			\end{array}
		}{$m\notin\DOM{\GAMMA} \ \MAND \ \TYPES{\GAMMA}{M}{\alpha}$}
		\NPLINE{\MIFF \ 
			r \notin \AN{\MMM}
			\ \MAND \
			\Mexists M^{\NAME}. 
			\begin{array}[t]{l}
				\AN{M}=\emptyset
				\
				\MAND \  
				\TYPES{\GAMMA}{M}{\alpha} 
				\\
				\MAND \ r \notin \AN{M\MMM}
				\\
				\MAND  (\AN{\MMM}, r, \  M\MMM) \CONV (\AN{\MMM},r,G', \ r) 
			\end{array}
		}{4cm}{$r \notin \AN{M\MMM}$ and \\ Lemma \ref{lem:Name_fresh_pre_eval_implies_name_fresh_post_eval} \\ imply contradiction}
		\NLASTLINE{\MIFF \ r \notin \AN{\MMM} \ \MIMPLIES \ \neg \Mexists M^{\NAME}. \LTCDERIVEDVALUE{M}{\GAMMA}{\MMM\cdot m:r}{r}}{}
	\end{NDERIVATION}
	\end{proof}

\end{lemma}

\begin{lemma}[\DONE Extensions give equal semantics of expressions]
	\label{lem:semantics_expressions_equal_under_model_extensions}
	\[
	\Mforall \GAMMA_1, \MMM_1^{\GAMMA_1}, \GAMMA_2, \MMM_{2}^{\GAMMA_2}, e. \
	\MMM_1 \EXTSTAR \MMM_{2}
	\ \MAND \ 
	\EXPRESSIONTYPES{\GAMMA_1}{e}{\alpha}
	\ \MIMPLIES \ 
	\SEM{e}{\MMM_1} \equiv \SEM{e}{\MMM_{2}}
	\]
	\begin{proof}
		Prove by induction on the structure of $e$ using semantics of expressions:
		\begin{itemize}
			\item[$e \equiv c$:]
			$\SEM{c}{\MMM_1} \equiv c \equiv \SEM{c}{\MMM_2} $ 
			
			\item[$e  \equiv x$:] 
			$x\in \GAMMA_1$ implies (Def $\EXTSTAR$) $\SEM{x}{\MMM_1} \equiv \MMM_1(x) \equiv \MMM_{2}(x)\equiv\SEM{c}{\MMM_2} $
			
			\item[$e \equiv \RAWPAIR{e_1 , e_2}$:] by IH on both $e_1$ and $e_2$  this clearly holds.
			
			\item[$e \equiv \pi_i(e')$:] by IH on $e'$ this clearly holds.
			
		\end{itemize}
	\end{proof}
\end{lemma}

\begin{lemma}[\DONE Sub LTC's implies subtype-contexts]
	\label{lem:subLTC_implies_sub_typing_context}
	\[
	\Mforall \GAMMA, \GAMMA_0, \MMM^{\GAMMA}. \ \TCTYPES{\GAMMA}{\GAMMA_0} \ \MIMPLIES \ \SEM{\GAMMA_0}{\MMM} \subseteq \SEM{\GAMMA}{\MMM}
	\]
	Proof: clearly holds through induction on structure of $\GAMMA$ and $\GAMMA_0$ and using the typing rules in Fig.~\ref{figure_typing_formulae}.
	\HIDDEN{
		\begin{proof}
			By induction on the structure of the rules for $\TCTYPES{}{}$ with the Def $\SEM{\GAMMA}{\MMM}$.
			\\
			\begin{itemize}
				\item $\GAMMA_0 \equiv \emptyset$ then clearly this holds.
				\item $\GAMMA \equiv \GAMMA' \PLUSV x:\alpha$ and $\GAMMA_0 \equiv \GAMMA_0' \PLUSV x:\alpha$ then by the typing rules and IH this holds.
				\item $\GAMMA \equiv \GAMMA' \PLUSTC \TCV$ and $\GAMMA_0 \equiv \GAMMA_0' \PLUSTC \TCV$ then by the typing rules and IH this holds as \\
				$\SEM{\GAMMA_0' \PLUSTC\TCV}{\MMM} \equiv \SEM{\GAMMA_0'}{\MMM}  \cup \SEM{\MMM(\TCV)}{\MMM}  \subseteq \SEM{\GAMMA'}{\MMM}  \cup \SEM{\MMM(\TCV)}{\MMM}   \equiv \SEM{\GAMMA' \PLUSTC \TCV}{\MMM}$
				\item $\GAMMA \equiv \GAMMA_0$ then clearly this holds.
				\item $\GAMMA \equiv \GAMMA' \PLUSV x:V$ and $\TCTYPES{\GAMMA'}{\GAMMA_0}$ this holds from IH and as $\SEM{\GAMMA'}{\MMM} \subseteq \SEM{\GAMMA' \PLUSV x:V}{\MMM}$
				\item $\GAMMA \equiv \GAMMA' \PLUSTC \TCV$ and $\TCTYPES{\GAMMA'}{\GAMMA_0}$ this holds from IH and as $\SEM{\GAMMA'}{\MMM} \subseteq \SEM{\GAMMA' \PLUSTC \TCV}{\MMM} \equiv \SEM{\GAMMA'}{\MMM} \cup \MMM(\TCV)$
			\end{itemize}
		\end{proof}
	}
\end{lemma}

\begin{lemma}[\DONE Values derivable from subset of LTC is also derivable from superset of LTC]
	\label{lem:LTCDERIVED_subset_implies_LTCDERIVED_supset}
	\[
	\Mforall \GAMMA,\MMM^{\GAMMA}, \GAMMA_0 \PLUSG \GAMMA_1, M. \ 
	\TCTYPES{\GAMMA}{\GAMMA_0} \ \MAND \ \TCTYPES{\GAMMA}{\GAMMA_0 \PLUSG \GAMMA_1}
	\ \MAND \
	\LTCDERIVEDVALUE{M}{\GAMMA_0}{\MMM}{V} \ \MIMPLIES \ \LTCDERIVEDVALUE{M}{\GAMMA_0 \PLUSG \GAMMA_1}{\MMM}{V}
	\]
	\begin{proof}
		By Lemma \ref{lem:subLTC_implies_sub_typing_context},  
		\\
		$\TCTYPES{\GAMMA}{\GAMMA_0} \ \MAND \ \TCTYPES{\GAMMA}{\GAMMA_0 \PLUSG \GAMMA_1}$
		\\
		$\MIMPLIES \ \TCTYPES{\GAMMA_0 \PLUSG \GAMMA_1}{\GAMMA_0} $
		\\
		$\MIMPLIES \ \SEM{\GAMMA_0}{\MMM} \subseteq \SEM{\GAMMA_0\ \PLUSG \GAMMA_1}{\MMM}$ 
		\\
		hence this must hold.
	\end{proof}
\end{lemma}

\begin{lemma}[\DONE Values derived from $\GAMMA$ maintain extension (single) when added to models which are extensions]
	\label{lem:Gamma_derived_terms_maintain_extension_single_when_added}
	\[
	\Mforall \GAMMA, \MMM^{\GAMMA}, \GAMMA', \MMM'^{\GAMMA'}, M^{\alpha}.  \
	\MMM \EXTSINGLE \MMM' 
	\ \MAND \ 
	\LTCDERIVEDVALUE{M}{\GAMMA}{\MMM}{V}
	\ \MAND \
	\AN{V} \cap \AN{\MMM'} \subseteq \AN{\MMM}
	\ \MIMPLIES \ 
	\MMM \cdot x:V \EXTSINGLE \MMM' \cdot x:V
	\]
	\begin{proof}
	\begin{NDERIVATION}{1}
		\NLINE{\text{Assume: $\GAMMA,\MMM^{\GAMMA},\GAMMA', \MMM^{\GAMMA'}, M^{\alpha}$ s.t. $\MMM \EXTSINGLE \MMM' \ \MAND \ \LTCDERIVEDVALUE{M}{\GAMMA}{\MMM}{V}$ }}{}
		\NLINE{\text{The case for $\MMM' \equiv \MMM \cdot \TCV:\GAMMA_0\REMOVETCVfrom$ is proven below:}}{}
		\NLINE{
			\parbox[t]{9cm}{
				assume some $\GAMMA_0$ s.t. $\TCTYPES{\GAMMA}{\GAMMA_0}$ with $\MMM \EXTSINGLE\MMM \cdot \TCV:\GAMMA_0\REMOVETCVfrom$
				\\
				then  $\TCTYPES{\GAMMA \PLUSV x}{\GAMMA_0}$ and hence $\MMM \cdot x:V \EXTSINGLE \MMM \cdot \TCV:\GAMMA_0 \REMOVETCVfrom \cdot x:V$}
			}{Sem. $\EXTSINGLE$}
		\NLINE{\text{The case for $\MMM' \equiv \MMM \cdot y:V'$ is proven below:}}{}
		\NLINE{(\MMM \EXTSINGLE \MMM' \ \MAND \ \LTCDERIVEDVALUE{M}{\GAMMA}{\MMM}{V}) \ \MIMPLIES \ \Mexists M'. \ \LTCDERIVEDVALUE{M'}{\GAMMA}{\MMM}{V'} \ \MAND \ \MMM' \equiv \MMM \cdot y':V'  \ \MAND \ \LTCDERIVEDVALUE{M}{\GAMMA}{\MMM}{V}}{Sem. $\EXTSINGLE$}
		\NLINE{\MIMPLIES \ \Mexists M'. \ \LTCDERIVEDVALUE{M'}{\GAMMA}{\MMM \cdot x:V}{V'} \ \MAND \ \MMM' \cdot x:V \equiv \MMM \cdot x:V \cdot y':V'}{Lemma \ref{lem:eval_under_extensions_are_equivalent} ($\AN{V} \cap \AN{V'} \subseteq \AN{\MMM}$)}
		\NLINE{\MIMPLIES \ \Mexists M'. \ \LTCDERIVEDVALUE{M'}{\GAMMA \PLUSV x:\alpha_x}{\MMM \cdot x:V}{V'} \ \MAND \ \MMM' \cdot x:V \equiv \MMM \cdot x:V \cdot y':V'}{Lemma \ref{lem:LTCDERIVED_subset_implies_LTCDERIVED_supset}}
		\NLASTLINE{\MIMPLIES \ \MMM \cdot x:V \EXTSINGLE \MMM' \cdot x:V}{Sem. $\EXTSINGLE$}
	\end{NDERIVATION}
	\end{proof}
\end{lemma}

\begin{lemma}[\DONE Values derived from $\GAMMA$ maintain extension (star) when added to models which are extensions]
	\label{lem:Gamma_derived_terms_maintain_extension_when_added}
	\[
	\Mforall \GAMMA, \MMM^{\GAMMA}, \GAMMA', \MMM'^{\GAMMA'}, M.  \
	\MMM \EXTSTAR \MMM' 
	\ \MAND \ 
	\LTCDERIVEDVALUE{M}{\GAMMA}{\MMM}{V}
	\ \MAND \
	\AN{V} \cap \AN{\MMM'} \subseteq \AN{\MMM}
	\ \MIMPLIES \ 
	\MMM \cdot x:V \EXTSTAR \MMM' \cdot x:V
	\]
	\begin{proof}
		Use Lemma \ref{lem:Gamma_derived_terms_maintain_extension_single_when_added} with Def $\EXTSTAR$ to see this clearly holds.	
	\end{proof}
\end{lemma}

\begin{lemma}[Names produced from an extension can be added to an extension and still hold]
	\label{lem:name_derived_from_extension_implies_extension_and_name_are_extension}
	\[
		\Mforall \GAMMA, \MMM^{\GAMMA}, \GAMMA', \MMM'^{\GAMMA'}, M, r.  \
		\LTCDERIVEDVALUE{M}{\GAMMA'}{\MMM'}{r} 
		\ \MAND \ 
		\MMM \EXTSTAR \MMM' 
		\ \MIMPLIES \ 
		\MMM \cdot x:r \EXTSTAR \MMM' \cdot x:r
	\]
	\begin{proof}
		Assume a sequence of terms $[\VEC{z}: \VEC{N}]$ that derive the (non-LTC) values of $\MMM'$ (the LTC are trivial),
		then three cases exist:
		\begin{itemize}
			\item[$r \in \AN{\MMM}$] then clearly this holds as the same sequence of terms $[\VEC{z}:\VEC{N}]$ derives $\MMM' \cdot x:r$ with the new LTC $\GAMMA \PLUSV x:\NAME$. Hence $\MMM \cdot x:r \EXTSTAR \MMM' \cdot x:r$
			
			\item[$r \notin \AN{\MMM} \MAND  r \notin \AN{\MMM'}$] then clearly $r$ is fresh and the same reasoning as above holds.
			
			\item[$r \notin \AN{\MMM} \MAND  r \in \AN{\MMM'}$]  then $r$ is derived from $\MMM'$ but not from $\MMM$ and hence the sequence $ [\VEC{z}:\VEC{N}\SUBST{x}{\GENSYM()}]$ also derives $\MMM'$ proves $\MMM \cdot x:r \EXTSTAR \MMM' \cdot x:r$ where we write  $\VEC{N}\SUBST{x}{\GENSYM()}$ for the term $\VEC{N}_i$ that initially produce $r$ to substitute $x$ for that $\GENSYM()$ which clearly must produce the same results. All other fresh names are irrelevant.

		\end{itemize}
	\end{proof}
	
\end{lemma}

\begin{lemma}[\DONE Extending the model by $\TCV$ maintains models]
	\label{lem:model_and_model_plus_TCV_models_equivalently_TCV-free_formula}
	($A$-\EXTINDEP \ not required)
	\[
	\Mforall \GAMMA, \MMM^{\GAMMA}, \GAMMA', \MMM'^{\GAMMA'},\GAMMA_a. \
	\TCTYPES{\GAMMA}{\GAMMA_a} \ \MIMPLIES \
	(\MMM \cdot \MMM' \models A^{-\TCV} \ \MIFF \ \MMM \cdot \TCV:\GAMMA_a\REMOVETCVfrom \cdot \MMM' \models A^{-\TCV})
	\]
	where $A^{-\TCV}$ means $\TCV$ does not occur syntactically in $A$.
	\begin{proof}
		By definition of $\EXTSINGLE$, $\MMM \EXTSINGLE \MMM \cdot \TCV:\GAMMA_a\REMOVETCVfrom$ holds and implies via Def $\EXTSTAR$ that $\MMM \cdot \MMM' \EXTSTAR \MMM \cdot \TCV:\GAMMA_a\REMOVETCVfrom \cdot \MMM'$ which allows the use of Lemma \ref{lem:semantics_expressions_equal_under_model_extensions} and  Lemma \ref{lem:eval_under_extensions_are_equivalent} ($\AN{\MMM \cdot \MMM'} \equiv \AN{\MMM \cdot \TCV:\GAMMA_a\REMOVETCVfrom \cdot \MMM'}$).
		\\\\
		The proof follows by IH on structure of $A$:
		\begin{itemize}
			\item $e_1 = e_2$ proven by noting that $\SEM{e_i}{\MMM \cdot \MMM'} \equiv \SEM{e_i}{\MMM\cdot \TCV:\GAMMA_a\REMOVETCVfrom \cdot \MMM'}$ (Lemma \ref{lem:semantics_expressions_equal_under_model_extensions}).
			\item $A_1 \PAND A_2$,  $A_1 \POR A_2$, $A_1 \PIMPLIES A_2$: IH on $A_1$, $A_2$.
			\item $\neg A_1$, proof by IH on $A_1$, assume $\MMM \cdot \TCV:\GAMMA_a\REMOVETCVfrom \cdot \MMM' \models \neg A_1$ then clearly 
			$\MMM \cdot \TCV:\GAMMA_a\REMOVETCVfrom \cdot \MMM' \not\models A_1$ and by IH on $A_1$ then $\MMM \cdot \MMM' \not\models A_1$ i.e. $\MMM \cdot \MMM' \models \neg A_1$.
			
			\item $\FRESH{x}{\GAMMA'}$ holds by Lemma \ref{lem:semantics_expressions_equal_under_model_extensions} and Lemma \ref{lem:eval_under_extensions_are_equivalent} ($\AN{\MMM \cdot \MMM'} \equiv \AN{\MMM \cdot \TCV:\GAMMA_a\REMOVETCVfrom \cdot \MMM'}$)
			which ensure $\SEM{x}{\MMM \cdot \MMM'} \equiv \SEM{x}{\MMM\cdot \TCV:\GAMMA_a\REMOVETCVfrom \cdot \MMM'}$ and
			\\
			$\Mexists M_x. \ \LTCDERIVEDVALUE{M_x}{\GAMMA'}{\MMM \cdot \MMM'}{V} \ \MIFF \ \Mexists M_x. \ \LTCDERIVEDVALUE{M_x}{\GAMMA'}{\MMM \cdot \TCV:\GAMMA_a\REMOVETCVfrom \cdot \MMM'}{V}$
			respectively, hence clearly
			\\
			$\neg \Mexists M_x. \ \LTCDERIVEDVALUE{M_x}{\GAMMA'}{\MMM \cdot \MMM'}{\SEM{x}{\MMM \cdot \MMM'}} \ \MIFF \ \neg\Mexists M_x. \ \LTCDERIVEDVALUE{M_x}{\GAMMA'}{\MMM \cdot \TCV:\GAMMA_a\REMOVETCVfrom \cdot \MMM'}{\SEM{x}{\MMM \cdot \TCV:\GAMMA_a\REMOVETCVfrom \cdot \MMM'}}$.
			
			\item $\ONEEVAL{u}{e}{m}{A_1}$ by IH on $A_1$ with $\SEM{u}{\MMM \cdot \MMM'} \equiv \SEM{u}{\MMM\cdot \TCV:\GAMMA_a\REMOVETCVfrom \cdot \MMM'}$ and $\SEM{e}{\MMM \cdot \MMM'} \equiv \SEM{e}{\MMM\cdot \TCV:\GAMMA_a\REMOVETCVfrom \cdot \MMM'}$ (Lemma \ref{lem:semantics_expressions_equal_under_model_extensions}) hence evaluation is equivalent and this holds.
			\item $\FORALL{x}{\GAMMA'} A_1$ knowing that $\TCV \notin \GAMMA'$ then the same possible terms are quantified over given Lemma \ref{lem:eval_under_extensions_are_equivalent} ($\LTCDERIVEDVALUE{M_x}{\GAMMA'}{\MMM \cdot \MMM'}{V} \MIFF \LTCDERIVEDVALUE{M_x}{\GAMMA'}{\MMM \cdot \TCV:\GAMMA_a\REMOVETCVfrom \cdot \MMM'}{V}$)  and by IH on $A_1$ the case holds.
			\item $\FAD{\TCV'} A_1$ by IH on $A_1$ this holds for any model containing $\TCV$ and $\TCV'$ then by the semantics the lemma holds.
		\end{itemize}
	\end{proof}
\end{lemma}

The following 2 definitions will be used in some of the proofs using equality, but are not required otherwise.

\begin{definition}[Similar extensions (single)]
	\label{def:similar_extsingle}
	Two model extensions are \EMPH{congruent model extensions (single)}  iff there exists an equivalent derivation for both models:
	\[
	\begin{array}[t]{l}
		\Mforall \GAMMA_i, \MMM_1^{\GAMMA_i}, \MMM_2^{\GAMMA_i}, \GAMMA_i', \MMM_1'^{\GAMMA_i'}, \MMM_2'^{\GAMMA_i'}.
		\\
		(\MMM_1 \EXTSINGLE \MMM_1') \CONGEXTSINGLE (\MMM_2 \EXTSINGLE \MMM_2') 
		\ \MIFF \
		\begin{array}[t]{l}
			\GAMMA_i' \equiv \GAMMA_i \PLUSTC \TCV_j \ \MAND \
			\MMM_1'(\TCV_j) \equiv 
			\MMM_2'(\TCV_j)
			\\
			\MOR
			\\
			\GAMMA_i' \equiv \GAMMA_i \PLUSV x_j: \alpha_j \ \MAND \
			\Mexists M_j^{\alpha_j}. \
			\LTCDERIVEDVALUE{M_j}{\GAMMA_i}{\MMM_1}{\MMM_1'(x_j)} 
			\ \MAND \
			\LTCDERIVEDVALUE{M_j}{\GAMMA_i}{\MMM_2}{\MMM_2'(x_j)} 
		\end{array}
	\end{array}
	\]
\end{definition}
\begin{definition}[Similar extensions (star)]
	\label{def:similar_extstar}
	Two model extensions are \EMPH{similarly derived extensions (star)}  iff there exists an equivalent derivation for both models:
	\[
	\begin{array}[t]{l}
		\Mforall \GAMMA_i, \MMM_1^{\GAMMA_i}, \MMM_2^{\GAMMA_i}, \GAMMA_i', \MMM_1'^{\GAMMA_i'}, \MMM_2'^{\GAMMA_i'}.
		\\
		(\MMM_1 \EXTSTAR \MMM_1') \CONGEXTSTAR (\MMM_2 \EXTSTAR \MMM_2') 
		\ \MIFF \
		\begin{array}[t]{l}
			\MMM_1 \equiv \MMM_1' \ \MAND \ \MMM_2 \equiv \MMM_2'
			\\
			\MOR 
			\\
			(\MMM_1 \EXTSINGLE \MMM_1') \CONGEXTSINGLE (\MMM_2 \EXTSINGLE \MMM_2') 
			\\
			\MOR
			\\
			\Mexists \MMM_{1j}^{\GAMMA_{j}}, \MMM_{2j}^{\GAMMA_{j}}.
			\begin{array}[t]{l}
				(\MMM_1 \EXTSTAR \MMM_{1j}) \CONGEXTSTAR (\MMM_2 \EXTSTAR \MMM_{2j}) 
				\\ \MAND \
				(\MMM_{1j} \EXTSTAR \MMM_1') \CONGEXTSTAR (\MMM_{2j} \EXTSTAR \MMM_2') 
			\end{array}
		\end{array} 
	\end{array}
	\]
\end{definition}

\begin{lemma}[\DONE Congruent extensions implies their evaluation added to a model, model equivalence equivalently]
	\label{lem:congruent_base_model_extensions_model_equivalence_equivelently}
	\[
	\begin{array}[t]{l}
		\Mforall \GAMMA, \MMM^{\GAMMA}, M_1^{\alpha}, M_2^{\alpha}.
		\\
		M_1\MMM \CONGCONTEXT{\alpha}{\AN{\MMM}} M_2\MMM
		\ \MAND \ 
		\LTCDERIVEDVALUE{M_1}{\GAMMA}{\MMM}{V_1}
		\ \MAND \ 
		\LTCDERIVEDVALUE{M_2}{\GAMMA}{\MMM}{V_2}
		\\ 
		\MIMPLIES 
		\\
		\Mforall \GAMMA_i, \MMM_1'^{\GAMMA_i}, \MMM_2'^{\GAMMA_i}, e=e'.
		\begin{array}[t]{l}
			(\MMM \cdot x:V_1 \EXTSTAR \MMM_1') 
			\ \CONGEXTSTAR \ 
			(\MMM \cdot x:V_2 \EXTSTAR \MMM_2')
			\ \MAND \
			\FORMULATYPES{\GAMMA_i}{e=e'}
			\\
			\MIMPLIES 
			\\
			\MMM_1' \models e=e' \ \MIFF \ \MMM_2' \models e=e'
		\end{array}
	\end{array}
	\]
	As an example consider $M_1 \equiv \GENSYM() \equiv M_2$ clearly any use of the names produced must be indistinguishable in any future similarly derived models as the names can be switched asimilar to $\alpha$-equivalence.
	\begin{proof}
		Required to prove 
		$\SEM{e}{\MMM_1'} \CONGCONTEXT{\alpha}{\AN{\MMM_1'}} \SEM{e'}{\MMM_1'} 
		\ \MIMPLIES \
		\SEM{e}{\MMM_2'} \CONGCONTEXT{\alpha}{\AN{\MMM_2'}}  \SEM{e'}{\MMM_2'}$
		\\
		\begin{NDERIVATION}{1}
			\NLINE{\text{Making the relevant assumptions about $\GAMMA$, $\MMM^{\GAMMA}$, $M_1$, $M_2$, $\GAMMA_i$, $\MMM_1'^{\GAMMA_i}$, $\MMM_2'^{\GAMMA_i}$, $e=e'$, s.t. ...}}{}
			\NLINE{M_1\MMM \CONGCONTEXT{\alpha}{\AN{\MMM}} M_2\MMM}{}
			\NLINE{\MIFF \ 
				\Mforall C[\cdot]^{\alpha}, b^{\BOOL}. \ 
				\left(
				\begin{array}{l}
					\TYPES{\emptyset}{C[\cdot]^{\alpha}}{\BOOL} 
					\\ \MAND \
					\AN{C[\cdot]}\subseteq \AN{\MMM}
				\end{array}
				\right)
				\ \MIMPLIES \
				\left(
				\begin{array}{c}
					(\AN{\MMM}, C[M_1\MMM]) \CONV (\AN{\MMM}, G_1, b)
					\\ \MIFF \\
					(\AN{\MMM}, C[M_2\MMM]) \CONV (\AN{\MMM}, G_2, b)
				\end{array}
				\right)
			}{Sem. $\CONGCONTEXT{}{}$}
			\NLINE{\parbox[t]{12cm}{
					Let $\MMM_j'[\cdot]$ be derived from the terms as follows $\VEC{z}: \VEC{N}$ and $x:V_i$ noting that $\VEC{N}_{i+1}$ may contain the variable $\VEC{z}_{i}$ and ignoring the inconsequential TCV-parts, then if $\EXPRESSIONTYPES{\GAMMA_i}{e}{\alpha_e}$ a subset of $C[\cdot]$ can be set for any $C'[\cdot]^{\alpha_e}$ as follows:
				}
			}{}
			\NPLINE{\MIMPLIES \ 
				\begin{array}[t]{l}
					\Mforall C[\cdot] \equiv (\LET{x}{[\cdot]_{\emptyset}^{\alpha}}{
						\LET{\VEC{z}}{\VEC{N}\MMM}{
							C'[e\MMM] =C'[e'\MMM]}}), b^{\BOOL}. 
						\\
						\TYPES{\emptyset}{C[\cdot]^{\alpha}}{\BOOL} 
						\ \MAND \
						\AN{C[\cdot]}\subseteq \AN{\MMM}
					\\ \MIMPLIES \\
					\hspace{-0.5cm}\left(
					\begin{array}{c}
						(\AN{\MMM}, \LET{x}{[M_1\MMM]_{\emptyset}^{\alpha}}{
							\LET{\VEC{z}}{\VEC{N}\MMM}{
								C'[e\MMM] =C'[e'\MMM]}}) \CONV (\AN{\MMM}, G_1, b)
						\\ \MIFF \\
						(\AN{\MMM}, \LET{x}{[M_2\MMM]_{\emptyset}^{\alpha}}{
							\LET{\VEC{z}}{\VEC{N}\MMM}{
								C'[e\MMM] =C'[e'\MMM]}}) \CONV (\AN{\MMM}, G_2, b)
					\end{array}
					\right)
				\end{array} \hspace{-1.5cm}
			}{4cm}{
				Subset $C[\cdot] \equiv ...$  
				\\
				$C'[\cdot]$-capture-avoiding variables in $\DOM{\GAMMA_i}$
			}
			\NPLINE{\MIMPLIES \Mforall b^{\BOOL}.
				\begin{array}[t]{l}
					\left(
					\begin{array}{l}
						\Mforall C[\cdot] \equiv (\LET{x}{[\cdot]_{\emptyset}^{\alpha}}{
							\LET{\VEC{z}}{\VEC{N}\MMM}{
								C'[e\MMM] =C'[e'\MMM]}}), b^{\BOOL}. 
						\\
						\TYPES{\emptyset}{C[\cdot]}{\BOOL} 
						\ \MAND \
						\AN{C[\cdot]} \subseteq \AN{\MMM}
						\\ \MIMPLIES \
						(\AN{\MMM}, C[M_1\MMM]) \CONV (\AN{\MMM}, G_1, b)
					\end{array}
					\right)
					\\ \MIFF \\
					\left(
					\begin{array}{l}
						\Mforall C[\cdot] \equiv (\LET{x}{[\cdot]_{\emptyset}^{\alpha}}{
							\LET{\VEC{z}}{\VEC{N}\MMM}{
								C'[e\MMM] =C'[e'\MMM]}}), b^{\BOOL}. 
						\\
						\TYPES{\emptyset}{C[\cdot]}{\BOOL} 
						\ \MAND \
						\AN{C[\cdot]} \subseteq \AN{\MMM}
						\\ \MIMPLIES \
						(\AN{\MMM}, C[M_2\MMM]) \CONV (\AN{\MMM}, G_2, b)
					\end{array}
					\right)
				\end{array} 
			}{4cm}{
			$(\Mforall x. A(x) \MIMPLIES (B(x) \MIFF C(x)))$ 
			\\ 
			$ \MIMPLIES \qquad (\Mforall x. A(x) \MIMPLIES B(x)) $ 
			\\ 
			$\MIFF (\Mforall x. A(x) \MIMPLIES C(x))$
			\\
			FOL
			}
			\NPLINE{\MIMPLIES
				\Mforall b^{\BOOL}.
				\begin{array}[t]{l}
					\left(
					\begin{array}{l}
						\Mforall \LET{\VEC{z}}{\VEC{N}}{C'[e\MMM_1] =C'[e'\MMM_1]}. \\
						\left(
						\begin{array}{l}
							\TYPES{\emptyset}{\LET{\VEC{z}}{\VEC{N}\MMM_1}{C'[e\MMM_1] =C'[e'\MMM_1]}}{\BOOL} 
							\\ \MAND \
							\AN{\LET{\VEC{z}}{\VEC{N}\MMM_1}{C'[e\MMM_1] =C'[e'\MMM_1]}} \subseteq \AN{\MMM_1}
						\end{array}
						\right)
						\\ \MIMPLIES \
						(\AN{\MMM_1}, \LET{\VEC{z}}{\VEC{N}\MMM_1}{C'[e\MMM_1] =C'[e'\MMM_1]}) \CONV (\AN{\MMM_1}, G_1, b)
					\end{array}
					\right)
					\\ \MIFF \\
					\left(
					\begin{array}{l}
						\Mforall \LET{\VEC{z}}{\VEC{N}\MMM_2}{C'[e\MMM_2] =C'[e'\MMM_2]}. \\
						\left(
						\begin{array}{l}
							\TYPES{\emptyset}{\LET{\VEC{z}}{\VEC{N}}{C'[e\MMM_2] =C'[e'\MMM_2]}}{\BOOL} 
							\\ \MAND \
							\AN{\LET{\VEC{z}}{\VEC{N}\MMM_2}{C'[e\MMM_2] =C'[e'\MMM_2]}}\subseteq \AN{\MMM_2}
						\end{array}
						\right)
						\\ \MIMPLIES \
						(\AN{\MMM_2}, \LET{\VEC{z}}{\VEC{N}\MMM_2}{C'[e\MMM_2] =C'[e'\MMM_2]}) \CONV (\AN{\MMM_2}, G_2, b)
					\end{array}
					\right)
				\end{array}
			}{4cm}{
				$\LTCDERIVEDVALUE{M_j}{\GAMMA}{\MMM}{V_j}$ 
				\\	
				then $\MMM_j \equiv \MMM \cdot x:V_j$
				\\
				($j=1,2$)
				\\
				$\LET{x}{M_j}{M\MMM} \RED M\MMM_j$
			}
			\NPLINE{
				\MIMPLIES
				\Mforall b^{\BOOL}.
				\begin{array}[t]{l}
					\left(
					\begin{array}{l}
						\Mforall C'[\cdot]^{\alpha}. 
						\left(
						\begin{array}{l}
							\TYPES{\emptyset}{C'[\cdot]^{\alpha_e}}{\BOOL} 
							\\ \MAND \
							\AN{C'[\cdot]}\subseteq \AN{\MMM'_1}
						\end{array}
						\right)
						\\ \MIMPLIES \
						(\AN{\MMM_1'}, C'[e\MMM_1'] =C'[e'\MMM_1']) \CONV (\AN{\MMM_1'}, G_1, b)
					\end{array}
					\right)
					\\ \MIFF \\
					\left(
					\begin{array}{l}
						\Mforall C'[\cdot]^{\alpha_e}. 
						\left(
						\begin{array}{l}
							\TYPES{\emptyset}{C'[\cdot]^{\alpha_e}}{\BOOL} 
							\\ \MAND \
							\AN{C'[\cdot]}\subseteq \AN{\MMM'_2}
						\end{array}
						\right)
						\\ \MIMPLIES \
						(\AN{\MMM_2'}, C'[e\MMM_2'] =C'[e'\MMM_2']) \CONV (\AN{\MMM_2'}, G_2, b)
					\end{array}
					\right)
				\end{array}
			}{5.5cm}{
				$\LET{\VEC{z}}{\VEC{N}\MMM_j}{M\MMM_j} \RED M\MMM_j'$
				\\
				$C[M_j\MMM] \ \RED \ C'[e\MMM_j]^{\alpha} =C'[e'\MMM_j]^{\alpha}$
			}
			\NPLINE{
			 \MIMPLIES
				\begin{array}[t]{l}
					\left(
					\begin{array}{l}
						\Mforall C'[\cdot]^{\alpha_e}, b'^{\BOOL}.
						\left(
						\begin{array}{l}
							\TYPES{\emptyset}{C'[\cdot]^{\alpha_e}}{\BOOL} 
							\\ \MAND \
							\AN{C'[\cdot]}\subseteq \AN{\MMM'_1}
						\end{array}
						\right)
						\\
						\MIMPLIES
						\left(
						\begin{array}{c}
							(\AN{\MMM_1'}, C'[e\MMM_1']) \CONV (\AN{\MMM_1'}, G_1, b')
							\\ \MIFF \
							(\AN{\MMM_1'}, C'[e'\MMM_1']) \CONV (\AN{\MMM_1'}, G_1, b')
						\end{array}
						\right)
					\end{array}
					\right)
					\\ \MIFF \\
					\left(
					\begin{array}{l}
						\Mforall C'[\cdot]^{\alpha_e}, b'^{\BOOL}. 
						\left(
						\begin{array}{l}
							\TYPES{\emptyset}{C'[\cdot]^{\alpha_e}}{\BOOL} 
							\\ \MAND \
							\AN{C'[\cdot]}\subseteq \AN{\MMM_2'}
						\end{array}
						\right)
						\\
						\MIMPLIES
						\left(
						\begin{array}{c}
							(\AN{\MMM_2'}, C'[e\MMM_2']) \CONV (\AN{\MMM_2'}, G_1, b')
							\\ \MIFF \
							(\AN{\MMM_2'}, C'[e'\MMM_2']) \CONV (\AN{\MMM_2'}, G_1, b')
						\end{array}
						\right)
					\end{array}
					\right)
				\end{array} \hspace{-1cm}
			}{6cm}{ \hspace{-1cm} Hence given $C'[e\MMM_i']\MMM_i': \BOOL$ this implies:
				\\
				(To think about this set $b \equiv \TRUE$ or $b=\FALSE$ then clearly this holds)
				\\
				$b=\TRUE$: implies the $=$ case of $\MIFF$
				\\
				$b=\FALSE$: implies the $\neq$ case of $\MIFF$
			}
			\NLASTLINE{\MMM_1' \models e=e' \ \MIFF \ \MMM_2' \models e=e'}{Sem. $=$}
		\end{NDERIVATION}
	\end{proof}
\end{lemma}

\begin{lemma}[\DONE Congruent extensions implies their evaluation added to a model, model equivalently]
	\label{lem:congruent_base_model_extensions_derive_values_equivalently}
	\[
	\begin{array}[t]{l}
	\Mforall \GAMMA, \MMM^{\GAMMA}, M_1^{\alpha}, M_2^{\alpha}.
	\\
	M_1\MMM \CONGCONTEXT{\alpha}{\AN{\MMM}} M_2\MMM
	\
	\MAND \ 
	\LTCDERIVEDVALUE{M_1}{\GAMMA}{\MMM}{V_1}
	\
	\MAND \ 
	\LTCDERIVEDVALUE{M_2}{\GAMMA}{\MMM}{V_2}
	\\ 
	\MIMPLIES 
	\\
	\Mforall \GAMMA_i, \MMM_1'^{\GAMMA_i}, \MMM_2'^{\GAMMA_i}, A.
	\begin{array}[t]{l}
	(\MMM \cdot x:V_1 \EXTSTAR \MMM_1') 
	\ \CONGEXTSTAR \ 
	(\MMM \cdot x:V_2 \EXTSTAR \MMM_2')
	\ \MAND \
	\FORMULATYPES{\GAMMA_i}{A}
	\\
	\MIMPLIES 
	\\
	\Mforall \GAMMA_0. \TCTYPES{\GAMMA_i}{\GAMMA_0} \ \MIMPLIES \
	(\Mexists M_u. \LTCDERIVEDVALUE{M_u}{\GAMMA_0}{\MMM_1'}{\SEM{u}{\MMM_1'}}
	\ \MIFF \
	\Mexists N_u. \LTCDERIVEDVALUE{N_u}{\GAMMA_0}{\MMM_2'}{\SEM{u}{\MMM_2'}})
	\end{array}
	\end{array}
	\]
	
		If $u \in \DOM{\MMM}$ then clearly this holds.
		\\
		If $u \notin \DOM{\MMM}$ then clearly the $u$ must be in $\MMM_j'$ and hence is derived from a series of extensions, which are characterised by $\VEC{z}:\VEC{N}$ therefore the same use in $M_u$ and $N_u$ of these variables derives the same $\SEM{u}{\MMM_j'}$.
		This relies on the fact that if a name is used in $\MMM_1'$ then it is the same name used in $\MMM_2'$ or it is fresh from $\MMM$ in both.
	\begin{proof}
		This is proven using the fact that $\MMM_1'$ and $\MMM_2'$ are similar and hence any $M_u$ that produces $\SEM{u}{\MMM_1'}$ can be used to derive $\SEM{u}{\MMM_2'}$ as follows:
		\\
		\begin{NDERIVATION}{1}
			\NLINE{
				\text{Let the $\CONGEXTSTAR$ -extension part be derived from by $\VEC{z}:\VEC{M}$ where $\VEC{M}_{i+1}$ may contain the variable $\VEC{z}_i$.}
			}{}
			\NLINE{
				\text{Assume some $\GAMMA, \MMM^{\GAMMA}, M_1^{\alpha}, M_2^{\alpha}$ s.t. }
				M_1\MMM \CONGCONTEXT{\alpha}{\AN{\MMM}} M_2\MMM
				\ \MAND \ 
				\LTCDERIVEDVALUE{M_1}{\GAMMA}{\MMM}{V_1}
				\ \MAND \ 
				\LTCDERIVEDVALUE{M_2}{\GAMMA}{\MMM}{V_2}
			}{}
			\NLINE{
				\text{also assume some $\GAMMA_i, \MMM_1'^{\GAMMA_i}, \MMM_2'^{\GAMMA_i}, A$ s.t. }
				(\MMM \cdot x:V_1 \EXTSTAR \MMM_1') 
				\ \CONGEXTSTAR \ 
				(\MMM \cdot x:V_2 \EXTSTAR \MMM_2')
				\ \MAND \
				\FORMULATYPES{\GAMMA_i}{A}
			}{}
			\NLINE{
				\parbox[t]{12cm}{
					also assume some $\GAMMA_0$ s.t. $\TCTYPES{\GAMMA_i}{\GAMMA_0}$,
					prove $\MIMPLIES$ to prove $\MIFF$-by symmetry
				}
			}{}
			\NLINE{
				\text{
					Assume there exists $M_u$ s.t. $\LTCDERIVEDVALUE{M_u}{\GAMMA_0}{\MMM_1'}{\SEM{u}{\MMM_1'}}$
				}
			}{}
			\NLINE{
				\text{
					then show that there exists $N_u$ s.t. $\LTCDERIVEDVALUE{N_u}{\GAMMA_0}{\MMM_2'}{\SEM{u}{\MMM_2'}}$
				}
			}{}
			\NLINE{
				M_1\MMM \CONGCONTEXT{\alpha}{\AN{\MMM}} M_2\MMM
			}{}
			\NLINE{ \MIFF \
				\Mforall C[\cdot]^{\alpha}, b^{\BOOL}. \ 
				\left(
				\begin{array}{l}
					\TYPES{\emptyset}{C[\cdot]^{\alpha}}{\BOOL} 
					\\ \MAND \
					\AN{C[\cdot]}\subseteq \AN{\MMM}
				\end{array}
				\right)
				\ \MIMPLIES \
				\left(
				\begin{array}{c}
					(\AN{\MMM}, C[M_1\MMM]) \CONV (\AN{\MMM}, G_1, b)
					\\ \MIFF \\
					(\AN{\MMM}, C[M_2\MMM]) \CONV (\AN{\MMM}, G_2, b)
				\end{array}
				\right)
			}{Sem. $\CONGCONTEXT{}{}$}
			\NPLINE{\Mforall b^{\BOOL}. 
			\begin{array}[t]{l}
			\Mforall C[\cdot]^{\alpha}. \ 
			\left(
			\begin{array}{l}
				\TYPES{\emptyset}{C[\cdot]^{\alpha}}{\BOOL} 
				\\ \MAND \
				\AN{C[\cdot]}\subseteq \AN{\MMM}
			\end{array}
			\right)
			\ \MIMPLIES \
			(\AN{\MMM}, C[M_1\MMM]) \CONV (\AN{\MMM}, G_1, b)
			\\
			\MIFF
			\\
			\Mforall C[\cdot]^{\alpha}. \ 
			\left(
			\begin{array}{l}
				\TYPES{\emptyset}{C[\cdot]^{\alpha}}{\BOOL} 
				\\ \MAND \
				\AN{C[\cdot]}\subseteq \AN{\MMM}
			\end{array}
			\right)
			\ \MIMPLIES \
			(\AN{\MMM}, C[M_2\MMM]) \CONV (\AN{\MMM}, G_2, b)
			\end{array}
			}{4cm}{
				 $(\Mforall x. A(x) \MIMPLIES (B(x) \MIFF C(x)))\MIMPLIES $ \\ $((\Mforall x. A(x) \MIMPLIES B(x)) \quad$ \\ $ \MIFF (\Mforall x. A(x) \MIMPLIES C(x)))$
			}
			\NLINE{
				\parbox[t]{13cm}{
					Assume  $\MMM_1'$ and $\MMM_2'$ are both derived from $\MMM$ via $\VEC{z}:\VEC{N}$  and $x:M_j$ (from $\CONGEXTSTAR$)
					\\
					Select $C[\cdot]$ as 
					$\LET{x}{[\cdot]^{\alpha}}{\LET{\VEC{z}}{\VEC{N}\MMM}{(C'[M_u\MMM]^{\NAME} = C'[u\MMM]^{\NAME})}}$
					\\
					with $C'[\cdot] $ not capturing $x,\VEC{z}$
					\\
					$\LTCDERIVEDVALUE{M_j}{\GAMMA}{\MMM}{V_j}$ implies  $C[M_j]\MMM \begin{array}[t]{l}\RED \LET{\VEC{z}}{\VEC{N}(\MMM \cdot x:V_j)}{C'[M_u(\MMM \cdot x:V_j)]^{\NAME} = C'[u(\MMM \cdot x:V_j)]^{\NAME}}
						\\
						 ... \RED \ C'[M_u\MMM'_j]^{\NAME} = C'[u\MMM'_j]^{\NAME}
						\end{array}$
				}
			}{}
			\NPLINE{\MIMPLIES \
				\begin{array}[t]{l}
					\Mforall b^{\BOOL}. 
					\\
					\Mforall C'[M_u\MMM'_1]^{\NAME} = C'[u\MMM'_1]^{\NAME}. \ 
					\begin{array}{l}
						\left(
						\begin{array}{l}
							\TYPES{\emptyset}{C'[M_u\MMM'_1]^{\NAME} = C'[u\MMM'_1]^{\NAME}}{\alpha} 
							\\ \MAND \
							\AN{C'[M_u\MMM'_1]^{\NAME} = C'[u\MMM'_1]^{\NAME}} \subseteq \AN{\MMM'_1}
						\end{array}
						\right)
						\\ \MIMPLIES \
						(\AN{\MMM_1'}, (C'[M_u\MMM'_1] = C'[u\MMM'_1]) \CONV (\AN{\MMM_1'}, G_1, b)
					\end{array}
					\\
					\MIFF
					\\
					\Mforall C'[M_u\MMM'_2]^{\NAME} = C'[u\MMM'_2]^{\NAME}. \ 
					\begin{array}{l}
						\left(
						\begin{array}{l}
							\TYPES{\emptyset}{C'[M_u\MMM'_2]^{\NAME} = C'[u\MMM'_2]^{\NAME}}{\alpha} 
							\\ \MAND \
							\AN{C'[M_u\MMM'_2]^{\NAME} = C'[u\MMM'_2]^{\NAME}} \subseteq \AN{\MMM'_2}
						\end{array}
						\right)
						\\ \MIMPLIES \
						(\AN{\MMM'_2}, (C'[M_u\MMM'_2] = C'[u\MMM'_2]) \CONV (\AN{\MMM'_2}, G_1, b)
					\end{array}
				\end{array}
			\hspace{-1.5cm}
			}{4cm}{Subset of $C[\cdot]$ \\ from line 10}
			\NPLINE{\!\! \MIMPLIES
				\begin{array}[t]{l}
				\Mforall C'[\cdot]^{\NAME}, b'^{\BOOL}.
				\begin{array}[t]{l}
					\left(
					\begin{array}{l}
					\TYPES{\emptyset}{C'[\cdot]^{\NAME}}{\BOOL} 
					\\ \MAND \
					\AN{C'[\cdot]^{\NAME}}\subseteq \AN{\MMM'_1}
					\end{array}
					\right)
					\\ \MIMPLIES
					\left(
					\begin{array}{c}
						(\AN{\MMM_1'}, C'[M_u\MMM_1']) \CONV (\AN{\MMM_1'}, G_2, b')
						\\ \MIFF \\
						(\AN{\MMM_1'}, C'[u\MMM_1']) \CONV (\AN{\MMM_1'}, G_2, b')
					\end{array}
					\right)
				\end{array}
				\\
				\MIFF
				\\
				\Mforall C'[\cdot]^{\NAME}, b'^{\BOOL}.
				\begin{array}[t]{l}
					\left(
					\begin{array}{l}
						\TYPES{\emptyset}{C'[\cdot]^{\NAME}}{\BOOL} 
						\\ \MAND \
						\AN{C'[\cdot]^{\NAME}}\subseteq \AN{\MMM'_2}
					\end{array}
					\right)
					\\ \MIMPLIES
					\left(
					\begin{array}{c}
						(\AN{\MMM_2'}, C'[M_u\MMM'_2]) \CONV (\AN{\MMM_2'}, G_2, b')
						\\ \MIFF \\
						(\AN{\MMM_2'}, C'[u\MMM_2']) \CONV (\AN{\MMM_2'}, G_2, b')
					\end{array}
					\right)
				\end{array}
				\end{array}
			\hspace{-2cm}
			}{5cm}{
				$b=\TRUE$: clearly holds.
				\\
				$b=\FALSE$: holds as 
				\\
				$
				\begin{array}{c}
					(C'[M_u\MMM'_j] = C'[u\MMM'_j]) \CONV \FALSE
					\\
					\MIFF
					\\
					C'[M_u \MMM'_j] \CONV b' \MIFF C'[u \MMM'_j] \CONV \neg b'
				\end{array}
				$
				\\
				implies $(\Mforall C'_{M_u, u}[\cdot]. \TRUTH\MIMPLIES \FALSITY) \MIFF (\Mforall C'_{M_u, u}[\cdot]. \TRUTH\MIMPLIES \FALSITY)$ 
				\\
				and hence the more general $\FALSITY \MIFF \FALSITY$
			}
			\NLASTLINE{
				\MIMPLIES \
				(\Mexists M_u. \LTCDERIVEDVALUE{M_u}{\GAMMA_0}{\MMM_1'}{\SEM{u}{\MMM_1'}}
				\ \MIFF \
				\Mexists N_u. \LTCDERIVEDVALUE{N_u}{\GAMMA_0}{\MMM_2'}{\SEM{u}{\MMM_2'}})
			}{with precisely the same $M_u \equiv N_u$.}
		\end{NDERIVATION}
	\end{proof}
\end{lemma}

\begin{lemma}[\DONE Congruent extensions implies their evaluation added to a model, model equivalently]
\label{lem:congruent_base_model_extensions_model_equivelently}
	\[	
	\begin{array}[t]{l}
	\Mforall \GAMMA, \MMM^{\GAMMA}, M_1^{\alpha}, M_2^{\alpha}.
	\\
	M_1\MMM \CONGCONTEXT{\alpha}{\AN{\MMM}} M_2\MMM
	\ \MAND \ 
	\LTCDERIVEDVALUE{M_1}{\GAMMA}{\MMM}{V_1}
	\ \MAND \ 
	\LTCDERIVEDVALUE{M_2}{\GAMMA}{\MMM}{V_2}
	\\ 
	\MIMPLIES 
	\\
	\Mforall \GAMMA_i, \MMM_1'^{\GAMMA_i}, \MMM_2'^{\GAMMA_i}, A.
	\begin{array}[t]{l}
		(\MMM \cdot x:V_1 \EXTSTAR \MMM_1') 
		\ \CONGEXTSTAR \ 
		(\MMM \cdot x:V_2 \EXTSTAR \MMM_2')
		\ \MAND \
		\FORMULATYPES{\GAMMA_i}{A}
		\\
		\MIMPLIES 
		\\
		\MMM_1' \models A \ \MIFF \ \MMM_2' \models A
	\end{array}
	\end{array}
	\]
	Use IH as the additions are all of the same form when proving $A$:
	\\
	\begin{proof}
		Proof by IH on the structure of $A$:
		\\
		\begin{NDERIVATION}{1}
			\NLINE{\text{Assume:  $\GAMMA$, $\MMM^{\GAMMA}$, $M_1^{\alpha}$, $M_2^{\alpha}$ s.t. $M_1\MMM \CONGCONTEXT{\alpha}{\AN{\MMM}} M_2\MMM
					\ \MAND \ 
					\LTCDERIVEDVALUE{M_1}{\GAMMA}{\MMM}{V_1}
					\ \MAND \ 
					\LTCDERIVEDVALUE{M_2}{\GAMMA}{\MMM}{V_2}$}
			}{}
			\NLINE{\text{Assume:  $\GAMMA_i$, $\MMM_1'^{\GAMMA_i}$, $\MMM_2'^{\GAMMA_i}$, $A$ s.t. $(\MMM \cdot x:V_1 \EXTSTAR \MMM_1') 
					\ \CONGEXTSTAR \ 
					(\MMM \cdot x:V_2 \EXTSTAR \MMM_2')
					\ \MAND \
					\FORMULATYPES{\GAMMA_i}{A}$}
			}{}
		\NLASTLINE{\text{Assume:  $\MMM_1' \models A$ to prove $\MMM_2' \models A$}}{}
		\end{NDERIVATION}
		Noting that $\Mforall \GAMMA_0. \TCTYPES{\GAMMA_i}{\GAMMA_0} \ \MIMPLIES \ \SEM{\GAMMA_0}{\MMM_1'} \equiv \SEM{\GAMMA_0}{\MMM_2'}$ as $\TCV$'s are unaffected by $\CONGEXTSTAR$, and the standard mappings are also unaffected.
		
		Next prove this final step by IH on the structure of $A$:
		\begin{itemize}
			\item [\DONE $A \equiv e=e'$] 
			Lemma \ref{lem:congruent_base_model_extensions_model_equivalence_equivelently}.
			\item[\DONE $A \equiv \neg A'$]
			By IH on $A'$.
			\item[\DONE $A \equiv A_1 \PAND A_2$]
			By IH on $A_1$ and $A_2$.
			\item[\DONE $A \equiv A_1 \POR A_2$]
			By IH on $A_1$ and $A_2$.
			\item[\DONE $A \equiv A_1 \PIMPLIES A_2$]
			By IH on $A_1$ and $A_2$.
			
			\item[\DONE $A \equiv \ONEEVAL{u}{e}{m}{A'}$]
			Assuming $\MMM_1' \models \ONEEVAL{u}{e}{m}{A'}$ iff $(\AN{\MMM_1'}, \ \LTCDERIVEDVALUE{ue}{\GAMMA}{\MMM_1'}{V_m} \ \MAND \ \MMM_1' \cdot m:V_m \models A'$ 
			\\
			then it is clear there exists $W_m$ s.t. $\LTCDERIVEDVALUE{ue}{\GAMMA}{\MMM_2'}{W_m}$
			\\
			thus 
			$(\MMM \cdot x:V_1 \EXTSTAR \MMM_1' \cdot m:V_m) 
			\ \CONGEXTSTAR \ 
			(\MMM \cdot x:V_2 \EXTSTAR \MMM_2' \cdot m:W_m)$
			\\
			hence by IH on $A'$: $\MMM_2' \cdot m:W_m \models A'$
			\\
			hence $\MMM_2' \models \ONEEVAL{u}{e}{m}{A'}$
			
			\item[\DONE $A \equiv \FRESH{u}{\GAMMA_0}$]
			Prove by direct use of Lemma \ref{lem:congruent_base_model_extensions_derive_values_equivalently}.
			
			\item[\DONE $A \equiv \FORALL{u}{\GAMMA_0} A'$]
			Assume $\MMM_1' \models \FORALL{u}{\GAMMA_0} A'$ iff $\Mforall M_u. \LTCDERIVEDVALUE{M_u}{\GAMMA_0}{\MMM_1'}{V_u} \MIMPLIES \MMM_1' \cdot u:V_u \models A'$
			\\
			Prove  $\MMM_2' \models \FORALL{u}{\GAMMA_0} A'$ iff $\Mforall M_u. \LTCDERIVEDVALUE{M_u}{\GAMMA_0}{\MMM_2'}{W_u} \MIMPLIES \MMM_2' \cdot u:W_u \models A'$
			\\
			Assume some $M_u$ in the $\MMM_2'$ such that $\LTCDERIVEDVALUE{M_u}{\GAMMA_0}{\MMM_2'}{W_u}$ then given $\SEM{\GAMMA_0}{\MMM_2'} \equiv \SEM{\GAMMA_0}{\MMM_1'}$ and the assumption then:
			for that $M_u$ then
			$\LTCDERIVEDVALUE{M_u}{\GAMMA_0}{\MMM_1'}{V_u}$ and hence $\MMM_1' \cdot u:V_u \models A'$
			and as $W_u$ and $V_u$ are derived from the same term then 
			$(\MMM \cdot x:V_1 \EXTSTAR \MMM_1' \cdot u:V_u) \CONGEXTSTAR(\MMM \cdot x:V_2 \EXTSTAR \MMM_2' \cdot u:W_u)$
			\\
			by induction on $A'$ this implies $\MMM_2' \cdot u:W_u \models A'$.
			\\
			Hence $\Mforall M_u. \LTCDERIVEDVALUE{M_u}{\GAMMA_0}{\MMM_2'}{W_u} \MIMPLIES \MMM_2' \cdot u:W_u \models A'$ hence $\MMM_2' \models \FORALL{u}{\GAMMA_0} A'$
			
			\item[\DONE $A \equiv \EXISTS{u}{\GAMMA_0} A'$] Similar proof as $\FORALL{u}{\GAMMA_0} A'$ above.
			
			\item[\DONE $A \equiv \FAD{\TCV} A'$]
			Assume  $\MMM_1' \models \FAD{\TCV} A'$ iff $\Mforall \MMM_1''^{\GAMMA_i''}. \MMM_1' \EXTSTAR \MMM_1'' \MIMPLIES \MMM_1'' \cdot \TCV:\GAMMA_i''\REMOVETCVfrom \models A'$
			\\
			Prove  $\MMM_2' \models \FAD{\TCV} A'$ iff $\Mforall \MMM_2''^{\GAMMA_i''}. \MMM_2' \EXTSTAR \MMM_2'' \MIMPLIES \MMM_2'' \cdot \TCV:\GAMMA_i''\REMOVETCVfrom \models A'$
			\\
			i.e. assume some $\MMM_2''^{\GAMMA_i''}$ such that $\MMM_2' \EXTSTAR \MMM_2''$ then for the same derivation required for $\MMM_2''$ from $\MMM_2'$ there is an similarly derived model $\MMM_1''$ derived from $\MMM_1'$ 
			\\
			and hence $(\MMM \cdot x:V_1 \EXTSTAR \MMM_1'') \CONGEXTSTAR(\MMM \cdot x:V_2 \EXTSTAR \MMM_2'')$ (By assumption)
			\\
			and hence $(\MMM \cdot x:V_1 \EXTSTAR \MMM_1'' \cdot \TCV:\GAMMA_i''\REMOVETCVfrom) \CONGEXTSTAR(\MMM \cdot x:V_2 \EXTSTAR \MMM_2'' \cdot \TCV:\GAMMA_i''\REMOVETCVfrom)$  (Def \ref{def:similar_extstar})
			\\
			and IH on $A'$ can now be used to show that given $\MMM_1'' \cdot \TCV:\GAMMA_i''\REMOVETCVfrom \models A'$ then $\MMM_2'' \cdot \TCV:\GAMMA_i''\REMOVETCVfrom \models A'$ hence the case holds.
		\end{itemize}
	\end{proof}
\end{lemma}

\begin{lemma}[\DONE Two extensions combine to make extensions of each other]
	\label{lem:two_extensions_combine_to_make_extensions_of_each_other}
	
	\[
	\begin{array}[t]{l}
	\Mforall \GAMMA, \GAMMA_1, \GAMMA_2, \MMM^{\GAMMA}, \MMM_1^{\GAMMA_1}, \MMM_2^{\GAMMA_2}.
	\\
	\MMM \EXTSTAR \MMM \cdot \MMM_1
	\ \MAND \
	\MMM \EXTSTAR \MMM \cdot \MMM_2
	\\ \MIMPLIES \ 
	\left(
	\MMM \cdot \MMM_1 \EXTSTAR \MMM \cdot \MMM_1 \cdot \MMM_2
	\ \MIFF \
	\MMM \cdot \MMM_2 \EXTSTAR \MMM \cdot \MMM_1 \cdot \MMM_2
	\right)
	\end{array} 
	\]
	i.e.
	\[
	\begin{array}[t]{l}
	\MMM \EXTSTAR \MMM \cdot \VEC{X}:\VEC{V} 
	\ \MAND \
	\MMM \EXTSTAR \MMM \cdot \VEC{Y}:\VEC{W}
	\\ \MIMPLIES \ 
	\left(
	\MMM \cdot \VEC{X}:\VEC{V} \EXTSTAR \MMM \cdot \VEC{X}:\VEC{V} \cdot \VEC{Y}:\VEC{W}
	\ \MIFF \
	\MMM \cdot \VEC{Y}:\VEC{W} \EXTSTAR \MMM \cdot \VEC{X}:\VEC{V} \cdot \VEC{Y}:\VEC{W}
	\right)
	\end{array} 
	\]
	\begin{proof}
		\begin{NDERIVATION}{1}
			\NLINE{\text{By symmetry this only needs to be proven in one direction, which is as follows: } }{}
			\NLINE{\text{Assume: } 
				\MMM \EXTSTAR \MMM \cdot \MMM_1
				\ \MAND \
				\MMM \EXTSTAR \MMM \cdot \MMM_2
				\ \MAND \ 
				\MMM \cdot \MMM_1 \EXTSTAR \MMM \cdot \MMM_1 \cdot \MMM_2 
			}{}
			\NLINE{\text{i.e. assume: } 
				\MMM \EXTSTAR \MMM \cdot \VEC{X}: \VEC{V}^{\VEC{\alpha}}
				\ \MAND \
				\MMM \EXTSTAR \MMM \cdot \MMM_2
				\ \MAND \ 
				\MMM \cdot \VEC{X}: \VEC{V}^{\VEC{\alpha}} \EXTSTAR \MMM \cdot \VEC{X}: \VEC{V}^{\VEC{\alpha}} \cdot \MMM_2
			}{}
			\NLINE{\parbox[t]{10cm}{From Def $\EXTSTAR$, $\EXTSINGLE$, and Lemma \ref{lem:model_and_model_plus_TCV_models_equivalently_TCV-free_formula}, the TCV can be discarded in the reasoning, hence only standard variables are considered}}{}
			\NLINE{
				\text{Line 3}
				\ \MIFF \
				\begin{array}[t]{l}
					\Mexists \VEC{M}. \begin{array}[t]{l}
						\LTCDERIVEDVALUE{\VEC{M}_1}{\GAMMA}{\MMM}{\VEC{V}_1}
						\
						\MAND \ 
						\LTCDERIVEDVALUE{\VEC{M}_2}{\GAMMA \PLUSV \VEC{X}_1}{\MMM \cdot \VEC{X}_1 : \VEC{V}_1}{\VEC{V}_2}
						\
						\MAND \ ...
						\\
						\MAND \ 
						\LTCDERIVEDVALUE{\VEC{M}_{i+1}}{\GAMMA \PLUSV \VEC{X}_1 \PLUSV ... \PLUSV \VEC{X}_i}{\MMM \cdot \VEC{X}_1 : \VEC{V}_1 \cdot ... \cdot \VEC{X}_i : \VEC{V}_i}{\VEC{V}_{i+1}}
						\
						\MAND \ ...
					\end{array} 
					\\ \MAND \
					\MMM \EXTSTAR \MMM \cdot \MMM_2
					\\ \MAND \ 
					\MMM \cdot \VEC{X}: \VEC{V}^{\VEC{\alpha}} \EXTSTAR \MMM \cdot \VEC{X}: \VEC{V}^{\VEC{\alpha}} \cdot \MMM_2
				\end{array}
			}{}
			\NPLINE{
				\MIMPLIES \
				\begin{array}[t]{l}
					\Mexists \VEC{M}. \begin{array}[t]{l}
						\LTCDERIVEDVALUE{\VEC{M}_1}{\GAMMA}{\MMM \cdot \MMM_2}{\VEC{V}_1}
						\
						\MAND \ 
						\LTCDERIVEDVALUE{\VEC{M}_2}{\GAMMA \PLUSV \VEC{X}_1}{\MMM \cdot \MMM_2 \cdot \VEC{X}_1 : \VEC{V}_1}{\VEC{V}_2}
						\
						\MAND \ ...
						\\
						\MAND \ 
						\LTCDERIVEDVALUE{\VEC{M}_{i+1}}{\GAMMA \PLUSV \VEC{X}_1 \PLUSV ... \PLUSV \VEC{X}_i}{\MMM \cdot \MMM_2 \cdot \VEC{X}_1 : \VEC{V}_1 \cdot ... \cdot \VEC{X}_i : \VEC{V}_i}{\VEC{V}_{i+1}}
						\
						\MAND \ ...
					\end{array} 
				\end{array}
			}{4.5cm}{\raggedleft
					Lemma \ref{lem:eval_under_extensions_are_equivalent}
					\
					$\MMM \EXTSTAR \MMM \cdot \MMM_2$
					\\
					($\AN{\VEC{V}_2} \cap \AN{\MMM_2} \subseteq \AN{\MMM \cdot \VEC{X}_1:\VEC{V}_1}$)
					\\
					$\MMM \cdot \VEC{X}_1:\VEC{V}_1 \EXTSTAR \MMM \cdot \VEC{X}_1:\VEC{V}_1 \cdot \MMM_2$
					\\
					...
			}
			\NLINE{
				\MIMPLIES \
				\begin{array}[t]{l}
					\Mexists \VEC{M}. \begin{array}[t]{l}
						\LTCDERIVEDVALUE{\VEC{M}_1}{\GAMMA \PLUSG \GAMMA_2}{\MMM \cdot \MMM_2}{\VEC{V}_1}
						\\
						\MAND \ 
						\LTCDERIVEDVALUE{\VEC{M}_2}{\GAMMA \PLUSG \GAMMA_2 \PLUSV \VEC{X}_1}{\MMM \cdot \MMM_2 \cdot \VEC{X}_1 : \VEC{V}_1}{\VEC{V}_2}
						\
						\MAND \ ...
						\\
						\MAND \ 
						\LTCDERIVEDVALUE{\VEC{M}_{i+1}}{\GAMMA \PLUSG \GAMMA_2 \PLUSV \VEC{X}_1 \PLUSV ... \PLUSV \VEC{X}_i}{\MMM \cdot \MMM_2 \cdot \VEC{X}_1 : \VEC{V}_1 \cdot ... \cdot \VEC{X}_i : \VEC{V}_i}{\VEC{V}_{i+1}}
						\
						\MAND \ ...
					\end{array} 
				\end{array}\hspace{-1cm}
			}{$\begin{array}[t]{r}
				\SEM{\GAMMA}{\MMM \cdot \MMM_2} \subset \SEM{\GAMMA \PLUSG \GAMMA_2}{\MMM \cdot \MMM_2}
				\\
				\SEM{\GAMMA \PLUSV \VEC{X}_1}{\MMM \cdot \MMM_2 \cdot \VEC{X}_1:\VEC{V}_1} \subset \SEM{\GAMMA \PLUSV \VEC{X}_1 \PLUSG \GAMMA_2}{\MMM \cdot \VEC{X}_1:\VEC{V}_1 \cdot \MMM_2}
				\\
				...
				\end{array}$}
			\NLASTLINE{
				\MIFF \
				\ \MMM \cdot \VEC{Y}: \VEC{W}^{\VEC{\beta}} \EXTSTAR \MMM \cdot \VEC{X}: \VEC{V}^{\VEC{\alpha}} \cdot \VEC{Y}: \VEC{W}^{\VEC{\beta}}
			}{}
		\end{NDERIVATION}
	\end{proof}
\end{lemma}

%% file: appendix/appendix_soundness_axioms.tex
\section{Soundness of Axioms}

We now prove the soundness of the axioms in Sec.~\ref{sec:axioms} split up into sections for each type of axiom.

\label{sec:apndx_soundness_axioms} \label{appendix_soundness_axioms}
\subsection{Soundness of Axioms for $e=e'$}

\input{appendix/soundness_proof_axiom_eq}

\subsection{Soundness of Axioms for $\FORALL{x}{\Gamma}A$}
\input{appendix/soundness_proof_axiom_u}

\subsection{Soundness of Axioms for $\FRESH{x}{\Gamma}$}

\input{appendix/soundness_proof_axiom_bb}

\subsection{Soundness of Axioms for $\FAD{\TCV}A$}

\input{appendix/soundness_proof_axiom_UTC}

\subsection{Soundness of Axioms for $\ONEEVAL{u}{e}{m}{A}$}

\input{appendix/soundness_proof_axiom_e}

%% file: appendix/soundness_proof_axiom_eq.tex
The following axioms are proven trivially: \\ $(eq2) \equiv x=x$, $(eq3) \equiv x=y \PIFF y=x$, $(eq4) \equiv x=y \PAND y=z \PIMPLIES x=z$
\subsubsection{\DONE Soundness proof of Axiom $(eq1)$:} 
\[
	A(x) \PAND x=e \PIMPLIES A(x) \LSUBST{e}{x}
\]

\begin{proof}
	Syntactically: if $x$ occurs in an LTC and $e$ contains a destructor as in Sec. \ref{def:sec:logical_substitution}, then $\GAMMA_0$ in $A$ then $\EXPRESSIONTYPES{\GAMMA_0}{e}{\alpha}$ holds, otherwise we just add all variables in $e$ to $\GAMMA_0$ s.t. the order is maintained by the global LTC.
\begin{NDERIVATION}{1}
	\NLINE{\text{Assume $\GAMMA$, $\MMM^{\GAMMA}$ s.t. $\FORMULATYPES{\GAMMA}{A(x) \PAND x=e \PIMPLIES A(x)\LSUBST{e}{x}}$ then:}}{}
	\NLINE{\MMM^{\GAMMA} \models A(x) \PAND x=e}{Assume}
	\NLINE{\text{$e$ free for $x$ in $A(x)$}}{Assume}
	\NLINE{\MMM \models A(x) \ \MAND \ \MMM \models x=e}{Sem. $\MAND$}
	\NLINE{\MMM \models A(x)
			\ \MAND \
			\SEM{x}{\MMM} \CONGCONTEXT{\alpha}{\AN{\MMM}} \SEM{e}{\MMM}
	}{Sem. $=$}
	\NLINE{\MMM \models A(x) 
		\ \MAND \
		\SEM{x}{\MMM} \CONGCONTEXT{\alpha}{\AN{\MMM}} \SEM{e}{\MMM}
		\ \MAND \
		(\AN{\MMM}, \SEM{e}{\MMM}) \CONV (G', V_e) 
	}{\ASSUMETERMINATION}
	\NLINE{\MMM \models A(x)
		\ \MAND \
		\SEM{x}{\MMM} \CONGCONTEXT{\alpha}{\AN{\MMM}} \SEM{e}{\MMM}
		\ \MAND \
		\SEM{e}{\MMM} \CONGCONTEXT{\alpha}{\AN{\MMM}} V_e
		\ \MAND \
		(\AN{\MMM}, \SEM{e}{\MMM}) \CONV (G', V_e) 
	}{Lemma \ref{lem:expression_Cong_evaluation}}
	\NLINE{\MMM \models A(x)
		\ \MAND \
		\SEM{x}{\MMM} \CONGCONTEXT{\alpha}{\AN{\MMM}} V_e
		\ \MAND \
		(\AN{\MMM}, \SEM{e}{\MMM}) \CONV (G', V_e) 
	}{$e \CONG e' \PAND e' \CONG e'' \MIMPLIES e \CONG e''$}
	\NPLINE{ \MIMPLIES 
		\begin{array}[t]{l}
			\MMM \models A(x)
			\ \MAND \
			(\AN{\MMM}, \SEM{e}{\MMM}) \CONV (G', V_e) 
			\\ \MAND \
			\Mforall x', A_0(x'). 
			\begin{array}[t]{l}
				x' \notin \DOM{\GAMMA} \MAND \FORMULATYPES{\GAMMA \cdot x':\GAMMA(x)}{A_0(x')} 
				\\ \MIMPLIES \
				(\MMM \cdot x':\MMM(x) \models A_0(x)\LSUBST{x'}{x}
				\MIFF
				\MMM \cdot x':V_e \models A_0(x)\LSUBST{x'}{x})
			\end{array}
		\end{array}
	\hspace{-1cm}
	}{4cm}{
		Lemma \ref{lem:congruent_base_model_extensions_model_equivelently}
		\\
		$B^{-x,x'}(x') \equiv B(x) \LSUBST{x'}{x}$
		\\
		($\SEM{e}{\MMM} \equiv e\MMM$)
	}
	\NPLINE{ \MIFF \
		\begin{array}[t]{l}
			\MMM \models A(x) 
			\ \MAND \
			(\AN{\MMM}, \SEM{e}{\MMM}) \CONV (G', V_e) 
			\\ \MAND \
			\Mforall x'. x' \notin \DOM{\GAMMA} \MIMPLIES 
			(\MMM \models A(x)
			\MIFF
			\MMM \cdot x':V_e \models A(x'))
		\end{array}
	}{4.5cm}{
		Let $A_0 \equiv A$ 
		\\
		$\MMM \cdot x':\MMM(x) \models A(x') \MIFF \MMM \models A(x)$
	}
	\NLINE{\MIMPLIES \
		(\AN{\MMM}, \SEM{e}{\MMM}) \CONV (G', V_e) 
		\ \MAND \
		\Mforall x'. x' \notin \DOM{\GAMMA} \MIMPLIES 
		\MMM \cdot x':V_e \models A(x')
	}{$B \MAND (B \MIFF C) \ \MIMPLIES \ C$ }
	\NPLINE{\MIMPLIES \
		\Mforall x'. x' \notin \DOM{\GAMMA} \ \MIMPLIES \
		((\AN{\MMM}, \SEM{e}{\MMM}) \CONV (G', V_e) 
		\ \MAND \
		\MMM \cdot x':V_e \models A(x'))
	}{3cm}{$A \MAND \Mforall x. (B \MIMPLIES C)$ \\ $\MIMPLIES \ \Mforall x. (B \MIMPLIES A \MAND C)$}
	\NLINE{\MIFF \
		\Mforall x'. x' \notin \DOM{\GAMMA} \ \MIMPLIES \
		\MMM \models A(x')\LSUBST{e}{x'}
	}{Sem. $A\LSUBST{e}{x'}$ ($x' \notin \FV{\MMM}$)}
	\NLINE{\MIFF \
		\MMM \models A(x)\LSUBST{e}{x}
	}{Sem. $A\LSUBST{e}{x}$ ($x \in \FV{\MMM}$)}
	\NLASTLINE{\text{Hence: } \Mforall \MMM^{\GAMMA}. \FORMULATYPES{\GAMMA}{A(x) \PAND x=e \PIMPLIES A(x)\LSUBST{e}{x}} \MIMPLIES \MMM \models A(x) \PAND x=e \PIMPLIES A(x)\LSUBST{e}{x}}{lines 1-14}
\end{NDERIVATION}
\end{proof}

%% file: appendix/soundness_proof_axiom_u.tex
\newpage
\subsubsection{\DONE Soundness proof of Axiom $(u1)$:}
\[
\FORALL{x^{\alpha}}{\GAMMA_0} A  \ \PAND \EXPRESSIONTYPES{\GAMMA_0}{e}{\alpha} \quad \PIMPLIES \quad A \LSUBST{e}{x}
\]
\begin{proof}
\begin{NDERIVATION}{1}
	\NLINE{\text{Assume: $\GAMMA$, $\MMM^{\GAMMA_d}$ s.t. $\CONSTRUCT{\GAMMA}{\MMM} \ \MAND \ \FORMULATYPES{\GAMMA}{\FORALL{x^{\alpha}}{\GAMMA_0} A  \ \PIMPLIES  \ A \LSUBST{e}{x}}$ and $\EXPRESSIONTYPES{\GAMMA_0}{e}{\alpha}$}}{}
	\NLINE{\text{Assume: } \MMM \models \FORALL{x^{\alpha}}{\GAMMA_0}A \text{ and } \TYPES{\GAMMA_0}{e}{\alpha}}{Note $e\in  \{x,c,Op(\VEC{e}), \PAIR{e_1}{e_2} \pi_i(e_1) \}$}
	\NLINE{
		\Mforall M_x.
			\LTCDERIVEDVALUE{M_x}{\GAMMA_0}{\MMM}{V_x}
			\ \MIMPLIES \  
			\MMM \cdot x:V_x \models A
	}{Sem.  $\FORALL{}{}$, 2}
	\NLINE{e-\text{expression} \ \MIMPLIES \ e-\text{term} \ \MIMPLIES \ \AN{e}=\emptyset}{Lemma \ref{lem:expressions_are_name_free}}
	\NLINE{
		\TCTYPES{\GAMMA}{\GAMMA_0} \ \MIMPLIES \ (\EXPRESSIONTYPES{\GAMMA_0}{e}{\alpha} \ \MIMPLIES \ \EXPRESSIONTYPES{\GAMMA}{e}{\alpha})
	}{Typing rules}
	\NLINE{
		\EXPRESSIONTYPES{\GAMMA_0}{e}{\alpha} \ \MIMPLIES \ \Mexists V_e. (\AN{\MMM}, \SEM{e}{\MMM}) \CONV (\AN{\MMM}, V_e)
	}{Lemma \ref{lem:expressions_cannot_create_new_names}, 6}
	\NLINE{
		\LTCDERIVEDVALUE{e}{\GAMMA_0}{\MMM}{V_e}
	}{Sem. $\LTCDERIVEDVALUE{}{}{}{}$, 4, 5, 6, $\SEM{e}{\MMM} \equiv e\MMM$}
	\NLINE{
		\LTCDERIVEDVALUE{e}{\GAMMA_0}{\MMM}{V_e} \ \MIMPLIES \ \MMM \cdot x:V_e \models A
	}{Instantiate $\Mforall M_x$ to $e$, 3}
	\NLINE{
		\MIMPLIES \ 
		\LTCDERIVEDVALUE{e}{\GAMMA}{\MMM}{V_e} \
			\MAND \ \MMM \cdot x:V_e \models A
	}{MP, 7, 8}
	\NLINE{
		\MMM \models A\LSUBST{e}{x}
	}{Sem. $\LSUBST{e}{x}$, $x \notin \DOM{\MMM}$}
	\NLINE{
		\text{Hence } \MMM \models \FORALL{x^{\alpha}}{\GAMMA_0}A \text{ and } \TYPES{\GAMMA_0}{e}{\alpha} \text{ implies }  \MMM \models A\LSUBST{e}{x}
	}{lines 2-10}
	\NLASTLINE{
		\text{Hence } 
		\begin{array}[t]{l} 
		\Mforall \GAMMA, \GAMMA_0, A, e. \
		\FORMULATYPES{\GAMMA}{\FORALL{x^{\alpha}}{\GAMMA_0} A  \ \PIMPLIES  \ A \LSUBST{e}{x}} 
		\ \MAND \ \EXPRESSIONTYPES{\GAMMA_0}{e}{\alpha} 
		\\ \MIMPLIES \
		\Mforall \MMM. \CONSTRUCT{\GAMMA}{\MMM} \ \MIMPLIES \
		\MMM \models ((\FORALL{x^{\alpha}}{\GAMMA_0}A)  \PIMPLIES (A\LSUBST{e}{x}))
		\end{array}
	}{lines 1-11}
\end{NDERIVATION}
\end{proof}
\subsubsection{\DONE Soundness proof of Axiom $(u2)$:}
\[
\FORALL{x^{\alpha}}{\GAMMA_0 \PLUSG \GAMMA_1} A \quad \PIMPLIES \quad (\FORALL{x^{\alpha}}{\GAMMA_0} A ) \PAND (\FORALL{x^{\alpha}}{\GAMMA_1} A )
\]
\begin{proof}
\begin{NDERIVATION}{1}
	\NLINE{\text{Assume: $\GAMMA$, $\MMM^{\GAMMA_d}$ s.t. $\CONSTRUCT{\GAMMA}{\MMM} \ \MAND \ \FORMULATYPES{\GAMMA}{\FORALL{x^{\alpha}}{\GAMMA \PLUSG \GAMMA'} A \quad \PIMPLIES \quad (\FORALL{x^{\alpha}}{\GAMMA}A ) \PAND (\FORALL{x^{\alpha}}{\GAMMA'} A )}$}}{}
	\NLINE{\text{Assume: } \MMM \models \FORALL{x^{\alpha}}{\GAMMA_0,\GAMMA_1}A }{}
	\NLINE{\text{Prove: } \MMM \models (\FORALL{x^{\alpha}}{\GAMMA_0}A ) \PAND (\FORALL{x^{\alpha}}{\GAMMA_1} A ) }{}
	\NLINE{
		\Mforall M.
		\LTCDERIVEDVALUE{M}{\GAMMA_0 \PLUSG \GAMMA_1}{\MMM}{V} \
		\MIMPLIES \ \MMM \cdot x:V \models A
	}{Sem. $\FORALL{x}{\GAMMA_0,\GAMMA_1}$,1}
	\NLINE{
		\MIMPLIES
		\Mforall M.
		\LTCDERIVEDVALUE{M}{\GAMMA_0}{\MMM}{V} \
		\MIMPLIES \ \MMM \cdot x:V \models A
	}{$\begin{array}[t]{r}
	\Mforall V. \TYPES{\GAMMA_i}{V}{\alpha}  \ ... \ \subseteq \ \Mforall V. \TYPES{\GAMMA_0 \PLUSG \GAMMA_1}{V}{\alpha} \ ... \end{array}$}
	\NLINE{\MMM \models \FORALL{x^{\alpha}}{\GAMMA_0}A }{Sem.  $\FORALL{}{}$}
	\NLINE{\MMM \models \FORALL{x^{\alpha}}{\GAMMA_1}A }{Same as 3-5 but with $\GAMMA_1$}
	\NLINE{\MIMPLIES \MMM \models (\FORALL{x^{\alpha}}{\GAMMA_0}A) \PAND (\FORALL{x^{\alpha}}{\GAMMA_1}A )}{5,6}
	\NLASTLINE{\text{Hence: } 
		\begin{array}[t]{l}
		\Mforall \GAMMA. \FORMULATYPES{\GAMMA}{\FORALL{x^{\alpha}}{\GAMMA_0 \PLUSG \GAMMA_1} A \ \PIMPLIES \ (\FORALL{x^{\alpha}}{\GAMMA_0} A ) \PAND (\FORALL{x^{\alpha}}{\GAMMA_1} A )}
		\\
		\MIMPLIES \ \Mforall \MMM. \CONSTRUCT{\GAMMA}{\MMM} 
		\ \MIMPLIES \
		\MMM \models (\FORALL{x^{\alpha}}{\GAMMA_0}A) \PAND (\FORALL{x^{\alpha}}{\GAMMA_1}A )
		\end{array}
	}{lines 1-8}
\end{NDERIVATION}
\end{proof}

\subsubsection{\DONE Soundness proof of Axiom $(u3)$:}

\[
{A^{-x}  \PIFF  \FORALL{x^{\alpha}}{\GAMMA_0}A}
\]
\begin{proof}
	Clearly $\FORALL{x}{\GAMMA_0} A \ \PIMPLIES \ A$ applying $(u1)$ given $A^{-x}\LSUBST{e}{x} \equiv A$ and the $x$ derived from $\GAMMA_0$ (say $V$) must imply an extension $\MMM \EXTSTAR \MMM \cdot x:V$ which given  $A$-\EXTINDEP implies this direction.
 	\\
 	Then to prove $A^{-x} \PIMPLIES \FORALL{x}{\GAMMA_0} A$:
 		\begin{NDERIVATION}{1}
 			\NLINE{\text{Assume: $\GAMMA, \GAMMA_0, A$ s.t. 
 					$\FORMULATYPES{\GAMMA}{A^{-x} \PIMPLIES \FORALL{x^{\alpha}}{\GAMMA_0}A}   \ \MAND \ A-\EXTINDEP$ }}{}
 			\NLINE{\text{Assume: $\MMM^{\GAMMA_d}$ s.t. $\CONSTRUCT{\GAMMA}{\MMM}$ and } \MMM \models A^{-x} }{}
 			\NLINE{
 				\MIMPLIES \
 				\Mforall M.\
 				\LTCDERIVEDVALUE{M}{\GAMMA_0}{\MMM}{V}
 				\ \MIMPLIES \ 
 				\MMM \models A 
 			}{Tautology}
 			\NLINE{
 				\MIMPLIES \
 				\Mforall M.\
 				\LTCDERIVEDVALUE{M}{\GAMMA_0}{\MMM}{V}
 				\ \MIMPLIES \ 
 				\MMM \cdot x:V \models A 
 			}{Sem. $\EXTSTAR$, Lemmas \ref{lem:LTCDERIVED_subset_implies_LTCDERIVED_supset}  and $A$-\EXTINDEP}
 			\NLINE{
 				\MMM \models \FORALL{x^{\alpha}}{\GAMMA_0} A
 			}{Sem.  $\FORALL{}{}$}
 			\NLINE{
 				\MMM \models A \ \MIMPLIES \ \MMM \models \FORALL{x^{\alpha}}{\GAMMA_0} A
 			}{2-9}
 			\NLINE{
 				\MMM \models A \ \PIMPLIES \ \FORALL{x^{\alpha}}{\GAMMA_0} A
 			}{Sem. $\PIMPLIES$}
 			\NLASTLINE{\text{Hence: }
 				\begin{array}[t]{l}
 					\Mforall \GAMMA. 
 					\FORMULATYPES{\GAMMA}{ A \PIMPLIES \FORALL{x^{\alpha}}{\GAMMA_0} A}
 					\\
 					\PAND \Mforall \MMM^{\GAMMA_d}. \CONSTRUCT{\GAMMA}{\MMM} \MIMPLIES \MMM \models  A \PIMPLIES \FORALL{x^{\alpha}}{\GAMMA_0} A
 				\end{array} 					
 			}{lines 1-7}
 		\end{NDERIVATION}
 \end{proof}

\newpage
\subsubsection{\DONE Soundness proof of Axiom $(u4)$:}
\[
(\FORALL{x^{\alpha}}{\GAMMA_0}(A \PAND B))  \PIFF   (\FORALL{x^{\alpha}}{\GAMMA_0}A) \PAND  (\FORALL{x^{\alpha}}{\GAMMA_0}B)
\]
\begin{proof}
	Assume $\GAMMA, \MMM^{\GAMMA_d}, A$
	\\ 
	such that 
	$\CONSTRUCT{\GAMMA}{\MMM} \ \MAND \ \FORMULATYPES{\GAMMA}{(\FORALL{x^{\alpha}}{\GAMMA_0}(A \PAND B))  \PIFF   (\FORALL{x^{\alpha}}{\GAMMA_0}A) \PAND  (\FORALL{x^{\alpha}}{\GAMMA_0}B)}  $
	\\
$\PIMPLIES$:
\begin{NDERIVATION}{1}
	\NLINE{\text{Assume: } \MMM \models (\FORALL{x^{\alpha}}{\GAMMA_0}(A \PAND B)) }{ }
	\NLINE{
		\Mforall M.
		\LTCDERIVEDVALUE{M}{\GAMMA_0}{\MMM}{V}
		\ \MIMPLIES \
		\left(
		\begin{array}{l}
			\MMM \cdot x:V\models A
			\\
			\MAND \ \MMM \cdot x:V\models B
		\end{array}
		\right)
	}{Sem.  $\FORALL{}{}$, $\MAND$}
	\NLINE{
		\begin{array}[t]{l}
			\Mforall M.
			\LTCDERIVEDVALUE{M}{\GAMMA_0}{\MMM}{V}
			\ \MIMPLIES \
			\left(
			\begin{array}{l}
				\MMM \cdot x:V\models A
				\\
				\MAND \ \MMM \cdot x:V\models B
			\end{array}
			\right)
			\\
			\MAND 
			\\
			\Mforall M.
			\LTCDERIVEDVALUE{M}{\GAMMA_0}{\MMM}{V}
			\ \MIMPLIES \
			\left(
			\begin{array}{l}
				\MMM \cdot x:V\models A
				\\
				\MAND \ \MMM \cdot x:V\models B
			\end{array}
			\right)
		\end{array}
	}{FOL $A \MIMPLIES A \MAND A$}
	\NLINE{ \MIMPLIES
	\begin{array}[t]{ll}
		\Mforall M.
		\LTCDERIVEDVALUE{M}{\GAMMA_0}{\MMM}{V}
		\ \MIMPLIES \
		\MMM \cdot x:V\models A
		\\
		\MAND 
		\\
		\Mforall M.
		\LTCDERIVEDVALUE{M}{\GAMMA_0}{\MMM}{V}
		\ \MIMPLIES \
		\MMM \cdot x:V\models B
	\end{array}
	}{FOL $\MAND$-elim}
	\NLINE{
		\MMM \models \FORALL{x^{\alpha}}{\GAMMA_0} A  \ \MAND \ \MMM \models \FORALL{x^{\alpha}}{\GAMMA_0} B
	}{Sem.  $\FORALL{}{}$}
	\NLASTLINE{
		\MMM \models (\FORALL{x^{\alpha}}{\GAMMA_0} A) \ \PAND \ (\FORALL{x^{\alpha}}{\GAMMA_0} B)
	}{Sem.  $\MAND$}
\end{NDERIVATION}

$\leftarrow$:
\begin{NDERIVATION}{1}
	\NLINE{\text{Assume: } \MMM \models (\FORALL{x^{\alpha}}{\GAMMA_0}A) \MAND  (\FORALL{x^{\alpha}}{\GAMMA_0}B)}{}
	\NLINE{ 
		\begin{array}[t]{ll}
			\Mforall M.
			\LTCDERIVEDVALUE{M}{\GAMMA_0}{\MMM}{V}
			\ \MIMPLIES \ \MMM \cdot x:V \models A
			\\
			\MAND 
			\\
			\Mforall M.
			\LTCDERIVEDVALUE{M}{\GAMMA_0}{\MMM}{V}
			\ \MIMPLIES \ \MMM \cdot x:V \models B
		\end{array}
	}{Sem.  $\FORALL{}{}$, $\MAND$, 1}
	\NLINE{ 
		\begin{array}[t]{ll}
			\Mforall M.
			\begin{array}[t]{ll}
				\LTCDERIVEDVALUE{M}{\GAMMA_0}{\MMM}{V}
				\ \MIMPLIES \ \MMM \cdot x:V \models A
				\\
				\MAND 
				\\
				\LTCDERIVEDVALUE{M}{\GAMMA_0}{\MMM}{V}
				\ \MIMPLIES \ \MMM \cdot x:V \models B
			\end{array}
		\end{array}
	}{extract the common $\Mforall M$}
	\NLINE{ 
	\begin{array}[t]{ll}
		\Mforall M.
		\LTCDERIVEDVALUE{M}{\GAMMA_0}{\MMM}{V}
		\ \MIMPLIES 
		\left(
		\begin{array}{l}
			\MMM \cdot x:V \models A
			\\
			\MAND \ \MMM \cdot x:V \models B
		\end{array}
		\right)
	\end{array}
	}{\parbox[t]{3.7cm}{\raggedleft FOL \\ $((A\MIMPLIES B) \MAND (A \MIMPLIES C))$ \\ $\MIMPLIES \ (A \MIMPLIES(B \MAND C))$}}
	\NLASTLINE{
		\MMM \models \FORALL{x^{\alpha}}{\GAMMA_0} (A\PAND B)
	}{Sem. $\MAND$, $\FORALL{x}{\GAMMA}$}
\end{NDERIVATION}
\end{proof}

\newpage
\subsubsection{\DONE Soundness proof of Axiom $(u5)$:}

\[
\FORALL{x^{\alpha}}{\GAMMA_0}  A \ \PIFF \ \FORALL{x^{\alpha}}{\emptyset}  A \qquad \text{ iff  $\alpha$ is $\NAME$-free}
\]
\begin{proof}
Trivial using Lemma \ref{lem:Nm-free_terms_have_equivalent_name_free_STLC-term} as they both quantify over the same set of values.
\end{proof}

\subsubsection{\DONE Soundness proof of Axiom $(ex1)$:}
\[
(\TCTYPES{\GAMMA}{\GAMMA_0} \MAND
\EXPRESSIONTYPES{\GAMMA_0}{e}{\alpha} )
\quad \MIMPLIES \quad
A\LSUBST{e}{x} \PIMPLIES \EXISTS{x}{\GAMMA_0} x = e \PAND  A
\]

\begin{proof}
	Assume $\GAMMA^{-x}, \MMM^{\GAMMA_d}, \GAMMA_0^{-x}, e, A$
	\\ 
	such that 
	$\FORMULATYPES{\GAMMA}{A\LSUBST{e}{x} \PIMPLIES \EXISTS{x}{\GAMMA} A}  $, 
	$\TCTYPES{\GAMMA}{\GAMMA_0}$ and
	$\EXPRESSIONTYPES{\GAMMA_0}{e}{\alpha}$
	\\
	\begin{NDERIVATION}{1}
		\NLINE{\text{Assume: $\GAMMA$ s.t. $\FORMULATYPES{\GAMMA}{A\LSUBST{e}{x} \PIMPLIES \EXISTS{x}{\GAMMA_0} x = e \PAND  A}$} }{}
		\NLINE{\text{Assume: $\GAMMA_0$, $e$ s.t. $\TCTYPES{\GAMMA}{\GAMMA_0}$ and $\EXPRESSIONTYPES{\GAMMA_0}{e}{\alpha}$ } }{}
		\NLINE{\text{Assume: $\MMM^{\GAMMA_d}$ s.t. $\CONSTRUCT{\GAMMA}{\MMM}$ and $\MMM\models A\LSUBST{e}{x}$}}{}
		\NLINE{\MIFF \
			\LTCDERIVEDVALUE{e}{\GAMMA}{\MMM}{V} \ \MAND \ \MMM \cdot x:V \models A
		}{Sem. $A\LSUBST{e}{x}$, $x \notin \DOM{\MMM}$}
		\NLINE{\MIFF \
			\LTCDERIVEDVALUE{e}{\GAMMA_0}{\MMM}{V} \ \MAND \ \MMM \cdot x:V \models A \ \MAND \ e(\MMM \cdot x:V) \CONGCONTEXT{\alpha}{\AN{\MMM \cdot x:V}} x(\MMM \cdot x:V)
		}{Lemma \ref{lem:expression_Cong_evaluation}}
		\NLINE{\MIFF \
			\LTCDERIVEDVALUE{e}{\GAMMA_0}{\MMM}{V} \ \MAND \ \MMM \cdot x:V \models x=e \PAND A
		}{Sem. $=$, $\PAND$}
		\NLINE{\MIMPLIES \
			\Mexists M_e. \LTCDERIVEDVALUE{M_e}{\GAMMA_0}{\MMM}{V} \ \MAND \ \MMM \cdot x:V \models  x=e \PAND A
		}{Clearly holds for $M_e = e$}
		\NLASTLINE{\MIFF \ \MMM \models \EXISTS{x}{\GAMMA_0}  x=e \PAND A}{Sem. $\EXISTS{x}{\GAMMA_0}$}
	\end{NDERIVATION}
\end{proof}

\subsubsection{\DONE Soundness proof of Axiom $(ex2)$:}
Assuming $\{ a,b \} \subseteq \DOM{\GAMMA_0}$
\[
\GAMMA \PLUSV x \PLUSTC \GAMMA_0 \Vdash 
\ONEEVAL{a}{b}{c}{c=x}
\
\PIMPLIES \ \EXISTS{x'}{\GAMMA_0} x=x'
\]

\begin{proof}
	\begin{NDERIVATION}{1}
		\NLINE{\text{Assume: $\MMM^{\GAMMA_d}$ s.t. $\CONSTRUCT{\GAMMA \PLUSV x \PLUSTC \GAMMA_0}{\MMM}$} }{}
		\NLINE{\text{Assume: } \MMM \models \ONEEVAL{a}{b}{c}{c=x}}{}
		\NLINE{\MIMPLIES \LTCDERIVEDVALUE{a b}{\GAMMA_0}{\MMM}{V_c} \ \MAND \ \MMM \cdot c:V_c \models x=c
		}{Sem. $\ONEEVAL{}{}{}{}$}
		\NLINE{\MIMPLIES \Mexists M_x'.  \LTCDERIVEDVALUE{M_x'}{\GAMMA_0}{\MMM}{V_c} \MAND \MMM \cdot c:V_c \models x=c
		}{$M_x' \equiv a b$ ($\{ a,b \} \subseteq \DOM{\GAMMA_0}$)} 
		\NLINE{\MIMPLIES 
			\MMM \models \EXISTS{x'}{\GAMMA_0} x=x'
		}{Sem. $\EXISTS{}{}$}
		\NLASTLINE{\text{Hence: }  
			\begin{array}[t]{l}
				\Mforall \GAMMA \PLUSV x \PLUSTC \GAMMA_0. 
				\GAMMA \PLUSV x \PLUSTC \GAMMA_0 \Vdash 
				\ONEEVAL{a}{b}{c}{c=x} \ \PIMPLIES \ \EXISTS{x'}{\GAMMA_0} x=x
				\\
				\MIMPLIES \
				\Mforall \MMM^{\GAMMA_d}. 
				\CONSTRUCT{\GAMMA \PLUSV x \PLUSTC \GAMMA_0}{\MMM} 
				\ \MIMPLIES \
				\MMM \models \ONEEVAL{a}{b}{c}{c=x} \ \PIMPLIES \ \EXISTS{x'}{\GAMMA_0} x=x'
			\end{array}
		}{}
	\end{NDERIVATION}
\end{proof}

\subsubsection{\DONE Soundness proof of Axiom $(ex3)$:}
\[
\GAMMA \PLUSV x \Vdash 
\FORALL{y^{\NAME}}{\emptyset}\EXISTS{z^{\NAME}}{\GAMMA_0 \PLUSV y} x=z
\
\PIMPLIES
\
\EXISTS{z}{\GAMMA_0} x=z
\]

\begin{proof}
	\begin{NDERIVATION}{1}
		\NLINE{\text{Assume some model $\MMM^{\GAMMA_d}$ s.t. $\CONSTRUCT{\GAMMA \PLUSV x}{\MMM}$ and $\MMM\models \FORALL{y^{\NAME}}{\emptyset}\EXISTS{z^{\NAME}}{\GAMMA_0 \PLUSV y} x=z$} }{}
		\NLINE{\MIMPLIES 
			\Mforall M_y^{\NAME}. \LTCDERIVEDVALUE{M_y}{\emptyset}{\MMM}{V_y} 
			\MIMPLIES 
			\Mexists M_z^{\NAME}. \LTCDERIVEDVALUE{M_z}{\GAMMA_0 \PLUSV y}{\MMM \cdot y:V_y}{V_z} 
			\MAND 
			\MMM \cdot y:V_y \cdot z:V_z \models x=z
		}{}
		\NLINE{\MIMPLIES 
			\Mforall M_y^{\NAME}. \LTCDERIVEDVALUE{M_y}{\emptyset}{\MMM}{V_y} 
			\MIMPLIES 
			\Mexists M_z^{\NAME}. \LTCDERIVEDVALUE{M_z}{\GAMMA_0}{\MMM \cdot y:V_y}{V_z} 
			\MAND 
			\MMM \cdot y:V_y \cdot z:V_z \models x=z
		}{replace $y$ with $M_y$}
		\NLINE{\MIMPLIES 
			\Mforall M_y^{\NAME}. \LTCDERIVEDVALUE{M_y}{\emptyset}{\MMM}{V_y} 
			\MIMPLIES 
			\Mexists M_z^{\NAME}. \LTCDERIVEDVALUE{M_z}{\GAMMA_0}{\MMM}{V_z} 
			\MAND 
			\SEM{x}{\MMM \cdot y:V_y \cdot z:V_z} \CONGCONTEXT{\NAME}{\AN{\MMM \cdot y:V_y \cdot z:V_z}} \SEM{z}{\MMM \cdot y:V_y \cdot z:V_z}
		}{Lemma \ref{lem:adding/remove_unused_names_maintains_evaluation}}
		\NLINE{\MIMPLIES 
			\Mforall M_y^{\NAME}. \LTCDERIVEDVALUE{M_y}{\emptyset}{\MMM}{V_y} 
			\MIMPLIES 
			\Mexists M_z^{\NAME}. \LTCDERIVEDVALUE{M_z}{\GAMMA_0}{\MMM}{V_z} 
			\MAND 
			\SEM{x}{\MMM \cdot z:V_z} \CONGCONTEXT{\NAME}{\AN{\MMM \cdot y:V_y \cdot z:V_z}} \SEM{z}{\MMM \cdot z:V_z}
		}{$\SEM{x}{\MMM \cdot y:V_y} = \SEM{x}{\MMM}$}
		\NLINE{\MIMPLIES 
			\Mforall M_y^{\NAME}. \LTCDERIVEDVALUE{M_y}{\emptyset}{\MMM}{V_y} 
			\MIMPLIES 
			\Mexists M_z^{\NAME}. \LTCDERIVEDVALUE{M_z}{\GAMMA_0}{\MMM}{V_z} 
			\MAND 
			\SEM{x}{\MMM \cdot z:V_z} \CONGCONTEXT{\NAME}{\AN{\MMM \cdot z:V_z}} \SEM{z}{\MMM \cdot z:V_z}
		}{Lemma \ref{lem:adding/removing_names_to_congruence_makes_no_difference}}
		\NLINE{\MIMPLIES
			\Mexists M_z^{\NAME}. \LTCDERIVEDVALUE{M_z}{\GAMMA_0}{\MMM}{V_z} 
			\MAND 
			\SEM{x}{\MMM \cdot z:V_z} \CONGCONTEXT{\NAME}{\AN{\MMM \cdot z:V_z}} \SEM{z}{\MMM \cdot z:V_z}
		}{FOL}
		\NLINE{\MIMPLIES 
			\MMM \models \EXISTS{z}{\GAMMA_0} x=z
		}{Sem. $\EXISTS{}{}$} 
		\NLASTLINE{\text{Hence: }  
			\begin{array}[t]{l}
				\Mforall \GAMMA \PLUSV x. 
				\GAMMA \PLUSV x \Vdash 
				\FORALL{y^{\NAME}}{\emptyset}\EXISTS{z^{\NAME}}{\GAMMA_0 \PLUSV y} x=z
				\ \PIMPLIES \
				\EXISTS{z}{\GAMMA_0} x=z
				\\
				\MIMPLIES \
				\Mforall \MMM^{\GAMMA_d}. 
				\CONSTRUCT{\GAMMA \PLUSV x}{\MMM} 
				\ \MIMPLIES \
				\MMM \models 
				\FORALL{y^{\NAME}}{\emptyset}\EXISTS{z^{\NAME}}{\GAMMA_0 \PLUSV y} x=z
				\ \PIMPLIES \
				\EXISTS{z}{\GAMMA_0} x=z
			\end{array}
		}{}
	\end{NDERIVATION}
\end{proof}

%% file: appendix/soundness_proof_axiom_bb.tex
\subsubsection{\DONE Soundness proof of Axiom $(f1)$:} 
\[
\FORMULATYPES{\GAMMA \PLUSV  x:\NAME \PLUSV  f:\alpha \FS \TYBASE 
}{ \FRESH{x}{\GAMMA} \MIMPLIES \FRESH{x}{\GAMMA \PLUSV f:\alpha \FS\TYBASE}}
\]
\begin{proof}
	We use Lemma \ref{lem:Base-result-functions_cannot_reveal_their_names} so show that $f$ cannot help produce $x$ if $f:\alpha \FS \TYBASE$.
	\begin{NDERIVATION}{1}
	\NLINE{\text{Assume: $\MMM_{xf}$ s.t. $\CONSTRUCT{\GAMMA \PLUSV  x:\NAME \PLUSV  f:\alpha \FS \TYBASE}{(\MMM^{\GAMMA_0} \cdot x:r_x \cdot f:V_f)^{\GAMMA_0  \PLUSV  x:\NAME \PLUSV  f:\alpha \FS \TYBASE} \equiv \MMM_{xf}}$ }}{}
	\NLINE{\CONSTRUCT{\GAMMA\PLUSV x \PLUSV f}{\MMM_{xf}} \ \MIMPLIES \ 
			\Mexists M_x. \LTCDERIVEDVALUE{M_x}{\GAMMA}{\MMM}{r_x}
			\ \MAND \
			\Mexists M_f. \LTCDERIVEDVALUE{M_f}{\GAMMA \PLUSV x}{\MMM\cdot x:r_x}{V_f}
	}{}
	\NLINE{\text{Assume: } \MMM_{xf} \models \FRESH{x}{\GAMMA}}{Assumption}
	\NLINE{\text{Hence: } 
		\Mexists M_x. \LTCDERIVEDVALUE{M_x}{\GAMMA}{\MMM}{r_x} 
		\ \MAND \ 
		\neg \Mexists N_x. \LTCDERIVEDVALUE{N_x}{\GAMMA}{\MMM_{xf}}{r_x} 
	}{Sem. $\FRESH{x}{\GAMMA}$, 2}
	\NLINE{\text{Hence: } 
		\Mexists M_x. \LTCDERIVEDVALUE{M_x}{\GAMMA}{\MMM}{r_x} 
		\ \MAND \ 
		\neg \Mexists N_x. \LTCDERIVEDVALUE{N_x}{\GAMMA}{\MMM \cdot x:r_x}{r_x} 
	}{Lemma \ref{lem:eval_under_extensions_are_equivalent}}
	\NLINE{
		\MIMPLIES \ r_x \notin \AN{\MMM} 
	}{Lemma \ref{lem:Gamma_derived_name_and_not_derivable_from_model_plus_name_implies_fresh_name}}
	\NLINE{\text{Assume: } r_x \notin \AN{V_f} }{Then trivial as $x$-fresh and  $f$ cannot output $r_x$, Lemma \ref{lem:Name_fresh_pre_eval_implies_name_fresh_post_eval}}
	\NLINE{\text{Assume: } r_x \in \AN{V_f} }{(proof by contradiction)}
	
	\NLINE{\text{Assume: } \Mexists P_x. \LTCDERIVEDVALUE{P_x}{\GAMMA \PLUSV f}{\MMM_{xf}}{r_x}}{}

	\NLINE{
		\MIFF \
		\Mexists P_x. 
		\begin{array}[t]{l}
			\AN{P_x}=\emptyset
			\\ \MAND \  
			\TYPES{\SEM{\GAMMA}{\MMM}, f:\alpha \FS \TYBASE}{P_x}{\NAME} 
			\\
			\MAND \ (\AN{\MMM_{xf}}, \ P_x\MMM_{xf}) \CONV (\AN{\MMM_{xf}},G', \ r_x) 
		\end{array}
	}{Sem. $\LTCDERIVEDVALUE{}{}{}{}$}
	\NLINE{
		\MIFF \
		\Mexists P_x. 
		\begin{array}[t]{l}
			\AN{P_x}=\emptyset
			\\ \MAND \  
			\TYPES{\SEM{\GAMMA}{\MMM}, f:\alpha \FS \TYBASE}{P_x}{\NAME} 
			\\
			\MAND \ (\AN{\MMM_{xf}}, \ P_x\PSUBST{V_f}{f}\MMM) \CONV (\AN{\MMM_{xf}},G', \ r_x) 
		\end{array}
	}{Def $\PSUBST{V}{x}$, $P_x^{-x}$}
	\NLINE{
		\MIFF \
		\Mexists P_x. 
		\begin{array}[t]{l}
			\AN{P_x}=\emptyset
			\\ \MAND \  
			\TYPES{\SEM{\GAMMA}{\MMM}, f:\alpha \FS \TYBASE}{P_x}{\NAME} 
			\\
			\MAND \ (\AN{\MMM_{xf}}, \ (P_x\MMM)\PSUBST{V_f}{f}) \CONV (\AN{\MMM_{xf}},G', \ r_x) 
		\end{array}
	}{Def closure, $V_f$-value}
	
	\NPLINE{
		\MIFF \
		\Mexists P_x. 
		\begin{array}[t]{l}
			\AN{P_x}=\emptyset
			\\ \MAND \  
			\TYPES{\SEM{\GAMMA}{\MMM}, f:\alpha \FS \TYBASE}{P_x}{\NAME} 
			\\
			\MAND \ (\AN{\MMM_{xf}}, \ (P_x\MMM)\PSUBST{V_f}{f}) \CONV (\AN{\MMM_{xf}},G', \ r_x) 
			\\
			\MAND \ \neg(\AN{\MMM_{xf}}, \ (P_x\MMM)\PSUBST{V_f}{f}) \CONV (\AN{\MMM_{xf}},G', \ r_x) 
		\end{array}
	}{4cm}{
	$r_x \notin \MMM \MIMPLIES r_x\notin \AN{P_x\MMM}$, 
	\\
	$r_x \in \AN{V_f}$
	\\
	Lemma \ref{lem:Base-result-functions_cannot_reveal_their_names} 
	$\MIMPLIES$  $\neg (... \CONV r_x)$
	}
	\NLINE{\text{Contradiction, hence: $\neg \Mexists P_x. \LTCDERIVEDVALUE{P_x}{\GAMMA \PLUSV f}{\MMM_{xf}}{r_x}$}}{}
	\NLASTLINE{\text{Hence: } \Mforall \MMM_{xf}. \CONSTRUCT{\GAMMA \PLUSV x \PLUSV f}{\MMM_{xf}} \MIMPLIES \MMM_{xf} \models \FRESH{x}{\GAMMA} \PIMPLIES \FRESH{x}{\GAMMA \PLUSV f}}{}
	\end{NDERIVATION}
\end{proof}

\newpage
\subsubsection{\DONE  Soundness proof of Axiom $(f2)$:}

\[
		(\FRESH{x}{\GAMMA_0} \PAND \FORALL{y^{\alpha_y}}{\GAMMA_0}A )  
		\
		\PIFF
		\
		\FORALL{y^{\alpha_y}}{\GAMMA_0}(\FRESH{x}{(\GAMMA_0 \PLUSV y:\alpha_y)} \PAND A)
\]
\begin{proof}
The $\PIMPLIEDBY$ direction is elementary.
\\
The $\PIMPLIES$ direction is as follows:

\begin{NDERIVATION}{1}
	\NLINE{\text{Assume some $\GAMMA_1, \MMM^{\GAMMA_1}$ s.t. 
			$
				\FORMULATYPES{\GAMMA_1}{
					\begin{array}[t]{l}
						(\FRESH{x}{\GAMMA_0} \PAND \FORALL{y^{\alpha_y}}{\GAMMA_0}A )  
						\\
						\PIMPLIES \ \FORALL{y^{\alpha_y}}{\GAMMA_0}(\FRESH{x}{(\GAMMA_0 \PLUSV y:\alpha_y)} \PAND A)
					\end{array}
		}$} }{}
	\NLINE{\text{Assume: } \MMM \models \FRESH{x}{\GAMMA_0} \AND (\FORALL{y}{\GAMMA_0}A )  }{}
	\NLINE{
		\begin{array}[t]{l}
			\neg \Mexists M_x. \LTCDERIVEDVALUE{M_x}{\GAMMA_0}{\MMM}{\SEM{x}{\MMM}}
			\\
			\MAND \
			\Mforall M_y. 
			\begin{array}[t]{l}
				\LTCDERIVEDVALUE{M_y}{\GAMMA_0}{\MMM}{V_y}
				\\
				\ \MIMPLIES \ \MMM \cdot y:V_y \models A
			\end{array}
		\end{array} 
	}{Sem. $\begin{array}[t]{r}
		\FRESH{x}{\GAMMA_0}
		\\
		\FORALL{y}{\GAMMA_0}
		\\
		A
		\end{array}$}
	\NLINE{
		\Mforall M_y. 
		\begin{array}[t]{l}
			\LTCDERIVEDVALUE{M_y}{\GAMMA_0}{\MMM}{V_y}
			\\
			\ \MIMPLIES \
			\begin{array}[t]{l}
				\neg \Mexists M_x. \LTCDERIVEDVALUE{M_x}{\GAMMA_0}{\MMM}{\SEM{x}{\MMM}}
				\\
				\MAND \
				\MMM \cdot y:V_y \models A
			\end{array}
		\end{array}
	}{$\begin{array}[t]{r}
		A^{-M_y} \AND \Mforall M_y. \ B \MIMPLIES C 
		\\
		\MIMPLIES
		\\
		\Mforall M_y. B \MIMPLIES (A \MAND C)
		\end{array}$}
	\NPLINE{
		\Mforall M_y. 
		\begin{array}[t]{l}
			\LTCDERIVEDVALUE{M_y}{\GAMMA_0}{\MMM}{V_y}
			\\
			\ \MIMPLIES \
			\begin{array}[t]{l}
				\neg \Mexists M_x. \LTCDERIVEDVALUE{M_x}{\GAMMA_0 \PLUSV y}{\MMM \cdot y:V_y}{\SEM{x}{\MMM \cdot y:V_y}}
				\\
				\MAND \ 
				\MMM \cdot y:V_y \models A
			\end{array}
		\end{array}\hspace{-1cm}
	}{7cm}{
		Lemma \ref{lem:LTC_derived_value_cannot_reveal_old_names} 
		$\LTCDERIVEDVALUE{}{\GAMMA_0}{\MMM}{s} \ \MIFF \ \LTCDERIVEDVALUE{}{\GAMMA_0 \PLUSV y}{\MMM_y}{s}$,
		\\
		Lemma \ref{lem:semantics_expressions_equal_under_model_extensions}, $\SEM{x}{\MMM} \equiv \SEM{x}{\MMM \cdot y:V_y}$
	}
	\NLINE{
		\Mforall M_y. 
		\begin{array}[t]{l}
			\LTCDERIVEDVALUE{M_y}{\GAMMA_0}{\MMM}{V_y}
			\\
			\ \MIMPLIES \
			\begin{array}[t]{l}
				\MMM \cdot y:V_y \models \FRESH{x}{\GAMMA \PLUSV y} 
				\\
				\MAND \ 
				\MMM \cdot y:V_y \models A
			\end{array}
		\end{array} 
	}{
		Sem. $\FRESH{}{}$
	}

	\NLASTLINE{ \MMM \models\FORALL{y}{\GAMMA_0} \FRESH{x}{(\GAMMA_0 \PLUSV y)} \AND A
	}{
		Sem. $\FORALL{}{}$, $\PAND$
	}
\end{NDERIVATION}
\end{proof}

\subsubsection{\DONE  Soundness proof of Axiom $(f3)$:}
\[
	\FRESH{x}{\GAMMA_0}  \PAND \EXPRESSIONTYPES{\GAMMA_0}{e}{\NAME} \quad\PIMPLIES\quad x \neq e
\]

\begin{proof}
Provable directly from syntactic definition of $\FRESH{x}{\GAMMA_0}$ and $(u1)$.
\end{proof}

\subsubsection{\DONE Soundness proof of Axiom $(f4)$}      
\[
	\FRESH{x}{(\GAMMA_0 \PLUSG \GAMMA_1)} \quad\PIMPLIES\quad \FRESH{x}{\GAMMA_0} \PAND \FRESH{x}{ \GAMMA_1}
\]
\begin{proof}
Provable directly from syntactic definition of $\FRESH{x}{\GAMMA_0}$, and $(u2)$.
\end{proof}

%% file: appendix/soundness_proof_axiom_UTC.tex
\subsubsection{\DONE Soundness proof of Axiom $(utc1)$: } 
\[
	\FORMULATYPES{\GAMMA}{ (\FAD{\TCV}{} A)  \quad \PIMPLIES \quad   A \LSUBST{\GAMMA}{\TCV}}
\]

\begin{proof}
\begin{NDERIVATION}{1}
	\NLINE{\text{Assume:  $\GAMMA$ s.t. $\FORMULATYPES{\GAMMA}{
			(\FAD{\TCV}{} A)  \quad \PIMPLIES \quad   A \LSUBST{\GAMMA}{\TCV}
	}$}}{}
	\NLINE{\text{Assume: $\MMM^{\GAMMA_d}$ s.t. $\CONSTRUCT{\GAMMA}{\MMM}$ and } 
		\MMM\models \FAD{\TCV}{} A
	}{}
	\NLINE{
		\MIFF \ 
			\Mforall \MMM_0^{\GAMMA_0}. \
				\MMM \EXTSTAR \MMM_0
				\
				\MIMPLIES \ \MMM_0 \cdot \TCV : \GAMMA_0 \models A
	}{Sem. $\FAD{\TCV}{}$}
	\NLINE{
		\MIMPLIES \ 
		\MMM \EXTSTAR \MMM 
		\ \MIMPLIES \ 
		\MMM \cdot \TCV : \GAMMA_d \models A
	}{Instantiate $\Mforall \MMM_0^{\GAMMA_0}$ with $\MMM$}
	\NLINE{\MIMPLIES \ \MMM \cdot \TCV : \GAMMA_d \models A}{FOL, $(\MMM \EXTSTAR \MMM) \quad  (\equiv \TRUTH)$}
	\NLINE{\MIMPLIES \ \MMM \cdot \TCV : \GAMMA \models A}{$\CONSTRUCT{\GAMMA}{\MMM} \ \MIMPLIES \ \SEM{\GAMMA}{\MMM} \equiv \SEM{\GAMMA_d}{\MMM}$}
	\NLASTLINE{\MIMPLIES \ \MMM \models A\LSUBST{\GAMMA}{\TCV}}{Sem. $\LSUBST{\GAMMA}{\TCV}$ }
\end{NDERIVATION}
\end{proof}

\subsubsection{\DONE Soundness proof of Axiom $(utc2)$:} 
\[
\FORMULATYPES{\GAMMA}{\FORALL{x^{\NAME}}{\GAMMA} A^{-\TCV} \quad \PIFF \quad  \FAD{\TCV} \FORALL{x^{\NAME}}{\GAMMA \PLUSTC \TCV}  A^{-\TCV}} \text{ if $A$-\EXTINDEP}
\] 
\begin{proof}
	$\PIMPLIEDBY$:
	Holds through $(utc1)$
	\\
	$\PIMPLIES$: 
	\begin{NDERIVATION}{1}
		\NLINE{\text{Assume: $\GAMMA$ s.t. } \FORMULATYPES{\GAMMA}{
				\FORALL{x}{\GAMMA} A \ \PIFF \  \FAD{\TCV} \FORALL{x}{\GAMMA \PLUSTC \TCV}  A^{-\TCV}
			}
		}{}
		\NLINE{\text{Assume: $\MMM^{\GAMMA_d}$ s.t. $\CONSTRUCT{\GAMMA}{\MMM}$ and $\MMM \models \FORALL{x^{\NAME}}{\GAMMA} A $}}{}
		\NLINE{	
			\MIFF \ 
			\Mforall M. 
			\LTCDERIVEDVALUE{M}{\GAMMA}{\MMM}{r_0}
			\ \MIMPLIES \ 
			\MMM \cdot x : r_0 \models A
		}{}
		\NLINE{
			\text{Assume arbitrary $\MMM'^{\GAMMA'}$, $M'$ s.t. $\MMM \EXTSTAR \MMM'$ and $\LTCDERIVEDVALUE{M'}{\GAMMA \PLUSTC \TCV}{\MMM' \cdot \TCV:\GAMMA'\REMOVETCVfrom }{r_1}$}
		}{`\ $\FAD{\TCV} \FORALL{x}{\TCV}$'}
		\NLINE{
			\MIFF
			\ \LTCDERIVEDVALUE{M'}{\GAMMA'}{\MMM'}{r_1}
		}{$\SEM{\GAMMA\PLUSTC \TCV}{\MMM' \cdot \TCV:\GAMMA'\REMOVETCVfrom} \equiv \GAMMA'$, Lemma \ref{lem:LTC_derived_values_unaffected_by_TCV_addition/removal}}
		\NLINE{
			r_1 \in \AN{\MMM} \MIMPLIES r_0 \equiv r_1 \MIMPLIES \MMM \cdot x:r_1 \models A
		}{Lemma \ref{lem:extensions_cannot_reveal_old_names} $\MIMPLIES$ $r_0 \equiv r_1$ obtainable}
		\NPLINE{
			r_1 \notin \AN{\MMM} \MIMPLIES \text{fresh-}r_0\equiv r_1 \MIMPLIES \MMM \cdot x:r_1 \models A 
		}{5cm}{Let $r_0 \equiv r_1$ as $r_1 \notin \AN{\MMM}$ \\(fresh names can be swapped)}
		\NLINE{\MIMPLIES \ \MMM' \cdot x:r_1 \models A
		}{Lemma \ref{lem:name_derived_from_extension_implies_extension_and_name_are_extension}, $A-\EXTINDEP$}
		\NLINE{
			\MIMPLIES \ \MMM' \cdot \TCV: \GAMMA'\REMOVETCVfrom \cdot x:r_1 \models A
		}{
			Lemma \ref{lem:model_and_model_plus_TCV_models_equivalently_TCV-free_formula}, $A^{-\TCV}$
		}
		\NLINE{
			\MIMPLIES \ 
			\MMM \models \FAD{\TCV} \FORALL{x^{\NAME}}{\GAMMA \PLUSTC \TCV}  A^{-\TCV}
		}{lines 4-9}
		\NLASTLINE{\text{Hence: } 
			\begin{array}[t]{l}
				\Mforall \GAMMA. \ \FORMULATYPES{\GAMMA}{\FORALL{x}{\GAMMA} A \ \PIFF \  \FAD{\TCV} \FORALL{x}{\GAMMA \PLUSTC \TCV}  A^{-\TCV}}
				\\
				\MIMPLIES\ 
				\Mforall \MMM^{\GAMMA_d}. \CONSTRUCT{\GAMMA}{\MMM} \ \MIMPLIES \ \MMM \models \FORALL{x^{\NAME}}{\GAMMA} A^{-\TCV} \ \PIMPLIES \  \FAD{\TCV} \FORALL{x^{\NAME}}{\GAMMA \PLUSTC \TCV}  A^{-\TCV}
				\\
				(\text{Assuming $A-\EXTINDEP$})
			\end{array}
		}{lines 1-10}
	\end{NDERIVATION}
	Essentially in line 6: if $r_1$ is in $\AN{\MMM}$ then it can be derived from $\MMM$, and in line 7 if $r_1$ is not in $\AN{\MMM}$ then instantiating $M\equiv \GENSYM()$  produces a fresh name which can easily be set as (or swapped for) $r_1$.
	
\end{proof}

\newpage
\subsubsection{\DONE  Soundness proof of Axiom $(utc3)$: } 
\[
	A-\text{\EXTINDEP \ } \ \MIMPLIES \  A^{-\TCV} \ \PIFF \ \FAD{\TCV}{} A
\]
\begin{proof}
Assume: $\GAMMA$, $\MMM^{\GAMMA}$ s.t. $\FORMULATYPES{\GAMMA}{A^{-\TCV} \ \PIFF \ \FAD{\TCV}{} A}$
and assume some $A^{-\TCV}$-\EXTINDEP, 
then:
\\
$\MIMPLIEDBY$:
Use $(utc1)$ knowing that $A^{-\TCV}\LSUBST{\GAMMA}{\TCV} \equiv A$.
\\
$\MIMPLIES$:
\begin{NDERIVATION}{1}
	\NLINE{\text{Assume: } \MMM \models A^{-\TCV}}{}
	\NLINE{\text{Prove: } \MMM \models  \FAD{\TCV}{} A }{}
	\NLINE{
		\MMM \models A^{-\TCV} 
		\ \MIMPLIES \ 
		\Mforall \MMM'^{\GAMMA'}. \
		\MMM \EXTSTAR \MMM' 
		\ \MIMPLIES \ 
		\MMM \models A^{-\TCV}
	}{trivial addition}
	\NLINE{
		\MMM \models A^{-\TCV} 
		\ \MIMPLIES \ 
		\Mforall \MMM'^{\GAMMA'}. \
		\MMM \EXTSTAR \MMM' 
		\ \MIMPLIES \ 
		\MMM' \models A^{-\TCV}
	}{$A$-\EXTINDEP}
	\NLINE{\MMM \models A^{-\TCV} 
		\ \MIMPLIES \ 
		\Mforall \MMM'^{\GAMMA'}.
		\MMM \EXTSTAR \MMM' 
		\
		\MAND \ 
		\MMM'\cdot \TCV:\GAMMA' \models A^{-\TCV}
	}{Lemma \ref{lem:model_and_model_plus_TCV_models_equivalently_TCV-free_formula}}
	\NLASTLINE{
		\MMM \models A^{-\TCV} 
		\ \MIMPLIES \ 
		\MMM \models  \FAD{\TCV}{} A 
	}{Sem. $\FAD{\TCV}$}
\end{NDERIVATION}
\end{proof}

\subsubsection{\DONE Soundness proof of Axiom $(utc4)$: } 
\[
	\FAD{\TCV}{} (A \PAND B) \ \PIFF \ (\FAD{\TCV}{} A) \PAND (\FAD{\TCV}{} B)
\]

\begin{proof}
	Assume $\GAMMA$, $\MMM^{\GAMMA}$ s.t. $\FORMULATYPES{\GAMMA}{\FAD{\TCV}{} (A \PAND B) \ \PIFF \ (\FAD{\TCV}{} A) \PAND (\FAD{\TCV}{} B)}$
	\\
$\MIMPLIES$:
\begin{NDERIVATION}{1}
	\NLINE{\text{Assume: } \MMM \models \FAD{\TCV}{} (A \MAND B)}{}
	\NLINE{\text{Prove: } \MMM  \models  (\FAD{\TCV}{} A) \MAND (\FAD{\TCV}{} B) }{}
	\NLINE{\Mforall \MMM'^{\GAMMA'}. \
			\MMM \EXTSTAR \MMM'
			\
			\MIMPLIES \ \MMM' \cdot \TCV : \GAMMA' \models A \MAND B
	}{Sem. $\FAD{}{}$}	
	\NLINE{\MIFF \ 
		\begin{array}[t]{l}
			\Mforall \MMM'^{\GAMMA'}. \
				\MMM \EXTSTAR \MMM'
				\
				\MIMPLIES \ \MMM' \cdot \TCV : \GAMMA' \models A \MAND B
			\\
			\MAND
			\Mforall \MMM'^{\GAMMA'}. \
				\MMM \EXTSTAR \MMM'
				\
				\MIMPLIES \ \MMM' \cdot \TCV : \GAMMA' \models A \MAND B
		\end{array}
	}{FOL $A \MIMPLIES (A \MAND A)$}
	\NLINE{\MIMPLIES \ 
		\begin{array}[t]{l}
			\Mforall \MMM'^{\GAMMA'}. \
				\MMM \EXTSTAR \MMM'
				\
				\MIMPLIES \ \MMM' \cdot \TCV : \GAMMA' \models A
			\\
			\MAND
			\Mforall \MMM'^{\GAMMA'}. \
				\MMM \EXTSTAR \MMM'
				\
				\MIMPLIES \ \MMM' \cdot \TCV : \GAMMA' \models B
		\end{array}
	}{$\MAND$-elim}
	
	\NLASTLINE{ \MIFF \
		\MMM  \models  (\FAD{\TCV}{} A) \MAND (\FAD{\TCV}{} B) 
	}{$\MAND$-elim}
\end{NDERIVATION}

$\MIMPLIEDBY$:
\begin{NDERIVATION}{1}
	\NLINE{\text{Assume: } \MMM \models  (\FAD{\TCV}{} A) \MAND (\FAD{\TCV}{} B) }{}
	\NLINE{\text{Prove: } \MMM \models \FAD{\TCV}{} (A \MAND B)}{}
	\NLINE{\IFF \ 
		\begin{array}[t]{l}
			\Mforall \MMM'^{\GAMMA'}. \
				\MMM \EXTSTAR \MMM'
				\
				\MIMPLIES \ \MMM' \cdot \TCV : \GAMMA' \models A
			\\
			\MAND
			\Mforall \MMM'^{\GAMMA'}. \
				\MMM \EXTSTAR \MMM'
				\
				\MIMPLIES \ \MMM' \cdot \TCV : \GAMMA' \models B
		\end{array}
	}{Sem. $\MAND$,$\FAD{}{}$}
	\NLINE{\MIMPLIES \ 
		\begin{array}[t]{l}
			\Mforall \MMM'^{\GAMMA'}. 
			\begin{array}[t]{l}
				\MMM \EXTSTAR \MMM'
				\\
				\MIMPLIES 
				\left(
				\begin{array}{l}
				 	\MMM' \cdot \TCV : \GAMMA' \models A
				 	\\
				 	\MAND \ \MMM' \cdot \TCV : \GAMMA' \models B
				\end{array}
				\right)
			\end{array}
		\end{array}
	}{$\Mforall\MMM'$ unifing }
	
	\NLINE{
		\MIFF \
		\Mforall \MMM'^{\GAMMA'}. 
		\begin{array}[t]{l}
			\MMM \EXTSTAR \MMM'
			\\
			\MIMPLIES \ \MMM' \cdot \TCV : \GAMMA' \models A \MAND B
		\end{array}
	}{Sem. $\MAND$}
	
	\NLASTLINE{ \MIFF \
		\MMM  \models  \FAD{\TCV}{} (A \MAND B) 
	}{$\MAND$-elim}
\end{NDERIVATION}
\end{proof}

%% file: appendix/soundness_proof_axiom_e.tex
\newpage
\subsubsection{\DONE  Soundness proof of Axiom $(e1)$}

\[
\ONEEVAL{e_1}{e_2}{m}{A^{-m} \PAND B} \ \PIFF \ (A \PAND \ONEEVAL{e_1}{e_2}{m}{B}) \qquad \text{ iff } \ m\notin \FV{A} \text{ and $A$-\EXTINDEP.}
\]

\PROOFFINISHED
{
	\begin{NDERIVATION}{1}
		\NLINE{\text{Assume: $\GAMMA$, $\MMM^{\GAMMA}$ s.t.}}{}
		\NLINE{\text{Assume: $\FORMULATYPES{\GAMMA}{\ONEEVAL{e_1}{e_2}{m}{A^{-m} \PAND B} \ \PIFF \ (A \PAND \ONEEVAL{e_1}{e_2}{m}{B})}$ with $m\notin \FV{A}$and $A$-\EXTINDEP.} }{}
		\NLINE{\MMM \models \ONEEVAL{e_1}{e_2}{m}{A^{-m} \PAND B}}{$ m \notin \DOM{\GAMMA}$}
		\NLINE{
			\MIFF \
			\LTCDERIVEDVALUE{e_1 e_2}{\GAMMA}{\MMM}{V}
			\ \MAND \ \MMM \cdot m : V \models A \PAND B  
		}{Sem. $\ONEEVAL{}{}{}{}$}
		\NLINE{
			\MIFF \
			\LTCDERIVEDVALUE{e_1 e_2}{\GAMMA}{\MMM}{V}
			\ \MAND \ \MMM \cdot m : V \models A 
			\ \MAND \ \MMM \cdot m : V \models B  
		}{Sem. $\PAND$}
		\NLINE{\MIFF \
\LTCDERIVEDVALUE{e_1 e_2}{\GAMMA}{\MMM}{V}
			\ \MAND \ \MMM \models A 
			\ \MAND \ \MMM \cdot m : V \models B  
		}{Sem. $\EXTSTAR$, Lemma \ref{lem:LTCDERIVED_subset_implies_LTCDERIVED_supset}  and $A$-\EXTINDEP}
		\NLINE{\MIFF \
			\MMM \models A 
			\ \MAND \ \MMM \models \ONEEVAL{e_1}{e_2}{m}{B}
		}{Sem. $\ONEEVAL{}{}{}{}$}
		\NLASTLINE{\MIFF \
			\MMM \models (A \PAND \ONEEVAL{e_1}{e_2}{m}{B})
		}{Sem. $\PAND$}
	\end{NDERIVATION}
}

\subsubsection{\DONE  Soundness proof of Axiom  $(e2)$} 
\[
	\ONEEVAL{e_1}{e_2}{m}{\FORALL{x}{\GAMMA_0}A} \ \PIFF \ \FORALL{x}{\GAMMA_0} \ONEEVAL{e_1}{e_2}{m}{A} \text{ iff } x \neq e_1,e_2,m \ \MAND \ e_1,e_2,m \notin \GAMMA_0
\]

\PROOFFINISHED
{
\begin{NDERIVATION}{1}
	\NLINE{\text{Assume: $\GAMMA$, $\MMM^{\GAMMA}$ s.t.}}{}
	\NLINE{\parbox{10cm}{Assume: $\FORMULATYPES{\GAMMA}{\ONEEVAL{e_1}{e_2}{m}{\FORALL{x}{\GAMMA_0}A} \ \PIFF \ \FORALL{x}{\GAMMA_0} \ONEEVAL{e_1}{e_2}{m}{A}}$ 
			\\
			with $x \neq e_1,e_2,m \ \MAND \ e_1,e_2,m \notin \GAMMA_0$} }{}
	\NLINE{\text{Assume: } \MMM \models \ONEEVAL{e_1}{e_2}{m}{\FORALL{x}{\GAMMA_0}A} }{}
	\NLINE{ 
		\MIFF \ 
		\begin{array}[t]{l}
			\LTCDERIVEDVALUE{e_1 e_2}{\GAMMA}{\MMM}{V_m}
			\
			\MAND \ 
			\Mforall M_x.\
				\LTCDERIVEDVALUE{M_x}{\GAMMA_0}{\MMM \cdot m : V_m }{V_x}
				\
				\MIMPLIES \ \MMM \cdot m : V_m  \cdot x:V_x \models A
		\end{array}
	}{Sem. $\ONEEVAL{}{}{}{}$, $\FORALL{}{}$}
	\NLINE{ 
		\MIFF \
		\begin{array}[t]{l}
			\LTCDERIVEDVALUE{e_1 e_2}{\GAMMA}{\MMM}{V_m}
			\
			\MAND \
			\Mforall M_x.\
			\LTCDERIVEDVALUE{M_x}{\GAMMA_0}{\MMM}{V_x}
			\
			\MIMPLIES \ \MMM \cdot m : V_m  \cdot x:V_x \models A
		\end{array}\hspace{-2cm}
	}{
	$\begin{array}[t]{r}
		m \notin \GAMMA_0 \ \MAND \ \MMM \EXTSTAR \MMM  \cdot m : V_m 
		\\
		\text{Lemma \ref{lem:eval_under_extensions_are_equivalent}  ($\AN{V_x} \cap \AN{V_m} \subseteq \AN{\MMM}$)}
	\end{array}
	$}
	\NLINE{
		\MIFF \ 
		\Mforall M_x.
		\
		\LTCDERIVEDVALUE{e_1 e_2}{\GAMMA}{\MMM}{V_m}
		\
		\MAND \
		(\LTCDERIVEDVALUE{M_x}{\GAMMA_0}{\MMM}{V_x}
		\
		\MIMPLIES \ \MMM \cdot m : V_m  \cdot x:V_x \models A)
	}{$
	\begin{array}[t]{r}
		A \MAND \Mforall X. B
		\\
		\MIFF
		\
		\Mforall X. A^{-X} \MAND B
	\end{array}
	$}
	\NPLINE{ 
		\MIFF \ 
		\Mforall M_x.
		\
		\LTCDERIVEDVALUE{M_x}{\GAMMA_0}{\MMM}{V_x}
		\ 
		\MIMPLIES \ 
		(\LTCDERIVEDVALUE{e_1e_2}{\GAMMA}{\MMM}{V_m}
		\
		\MAND \ \MMM \cdot m : V_m  \cdot x:V_x \models A) 
	}{5.5cm}{
		$\MIMPLIES$: \ $A \MAND (B \MIMPLIES C) \ \MIMPLIES \ B \MIMPLIES (C \MAND A)$
		\\
		$\MIMPLIEDBY$: \ $\SEM{e_1 e_2}{}\CONV V_m $  terminates
	}
	\NPLINE{ 
		\MIFF \ 
		\begin{array}[t]{l}
			\Mforall M_x.
			\
			\LTCDERIVEDVALUE{M_x}{\GAMMA_0}{\MMM}{V_x}
			\ 
			\MIMPLIES 
			\
			(\LTCDERIVEDVALUE{e_1 e_2}{\GAMMA}{\MMM \cdot x:V_x}{V_m}
			\
			\MAND \ \MMM \cdot m : V_m  \cdot x:V_x \models A) 
		\end{array}
	}{4cm}{$x \neq e_1,e_2$,\\ Lemma \ref{lem:eval_under_extensions_are_equivalent} \\ ($\AN{V_x} \cap \AN{V_m} \subseteq \AN{\MMM \REMOVEVARIABLE x}$)}
	\NLASTLINE{\MIFF \ \MMM  \models \FORALL{x}{\GAMMA_0} \ONEEVAL{e_1}{e_2}{m}{A} }{Sem. $\FORALL{}{}$, $\ONEEVAL{}{}{}{}$}
\end{NDERIVATION}

}

\subsubsection{\DONE  Soundness proof of Axiom $(e3)$} 
\[
\ONEEVAL{e_1}{e_2}{m^{\TYBASE}}{\FAD{\TCV} A} \ \PIFF \ \FAD{\TCV} \ONEEVAL{e_1}{e_2}{m^{\TYBASE}}{A} \text{ iff } A-\EXTINDEP
\]

\PROOFFINISHED
{
	\begin{NDERIVATION}{1}
		\NLINE{\text{Assume: $\GAMMA$, $\MMM^{\GAMMA}$ s.t.}}{}
		\NLINE{\text{Assume: $\FORMULATYPES{\GAMMA}{\ONEEVAL{e_1}{e_2}{m^{\TYBASE}}{\FAD{\TCV} A} \ \PIFF \ \FAD{\TCV} \ONEEVAL{e_1}{e_2}{m^{\TYBASE}}{A}}$ with $A$-\EXTINDEP} }{}
		\NLINE{\text{Assume: } \MMM \models \ONEEVAL{e_1}{e_2}{m^{\TYBASE}}{\FAD{\TCV} A} }{}
		\NLINE{ 
			\MIFF \ 
				\LTCDERIVEDVALUE{e_1 e_2}{\GAMMA}{\MMM}{V_m}
				\
				\MAND \ 
				\Mforall \MMM_m'^{\GAMMA'_m}.\
				\MMM \cdot m:V_m \EXTSTAR \MMM'_m
				\
				\MIMPLIES \ \MMM'_m \cdot \TCV:\GAMMA'_m\REMOVETCVfrom \models A
		}{Sem. $\ONEEVAL{}{}{}{}$, $\FAD{\TCV}$ }
		\NPLINE{ 
			\MIFF \ 
				\LTCDERIVEDVALUE{e_1 e_2}{\GAMMA}{\MMM}{V_m}
				\
				\MAND \ 
				\Mforall \MMM'^{\GAMMA'}.\
				\MMM \cdot m:V_m \EXTSTAR \MMM' \cdot m:V_m
				\
				\MIMPLIES \ \MMM' \cdot m:V_m \cdot \TCV:\GAMMA'\REMOVETCVfrom \models A
		}{2cm}{
		Rewrite $\MMM'_m$ 
		\\
		Lemma \ref{lem:Base_Values_can_be_added/removed_from_TCV:TC}
		}
		\NLINE{ 
			\MIFF \ 
				\LTCDERIVEDVALUE{e_1 e_2}{\GAMMA}{\MMM}{V_m}
				\
				\MAND \ 
				\Mforall \MMM'^{\GAMMA'}.\
				\MMM \EXTSTAR \MMM'
				\
				\MIMPLIES \ \MMM' \cdot m:V_m \cdot \TCV:\GAMMA'\REMOVETCVfrom \models A
		}{
				Lemma \ref{lem:Base_types_added/removed_maintain_extension}
		}
		\NLINE{ 
			\MIFF \ 
			\Mforall \MMM'^{\GAMMA'}.
			\MMM \EXTSTAR \MMM'
			\
			\MIMPLIES 
				(\LTCDERIVEDVALUE{e_1 e_2}{\GAMMA}{\MMM}{V_m}
				\
				\MAND \ 
				\MMM' \cdot m:V_m \cdot \TCV:\GAMMA'\REMOVETCVfrom \models A)
		}{ $\SEM{e_1e_2}{\MMM}\CONV V_m$ terminates
		}
		\NLINE{ 
			\MIFF \ 
			\Mforall \MMM'^{\GAMMA'}.\
			\MMM \EXTSTAR \MMM'
			\
			\MIMPLIES
				(\LTCDERIVEDVALUE{e_1 e_2}{\GAMMA}{\MMM'}{V_m}
				\ 
				\MAND \ 
				\MMM' \cdot m:V_m \cdot \TCV:\GAMMA'\REMOVETCVfrom \models A)
		}{Lemmas \ref{lem:BaseValuesAreNameFree}, \ref{lem:eval_under_extensions_are_equivalent}}
		\NLINE{ 
			\MIFF \ 
			\Mforall \MMM'^{\GAMMA'}.
			\MMM \EXTSTAR \MMM'
			\
			\MIMPLIES \
				(\LTCDERIVEDVALUE{e_1e_2}{\GAMMA}{\MMM' \cdot \TCV:\GAMMA'\REMOVETCVfrom}{V_m}
				\
				\MAND \ 
				\MMM' \cdot \TCV:\GAMMA'\REMOVETCVfrom  \cdot m:V_m  \models A)
		}{Lemma \ref{lem:LTC_derived_values_unaffected_by_TCV_addition/removal}}
		\NLINE{ 
			\MIFF \ 
			\Mforall \MMM'^{\GAMMA'}.
			\MMM \EXTSTAR \MMM'
			\
			\MIMPLIES \
			(\LTCDERIVEDVALUE{e_1 e_2}{\GAMMA' \PLUSTC \TCV}{\MMM' \cdot \TCV:\GAMMA'\REMOVETCVfrom}{V_m}
			\
			\MAND \ 
			\MMM' \cdot \TCV:\GAMMA'\REMOVETCVfrom  \cdot m:V_m  \models A)
		}{$\EXPRESSIONTYPES{\GAMMA}{e_1 e_2}{\TYBASE}$}
		\NLASTLINE{
			\MIFF \ \MMM \models \FAD{\TCV} \ONEEVAL{e_1}{e_2}{m^{\TYBASE}}{A} 
		}{Sem. $\ONEEVAL{}{}{}{}$, $\FAD{\TCV}$ }
	\end{NDERIVATION}
	
}

\subsubsection{Soundness proof of Axiom $(ext)$} 
\label{ax_ext_soundness}
For all $e_1, e_2 : \alpha_1 \FS\alpha_2$ s.t. $\alpha_1$ and $\alpha_2$ are both $\NAME$-free:
\[
(\FORALL{x^{\alpha_1}}{\emptyset} \ONEEVAL{e_1}{x}{m_1^{\alpha_2}}{\ONEEVAL{e_2}{x}{m_2^{\alpha_2}}{m_1 = m_2}}) \ \PIFF \ e_1 =^{\alpha_1 \FS \alpha_2} e_2
\]
\begin{proof}
	Use Lemma \ref{lem:Nm-free_terms_have_equivalent_name_free_STLC-term} and the proof for $(ext)$ in STLC to see this holds

\end{proof}

%% file: appendix/appendix_soundness_rules.tex
\section{Soundness of Rules}
\label{appendix_soundness_rules}
\label{sec:apndx_soundness_rules}

In this section soundness proofs of the rules are introduced.

Given the logic is limited to compile-time syntax, the condition then for any term s.t. $\TCTYPES{\GAMMA}{M}{\alpha}$ where $\AN{M}=\emptyset$ is implied then clearly $(\AN{\MMM}, M\MMM) \CONV (\AN{\MMM}, G', \ V) \ \MIFF \ \LTCDERIVEDVALUE{M}{\GAMMA}{\MMM}{V}$ hence this will be used equivalently for brevity.

\subsection{\DONE Soundness of \RULENAME{[Var]}}
$
\ZEROPREMISERULENAMEDRIGHT
{
	\ASSERT{A\LSUBST{x}{m}}{x}{m}{A}
}{[Var]}
$
\PROOFFINISHED
{
	Proof:
	
	\begin{NDERIVATION}{1}
		\NLINE{\text{Assume: $\GAMMA$ s.t. $\JUDGEMENTTYPES{\GAMMA}{\ASSERT{A\LSUBST{x}{m}}{x}{m}{A}}$ s.t. }}{}
		\NLINE{\text{Assume $\MMM$ s.t. $\CONSTRUCT{\GAMMA}{\MMM}$ and $ \MMM \models A\LSUBST{x}{m}$} }{}
		\NLINE{\MIMPLIES \ \LTCDERIVEDVALUE{x}{\GAMMA}{\MMM}{\MMM(x)}}{$x\MMM \equiv \MMM(x)$}
		\NLINE{\MMM \models A \LSUBST{x}{m} \ \MIFF \ \MMM \cdot m:\MMM(x) \models A}{Sem. $\LSUBST{x}{m}$, $m$-fresh}
		\NLINE{\MMM \models A \LSUBST{x}{m} \ \MIMPLIES \ \LTCDERIVEDVALUE{x}{\GAMMA}{\MMM}{\MMM(x)} \ \MAND \ \MMM \cdot m:\MMM(x) \models A}{$\MMM(x)$ a value}
		\NLINE{\Mforall \MMM. \CONSTRUCT{\GAMMA}{\MMM} \MIMPLIES \MMM \models A \LSUBST{x}{m} \MIMPLIES \LTCDERIVEDVALUE{x}{\GAMMA}{\MMM}{\MMM(x)} \ \MAND \ \MMM \cdot m:\MMM(x) \models A}{2, 5}
		\NLASTLINE{\text{Hence: } \ZEROPREMISERULE{\models \ASSERT{A\LSUBST{x}{m}}{x}{m}{A}}}{Sem. valid triple}

	\end{NDERIVATION}
}

\subsection{\DONE Soundness of \RULENAME{[Const]}}
$
\ZEROPREMISERULENAMEDRIGHT
{
	\ASSERT{A\LSUBST{c}{m}}{c}{m}{A}
}{[Const]}
$
\\

\PROOFFINISHED
{
	Proof:
	
	\begin{NDERIVATION}{1}
		\NLINE{\text{Assume: $\GAMMA$ s.t. $\JUDGEMENTTYPES{\GAMMA}{\ASSERT{A\LSUBST{c}{m}}{c}{m}{A}}$}}{}
		\NLINE{\text{Assume  $\MMM$ s,t, $\CONSTRUCT{\GAMMA}{\MMM}$ and $\MMM\models A \LSUBST{c}{m}$}}{}
		\NLINE{\MIMPLIES \ \LTCDERIVEDVALUE{c}{\GAMMA}{\MMM}{c}}{$c\MMM \equiv c$}
		\NLINE{\MMM \models A \LSUBST{c}{m} \ \MIFF \ \MMM \cdot m:c \models A}{Sem. $\LSUBST{c}{m}$}
		\NLINE{\MMM \models A \LSUBST{c}{m} \ \MIMPLIES \ \LTCDERIVEDVALUE{c}{\GAMMA}{\MMM}{c} \ \MAND \ \MMM \cdot m:c \models A}{$c$ a value}
		\NLINE{\Mforall \MMM. \CONSTRUCT{\GAMMA}{\MMM} \MIMPLIES \MMM \models A \LSUBST{c}{m} \MIMPLIES \LTCDERIVEDVALUE{c}{\GAMMA}{\MMM}{c} \ \MAND \ \MMM \cdot m:c \models A}{1,5}
		\NLASTLINE{\text{Hence: } \ZEROPREMISERULE{\models \ASSERT{A\LSUBST{c}{m}}{c}{m}{A}}}{Sem. valid triple}

	\end{NDERIVATION}
	
}

\subsection{\DONE Soundness of \RULENAME{[Eq]}}
\[
\THREEPREMISERULENAMEDRIGHT
{
	\ASSERT{A}{M}{m}{B}
}
{
	\ASSERT{B}{N}{n}{C \LSUBST{\EQA{m}{n}}{u}}
}
{
	C\THINWRT{m,n}
}
{
	\ASSERT{A}{M = N}{u}{C}
}{[Eq]}
\]

\PROOFFINISHED
{
	Proof:
	\begin{NDERIVATION}{1}
		\NLINE{\text{Assume: } \Mforall \GAMMA_0, \MMM_1^{\GAMMA_1}. \ 
			\JUDGEMENTTYPES{\GAMMA_0}{\ASSERT{A}{M}{m}{B}} 
			\ \MAND \
			\CONSTRUCT{\GAMMA_0}{\MMM_1}
			\ \MIMPLIES \
			\MMM_1 \models \ASSERT{A}{M}{m}{B}
		}{IH(1)}
		\NLINE{\text{Assume: } \Mforall \GAMMA_0, \MMM_1^{\GAMMA_1}. \ 
			\begin{array}[t]{l}
				\JUDGEMENTTYPES{\GAMMA_0}{\ASSERT{B}{N}{n}{C\LSUBST{m=n}{u}}} 
				\ \MAND \
				\CONSTRUCT{\GAMMA_0}{\MMM_1}
				\\ \MIMPLIES \
				\MMM_1 \models \ASSERT{B}{N}{n}{C\LSUBST{m=n}{u}}
			\end{array}
		}{IH(2)}
	\NLINE{\text{Assume: $\GAMMA$ s.t. $\JUDGEMENTTYPES{\GAMMA}{\ASSERT{A}{M = N}{u}{C}}$}}{}
	\NLINE{\text{Assuming the typing constraints of this rule hold:} }{}
	\NLINE{\text{Assume  $\MMM$ s,t, $\CONSTRUCT{\GAMMA}{\MMM} \ \MAND \ \MMM \models A$}}{}
		\NLINE{\MIMPLIES \ (\AN{\MMM}, \ M\MMM) \CONV (\AN{\MMM}, G_m, \ V_m) \ \MAND \ \MMM \cdot m:V_m \models B}{IH(1)}
		\NLINE{\MIMPLIES 
			\begin{array}[t]{l} 
				(\AN{\MMM}, \ M\MMM) \CONV (\AN{\MMM}, G_m, \ V_m) 
				\\ \MAND \ 
				(\AN{\MMM \cdot m:V_m}, \ N(\MMM \cdot m:V_m)) \CONV (\AN{\MMM \cdot m:V_m}, G_n, \ V_n)  
				\\ \MAND \ 
				\MMM \cdot m:V_m \cdot n:V_n \models C\LSUBST{{m}={n}}{u} 
			\end{array}
		}{IH(2)}
		\NLINE{
			\ \MIMPLIES \ 
			\begin{array}[t]{ll} 
				(\AN{\MMM}, \ M\MMM) \CONV (\AN{\MMM}, G_m, \ V_m) 
				\\ \MAND \ 
				(\AN{\MMM \cdot m:V_m}, \ N(\MMM \cdot m:V_m)) \CONV (\AN{\MMM \cdot m:V_m}, G_n, \ V_n) 
				\\ \MAND \ 
				\LTCDERIVEDVALUE{(m=n)}{\GAMMA \PLUSV m \PLUSV n}{\MMM \cdot m:V_m \cdot n:V_n}{V_u}
				\\ \MAND \ 
				\MMM_{mn} \cdot u:V_u \models C
			\end{array}
		}{Sem. $C\LSUBST{\EQA{m}{n}}{u}$}
		\NLINE{
			\ \MIMPLIES \ 
			\begin{array}[t]{ll} 
				(\AN{\MMM}, \ M\MMM) \CONV (\AN{\MMM}, G_m, \ V_m) 
				\\ \MAND \ 
				(\AN{\MMM \cdot m:V_m}, \ N(\MMM \cdot m:V_m)) \CONV (\AN{\MMM \cdot m:V_m}, G_n, \ V_n) 
				\\ \MAND \ 
				\LTCDERIVEDVALUE{(m=n)}{\GAMMA \PLUSV m \PLUSV n}{\MMM \cdot m:V_m \cdot n:V_n}{V_u}
				\\ \MAND \ 
				\MMM \cdot u:V_u \models C
			\end{array}
		}{$C\THINWRT{m,n}$}
		\NLINE{
			\ \MIMPLIES \ 
			(\AN{\MMM}, \ (M=N)\MMM ) \CONV (\AN{\MMM}, G', \ V_u) 
			\ \MAND \ 
			\MMM \cdot u:V_u \models C
		}{
			Op. Sem.(=)
		}
		\NLINE{\MIMPLIES \ \Mforall \MMM. \CONSTRUCT{\GAMMA}{\MMM} \ \MIMPLIES \ \MMM \models \ASSERT{A}{M=N}{u}{C}}{Assumption, 5-10}
		\NLASTLINE{\text{Hence: }  
			\THREEPREMISERULE
			{
				\models \ASSERT{A}{M}{m}{B}
			}
			{
				\models \ASSERT{B}{N}{n}{C^{-m,n} \LSUBST{\EQA{m}{n}}{u}}
			}
			{
				C \THINWRT{m,n}
			}
			{
				\models \ASSERT{A}{M = N}{u}{C}
			}
		}{lines 1-11}
	\end{NDERIVATION}
}

\newpage
\subsection{\DONE Soundness of \RULENAME{[Gensym]}}
\[
\ZEROPREMISERULENAMEDRIGHT
{\ASSERT{\TRUTH}{\GENSYM}{u}{\FAD{\TCV} \ONEEVAL{u}{()}{m}{\FRESH{m}{\TCV}}}}{[Gensym]}
\]
\PROOFFINISHED
{
	Soundness proof:
	\begin{NDERIVATION}{1}
		\NLINE{\text{Assume: $\GAMMA$ s.t. $\JUDGEMENTTYPES{\GAMMA}{\ASSERT{\TRUTH}{\GENSYM}{u}{\FAD{\TCV} \ONEEVAL{u}{()}{m}{\FRESH{m}{\TCV}}}}$ }}{}
		\NLINE{\text{Assume: $\MMM$ s.t. $\CONSTRUCT{\GAMMA}{\MMM} \ \MAND \ \MMM \models \TRUTH$}}{}
		\NLINE{\text{Let: } \MMM_u^{\GAMMA_u} \equiv \MMM \cdot u:\GENSYM}{$\GENSYM$-value and $u\notin \DOM{\MMM}$}
		\NLINE{ 
			\Mforall \MMM_x^{\GAMMA_x}.   \
			(\AN{\MMM_x}, \ \GENSYM ()) \CONV ((\AN{\MMM_x},r), \ r) \ \MAND \ r \notin \AN{\MMM_x}
		}{Op. Sem. $\GENSYM()$}
		\NPLINE{ 
			\Mforall \MMM_1^{\GAMMA_1}. 
			\begin{array}[t]{l}
				\MMM  \cdot u:\GENSYM \EXTSTAR  \MMM_1 
				\\
				\MIMPLIES 
				\begin{array}[t]{l}
					(\AN{\MMM_1}, \ \SEM{u}{\MMM_1} ()) \CONV ((\AN{\MMM_1},r), \ r)
					\\ 
					\MAND \ r \notin \AN{\MMM_1}
				\end{array}	
			\end{array}	
		}{3cm}{
			restrict the quantification, 4
			\\
			$\SEM{u}{\MMM_1} \equiv \GENSYM$
		}
		\NLINE{ 
			\MIMPLIES \
			\Mforall \MMM_1^{\GAMMA_1}. 
			\begin{array}[t]{l}
				\MMM \cdot u:\GENSYM \EXTSTAR  \MMM_1 
				\\
				\MIMPLIES 
				\begin{array}[t]{l}
					(\AN{\MMM_1}, \ \SEM{u}{\MMM_1} ()) \CONV ((\AN{\MMM_1},r), \ r) 
					\\
					\MAND \ \neg \Mexists M_m^{\NAME}. \LTCDERIVEDVALUE{M_m}{\GAMMA_1}{\MMM_1 \cdot m:r}{r}
				\end{array}	
			\end{array}	
		}{Lemma \ref{lem:fresh_name_underivable_from_model+name}}
		\NLINE{ 
			\MIMPLIES \
			\Mforall \MMM_1^{\GAMMA_1}. 
			\begin{array}[t]{l}
				\MMM \cdot u:\GENSYM \EXTSTAR  \MMM_1 
				\\
				\MIMPLIES 
				\begin{array}[t]{l}
					(\AN{\MMM_1}, \ \SEM{u}{\MMM_1} ()) \CONV ((\AN{\MMM_1},r), \ r) 
					\\
					\MAND \ \neg \Mexists M_m^{\NAME}. \LTCDERIVEDVALUE{M_m}{\GAMMA_1}{\MMM_1 \cdot m:r}{\SEM{m}{\MMM_1 \cdot \TCV:\GAMMA_1\REMOVETCVfrom \cdot m:r}}
				\end{array}	
			\end{array}	
		}{$r \equiv \SEM{m}{\MMM_1 \cdot \TCV:\GAMMA_1\REMOVETCVfrom\cdot m:r}$}
		\NLINE{ 
			\MIMPLIES \
			\Mforall \MMM_1^{\GAMMA_1}. 
			\begin{array}[t]{l}
				\MMM \cdot u:\GENSYM \EXTSTAR  \MMM_1 
				\\
				\MIMPLIES 
				\begin{array}[t]{l}
					(\AN{\MMM_1}, \ \SEM{u}{\MMM_1} ()) \CONV ((\AN{\MMM_1},r), \ r) 
					\\
					\MAND \ \neg \Mexists M_m^{\NAME}. \LTCDERIVEDVALUE{M_m}{\GAMMA_1}{\MMM_1 \cdot \TCV:\GAMMA_1\REMOVETCVfrom \cdot m:r}{\SEM{m}{\MMM_1 \cdot \TCV:\GAMMA_1\REMOVETCVfrom \cdot m:r}}
				\end{array}	
			\end{array}	
		}{
		Lemma \ref{lem:LTC_derived_values_unaffected_by_TCV_addition/removal}
		}
		\NPLINE{ 
			\MIMPLIES \
			\Mforall \MMM_1^{\GAMMA_1}. 
			\begin{array}[t]{l}
				\MMM \cdot u:\GENSYM \EXTSTAR  \MMM_1 
				\\
				\MIMPLIES 
				\begin{array}[t]{l}
					(\AN{\MMM_1}, \ \SEM{u}{\MMM_1} ()) \CONV ((\AN{\MMM_1},r), \ r) 
					\\
					\MAND \ \neg \Mexists M_m^{\NAME}. \LTCDERIVEDVALUE{M_m}{\GAMMA_1\PLUSTC \TCV}{\MMM_1 \cdot \TCV:\GAMMA_1\REMOVETCVfrom \cdot m:r}{\SEM{m}{\MMM_1 \cdot \TCV:\GAMMA_1\REMOVETCVfrom \cdot m:r}}
				\end{array}	
			\end{array}	
		}{4cm}{$\SEM{\GAMMA_1}{\MMM_1 \cdot \TCV:\GAMMA_1\REMOVETCVfrom \cdot m:r} \equiv $ \\ $\SEM{\GAMMA_1\PLUSTC \TCV}{\MMM_1 \cdot \TCV:\GAMMA_1\REMOVETCVfrom \cdot m:r}$}
		\NLINE{ 
			\Mforall \MMM_1^{\GAMMA_1}. 
			\begin{array}[t]{l}
				\MMM \cdot u:\GENSYM \EXTSTAR  \MMM_1 
				\\
				\MIMPLIES 
				\begin{array}[t]{l}
					(\AN{\MMM_1}, \ \SEM{u}{\MMM_1} ()) \CONV ((\AN{\MMM_1},r), \ r) 
					\\
					\MAND \  \MMM_1 \cdot \TCV:\GAMMA_1\REMOVETCVfrom  \cdot m:r \models \FRESH{m}{(\GAMMA_1 \PLUSTC \TCV)}
				\end{array}	
			\end{array}	
		}{Sem. $\FRESH{}{}$}
		\NPLINE{ 
			\Mforall \MMM_1^{\GAMMA_1}. 
			\begin{array}[t]{l}
				\MMM \cdot u:\GENSYM \EXTSTAR  \MMM_1 
				\ \MAND \ 
				\MMM_d \equiv \MMM_1 \cdot \TCV:\GAMMA_1 \REMOVETCVfrom
				\\
				\MIMPLIES 
				\begin{array}[t]{l}
					(\AN{\MMM_d }, \ \SEM{u}{\MMM_d } ()) \CONV ((\AN{\MMM_d },r), \ r) 
					\\
					\MAND \  \MMM_d \cdot m:r \models \FRESH{m}{\TCV}
				\end{array}	
			\end{array}	
		}{3.5cm}{
			Lemma \ref{lem:eval_under_extensions_are_equivalent}  
			\\
			$\AN{\MMM_d} \equiv \AN{\MMM}$
			\\
			(Shorthand $\TCV$)
		}
		\NLINE{ 
			\MMM \cdot u:\GENSYM \models \FAD{\TCV}\ONEEVAL{u}{()}{m}{\FRESH{m}{\TCV}}
		}{Sem. $\FAD{\TCV}\ONEEVAL{u}{()}{m}{\FRESH{m}{\TCV}}$}
		\NLINE{
			\Mforall \MMM^{\GAMMA}. 
			\CONSTRUCT{\GAMMA}{\MMM} \ \MIMPLIES \
			\MMM \models \TRUTH \MIMPLIES
			\begin{array}[t]{l}
				(\AN{\MMM},\ \GENSYM \MMM) \CONV (\AN{\MMM}, \ \GENSYM) 
				\\
				\MAND \ \MMM \cdot u:\GENSYM \models \FAD{\TCV} \ONEEVAL{u}{()}{m}{\FRESH{m}{
						\GAMMA \PLUSV u \PLUSTC
						\TCV}}
			\end{array}
		}{2-12}
		\NLASTLINE{\text{Hence: } \ZEROPREMISERULE{\models \ASSERT{\TRUTH}{\GENSYM}{u}{\FAD{\TCV} \ONEEVAL{u}{()}{m}{\FRESH{m}{\TCV}}}}}{Sem. valid triple}
	\end{NDERIVATION}
	
}

\newpage
\subsection{\DONE Soundness of \RULENAME{[Lam]}}
\[
\TWOPREMISERULENAMEDRIGHT
{
	A-\EXTINDEP
}
{
	(\GAMMA \PLUSTC \TCV \PLUSV x:\alpha)  \Vdash \ASSERT{A^{\MINUS x} \PAND B}{M}{m}{C}
}
{
	\GAMMA \Vdash 
	\ASSERT
	{A}
	{\lambda x^{\alpha}. M}{u}
	{\FAD{\TCV} 
		\FORALL{x^{\alpha}}{\TCV} (B \PIMPLIES \ONEEVAL{u}{x}{m}{C})}
}{[Lam]}
\]
\PROOFFINISHED
{
	\begin{NDERIVATION}{1}
		\NLINE{\text{Assume: $\GAMMA$ s.t.  $\JUDGEMENTTYPES{\GAMMA}{\ASSERT
					{A}
					{\lambda x^{\alpha}. M}{u}
					{\FAD{\TCV} 
						\FORALL{x^{\alpha}}{\TCV} (B \PIMPLIES \ONEEVAL{u}{x}{m}{C})}}$}}{}
		\NLINE{\text{and  $\JUDGEMENTTYPES{\GAMMA \PLUSTC \TCV \PLUSV x}{\ASSERT{A^{\MINUS x} \PAND B}{M}{m}{C}}$}
			}{}
		\NLINE{\text{Assume: }
			\begin{array}[t]{l}
				\Mforall \MMM_0^{\GAMMA_1 \PLUSTC \TCV \PLUSV x} .
				\
				\CONSTRUCT{\GAMMA \PLUSTC \TCV \PLUSV x}{\MMM_0} 
				\ \MAND \
				\MMM_0 \models A^{\MINUS x} \PAND B \ \MIMPLIES \
				\begin{array}[t]{l} 
					(\AN{\MMM_0}, \ M\MMM_0) \CONV ( G',\ V_m)
					\\
					\MAND \ \MMM_0 \cdot m:V_m \models C	
				\end{array}	
			\end{array}
		}{IH(1)}
		\NLINE{		
			\Mforall \MMM^{\GAMMA_1}.\ \CONSTRUCT{\GAMMA}{\MMM} 
			\ \MIMPLIES \
			\MMM \models A \ \MIMPLIES \
			\begin{array}[t]{l} 
				(G, \ (\lambda x^{\alpha}. M)\MMM) \CONV (G, G',\ V_u)
				\\
				\MAND \ \MMM \cdot u:V_u \models \FAD{\TCV}\FORALL{x^{\alpha}}{\TCV} (B \PIMPLIES \ONEEVAL{u}{x}{m}{C})
			\end{array}
		}{Need to Prove}
		
		\NLINE{\text{Assume: $\MMM^{\GAMMA_1}$ s.t. } 
			\CONSTRUCT{\GAMMA}{\MMM} \ \MAND \ 
			\MMM \models A 
		}{Assume}
		\NLINE{	
			\MIMPLIES \	
			(\AN{\MMM}, \ (\lambda x. M)\MMM) \CONV (\AN{\MMM}, \ \lambda x. (M\MMM \REMOVEVARIABLE x))
		}{Op. Sem. $\lambda x. M$}
		\NLINE{\text{For brevity write: } \MMM_{uddx} \equiv \MAND \ \MMM_{ud} \cdot \TCV:\GAMMA_{ud}\REMOVETCVfrom \cdot x:V_x}{}
		\NLINE{		
			\MIMPLIES \
			\begin{array}[t]{l} 
				(\AN{\MMM}, \ (\lambda x. M)\MMM) \CONV (\AN{\MMM}, \ \lambda x. (M\MMM \REMOVEVARIABLE x))
				\\
				\MAND \ \Mforall \MMM_{ud}^{\GAMMA_{ud}}. 
				\begin{array}[t]{l} 
					\MMM \cdot u:\lambda x. (M \MMM \REMOVEVARIABLE x) \EXTSTAR \MMM_{ud}
					\\
					\MAND \ \Mforall P_x. 
					\begin{array}[t]{l}
						\LTCDERIVEDVALUE{P_x}{\GAMMA_{ud} \PLUSTC \TCV}{\MMM_{ud} \cdot \TCV:\GAMMA_{ud}\REMOVETCVfrom}{V_x}
						\\
						\MAND \ \MMM_{uddx} \models B 
						\ \MIMPLIES \ \MMM_{uddx} \models B 
					\end{array}
				\end{array}
			\end{array}
		}{
			\RBOX{	
				Assume some
				\\
				$u$, $\TCV$, $x$-extension 
				\\
				that models $B \MIMPLIES B$
				\\
				Tautology
			}
		}
		\NPLINE{		
			\MIMPLIES \
			\begin{array}[t]{l} 
				(\AN{\MMM}, \ (\lambda x. M)\MMM) \CONV (\AN{\MMM}, \ \lambda x. (M\MMM \REMOVEVARIABLE x))
				\\
				\MAND \ \Mforall \MMM_{ud}^{\GAMMA_{ud}}. 
				\begin{array}[t]{l} 
					\MMM \cdot u:\lambda x. (M \MMM \REMOVEVARIABLE x) \EXTSTAR \MMM_{ud}
					\\
					\MAND \ \Mforall P_x. 
					\begin{array}[t]{l}
						\LTCDERIVEDVALUE{P_x}{\GAMMA_{ud} \PLUSTC \TCV}{\MMM_{ud} \cdot \TCV:\GAMMA_{ud}\REMOVETCVfrom}{V_x}
						\\
						\MAND \ \MMM_{uddx} \models B 
						\ \MIMPLIES \ \MMM_{uddx} \models A \MAND B 
					\end{array}
				\end{array}
			\end{array}
		}{5cm}{
			line 5 ($\MMM \models A$)
			\\
			$\MMM \EXTSTAR \MMM_{uddx}$
			\\ 
			$A-$\EXTINDEP 
			\\
			Sem. $\PAND$
		}
		\NPLINE{		
			\MIMPLIES \
			\begin{array}[t]{l} 
				(\AN{\MMM}, \ (\lambda x. M)\MMM) \CONV (\AN{\MMM}, \ \lambda x. (M\MMM \REMOVEVARIABLE x))
				\\
				\MAND \ \Mforall \MMM_{ud}^{\GAMMA_{ud}}. 
				\begin{array}[t]{l} 
					\MMM \cdot u:\lambda x. (M \MMM \REMOVEVARIABLE x) \EXTSTAR \MMM_{ud}
					\\
					\MAND \ \Mforall P_x. 
					\begin{array}[t]{l}
						\LTCDERIVEDVALUE{P_x}{\GAMMA_{ud} \PLUSTC \TCV}{\MMM_{ud} \cdot \TCV:\GAMMA_{ud}\REMOVETCVfrom}{V_x}
						\\
						\MAND \ \MMM_{uddx} \models B 
						\ \MIMPLIES \ 
						\begin{array}[t]{l}
							(\AN{\MMM_{uddx}}, \ M\MMM_{uddx}) \CONV (\AN{\MMM_{uddx}},G_m, \ V_m) 
							\hspace{-1cm}
							\\
							\MAND \ \MMM_{uddx} \cdot m:V_m \models C
						\end{array}
					\end{array}
				\end{array}
			\end{array}\hspace{-1cm}
		}{4cm}{
				IH(1)
				\\
				$\CONSTRUCT{\GAMMA \PLUSV u \PLUSTC \TCV \PLUSV x}{\MMM_{uddx}}$
		}
		\NPLINE{		
			\MIMPLIES \
			\begin{array}[t]{l} 
				(\AN{\MMM}, \ (\lambda x. M)\MMM) \CONV (\AN{\MMM}, \ \lambda x. (M\MMM \REMOVEVARIABLE x))
				\\
				\MAND \ \Mforall \MMM_{ud}^{\GAMMA_{ud}}. 
				\begin{array}[t]{l} 
					\MMM \cdot u:\lambda x. (M \MMM \REMOVEVARIABLE x) \EXTSTAR \MMM_{ud}
					\\
					\MAND \ \Mforall P_x. 
					\begin{array}[t]{l}
						\LTCDERIVEDVALUE{P_x}{\GAMMA_{ud} \PLUSTC \TCV}{\MMM_{ud} \cdot \TCV:\GAMMA_{ud}\REMOVETCVfrom}{V_x}
						\\
						\MAND \ \MMM_{uddx} \models B 
						\ \MIMPLIES \ 
						\begin{array}[t]{l}
							(\AN{\MMM_{uddx}}, \ \SEM{u}{\MMM_{uddx}} \SEM{x}{\MMM_{uddx}}) \CONV (\AN{\MMM_{uddx}},G_m, \ V_m)
							\hspace{-1cm}
							\\
							\MAND \ \MMM_{uddx} \cdot m:V_m \models C
						\end{array}
					\end{array}
				\end{array}
			\end{array}\hspace{-3cm}
		}{6cm}{
			$M\MMM \equiv ((\lambda x. M)x)\MMM$ $\equiv \SEM{u}{\MMM_{uddx}} \SEM{x}{\MMM_{uddx}}$
		}
		\NPLINE{		
			\MIMPLIES \
			\begin{array}[t]{l} 
				(\AN{\MMM}, \ (\lambda x. M)\MMM) \CONV (\AN{\MMM}, \ \lambda x. (M\MMM \REMOVEVARIABLE x))
				\\
				\MAND \ \Mforall \MMM_{ud}^{\GAMMA_{ud}}. 
				\begin{array}[t]{l} 
					\MMM \cdot u:\lambda x. (M \MMM \REMOVEVARIABLE x) \EXTSTAR \MMM_{ud}
					\\
					\MAND \ \Mforall P_x. \
						\LTCDERIVEDVALUE{P_x}{\GAMMA_{ud} \PLUSTC \TCV}{\MMM_{ud} \cdot \TCV:\GAMMA_{ud}\REMOVETCVfrom}{V_x}
						\
						\MIMPLIES \ \MMM_{uddx} \models B \PIMPLIES \ONEEVAL{u}{x}{m}{C}
				\end{array}
			\end{array} \hspace{-2cm}
		}{4cm}{
			Sem.  $\ONEEVAL{u}{x}{m}{C}$, $\PIMPLIES$
			\\
			$\MAND$ implies $\MIMPLIES$
		}
		\NPLINE{		
			\MIMPLIES \
				(\AN{\MMM}, \ (\lambda x. M)\MMM) \CONV (\AN{\MMM}, \ V_u)
				\
				\MAND \ \MMM \cdot u: V_u \models \FAD{\TCV} \FORALL{x}{\TCV} B \PIMPLIES \ONEEVAL{u}{x}{m}{C}
		}{3cm}{
			Sem. $\FAD{\TCV} \FORALL{x}{\TCV}$
		}
		
		\NLINE{		
			\models \ASSERT{A}{\lambda x. M}{u}{\FAD{\TCV} \FORALL{x}{\GAMMA \PLUSV u \PLUSTC \TCV} (B \PIMPLIES \ONEEVAL{u}{x}{m}{C})}
		}{
			line 5, Sem. $\ASSERT{A}{M}{u}{B}$
		}
		\NLASTLINE{	
			\text{Hence: } 
			\TWOPREMISERULE{
				A-\EXTINDEP
			}{
				\models \ASSERT{A \MAND B}{M}{m}{C}
			}{
				\models \ASSERT{A}{\lambda x. M}{u}{\FAD{\TCV} \FORALL{x}{\GAMMA \PLUSV u \PLUSTC \TCV} (B \PIMPLIES \ONEEVAL{u}{x}{m}{C})}
			}
		}{
			lines 1-13
		}
	\vspace{-3cm}
		
	\end{NDERIVATION}}

\newpage

\subsection{\DONE Soundness of \RULENAME{[App]}}
\[
\THREEPREMISERULENAMEDRIGHT
{
	\ASSERT{A}{M}{m}{B}
}
{
	\ASSERT{B}{N}{n}{\ONEEVAL{m}{n}{u}{C}}
}
{
	C-\THINWRT{m,n}
}
{
	\ASSERT{A}{MN}{u}{C}
}{[App]}
\]

\PROOFFINISHED
{
	Proof:
	
	\begin{NDERIVATION}{1}
		\NLINE{\text{Assume: $\GAMMA$ s.t.  $\JUDGEMENTTYPES{\GAMMA}{\ASSERT{A}{MN}{u}{C}}$}}{}
		\NLINE{\MIMPLIES \ 
			\JUDGEMENTTYPES{\GAMMA}{\ASSERT{A}{M}{m}{B}} 
			\ \MAND \ 
			\JUDGEMENTTYPES{\GAMMA \PLUSV m}{\ASSERT{B}{N}{n}{\ONEEVAL{m}{n}{u}{C}}}
		}{}
		\NLINE{\TYPES{\GAMMA\LTCtoSTC}{MN}{\beta} \ \MIMPLIES 
			\
			\TYPES{\GAMMA\LTCtoSTC}{M}{\alpha \FS \beta}
			\ \MAND \
			\TYPES{\GAMMA\LTCtoSTC}{N}{\alpha}}{}
		\NLINE{\Mforall \MMM_1^{\GAMMA_1}. \CONSTRUCT{\GAMMA}{\MMM_1} \ \MIMPLIES \ \MMM_1 \models A \MIMPLIES (\AN{\MMM_1}, \ M\MMM_1) \CONV (\AN{\MMM_1},G_m, \ V_m) \ \MAND \ \MMM_1 \cdot m:V_m \models B}{IH(1)}
		\NLINE{\Mforall \MMM_2^{\GAMMA_2}. \CONSTRUCT{\GAMMA \PLUSV m}{\MMM_2} \ \MIMPLIES \ \MMM_2 \models B \MIMPLIES 
			\begin{array}[t]{l} 
				(\AN{\MMM_2}, \ N\MMM_2) \CONV (\AN{\MMM_2}, G_n \ V_n) 
				\\ \MAND \ \MMM_2 \cdot n:V_n\models \ONEEVAL{m}{n}{u}{C}
			\end{array}
		}{IH(2)}
		\NLINE{\text{Assume: $\MMM^{\GAMMA_0}$ s.t. } \CONSTRUCT{\GAMMA}{\MMM} \ \MAND \ \MMM \models A}{}
		\NLINE{\MIMPLIES (\AN{\MMM}, \ M\MMM) \CONV (\AN{\MMM},G_m, \ V_m) \ \MAND \ \MMM \cdot m:V_m \models B}{IH(1)}
		\NLINE{
			\MIMPLIES
			\begin{array}[t]{l} 
				(\AN{\MMM}, \ M\MMM) \CONV (\AN{\MMM},G_m, \ V_m) 
				\\ 
				\MAND \
				(\AN{\MMM \cdot m:V_m}, \ N(\MMM\cdot m:V_m )) \CONV (\AN{\MMM\cdot m:V_m}, G_n \ V_n) 
				\\ 
				\MAND \ 
				\MMM \cdot m:V_m  \cdot n:V_n \models \ONEEVAL{m}{n}{u}{C}
			\end{array}
		}{IH(2)}
		\NLINE{
			\MIMPLIES
			\begin{array}[t]{l} 
				(\AN{\MMM}, \ M\MMM) \CONV (\AN{\MMM},G_m, \ V_m) 
				\\ 
				\MAND \
				(\AN{\MMM\cdot m:V_m}, \ N(\MMM\cdot m:V_m )) \CONV (\AN{\MMM\cdot m:V_m}, G_n \ V_n) 
				\\ 
				\MAND \
				(\AN{\MMM \cdot m:V_m  \cdot n:V_n}, \ \SEM{m}{\MMM \cdot m:V_m  \cdot n:V_n } \SEM{n}{\MMM \cdot m:V_m  \cdot n:V_n }) \CONV (G''', \ V_u)
				\\
				\MAND \ 
				\MMM \cdot m:V_m  \cdot n:V_n \cdot u:V_u \models C
			\end{array}
		}{Sem. $\ONEEVAL{}{}{}{}$}
		\NLINE{
			\MIMPLIES
			\begin{array}[t]{l} 
				(\AN{\MMM}, \ M\MMM) \CONV (\AN{\MMM},G_m, \ V_m) 
				\\ 
				\MAND \
				(\AN{\MMM\cdot m:V_m}, \ N(\MMM\cdot m:V_m )) \CONV (\AN{\MMM\cdot m:V_m}, G_n \ V_n) 
				\\ 
				\MAND \
				(\AN{\MMM \cdot m:V_m  \cdot n:V_n}, \ V_m V_n) \CONV (G''', \ V_u)
				\\
				\MAND \ 
				\MMM \cdot m:V_m  \cdot n:V_n \cdot u:V_u \models C
			\end{array}
		}{Sem. $\SEM{x}{\MMM}$}
		\NLINE{
			\MIMPLIES
			\begin{array}[t]{l} 
				(\AN{\MMM}, \ M\MMM) \CONV (\AN{\MMM},G_m, \ V_m) 
				\\ 
				\MAND \
				(\AN{\MMM\cdot m:V_m}, \ N(\MMM\cdot m:V_m )) \CONV (\AN{\MMM\cdot m:V_m}, G_n \ V_n) 
				\\ 
				\MAND \
				(\AN{\MMM \cdot m:V_m  \cdot n:V_n}, \ V_m V_n) \CONV (G''', \ V_u)
				\\
				\MAND \ 
				\MMM \cdot u:V_u \models C
			\end{array}
		}{$C-\THINWRT{m,n}$}
		\NLINE{
			\MIMPLIES
			\begin{array}[t]{l} 
				(\AN{\MMM}, \ (MN)\MMM) \CONV (G''', \ V_u)
				\\
				\MAND \ 
				\MMM \cdot u:V_u \models C
			\end{array}
		}{Op. Sem.  (App)}
		\NLINE{
			\MIMPLIES \
			\Mforall \MMM^{\GAMMA_0}.\ 
			\CONSTRUCT{\GAMMA}{\MMM} 
			\ \MIMPLIES \
			\MMM \models A 
			\ \MIMPLIES 
			\begin{array}[t]{l} 
				(\AN{\MMM}, \ (MN)\MMM) \CONV (G''', \ V_u)
				\\
				\MAND \ 
				\MMM \cdot u:V_u \models C
			\end{array}
		}{Assumption, line 6}
		\NLASTLINE{
			\text{Hence: }
			\THREEPREMISERULE
			{
				\models \ASSERT{A}{M}{m}{B}
			}
			{
				\models\ASSERT{B}{N}{n}{\ONEEVAL{m}{n}{u}{C}}
			}
			{
				C-\THINWRT{m,n}
			}
			{
				\models \ASSERT{A}{MN}{u}{C}
			}
		}{lines 1-13}

	\end{NDERIVATION}
}

\newpage
\subsection{\DONE Soundness of \RULENAME{[Pair]}}
\[
\THREEPREMISERULENAMEDRIGHT
{
	\ASSERT{A}{M}{m}{B}
}
{
	\ASSERT{B}{N}{n}{C\LSUBSTLTC{\PAIR{m}{n}}{u}{\GAMMA \PLUSV m \PLUSV n}}
}
{
	C \THINWRT{x}
}
{
	\ASSERT{A}{\PAIR{M}{N}}{u}{C}
}{[Pair]}
\]
Proof:
\\
\begin{NDERIVATION}{1}
	\NLINE{\text{Assume: $\GAMMA$ s.t.  $\JUDGEMENTTYPES{\GAMMA}{\ASSERT{A}{\PAIR{M}{N}}{u}{C\LSUBSTLTC{\PAIR{m}{n}}{u}{\GAMMA \PLUSV m \PLUSV n}}}$}}{}
	\NLINE{\MIMPLIES \ 
		\JUDGEMENTTYPES{\GAMMA}{\ASSERT{A}{M}{m}{B}} 
		\ \MAND \ 
		\JUDGEMENTTYPES{\GAMMA \PLUSV m}{\ASSERT{B}{N}{n}{C\LSUBSTLTC{\PAIR{m}{n}}{u}{\GAMMA \PLUSV m \PLUSV n}}}
		\ \MAND \
		\TYPES{\GAMMA}{N}{\alpha}
	}{}
	\NLINE{\Mforall \GAMMA, \MMM^{\GAMMA_0}. \CONSTRUCT{\GAMMA}{\MMM} \ \MIMPLIES \ \MMM \models A \MIMPLIES (\AN{\MMM}, \ M\MMM) \CONV (\AN{\MMM},G_m, \ V_m) \ \MAND \ \MMM \cdot m:V_m \models B}{IH(1)}
	\NLINE{\Mforall \GAMMA, \MMM_m^{\GAMMA_0 \PLUSV m}. \CONSTRUCT{\GAMMA \PLUSV m}{\MMM_m} \ \MIMPLIES \ \MMM_m \models B \MIMPLIES 
		\begin{array}[t]{l} 
			(\AN{\MMM_m}, \ N\MMM_m) \CONV (\AN{\MMM_m}, G_n \ V_n) 
			\\ \MAND \ \MMM_n \cdot n:V_n\models C\LSUBSTLTC{\PAIR{m}{n}}{u}{\GAMMA \PLUSV m \PLUSV n}
		\end{array}
	}{IH(2)}
	\NLINE{\text{Assume: $\MMM^{\GAMMA_0}$ s.t. } \CONSTRUCT{\GAMMA}{\MMM} \ \MAND \ \MMM \models A}{}
	\NLINE{\MIMPLIES (\AN{\MMM}, \ M\MMM) \CONV (\AN{\MMM},G_m, \ V_m) \ \MAND \ \MMM \cdot m:V_m \models B}{IH(1)}
	\NLINE{
		\MIMPLIES
		\begin{array}[t]{l} 
			(\AN{\MMM}, \ M\MMM) \CONV (\AN{\MMM},G_m, \ V_m) 
			\\ 
			\MAND \
			(\AN{\MMM \cdot m:V_m}, \ N(\MMM\cdot m:V_m )) \CONV (\AN{\MMM\cdot m:V_m}, G_n \ V_n) 
			\\ 
			\MAND \ 
			\MMM \cdot m:V_m  \cdot n:V_n \models C\LSUBSTLTC{\PAIR{m}{n}}{u}{\GAMMA \PLUSV m \PLUSV n}
		\end{array}
	}{IH(2)}
	\NPLINE{
		\MIMPLIES
		\begin{array}[t]{l} 
			(\AN{\MMM}, \ M\MMM) \CONV (\AN{\MMM},G_m, \ V_m) 
			\\ 
			\MAND \
			(\AN{\MMM\cdot m:V_m}, \ N(\MMM\cdot m:V_m )) \CONV (\AN{\MMM\cdot m:V_m}, G_n \ V_n) 
			\\ 
			\MAND \
			(\AN{\MMM \cdot m:V_m  \cdot n:V_n}, \ \SEM{\PAIR{m}{n}}{\MMM \cdot m:V_m  \cdot n:V_n }) \CONV (G''', \ V_u)
			\\
			\MAND \ 
			\MMM \cdot m:V_m  \cdot n:V_n \cdot u:V_u \models C
		\end{array}
	}{3cm}{Sem. $\LSUBST{e}{x}$,\\ $u \notin \DOM{\MMM \cdot m \cdot n}$}
	\NLINE{
		\MIMPLIES
		\begin{array}[t]{l} 
			(\AN{\MMM}, \ M\MMM) \CONV (\AN{\MMM},G_m, \ V_m) 
			\\ 
			\MAND \
			(\AN{\MMM\cdot m:V_m}, \ N(\MMM\cdot m:V_m )) \CONV (\AN{\MMM\cdot m:V_m}, G_n \ V_n) 
			\\ 
			\MAND \
			(\AN{\MMM \cdot m:V_m  \cdot n:V_n}, \ \PAIR{V_m}{V_n}) \CONV (G''', \ V_u)
			\\
			\MAND \ 
			\MMM \cdot m:V_m  \cdot n:V_n \cdot u:V_u \models C
		\end{array}
	}{Sem. $\SEM{e}{\MMM}$}
	\NLINE{
		\MIMPLIES
		\begin{array}[t]{l} 
			(\AN{\MMM}, \ M\MMM) \CONV (\AN{\MMM},G_m, \ V_m) 
			\\ 
			\MAND \
			(\AN{\MMM\cdot m:V_m}, \ N(\MMM\cdot m:V_m )) \CONV (\AN{\MMM\cdot m:V_m}, G_n \ V_n) 
			\\ 
			\MAND \
			(\AN{\MMM \cdot m:V_m  \cdot n:V_n}, \ \PAIR{V_m}{V_n}) \CONV (G''', \ V_u)
			\\
			\MAND \ 
			\MMM \cdot u:V_u \models C
		\end{array}
	}{$C$ \THINWRT{m,n}}
	\NLINE{
		\MIMPLIES
		\begin{array}[t]{l} 
			(\AN{\MMM}, \ \PAIR{M}{N}\MMM) \CONV (G''', \ V_u)
			\\
			\MAND \ 
			\MMM \cdot u:V_u \models C
		\end{array}
	}{Op. Sem.  (Pair)}
	\NLINE{
		\MIMPLIES \
		\Mforall \MMM^{\GAMMA_0}.\ 
		\CONSTRUCT{\GAMMA}{\MMM} 
		\ \MIMPLIES \
		\MMM \models A 
		\ \MIMPLIES 
		\begin{array}[t]{l} 
			(\AN{\MMM}, \ \PAIR{M}{N}\MMM) \CONV (G''', \ V_u)
			\\
			\MAND \ 
			\MMM \cdot u:V_u \models C
		\end{array}
	}{Assumption, (line 5)}
	\NLASTLINE{
		\text{Hence: } \
		\THREEPREMISERULE
		{
			\models \ASSERT{A}{M}{m}{B}
		}
		{
			\models\ASSERT{B}{N}{n}{C\LSUBSTLTC{\PAIR{m}{n}}{u}{\GAMMA \PLUSV m \PLUSV n}}
		}
		{
			C-\THINWRT{m,n}
		}
		{
			\models \ASSERT{A}{\PAIR{M}{N}}{u}{C}
		}
	}{lines 1- 12}

\end{NDERIVATION}

\newpage
\subsection{\DONE Soundness of \RULENAME{[Proj($i$)]}}
\[
\TWOPREMISERULENAMEDRIGHT
{
	\ASSERT{A}{M}{m}{C\LSUBSTLTC{\pi_i(m)}{u}{\GAMMA \PLUSV m}}
}
{
	C \THINWRT{m}
}
{
	\ASSERT{A}{\pi_i(M)}{u}{C}
}{[Proj($i$)]}
\]
Proof:
\\
\begin{NDERIVATION}{1}
	\NLINE{\text{Assume: $\GAMMA$ s.t.  $\JUDGEMENTTYPES{\GAMMA}{\ASSERT{A}{\pi_i(M)}{u}{C}}$}}{}
	\NLINE{\MIMPLIES \ 
		\JUDGEMENTTYPES{\GAMMA}{\ASSERT{A}{M}{m}{C\LSUBSTLTC{\pi_i(m)}{u}{\GAMMA \PLUSV m}}} 
	}{Typing rules}
	\NLINE{\Mforall \MMM^{\GAMMA_0}. \CONSTRUCT{\GAMMA}{\MMM} \ \MIMPLIES \ \MMM \models A \MIMPLIES (\AN{\MMM}, \ M\MMM) \CONV (\AN{\MMM},G_m, \ V_m) \ \MAND \ \MMM \cdot m:V_m \models C\LSUBSTLTC{\pi_i(m)}{u}{\GAMMA \PLUSV m}
	}{IH(1)}
	\NLINE{\text{Assume: $\MMM^{\GAMMA_0}$ s.t. } \CONSTRUCT{\GAMMA}{\MMM} \ \MAND \ \MMM \models A}{}
	\NLINE{
		\MIMPLIES 
		(\AN{\MMM}, \ M\MMM) \CONV (\AN{\MMM},G_m, \ V_m) 
		\ \MAND \ 
		\MMM \cdot m:V_m \models C\LSUBSTLTC{{\pi}_i(m)}{u}{\GAMMA \PLUSV m}
	}{IH(1)}
	\NPLINE{
		\MIMPLIES
		\begin{array}[t]{l} 
			(\AN{\MMM}, \ M\MMM) \CONV (\AN{\MMM},G_m, \ V_m) 
			\\ 
			\MAND \
			(\AN{\MMM \cdot m:V_m}, \ \SEM{{\pi}_i(m)}{\MMM \cdot m:V_m}) \CONV (G''', \ V_u)
			\\
			\MAND \ 
			\MMM \cdot m:V_m  \cdot u:V_u \models C
		\end{array}
	}{3cm}{Sem. $\LSUBST{e}{x}$,\\ $u \notin \DOM{\MMM \cdot m \cdot n}$}
	\NLINE{
		\MIMPLIES
		\begin{array}[t]{l} 
			(\AN{\MMM}, \ M\MMM) \CONV (\AN{\MMM},G_m, \ V_m) 
			\\ 
			\MAND \
			(\AN{\MMM \cdot m:V_m}, \ \pi_i(V_m)) \CONV (G''', \ V_u)
			\\
			\MAND \ 
			\MMM \cdot m:V_m  \cdot u:V_u \models C
		\end{array}
	}{Sem. $\SEM{e}{\MMM}$}
	\NLINE{
		\MIMPLIES
		\begin{array}[t]{l} 
			(\AN{\MMM}, \ M\MMM) \CONV (\AN{\MMM},G_m, \ V_m) 
			\\ 
			\MAND \
			(\AN{\MMM \cdot m:V_m}, \ \pi_i(V_m)) \CONV (G''', \ V_u)
			\\
			\MAND \ 
			\MMM \cdot u:V_u \models C
		\end{array}
	}{$C$ \THINWRT{m}}
	\NLINE{
		\MIMPLIES \
			(\AN{\MMM}, \ \pi_i(M)\MMM) \CONV (G''', \ V_u)
			\
			\MAND \ 
			\MMM \cdot u:V_u \models C
	}{Op. Sem.  (Pair)}
	\NLINE{
		\MIMPLIES \
		\Mforall \MMM^{\GAMMA_0}.\ 
		\CONSTRUCT{\GAMMA}{\MMM}
		\ \MIMPLIES \
		\MMM \models A 
		\ \MIMPLIES 
		\begin{array}[t]{l} 
			(\AN{\MMM}, \ \PAIR{M}{N}\MMM) \CONV (G''', \ V_u)
			\\
			\MAND \ 
			\MMM \cdot u:V_u \models C
		\end{array}
	}{Assumption, (line 4)}
	\NLASTLINE{
		\text{Hence: } 
		\TWOPREMISERULE
		{
			\models \ASSERT{A}{M}{m}{C\LSUBSTLTC{\pi_i(m)}{u}{\GAMMA \PLUSV m}}
		}
		{
			C \THINWRT{m}
		}
		{
			\models \ASSERT{A}{\pi_i(M)}{u}{C}
		}
	}{lines 1-10}

\end{NDERIVATION}

\subsection{Soundness of \RULENAME{[If]}}
\[
\FIVEPREMISERULENAMEDRIGHT
{
	\ASSERT{A}{M}{m}{B}
}
{
	\ASSERT{B\LSUBST{b_i}{m}}{N_i}{u}{C}
}
{
	b_1 = \TRUE
}
{
	b_2 = \FALSE
}
{
	i = 1, 2
}
{
	\ASSERT{A}{\IFTHENELSE{M}{N_1}{N_2}}{u}{C}
}{[If]}
\]

Proof: standard. Note that substitution is equivalent to standard substitution for $b_i$-values of type $\BOOL$.

\newpage

\subsection{\DONE Soundness of \RULENAME{[Let]}}
\[
\THREEPREMISERULENAMEDRIGHT
{
	\ASSERT{A}{M}{x}{B}
}
{
	\ASSERT{B}{N}{u}{C}
}
{
	C \THINWRT{x}
}
{
	\ASSERT{A}{\LET{x}{M}{N}}{u}{C}
}{[Let]}
\]
\PROOFFINISHED
{
	\begin{NDERIVATION}{1}
		\NLINE{\text{Assume: $\GAMMA$ s.t.  $\JUDGEMENTTYPES{\GAMMA}{\ASSERT{A}{\LET{x}{M}{N}}{u}{C}}$}}{}
		\NLINE{\MIMPLIES \ 
			\TYPES{\GAMMA\LTCtoSTC}{M}{\alpha} \ \MAND \ \TYPES{\GAMMA\LTCtoSTC, x:\alpha}{N}{\beta}
		}{Typing terms}
		\NLINE{\text{Assume: } \Mforall \MMM^{\GAMMA_0}. \ \CONSTRUCT{\GAMMA}{\MMM} \MIMPLIES \MMM \models A \ \MIMPLIES \ (\AN{\MMM}, M\MMM) \CONV (G_m', V_m) \ \MAND \ \MMM \cdot x:V_m \models B}{IH(1)}
		\NLINE{\text{Assume: } \Mforall \MMM_x^{\GAMMA_0 \PLUSV x:\alpha}.  \CONSTRUCT{\GAMMA \PLUSV x}{\MMM_x} \MIMPLIES\ \MMM_x \models B \ \MIMPLIES \ (\AN{\MMM_x}, N\MMM_x) \CONV (G_n', V_n) \ \MAND \ \MMM_x \cdot u:V_n \models C}{IH(2)}
		\NLINE{\text{Assume: $\MMM^{\GAMMA_0}$ s.t: } \CONSTRUCT{\GAMMA}{\MMM} \ \MAND \ \MMM \models A}{}
		\NLINE{\MIMPLIES \  (\AN{\MMM}, \ M\MMM ) \CONV (G', \ V_m)  \ \MAND \ \MMM \cdot x:V_m \models B}{IH(1)} 
		\NLINE{\MIMPLIES \ 
			\begin{array}[t]{l}
				(\AN{\MMM}, \ M\MMM ) \CONV (G', \ V_m)  
				\\
				\MAND \ (\AN{\MMM \cdot x:V_m}, \ N(\MMM \cdot x:V_m)) \CONV (G_n', \ V_n) 
				\\ \MAND \ 
				\MMM \cdot x:V_m \cdot u:V_n \models C
			\end{array}
		}{IH(2)}
		\NLINE{\MIMPLIES \ 
			\begin{array}[t]{l}
				(\AN{\MMM}, \ M\MMM ) \CONV (G', \ V_m)  
				\\
				\MAND \ (\AN{\MMM \cdot x:V_m}, \ N(\MMM \cdot x:V_m)) \CONV (G_n', \ V_n) 
				\\ \MAND \ 
				\MMM \cdot u:V_n \models C
			\end{array}
		}{$C \THINWRT{x}$}
		\NLINE{\MIMPLIES \ 
			\begin{array}[t]{l}
				(\AN{\MMM}, \ M\MMM ) \CONV (G', \ V_m)  
				\\
				\MAND \ (\AN{\MMM \cdot x:V_m}, \ N\PSUBST{V_m}{x}\MMM) \CONV (G_n', \ V_n)  
				\\
				\MAND \ \MMM \cdot u:V_n \models C
			\end{array}
		}{Sem. $\PSUBST{V_m}{x}$}
		\NLINE{\MIMPLIES \ 
				(\AN{\MMM}, \ (\LET{x}{M}{N})\MMM) \CONV (G'', \ V_n)  
				\
				\MAND \ \MMM \cdot u:V_n \models C
		}{Op. Sem. $\LET{x}{M}{N}$}
		\NLINE{\Mforall \MMM. \CONSTRUCT{\GAMMA}{\MMM} \ \MIMPLIES \ \MMM \models \ASSERT{A}{\LET{x}{M}{N}}{m}{C} 
		}{Sem. valid triple}
		\NLASTLINE{\text{Hence: } \THREEPREMISERULE{\models \ASSERT{A}{M}{x}{B}}{\models \ASSERT{B}{N}{u}{C}}{C \THINWRT{x}}{\models \ASSERT{A}{\LET{x}{M}{N}}{u}{C}}}{lines 1-11}
	\end{NDERIVATION}
	
}

\newpage
\subsection{\DONE Soundness of derived rule \RULENAME{[LetFresh]}}

\[
\THREEPREMISERULENAMEDRIGHT
{
	A-\EXTINDEP
}
{
	\GAMMA \PLUSV x  \Vdash \ASSERT{A \PAND \FRESH{x}{\GAMMA}}{M}{m}{C}
}
{
	C \THINWRT{x}
}
{
	\GAMMA \Vdash \ASSERT{A}{\LET{x}{\GENSYM\ ()}{M}}{m}{C}
}{[LetFresh]}
\]

\PROOFFINISHED
{	\begin{NDERIVATION}{1}
		\NLINE{\text{Assume: $\GAMMA$ s.t.  $\JUDGEMENTTYPES{\GAMMA}{\ASSERT{A}{\LET{x}{\GENSYM()}{N}}{u}{C}}$}}{}
		\NLINE{\GAMMA \Vdash \ASSERT{\TRUTH}{\GENSYM()}{x}{\FRESH{x}{\GAMMA}}}{See Example 1}
		\NLINE{\GAMMA \Vdash \ASSERT{A}{\GENSYM()}{x}{A \PAND \FRESH{x}{\GAMMA}}}{\RULENAME{[Invar]}, 7, $A-\EXTINDEP$}
		\NLINE{\GAMMA \PLUSV x \Vdash \ASSERT{A \PAND \FRESH{x}{\GAMMA}}{M}{m}{C}}{Assumption}
		\NLINE{C-\THINWRT{x}}{Assumption}
		\NLINE{\GAMMA \Vdash \ASSERT{A}{\LET{x}{\GENSYM()}{M}}{m}{C} 
		}{\RULENAME{[Let]}, 3, 4, 5}
		\NLASTLINE{\text{Hence: } \THREEPREMISERULE{A-\EXTINDEP}{\models \ASSERT{A \PAND \FRESH{x}{\GAMMA}}{N}{u}{C}}{C \THINWRT{x}}{\models \ASSERT{A}{\LET{x}{M}{N}}{u}{C}}}{lines 1-6}
	\end{NDERIVATION}
}

\newpage
\subsection{Soundness of Key structural rules:}

\subsection{\DONE Soundness of structural rule \RULENAME{[Conseq]}}

\[
\THREEPREMISERULENAMEDRIGHT
{
	A \PIMPLIES A'
}
{
	\ASSERT{A'}{M}{m}{B'}
}
{
	B' \PIMPLIES B
}
{
	\ASSERT{A}{M}{m}{B}
}{[Conseq]}
\]

\PROOFFINISHED
{
	Proof, simply by application of the assumptions:
	\\
	\begin{NDERIVATION}{1}
		\NLINE{\text{Assume $\GAMMA$, s.t. $\JUDGEMENTTYPES{\GAMMA}{\ASSERT{A}{M}{m}{B}}$}}{}
		\NLINE{
			\begin{array}[t]{l}
				\text{Assume typing holds for assumptions i.e. }
				\\
				\FORMULATYPES{\GAMMA}{A \PIMPLIES A'} 
				\ \MAND \ 
				\JUDGEMENTTYPES{\GAMMA}{\ASSERT{A'}{M^{\alpha}}{m}{B'}}
				\ \MAND \ 
				\FORMULATYPES{\GAMMA \PLUSV m :\alpha}{B' \PIMPLIES B}
			\end{array}
		}{}
		\NLINE{\text{Assume: $\Mforall \MMM_0^{\GAMMA_0}. \ \CONSTRUCT{\GAMMA}{\MMM_0} \ \MIMPLIES \ \MMM_0 \models A \PIMPLIES A'$}}{IH(1)}
		\NLINE{\text{Assume: $\Mforall \MMM_0^{\GAMMA_0}. \ \CONSTRUCT{\GAMMA}{\MMM_0} \ \MIMPLIES \ \MMM_0 \models \ASSERT{A'}{M}{m}{B'}$}}{IH(2)}
		\NLINE{\text{Assume: $\Mforall \MMM_1^{\GAMMA_1}. \  \CONSTRUCT{\GAMMA \PLUSV m}{\MMM_1} \ \MIMPLIES \ \MMM_1 \models B' \PIMPLIES B$}
		}{IH(3)}
		\NLINE{\text{Assume: $\MMM^{\GAMMA'}$ s.t. } \CONSTRUCT{\GAMMA}{\MMM} \ \MAND \ \MMM \models A}{}
		\NLINE{\MIMPLIES \ \MMM \models A'}{MP, IH(1)}
		\NLINE{\MIMPLIES \ (\AN{\MMM}, \ M\MMM) \CONV (\AN{\MMM}, G', \ V) \ \MAND \ \MMM \cdot m:V \models B'}{MP, IH(2)}
		\NLINE{\MIMPLIES \ (\AN{\MMM}, \ M\MMM) \CONV (\AN{\MMM}, G', \ V) \ \MAND \ \MMM \cdot m:V \models B}{MP, IH(3)}
		\NLINE{\MIMPLIES \ \MMM \models A \ \MIMPLIES \ (\AN{\MMM}, \ M\MMM) \CONV (\AN{\MMM}, G', \ V) \MAND \MMM \cdot m:V \models B}{Assumption 6-9 }
		\NLINE{\MIMPLIES \ \Mforall \MMM^{\GAMMA'}. \CONSTRUCT{\GAMMA}{\MMM} \ \MIMPLIES \ \MMM \models \ASSERT{A}{M}{m}{B}}{Assumption 3-10}
		\NLASTLINE{\text{Hence: }
			\THREEPREMISERULE
			{
				A \PIMPLIES A'
			}
			{
				\models \ASSERT{A'}{M}{m}{B'}
			}
			{
				B' \PIMPLIES B
			}
			{
				\models \ASSERT{A}{M}{m}{B}
			}
		}{Assumption 1-11}
	\end{NDERIVATION}
}

\newpage
\subsection{\DONE Soundness of structural rule \RULENAME{[Invar]}}
\[
\TWOPREMISERULENAMEDRIGHT
{
	C-\EXTINDEP
}
{
	\ASSERT{A}{M}{m}{B}
}
{
	\ASSERT{A \PAND C}{M}{m}{B \PAND C}
}{[Invar]}
\]

\PROOFFINISHED
{
	Proof: 
	\begin{NDERIVATION}{1}
		\NLINE{\text{Assume: $\GAMMA$ s.t.  $\JUDGEMENTTYPES{\GAMMA}{\ASSERT{A \PAND C}{M}{m}{B \PAND C}}$}}{}
		\NLINE{\MIMPLIES \ 
			\JUDGEMENTTYPES{\GAMMA}{\ASSERT{A}{M}{m}{B}} 
		}{}
		\NLINE{\text{Assume: } 
			\Mforall \MMM^{\GAMMA_0}. \CONSTRUCT{\GAMMA}{\MMM} \ \MIMPLIES \  \MMM \models A
			\MIMPLIES 
			\begin{array}[t]{ll}
				(G, \ M \MMM) \CONV (G,G', \ V)
				\\
				\MAND \ \MMM \cdot m:V \models B
			\end{array}
		}{IH(1)}
%
		\NLINE{\text{Assume for some model $\MMM^{\GAMMA_0}$ s.t. } \CONSTRUCT{\GAMMA}{\MMM} \ \MAND \ \MMM \models A \PAND C}{}
		\NLINE{\MMM \models A \quad \MAND \quad \MMM \models C }{Sem. $\PAND$}
		\NLINE{
			\begin{array}[t]{ll}
				(G, \ M \MMM) \CONV (G,G', \ V)
				\\
				\MAND \ \MMM \cdot m:V \models B
			\end{array}
		\
		\MAND \ \MMM \models C
		}{IH(1)}
		\NLINE{
			\begin{array}[t]{ll}
				(G, \ M \MMM) \CONV (G,G', \ V)
				\\
				\MAND \ \MMM \cdot m:V \models B
				\\
				\MAND \ \MMM \cdot m:V \models C
			\end{array}
		}{
			$\begin{array}[t]{r}
				C-\EXTINDEP
				\\
				\text{Sem. $\LTCDERIVEDVALUE{}{}{}{}$}
				\
				\MIMPLIES \ \LTCDERIVEDVALUE{M}{\GAMMA}{\MMM}{V}
				\\ 
				\text{Sem. $\EXTSTAR$}
				\
				\MIMPLIES \ \MMM \EXTSTAR \MMM\cdot m:V
			\end{array}$
		}
		\NLINE{
				(G, \ M \MMM) \CONV (G,G', \ V)
				\
				\MAND \ \MMM \cdot m:V \models B \PAND  C
		}{Sem. $\PAND$}
		\NLINE{
			\Mforall \MMM^{\GAMMA_0}. \CONSTRUCT{\GAMMA}{\MMM} \ \MIMPLIES \ \MMM \models \ASSERT{A \PAND C}{M}{m}{B \PAND C} 
		}{lines 5-8}
		\NLASTLINE{\text{Hence: }
			\TWOPREMISERULE
			{
				C-\EXTINDEP
			}
			{
				\models
				\ASSERT{A}{M}{m}{B}
			}
			{
				\models
				\ASSERT{A \PAND C}{M}{m}{B \PAND C}
			}
		}{lines 1-9}
	\end{NDERIVATION}
	
}

%% file: appendix/appendix_EXTINDEP.tex
\section{\EXTINDEP \ Formulae Construction Lemmas}
\label{appendix_EXTINDEP}
Syntactically define \EXTINDEP \ formulae and how to construct them:
\\
Some are not core constructions but specific instances used in the reasoning examples in Sec.~\ref{reasoning}.

\input{appendix/syntactic_EXTINDEP}

We introduce the following 2 lemmas for the specific cases in Def.~\ref{def:syntactic-EXTINDEP}
\begin{lemma}[\DONE  Constructing \EXTINDEP formulae from $\FAD{\TCV} \FORALL{x}{\TCV}$]
	\label{lem:EXTINDEP-CONSTR-FADForall}
	\[
	A^{-\TCV}-\text{\EXTINDEP} 
	\ \MIMPLIES \
	\FAD{\TCV} \FORALL{x}{\TCV} A^{-\TCV}-\text{\EXTINDEP} 
	\]
	\PROOFFINISHED
	{ 
		\begin{NDERIVATION}{1}
			\NLINE{\text{Assume: $A^{-\TCV}$-\EXTINDEP}}{}
			\NLINE{
				\MIFF \
				\Mforall \GAMMA, \MMM_x^{\GAMMA \PLUSV x:\alpha}, \MMM_x'^{\GAMMA',x:\alpha}.  \
					(\FORMULATYPES{\GAMMA \PLUSV x:\alpha}{A} 
					\
					\MAND \ \MMM_x \EXTSTAR \MMM'_x )
					\
					\MIMPLIES (\MMM_x \models A \ \MIFF \ \MMM'_x \models A)
			}{}
			\NLINE{
					\text{Show: } 
					\
					\Mforall \GAMMA, \MMM^{\GAMMA}, \MMM'^{\GAMMA'}. 
						\FORMULATYPES{\GAMMA}{\FAD{\TCV} \FORALL{x}{\TCV} A^{-\TCV}} 
						\
						\MAND \ \MMM \EXTSTAR \MMM' 
						\
						\MIMPLIES 
						\left(
						\begin{array}{l}
							\MMM \models \FAD{\TCV} \FORALL{x}{\TCV} A 
							\\ 
							\MIFF \ \MMM' \models \FAD{\TCV} \FORALL{x}{\TCV} A
						\end{array}
						\right)
			}{}
			\NLINE{
					\text{Hence assume: } 
					\GAMMA, \MMM^{\GAMMA}, \MMM'^{\GAMMA'} 
					\text{ s.t. } 
					\FORMULATYPES{\GAMMA}{\FAD{\TCV} \FORALL{x}{\TCV} A} 
					\ \MAND \ 
					\MMM \EXTSTAR \MMM'
			}{}	
			\NLINE{
				\MIMPLIES: \
					\text{Show that: } 
					\MMM \models \FAD{\TCV} \FORALL{x}{\TCV} A
					\ \MIMPLIES \
					\MMM' \models \FAD{\TCV} \FORALL{x}{\TCV} A
			}{line 7 below}
			\NLASTLINE{
				\MIMPLIEDBY: \
				\text{Show that: } 
				\MMM' \models \FAD{\TCV} \FORALL{x}{\TCV} A
				\ \MIMPLIES \
				\MMM \models \FAD{\TCV} \FORALL{x}{\TCV} A
			}{line 13 below}
			
		\end{NDERIVATION}
	$\MIMPLIES$:
		\begin{NDERIVATION}{7}
			\NLINE{\text{Let: }  \MMM_{1d} \equiv \MMM_1 \cdot \TCV:\GAMMA_1 \REMOVETCVfrom  \text{ and } \MMM_{2d}' \equiv \MMM_2' \cdot \TCV:\GAMMA_2 \REMOVETCVfrom }{}
			\NLINE{
					\MMM \EXTSTAR \MMM' 
					\ \MAND \ 
					\MMM \models \FAD{\TCV} \FORALL{x}{\TCV} A
			}{Ignore Typing Constraints as these hold}
			\NLINE{
				\MIFF \
					\MMM \EXTSTAR \MMM' 
					\ \MAND \ 
					\Mforall \MMM_1^{\GAMMA_1}.  \
						\MMM \EXTSTAR \MMM_1
						\
						\MIMPLIES \ 
						\Mforall M. \
							\LTCDERIVEDVALUE{M}{\TCV}{\MMM_{1d}}{V}
							\
							\MIMPLIES \ \MMM_{1d} \cdot x:V \models A
			}{Sem. $\FAD{\TCV} \FORALL{x}{\TCV}$}
			\NLINE{
				\MIFF \
					\MMM \EXTSTAR \MMM' 
					\ \MAND \ 
					\Mforall \MMM_1^{\GAMMA_1}.  \
						\MMM \EXTSTAR \MMM_1
						\
						\MIMPLIES \ 
						\Mforall M. \
							\LTCDERIVEDVALUE{M}{\TCV}{\MMM_{1d}}{V}
							\
							\MIMPLIES \ \MMM_1 \cdot x:V \models A
			}{
					Lemma \ref{lem:model_and_model_plus_TCV_models_equivalently_TCV-free_formula}
			}
			\NLINE{
				\MIMPLIES \
					\MMM \EXTSTAR \MMM' 
					\ \MAND \ 
					\Mforall \MMM_2'^{\GAMMA_2'}. \
						\MMM' \EXTSTAR \MMM'_2
						\
						\MIMPLIES \ 
						\Mforall M. \
							\LTCDERIVEDVALUE{M}{\TCV}{\MMM_{2d}'}{V}
							\
							\MIMPLIES \ \MMM'_2 \cdot x:V \models A
			}{
					select $\MMM_2'$ s.t. $\MMM' \EXTSTAR \MMM_2'$ 
			}
			\NLASTLINE{
				\MIFF \
				\MMM' \models \FAD{\TCV} \FORALL{x}{\TCV} A
			}{Lemma \ref{lem:model_and_model_plus_TCV_models_equivalently_TCV-free_formula}, Sem. $\FAD{\TCV} \FORALL{x}{\TCV}$}
		\end{NDERIVATION}
	$\MIMPLIEDBY$:
		\begin{NDERIVATION}{13}
			\NLINE{\text{Let: }  \MMM_{2d}' \equiv \MMM_2' \cdot \TCV:\GAMMA_2 \REMOVETCVfrom 
				\text{ and }  \MMM_{1d} \equiv \MMM_1 \cdot \TCV:\GAMMA_1 \REMOVETCVfrom 
				 \text{ and } \MMM_{1d}' \equiv \MMM_1 \cdot \TCV:\GAMMA_1' \REMOVETCVfrom 
			 }{}			
			\NLINE{
					\MMM \EXTSTAR \MMM' 
					\ \MAND \ 
					\MMM' \models \FAD{\TCV} \FORALL{x}{\TCV} A
			}{Ignore Typing Constraints as these hold}
			\NLINE{
				\MIFF \
					\MMM \EXTSTAR \MMM' 
					\ \MAND \ 
					\Mforall \MMM_2'^{\GAMMA_2'}.  \
						\MMM' \EXTSTAR \MMM_2'
						\ 
						\MIMPLIES \ 
						\Mforall M.
						\LTCDERIVEDVALUE{M}{\TCV}{\MMM_{2d}'}{V}
						\
						\MIMPLIES \ \MMM_{2d}'\cdot x:V \models A
			}{Sem. $\FAD{\TCV} \FORALL{x}{\TCV}$}
			\NLINE{
				\MIFF \
				\MMM \EXTSTAR \MMM' 
				\ \MAND \ 
						\Mforall \MMM_2'^{\GAMMA_2'}. \
							\MMM' \EXTSTAR \MMM_2'
							\
							\MIMPLIES \ 
							\Mforall M.
								\LTCDERIVEDVALUE{M}{\TCV}{\MMM_{2d}'}{V}
								\
								\MIMPLIES \ \MMM_2' \cdot x:V \models A
			}{Lemma \ref{lem:model_and_model_plus_TCV_models_equivalently_TCV-free_formula}}
			\NLINE{
				\MIFF \
				\MMM \EXTSTAR \MMM' 
				\ \MAND 
					\begin{array}[t]{l}
					\Mforall \MMM_1^{\GAMMA_1}. \ 
					\MMM \EXTSTAR \MMM_1
					\\ \MIMPLIES \ 
					\Mforall \MMM_2'^{\GAMMA_2'}. 
						\MMM' \EXTSTAR \MMM_2'
						\
						\MIMPLIES \ 
						\Mforall M. \
							\LTCDERIVEDVALUE{M}{\TCV}{\MMM_{2d}' }{V}
							\
							\MIMPLIES \ \MMM_2'  \cdot x:V \models A
				\end{array}
			}{Intro $\MMM_1$}
			\NPLINE{
				\MIFF \
				\MMM \EXTSTAR \MMM \cdot \VEC{\MMM}'
				\ \MAND 
				\begin{array}[t]{l}
					\Mforall (\MMM \cdot \VEC{\MMM}_1)^{\GAMMA_1}.
					\
					\MMM \EXTSTAR \MMM \cdot \VEC{\MMM}_1
					\\ 
					\hspace{-1cm}
					\MIMPLIES \ 
					\Mforall (\MMM \cdot \VEC{\MMM}' \cdot \VEC{\MMM}'_2)^{\GAMMA_2'}. \
					\begin{array}[t]{l}
						\MMM \cdot \VEC{\MMM}' \EXTSTAR \MMM \cdot \VEC{\MMM}' \cdot \VEC{\MMM}'_2
						\\
						\MIMPLIES \ 
						\Mforall M. \
							\LTCDERIVEDVALUE{M}{\TCV}{\MMM_{2d}'}{V}
							\
							\MIMPLIES \ \MMM_{2d}' \cdot x:V \models A
					\end{array}
				\end{array}
			}{3cm}{
				write $\MMM \equiv \MMM \cdot \VEC{\MMM}'$
				\\
				write $\MMM_1 \equiv \MMM \cdot \VEC{\MMM}_1$
				\\
				write $\MMM_2' \equiv \MMM' \cdot \VEC{\MMM}'_2$
			}
			\NPLINE{
				\MIFF \
					\MMM \EXTSTAR \MMM \cdot \VEC{\MMM}'
					\ \MAND \
					\Mforall \VEC{\MMM}_1.
					\begin{array}[t]{l}
						\MMM \EXTSTAR (\MMM \cdot \VEC{\MMM}_1)^{\GAMMA_1}
						\ \MAND \ 
						\MMM \cdot \VEC{\MMM}' \EXTSTAR (\MMM \cdot \VEC{\MMM}' \cdot \VEC{\MMM}_1)^{\GAMMA_1'}
						\\
						\MIMPLIES \ 
						\Mforall M.
							\LTCDERIVEDVALUE{M}{\TCV}{\MMM_{1d}' 
							}{V}
							\
							\MIMPLIES \ \MMM \cdot \VEC{\MMM}' \cdot \VEC{\MMM}_1 \cdot x:V \models A
					\end{array}
			}{3cm}{
				Select $\VEC{\MMM}'_2$ as $\VEC{\MMM}_1$
				\\
				$\GAMMA_2' \equiv \GAMMA_1'$
				\
				FOL
			}
			\NLINE{
				\MIFF \
					\MMM \EXTSTAR \MMM \cdot \VEC{\MMM}'
					\ \MAND \ 
					\Mforall \VEC{\MMM}_1.
					\begin{array}[t]{l}
						\MMM \EXTSTAR (\MMM \cdot \VEC{\MMM}_1)^{\GAMMA_1}
						\ \MAND \ 
						\MMM \cdot \VEC{\MMM}_1 \EXTSTAR (\MMM \cdot \VEC{\MMM}_1 \cdot \VEC{\MMM}')^{\GAMMA_1'}
						\\
						\MIMPLIES \ 
						\Mforall M. \
							\LTCDERIVEDVALUE{M}{\TCV}{\MMM_{1d}'
							 }{V}
							\
							\MIMPLIES \ \MMM \cdot \VEC{\MMM}_1 \cdot \VEC{\MMM}' \cdot x:V \models A
					\end{array}
			}{
				Lemma \ref{lem:two_extensions_combine_to_make_extensions_of_each_other}
			}
			\NPLINE{
				\MIMPLIES \ 
					\MMM \EXTSTAR \MMM \cdot \VEC{\MMM}'
					\ \MAND \ 
					\Mforall \VEC{\MMM}_1.
					\begin{array}[t]{l}
						\MMM \EXTSTAR (\MMM \cdot \VEC{\MMM}_1)^{\GAMMA_1}
						\ \MAND \ 
						\MMM \cdot \VEC{\MMM}_1 \EXTSTAR (\MMM \cdot \VEC{\MMM}_1 \cdot \VEC{\MMM}')^{\GAMMA_1'}
						\\
						\MIMPLIES \ 
						\Mforall M.
							\LTCDERIVEDVALUE{M}{\TCV}{\MMM _{1d}
							}{V}
							\
							\MIMPLIES \ \MMM \cdot \VEC{\MMM}_1 \cdot \VEC{\MMM}' \cdot x:V \models A
					\end{array}
\hspace{-1cm}
			}{4cm}{
				Subset of possible $M$'s
				\\
				$\TCTYPES{\GAMMA_1'}{\GAMMA_1}$
				\\
				Lemma \ref{lem:eval_under_extensions_are_equivalent}
			}
			\NLINE{
				\MIMPLIES \
					\Mforall \MMM_1.
					\begin{array}[t]{l}
						\MMM \EXTSTAR \MMM_1^{\GAMMA_1}
						\ \MAND \
						\MMM_1 \EXTSTAR (\MMM_1 \cdot \VEC{\MMM}')^{\GAMMA_1'}
						\\
						\MIMPLIES \ 
						\Mforall M.
							\LTCDERIVEDVALUE{M}{\TCV}{\MMM_{1d} 
							}{V}
							\
							\MIMPLIES \ \MMM_1\cdot x:V \models A
					\end{array}
			}{
					Lemma \ref{lem:Gamma_derived_terms_maintain_extension_when_added}
					\
					$A$-\EXTINDEP
			}
			\NPLINE{
				\MIMPLIES \
					\Mforall \MMM_1.\
						\MMM \EXTSTAR \MMM_1^{\GAMMA_1}
					\
					\MIMPLIES \ 
					\Mforall M.\
						\LTCDERIVEDVALUE{M}{\TCV}{\MMM_{1d}
						 }{V}
						\ \MIMPLIES \ 
						\MMM_1\cdot x:V \models A
			}{3.5cm}{
				$\AN{\VEC{\MMM}_1} \cap \AN{\VEC{\MMM}'} \subseteq \AN{\MMM}$
				\\
				guarantees constraint
			}
			\NLINE{
				\MIFF \
					\Mforall \MMM_1.\
						\MMM \EXTSTAR \MMM_1^{\GAMMA_1}
						\
						\MIMPLIES \ 
						\Mforall M.\
						\LTCDERIVEDVALUE{M}{\TCV}{\MMM_{1d} 
						}{V}
						\ \MIMPLIES \ 
						\MMM_{1d} 
						\models A
			}{Lemma \ref{lem:model_and_model_plus_TCV_models_equivalently_TCV-free_formula}, $A^{-\TCV}$}
			\NLASTLINE{
				\MIFF \ \MMM  \models \FAD{\TCV} \FORALL{x}{\TCV} A
			}{
				Sem. $\FAD{\TCV} \FORALL{x}{\GAMMA\PLUSTC \TCV}$
			}
		\end{NDERIVATION}
}
\end{lemma}

\begin{lemma}[\DONE The formula in the postcondition of \RULENAME{[Gensym]} is $\EXTINDEP$]
	\label{lem:GS_formula_is_EXTINDEP}
	\[
	\FAD{\TCV} \ONEEVAL{f}{()}{b}{\FRESH{b}{\TCV}}-\EXTINDEP
	\]
	\PROOFFINISHED
	{
		Proof:
		\\
		\[
		\Mforall \MMM^{\GAMMA}, \MMM'. \MMM \EXTSTAR \MMM' \MIMPLIES (\MMM \models \FAD{\TCV} \ONEEVAL{f}{()}{b}{\FRESH{b}{\TCV}} \ \MIFF \ \MMM' \models \FAD{\TCV} \ONEEVAL{f}{()}{b}{\FRESH{b}{\TCV}})
		\]
		\\
		{\DONE Extending} ($\MMM \models \FAD{\TCV} \ONEEVAL{f}{()}{b}{\FRESH{b}{\TCV}} \ \MIMPLIES \ \MMM' \models \FAD{\TCV} \ONEEVAL{f}{()}{b}{\FRESH{b}{\TCV}}$):
		\begin{NDERIVATION}{1}
			\NLINE{\text{Assume: $\GAMMA$, s.t. $\FORMULATYPES{\GAMMA}{\FAD{\TCV} \ONEEVAL{f}{()}{b}{\FRESH{b}{\GAMMA}}}$ }}{}
			\NLINE{\text{Assume some $\MMM^{\GAMMA_d}$, $\MMM'^{\GAMMA'}$ s.t. $\CONSTRUCT{\GAMMA}{\MMM}$ $\MMM \EXTSTAR\MMM'$ and $\MMM \models \FAD{\TCV} \ONEEVAL{f}{()}{b}{\FRESH{b}{\TCV}}$ }}{}
			\NLINE{\MIFF \ \Mforall \MMM_1^{\GAMMA_1}. \MMM\EXTSTAR \MMM_1 \MIMPLIES \MMM_1 \cdot \TCV:\GAMMA_1\REMOVETCVfrom \models  \ONEEVAL{f}{()}{b}{\FRESH{b}{\TCV}}  }{Sem. $\FAD{\TCV}$}
			\NLINE{\MIMPLIES \ \Mforall \MMM_1^{\GAMMA_1}. \MMM \EXTSTAR \MMM' \EXTSTAR \MMM_1 \MIMPLIES \MMM_1 \cdot \TCV:\GAMMA_1\REMOVETCVfrom \models  \ONEEVAL{f}{()}{b}{\FRESH{b}{\TCV}}  }{Subset $\Mforall \MMM_1$}
			\NLINE{\MIMPLIES \ \Mforall \MMM_1^{\GAMMA_1}. \MMM' \EXTSTAR \MMM_1 \MIMPLIES \MMM_1 \cdot \TCV:\GAMMA_1\REMOVETCVfrom \models  \ONEEVAL{f}{()}{b}{\FRESH{b}{\TCV}}  }{Remove $\MMM\EXTSTAR$}
			\NLASTLINE{\MIMPLIES \ \MMM' \models \FAD{\TCV} \ONEEVAL{f}{()}{b}{\FRESH{b}{\TCV} } }{}
		\end{NDERIVATION}
		{\DONE  Contracting}  ($\MMM \models \FAD{\TCV} \ONEEVAL{f}{()}{b}{\FRESH{b}{\TCV}} \ \MIMPLIEDBY \ \MMM' \models \FAD{\TCV} \ONEEVAL{f}{()}{b}{\FRESH{b}{\TCV}}$):
		\begin{NDERIVATION}{7}
			\NLINE{\text{Assume: $\GAMMA$,s.t. $\FORMULATYPES{\GAMMA}{\FAD{\TCV} \ONEEVAL{f}{()}{b}{\FRESH{b}{\GAMMA}}}$ }}{}
			\NLINE{\text{Assume some  $\MMM^{\GAMMA_d}$, $\MMM'$ s.t. $\CONSTRUCT{\GAMMA}{\MMM}$ and $\MMM \EXTSTAR\MMM'$ and $\MMM' \models \FAD{\TCV} \ONEEVAL{f}{()}{b}{\FRESH{b}{\TCV}}$ }}{}
			\NLINE{\MIFF \ \Mforall \MMM_1^{\GAMMA_1}. \MMM'\EXTSTAR \MMM_1 \MIMPLIES \MMM_1 \cdot \TCV:\GAMMA_1\REMOVETCVfrom \models  \ONEEVAL{f}{()}{b}{\FRESH{b}{\TCV}}  }{Sem. $\FAD{\TCV}$}
			\NPLINE{\MIFF \ \Mforall \VEC{\MMM}_1^{\VEC{\GAMMA}_1}. 
				\begin{array}[t]{l}
					\MMM \cdot \VEC{\MMM}' \EXTSTAR \MMM \cdot \VEC{\MMM}' \cdot \VEC{\MMM}_{1}
					\\
					\MIMPLIES \
						\MMM \cdot \VEC{\MMM}' \cdot \VEC{\MMM}_{1} \cdot \TCV:(\GAMMA \PLUSG \VEC{\GAMMA}' \PLUSG \VEC{\GAMMA}_{1}) \REMOVETCVfrom
						\models  \ONEEVAL{f}{()}{b}{\FRESH{b}{\TCV}} 
				\end{array}
			}{3cm}{
				$\MMM'^{\GAMMA'} \equiv \MMM^{\GAMMA} \cdot \VEC{\MMM}'^{\VEC{\GAMMA}'}$
				\\
				$\MMM_{1}^{\GAMMA_{1}} \equiv \MMM^{\GAMMA} \cdot \VEC{\MMM}'^{\VEC{\GAMMA}'} \cdot \VEC{\MMM}_{1}^{\VEC{\GAMMA}_1}$
			}
			\NPLINE{\MIMPLIES \ \Mforall \VEC{\MMM}_1^{\VEC{\GAMMA}_1}. 
				\begin{array}[t]{l}
					\MMM \cdot \VEC{\MMM}' \EXTSTAR \MMM \cdot \VEC{\MMM}' \cdot \VEC{\MMM}_{1}
					\
					\MAND \
					\MMM \EXTSTAR \MMM \cdot \VEC{\MMM}_{1}
					\\
					\MIMPLIES \
						\MMM \cdot \VEC{\MMM}' \cdot \VEC{\MMM}_{1} \cdot \TCV:(\GAMMA \PLUSG \VEC{\GAMMA}' \PLUSG \VEC{\GAMMA}_{1}) \REMOVETCVfrom \models  \ONEEVAL{f}{()}{b}{\FRESH{b}{\TCV}} 
				\end{array}
			}{4cm}{
				Subset $\Mforall \VEC{\MMM}_1$
				\\
				only $\VEC{\MMM}_1$ s.t. $\MMM \EXTSTAR \MMM \cdot \VEC{\MMM}_1$
			}
			\NPLINE{\MIMPLIES \ \Mforall \VEC{\MMM}_1^{\VEC{\GAMMA}_1}. 
				\begin{array}[t]{l}
					\MMM \cdot \VEC{\MMM}' \EXTSTAR \MMM \cdot \VEC{\MMM}' \cdot \VEC{\MMM}_{1}
					\
					\MAND \
					\MMM \EXTSTAR \MMM \cdot \VEC{\MMM}_{1}
					\\
					\MIMPLIES \
						\MMM \cdot \VEC{\MMM}_{1} \cdot \TCV:(\GAMMA \PLUSG \VEC{\GAMMA}_{1}) \REMOVETCVfrom
						\ONEEVAL{f}{()}{b}{\FRESH{b}{\TCV}} 
				\end{array}
			}{4cm}{
				See below: Line 16
				\\
				(reducing $\TCV$ holds)
			}
			\NPLINE{\MIMPLIES \ \Mforall \VEC{\MMM}_1^{\VEC{\GAMMA}_1}. 
				\MMM \EXTSTAR \MMM \cdot \VEC{\MMM}_{1}
				\
				\MIMPLIES \
				\MMM \cdot \VEC{\MMM}_{1} \cdot \TCV:(\GAMMA \PLUSG \VEC{\GAMMA}_{1})  \REMOVETCVfrom \models \ONEEVAL{f}{()}{b}{\FRESH{b}{\TCV}}
			}{5cm}{
				Lemma \ref{lem:two_extensions_combine_to_make_extensions_of_each_other}, 
				$\AN{\VEC{\MMM}_1} \cap \AN{\VEC{\MMM}'} \subseteq \AN{\MMM}$
				\\
				guarantees constraints
			}
			\NLINE{\MIMPLIES \ \Mforall \MMM_1^{\GAMMA_1}. 
				\MMM \EXTSTAR \MMM_1
				\
				\MIMPLIES \
				\MMM_1 \cdot \TCV:\GAMMA_1\REMOVETCVfrom   \models \ONEEVAL{f}{()}{b}{\FRESH{b}{\TCV}}
			}{
				$\GAMMA_1 \equiv \GAMMA \PLUSG \VEC{\GAMMA}_1$,
				\
				$\MMM_1 \equiv \MMM\cdot \VEC{\MMM}_1$
			}
			\NLASTLINE{\MIMPLIES \ \MMM  \models \FAD{\TCV} \ONEEVAL{f}{()}{b}{\FRESH{b}{\TCV}}
			}{
				Sem. $\FAD{\TCV}$
			}
			
		\end{NDERIVATION}

	Proof of: 
	\[
	\begin{array}{l}
		\Mforall \MMM^{\GAMMA}, \VEC{\MMM}_1^{\VEC{\GAMMA}_1}, \VEC{\MMM}'^{\VEC{\GAMMA}'}, V_x.
		\\
		\MMM \EXTSTAR \MMM \cdot \VEC{\MMM}' 
		\ \MAND \
		\MMM \cdot \VEC{\MMM}' \EXTSTAR \MMM \cdot \VEC{\MMM}' \cdot \VEC{\MMM}_1
		\ \MAND \
		\MMM \EXTSTAR \MMM \cdot \VEC{\MMM}_1
		\\
		\MAND \
		\MMM \cdot \VEC{\MMM}' \cdot \VEC{\MMM}_{1} \cdot \TCV:(\GAMMA \PLUSG \VEC{\GAMMA}' \PLUSG \VEC{\GAMMA}_{1})\REMOVETCVfrom  \models  \ONEEVAL{f}{()}{b}{\FRESH{b}{\TCV}} 
		\\
		\MIMPLIES
		\\
		\MMM \cdot \VEC{\MMM}_{1} \cdot \TCV:(\GAMMA \PLUSG \VEC{\GAMMA}_{1}) \REMOVETCVfrom \models  \ONEEVAL{f}{()}{b}{\FRESH{b}{\TCV}} 
	\end{array}
	\]
	
	\PROOFFINISHED
	{
		Proof:
		\begin{NDERIVATION}{16}
			\NLINE{
				\MMM  \cdot \VEC{\MMM}' \EXTSTAR \MMM \cdot \VEC{\MMM}' \cdot \VEC{\MMM}_{1}
				\ \MIFF \ 
				\MMM  \cdot \VEC{\MMM}_{1} \EXTSTAR \MMM \cdot \VEC{\MMM}' \cdot \VEC{\MMM}_{1}
			}{Lemma \ref{lem:two_extensions_combine_to_make_extensions_of_each_other}}
			\NLINE{\MIMPLIES \
				\MMM  \cdot \VEC{\MMM}_{1} \cdot b:r_b \EXTSTAR \MMM \cdot \VEC{\MMM}' \cdot \VEC{\MMM}_{1} \cdot b:r_b
			}{Lemma \ref{lem:Gamma_derived_terms_maintain_extension_when_added}}
			\NLINE{\text{
					Let $\MMM_3 \equiv \MMM \cdot \VEC{\MMM}' \cdot \VEC{\MMM}_{1} \cdot \TCV:(\GAMMA \PLUSG \VEC{\GAMMA}' \PLUSG \VEC{\GAMMA}_{1}) \REMOVETCVfrom$
			and  $\MMM_2 \equiv \MMM \cdot \VEC{\MMM}_{1} \cdot \TCV:(\GAMMA \PLUSG \VEC{\GAMMA}_{1}) \REMOVETCVfrom$	
			}}{}
			\NLINE{\text{
					Assume: $\MMM_3  \models  \ONEEVAL{f}{()}{b}{\FRESH{b}{\TCV}} $
			}}{}
			\NLINE{\MIFF \ 
				\LTCDERIVEDVALUE{f()}{\GAMMA_3}{\MMM_3}{r_b} 
				\ \MAND \ 
				\neg \Mexists M_b. \LTCDERIVEDVALUE{M_b}{\TCV}{\MMM_3 \cdot b:r_b}{r_b}
			}{Sem. $\ONEEVAL{}{}{}{}$, $\FRESH{}{}$}
			\NLINE{\MIMPLIES \ 
				\LTCDERIVEDVALUE{f()}{\GAMMA_2}{\MMM_2}{r_b}
				\ \MAND \ 
				\neg \Mexists M_b. \LTCDERIVEDVALUE{M_b}{\TCV}{\MMM_3 \cdot b:r_b}{r_b}
			}{$f \in \SEM{\GAMMA}{\MMM} \subseteq \SEM{\GAMMA_{2}}{\MMM_{2}} \subseteq \SEM{\GAMMA_{3}}{\MMM_{3}}$, Lemmas \ref{lem:eval_under_extensions_are_equivalent}}
			\NLINE{\MIMPLIES \
				\LTCDERIVEDVALUE{f()}{\GAMMA_2}{\MMM_2}{r_b}
				\ \MAND \ 
				\neg \Mexists M_b. \LTCDERIVEDVALUE{M_b}{\GAMMA \PLUSG \VEC{\GAMMA}' \PLUSG \VEC{\GAMMA}_1 }{\MMM \cdot \VEC{\MMM}' \cdot \VEC{\MMM}_{1} \cdot b:r_b}{r_b}
			}{$\LTCDERIVEDVALUE{M}{\TCV}{\MMM \cdot \TCV:\GAMMA \REMOVETCVfrom}{V} \equiv \LTCDERIVEDVALUE{M}{\GAMMA}{\MMM}{V}$}
			\NLINE{\MIMPLIES \
				\LTCDERIVEDVALUE{f()}{\GAMMA_2}{\MMM_2}{r_b}
				\ \MAND \ 
				\neg \Mexists M_b. \LTCDERIVEDVALUE{M_b}{\GAMMA \PLUSG \VEC{\GAMMA}_1 }{\MMM \cdot \VEC{\MMM}' \cdot \VEC{\MMM}_{1} \cdot b:r_b}{r_b}
			}{Subset, $\TCTYPES{\GAMMA \PLUSG \VEC{\GAMMA}' \PLUSG \VEC{\GAMMA}_1}{\GAMMA \PLUSG \VEC{\GAMMA}_1}$}
			\NLINE{\MIMPLIES \
				\LTCDERIVEDVALUE{f()}{\GAMMA_2}{\MMM_2}{r_b}
				\ \MAND \ 
				\neg \Mexists M_b. \LTCDERIVEDVALUE{M_b}{\GAMMA \PLUSG \VEC{\GAMMA}_1 }{\MMM \cdot \VEC{\MMM}_{1} \cdot b:r_b}{r_b}
			}{Lemma \ref{lem:eval_under_extensions_are_equivalent}, line 17}
			\NLINE{\MIMPLIES \ 
				\LTCDERIVEDVALUE{f()}{\GAMMA_2}{\MMM_2}{r_b}
				\ \MAND \ 
				\neg \Mexists M_b. \LTCDERIVEDVALUE{M_b}{\TCV}{\MMM_2 \cdot b:r_b}{r_b}
			}{$\LTCDERIVEDVALUE{M}{\TCV}{\MMM \cdot \TCV:\GAMMA \REMOVETCVfrom \cdot b:r_b }{V} \equiv \LTCDERIVEDVALUE{M}{\GAMMA}{\MMM \cdot b:r_b}{V}$}
			\NLASTLINE{\MIMPLIES \ 
				\MMM_2  \models  \ONEEVAL{f}{()}{b}{\FRESH{b}{\TCV}}
			}{}
		\end{NDERIVATION}
	}
	
	}
\end{lemma}

\begin{lemma}[\DONE The formula used in the reasoning for $\lambda x. \GENSYM()$ is $\EXTINDEP$]
	\label{lem:GS_formula_Plus_FORALL_is_EXTINDEP}
	\[
	\FAD{\TCV} \FORALL{x}{\TCV} \ONEEVAL{f}{x}{b}{\FRESH{b}{\TCV \PLUSV x}}-\EXTINDEP
	\]
	(Similar to proof above)
	\\
	\PROOFFINISHED
	{
		Proof:
		\\
		\[
		\Mforall \MMM^{\GAMMA}, \MMM'. \MMM \EXTSTAR \MMM' \MIMPLIES \MMM \models \FAD{\TCV} \FORALL{x}{\TCV} \ONEEVAL{f}{x}{b}{\FRESH{b}{\TCV \PLUSV x}} \ \MIFF \ \MMM' \models \FAD{\TCV} \FORALL{x}{\TCV} \ONEEVAL{f}{x}{b}{\FRESH{b}{\TCV \PLUSV x}}
		\]
		\\
		{\DONE Extending} ($\MMM \models \FAD{\TCV} \FORALL{x}{\TCV} \ONEEVAL{f}{x}{b}{\FRESH{b}{\TCV \PLUSV x}} \ \MIMPLIES \ \MMM' \models \FAD{\TCV} \FORALL{x}{\TCV} \ONEEVAL{f}{x}{b}{\FRESH{b}{\TCV \PLUSV x}}$):
		\begin{NDERIVATION}{1}
			\NLINE{\text{Assume: $\GAMMA$ s.t. $\FORMULATYPES{\GAMMA}{\FAD{\TCV} \FORALL{x}{\TCV} \ONEEVAL{f}{x}{b}{\FRESH{b}{\TCV \PLUSV x}}}$ }}{}
			\NLINE{\text{Assume some $\MMM^{\GAMMA_d}$, $\MMM'$ s.t. $\CONSTRUCT{\GAMMA}{\MMM}$ and  $\MMM \EXTSTAR\MMM'$ and $\MMM \models \FAD{\TCV} \FORALL{x}{\TCV} \ONEEVAL{f}{x}{b}{\FRESH{b}{\TCV \PLUSV x}}$ }}{}
			\NLINE{\MIFF \ \Mforall \MMM_1^{\GAMMA_1}. \MMM\EXTSTAR \MMM_1 \MIMPLIES \MMM_1 \cdot \TCV:\GAMMA_1 \REMOVETCVfrom \models \FORALL{x}{\TCV} \ONEEVAL{f}{x}{b}{\FRESH{b}{\TCV \PLUSV x}}  }{Sem. $\FAD{\TCV}$}
			\NLINE{\MIMPLIES \ \Mforall \MMM_1^{\GAMMA_1}. \MMM \EXTSTAR \MMM' \EXTSTAR \MMM_1 \MIMPLIES \MMM_1 \cdot \TCV:\GAMMA_1 \REMOVETCVfrom \models \FORALL{x}{\TCV} \ONEEVAL{f}{x}{b}{\FRESH{b}{\TCV \PLUSV x}} }{Subset $\Mforall \MMM_1$}
			\NLINE{\MIMPLIES \ \Mforall \MMM_1^{\GAMMA_1}. \MMM' \EXTSTAR \MMM_1 \MIMPLIES \MMM_1 \cdot \TCV:\GAMMA_1 \REMOVETCVfrom \models  \FORALL{x}{\TCV} \ONEEVAL{f}{x}{b}{\FRESH{b}{\TCV \PLUSV x}}  }{Remove $\MMM\EXTSTAR$}
			\NLASTLINE{\MIMPLIES \ \MMM' \models \FAD{\TCV} \FORALL{x}{\TCV} \ONEEVAL{f}{x}{b}{\FRESH{b}{\TCV \PLUSV x}}  }{}
		\end{NDERIVATION}

		{\DONE  Contracting}  ($\MMM \models \FAD{\TCV} \FORALL{x}{\TCV} \ONEEVAL{f}{x}{b}{\FRESH{b}{\TCV \PLUSV x}} \ \MIMPLIEDBY \ \MMM' \models \FAD{\TCV} \FORALL{x}{\TCV} \ONEEVAL{f}{x}{b}{\FRESH{b}{\TCV \PLUSV x}}$):
		\begin{NDERIVATION}{7}
			\NLINE{\text{Assume: $\GAMMA$, $\MMM^{\GAMMA}$ s.t. $\FORMULATYPES{\GAMMA}{\FAD{\TCV} \FORALL{x}{\TCV} \ONEEVAL{f}{x}{b}{\FRESH{b}{\TCV \PLUSV x}}}$ }}{}
			\NLINE{\text{Assume some $\MMM'$ s.t. $\MMM \EXTSTAR\MMM'$ and $\MMM' \models \FAD{\TCV} \FORALL{x}{\TCV} \ONEEVAL{f}{x}{b}{\FRESH{b}{\TCV \PLUSV x}}$ }}{}
			\NLINE{\text{Let: } \MMM_{1d}' \equiv \MMM_1^{\GAMMA_1} \cdot \TCV:\GAMMA_1\REMOVETCVfrom 
				\qquad 
				(\equiv \MMM \cdot \VEC{\MMM}' \cdot \VEC{\MMM}_1 \cdot \TCV:(\GAMMA \PLUSG \VEC{\GAMMA}' \PLUSG \VEC{\GAMMA}_1)\REMOVETCVfrom)
				\text{ and } \MMM_{1d}' 
				
			}{}
			\NLINE{\MIFF \ \Mforall \MMM_1^{\GAMMA_1}.  \
					\MMM' \EXTSTAR \MMM_1 
						\
						\MIMPLIES \
						\Mforall M_x. 
							\LTCDERIVEDVALUE{M_x}{\TCV}{\MMM_{1d}'}{V_x}
							\
							\MIMPLIES \
							\MMM_{1d}' \cdot x:V_x \models \ONEEVAL{f}{x}{b}{\FRESH{b}{\TCV \PLUSV x}}
			}{Sem. $\FAD{\TCV}$, $\FORALL{x}{\TCV}$}
			\NPLINE{\MIFF \ \Mforall \VEC{\MMM}_1^{\VEC{\GAMMA}_1}.
				\begin{array}[t]{l}
					\MMM \cdot \VEC{\MMM}' \EXTSTAR \MMM \cdot \VEC{\MMM}' \cdot \VEC{\MMM}_1 
						\\
						\MIMPLIES \
						\Mforall M_x. 
							\LTCDERIVEDVALUE{M_x}{\TCV}{\MMM_{1d}'}{V_x}
							\
							\MIMPLIES \
							\MMM_{1d}'  \cdot x:V_x  \models \ONEEVAL{f}{x}{b}{\FRESH{b}{\TCV \PLUSV x}}
					\end{array}
			}{3cm}{
				$\MMM'^{\GAMMA'} \equiv \MMM^{\GAMMA} \cdot \VEC{\MMM}'^{\VEC{\GAMMA}'}$
				\\
				$\MMM_{1}^{\GAMMA_{1}} \equiv \MMM^{\GAMMA} \cdot \VEC{\MMM}'^{\VEC{\GAMMA}'} \cdot \VEC{\MMM}_{1}^{\VEC{\GAMMA}_1}$
			}
			\NPLINE{\MIMPLIES \ \Mforall \VEC{\MMM}_1^{\VEC{\GAMMA}_1}.
				\begin{array}[t]{l}
					\MMM \cdot \VEC{\MMM}' \EXTSTAR \MMM \cdot \VEC{\MMM}' \cdot \VEC{\MMM}_1 
					\ \MAND \
					\MMM \EXTSTAR \MMM \cdot \VEC{\MMM}_{1}
					\\
						\MIMPLIES \
						\Mforall M_x. \
							\LTCDERIVEDVALUE{M_x}{\TCV}{\MMM_{1d}'}{V_x}
							\
							\MIMPLIES \
							\MMM_{1d}'  \cdot x:V_x \models \ONEEVAL{f}{x}{b}{\FRESH{b}{\TCV \PLUSV x}}
				\end{array}
				\hspace{-1cm}
			}{3.5cm}{
				Subset $\Mforall \VEC{\MMM}_1$
				\\
				only $\VEC{\MMM}_1$ s.t. $\MMM \EXTSTAR \MMM \cdot \VEC{\MMM}_1$
			}
			\NPLINE{\MIMPLIES \ \Mforall \VEC{\MMM}_1^{\VEC{\GAMMA}_1}.
				\begin{array}[t]{l}
					\MMM \cdot \VEC{\MMM}' \EXTSTAR \MMM \cdot \VEC{\MMM}' \cdot \VEC{\MMM}_1 
					\ \MAND \
					\MMM \EXTSTAR \MMM \cdot \VEC{\MMM}_{1}
						\\
						\MIMPLIES \
						\Mforall M_x. 
							\LTCDERIVEDVALUE{M_x}{\GAMMA \PLUSG \VEC{\GAMMA}' \PLUSG \VEC{\GAMMA}_1}{\MMM_{1d}' \REMOVEVARIABLE \TCV}{V_x}
							\
							\MIMPLIES \ \MMM_{1d}'  \cdot x:V_x  \models \ONEEVAL{f}{x}{b}{\FRESH{b}{\TCV \PLUSV x}}
				\end{array} 
			}{3cm}{
				$\SEM{\TCV}{\MMM_{1d}} \equiv \GAMMA \PLUSG \VEC{\GAMMA}' \PLUSG \VEC{\GAMMA}_1$
				\\
				Lemma \ref{lem:LTC_derived_values_unaffected_by_TCV_addition/removal}
			}
			\NLINE{\text{Let: }
				\MMM_{2} \equiv \MMM \cdot \VEC{\MMM}_{1} \cdot \TCV:(\GAMMA \PLUSG \VEC{\GAMMA}_{1}) \REMOVETCVfrom	
			}{}
			\NPLINE{\MIMPLIES \ \Mforall \VEC{\MMM}_1^{\VEC{\GAMMA}_1}.
				\begin{array}[t]{l}
					\MMM \cdot \VEC{\MMM}' \EXTSTAR \MMM \cdot \VEC{\MMM}' \cdot \VEC{\MMM}_1 
					\ \MAND \
					\MMM \EXTSTAR \MMM \cdot \VEC{\MMM}_{1}
					\\
					(\MIMPLIES \
					\Mforall M_x. 
					\LTCDERIVEDVALUE{M_x}{\GAMMA \PLUSG \VEC{\GAMMA}_{1}}{\MMM \cdot \VEC{\MMM}_{1}}{V_x}
					\
					\MIMPLIES \ \MMM_{1d}'  \cdot x:V_x  \models \ONEEVAL{f}{x}{b}{\FRESH{b}{\TCV \PLUSV x}})
					\\
					\MIMPLIES \
					\Mforall M_x. 
					\LTCDERIVEDVALUE{M_x}{\TCV}{\MMM_2}{V_x}
					\
					\MIMPLIES \ \MMM_{1d}'  \cdot x:V_x  \models \ONEEVAL{f}{x}{b}{\FRESH{b}{\TCV \PLUSV x}}
				\end{array} 
			}{3cm}{
				Subset $\Mforall M_x$
				\\
				$\MMM_2 \REMOVEVARIABLE \TCV \subseteq \MMM_{1d} \REMOVEVARIABLE \TCV $
			}
			\NPLINE{\MIMPLIES \ \Mforall \VEC{\MMM}_1^{\VEC{\GAMMA}_1}.
				\begin{array}[t]{l}
					\MMM \cdot \VEC{\MMM}' \EXTSTAR \MMM \cdot \VEC{\MMM}' \cdot \VEC{\MMM}_1 
					\ \MAND \
					\MMM \EXTSTAR \MMM \cdot \VEC{\MMM}_{1}
					\\
					\MIMPLIES \
					\Mforall M_x. 
					\LTCDERIVEDVALUE{M_x}{\TCV}{\MMM_2}{V_x}
					\
					\MIMPLIES \ \MMM_2  \cdot x:V_x  \models \ONEEVAL{f}{x}{b}{\FRESH{b}{\TCV \PLUSV x}}
				\end{array} 
			}{3cm}{
				See below, Line 19
			}
			\NPLINE{\MIMPLIES \ \Mforall \VEC{\MMM}_1^{\VEC{\GAMMA}_1}. 
				\begin{array}[t]{l}
				\MMM \EXTSTAR \MMM \cdot \VEC{\MMM}_{1}
				\\
				\MIMPLIES \
				\MMM \cdot \VEC{\MMM}_{1} \cdot \TCV:(\GAMMA \PLUSG \VEC{\GAMMA}_{1})\REMOVETCVfrom   \models \FORALL{x}{\TCV} \ONEEVAL{f}{x}{b}{\FRESH{b}{\TCV \PLUSV x}}
				\end{array}
			}{5cm}{
				Lemma \ref{lem:two_extensions_combine_to_make_extensions_of_each_other}, 
				$\AN{\VEC{\MMM}_1} \cap \AN{\VEC{\MMM}'} \subseteq \AN{\MMM}$
				\\
				guarantees constraints
			}
			\NLINE{\MIMPLIES \ \Mforall \MMM_1^{\GAMMA_1}. 
				\MMM \EXTSTAR \MMM_1
				\
				\MIMPLIES \
				\MMM_1 \cdot \TCV:\GAMMA_1\REMOVETCVfrom   \models \ONEEVAL{f}{x}{b}{\FRESH{b}{\TCV \PLUSV x}}
			}{
				$\MMM_1^{\GAMMA_1} \equiv \MMM^{\GAMMA} \cdot \VEC{\MMM}_1^{\VEC{\GAMMA}_1}$
			}
			\NLASTLINE{\MIMPLIES \ \MMM  \models \FAD{\TCV} \ONEEVAL{f}{x}{b}{\FRESH{b}{\TCV \PLUSV x}}
			}{
				Sem. $\FAD{\TCV}$
			}
			
		\end{NDERIVATION}
	
	Proof of
	\[
	\Mforall \VEC{\MMM}_1^{\VEC{\GAMMA}_1}.
	\begin{array}[t]{l}
		\MMM \cdot \VEC{\MMM}' \EXTSTAR \MMM \cdot \VEC{\MMM}' \cdot \VEC{\MMM}_1 
		\ \MAND \
		\MMM \EXTSTAR \MMM \cdot \VEC{\MMM}_{1}
		\\
		\MIMPLIES 
		\begin{array}[t]{l}
			\MMM_{1d} \equiv \MMM \cdot \VEC{\MMM}' \cdot \VEC{\MMM}_1 \cdot \TCV:(\GAMMA \PLUSG \VEC{\GAMMA}' \PLUSG \VEC{\GAMMA}_1)\REMOVETCVfrom
			\\
			\MAND \
			\MMM_{2} \equiv \MMM \cdot \VEC{\MMM}_{1} \cdot \TCV:(\GAMMA \PLUSG \VEC{\GAMMA}_{1}) \REMOVETCVfrom
			\\
			\MAND \
			\Mforall M_x. 
			\begin{array}[t]{l}
				\LTCDERIVEDVALUE{M_x}{\GAMMA \PLUSG \VEC{\GAMMA}_{1}}{\MMM \cdot \VEC{\MMM}_{1}}{V_x}  
				\\
				\MIMPLIES 
				\\
				(\MMM_{1d}  \cdot x:V_x  \models \ONEEVAL{f}{x}{b}{\FRESH{b}{\TCV \PLUSV x}}
				\ \MIMPLIES \
				\MMM_2  \cdot x:V_x  \models \ONEEVAL{f}{x}{b}{\FRESH{b}{\TCV \PLUSV x}})
			\end{array}
		\end{array}
	\end{array}
	\]
	\PROOFFINISHED
	{
		\begin{NDERIVATION}{19}
			\NLINE{
				\MMM  \cdot \VEC{\MMM}' \EXTSTAR \MMM \cdot \VEC{\MMM}' \cdot \VEC{\MMM}_{1}
				\ \MIFF \ 
				\MMM  \cdot \VEC{\MMM}_{1} \EXTSTAR \MMM \cdot \VEC{\MMM}' \cdot \VEC{\MMM}_{1}
			}{Lemma \ref{lem:two_extensions_combine_to_make_extensions_of_each_other}}
			\NLINE{\MIMPLIES \
				\MMM  \cdot \VEC{\MMM}_{1} \cdot b:r_b \EXTSTAR \MMM \cdot \VEC{\MMM}' \cdot \VEC{\MMM}_{1} \cdot b:r_b
			}{Lemma \ref{lem:Gamma_derived_terms_maintain_extension_when_added}}
			\NLINE{\text{
					Let $\MMM_{3x}^{\GAMMA_{3x}} \equiv \MMM \cdot \VEC{\MMM}' \cdot \VEC{\MMM}_{1} \cdot \TCV:(\GAMMA \PLUSG \VEC{\GAMMA}' \PLUSG \VEC{\GAMMA}_{1})\REMOVETCVfrom \cdot x:V_x$ 
				}}{}
			\NLINE{\text{
					Let 
					$\MMM_{2x}^{\GAMMA_{2x}} \equiv \MMM^{\GAMMA} \cdot \VEC{\MMM}_{1} \cdot \TCV:(\GAMMA \PLUSG \VEC{\GAMMA}_{1}) \REMOVETCVfrom \cdot  x:V_x$	
			}}{}
			\NLINE{\text{
					Assume: $\MMM \EXTSTAR \MMM \cdot \VEC{\MMM}' 
					\ \MAND \
					\MMM \cdot \VEC{\MMM}' \EXTSTAR \MMM \cdot \VEC{\MMM}' \cdot \VEC{\MMM}_1
					\ \MAND \
					\MMM \EXTSTAR \MMM \cdot \VEC{\MMM}_1$
			}}{}
			\NLINE{\text{
					Assume: $\MMM_{3x}  \models  \ONEEVAL{f}{x}{b}{\FRESH{b}{\TCV \PLUSV x}} $
			}}{}
			\NLINE{\MIFF \
				\LTCDERIVEDVALUE{fx}{\GAMMA_{3x}}{\MMM_{3x}}{r_b} 
				\ \MAND \ 
				\neg \Mexists M_b. \LTCDERIVEDVALUE{M_b}{\TCV \PLUSV x}{\MMM_{3x} \cdot b:r_b}{r_b}
			}{Sem. $\ONEEVAL{}{}{}{}$, $\FRESH{}{}$}
			\NLINE{\MIFF \
				\LTCDERIVEDVALUE{fx}{\GAMMA_{2x}}{\MMM_{2x}}{r_b} 
				\ \MAND \ 
				\neg \Mexists M_b. \LTCDERIVEDVALUE{M_b}{\TCV \PLUSV x}{\MMM_{3x} \cdot b:r_b}{r_b}
			}{$f,x \in \SEM{\GAMMA}{\MMM} \subseteq \SEM{\GAMMA_{2x}}{\MMM_{2x}} \subseteq \SEM{\GAMMA_{3x}}{\MMM_{3x}}$, Lemma \ref{lem:eval_under_extensions_are_equivalent}}
			\NLINE{\MIFF \
				\LTCDERIVEDVALUE{fx}{\GAMMA_{2x}}{\MMM_{2x}}{r_b} 
				\ \MAND \ 
				\neg \Mexists M_b. \LTCDERIVEDVALUE{M_b}{\GAMMA_{3x}}{(\MMM_{3x} \REMOVEVARIABLE\TCV) \cdot b:r_b}{r_b}
			}{$\SEM{\TCV}{\MMM_{3x}} = \GAMMA_{3x}$, Lemma \ref{lem:LTC_derived_values_unaffected_by_TCV_addition/removal}}
			\NLINE{\MIMPLIES \
				\LTCDERIVEDVALUE{fx}{\GAMMA_{2x}}{\MMM_{2x}}{r_b} 
				\ \MAND \ 
				\neg \Mexists M_b. \LTCDERIVEDVALUE{M_b}{\GAMMA_{2x}}{(\MMM_{3x} \REMOVEVARIABLE \TCV) \cdot b:r_b}{r_b}
			}{Subset, $\TCTYPES{\GAMMA \PLUSG \VEC{\GAMMA}' \PLUSG \VEC{\GAMMA}_1}{\GAMMA \PLUSG \VEC{\GAMMA}_1}$}
			\NLINE{\MIMPLIES \
				\LTCDERIVEDVALUE{fx}{\GAMMA_{2x}}{\MMM_{2x}}{r_b} 
				\ \MAND \ 
				\neg \Mexists M_b. \LTCDERIVEDVALUE{M_b}{\GAMMA_{2x}}{(\MMM_{2x} \REMOVEVARIABLE \TCV) \cdot b:r_b}|{r_b}
			}{Lemma \ref{lem:extensions_cannot_reveal_old_names}, line 20}
			\NLINE{\MIMPLIES \ 
				\LTCDERIVEDVALUE{fx}{\GAMMA_{2x}}{\MMM_{2x}}{r_b} 
				\ \MAND \ 
				\neg \Mexists M_b. \LTCDERIVEDVALUE{M_b}{\TCV \PLUSV x}{\MMM_{2x} \cdot b:r_b}{r_b}
			}{$\SEM{\TCV}{\MMM_{2x}} = \GAMMA_{2x}$, Lemma \ref{lem:LTC_derived_values_unaffected_by_TCV_addition/removal}}
			\NLASTLINE{\MIMPLIES \ 
				\MMM_{2x}  \models  \ONEEVAL{f}{x}{b}{\FRESH{b}{\TCV \PLUSV x}}
			}{}
		\end{NDERIVATION}
	}
	}
\end{lemma}

\label{sec:apndx_formulae_properties}

\begin{lemma}[\DONE  \SYNEXTINDEP \ implies \EXTINDEP ]
	\label{lem:apndx_SIMPLEEXTINDEP_implies_EXTINDEP}
	\[
	A \text{ is \SYNEXTINDEP} \ \MIMPLIES \ A \text{ is \EXTINDEP}
	\]
	
	\begin{proof}
		i.e.
		$
		\begin{array}[t]{l}
			A -\text{\SYNEXTINDEP}
			\\
			\MIMPLIES
			\
			\Mforall \GAMMA, \MMM^{\GAMMA}, \GAMMA', \MMM'^{\GAMMA'}. \FORMULATYPES{\GAMMA}{A} \ \MAND \ \MMM \EXTSTAR \MMM' \MIMPLIES (\MMM \models A \ \MIFF \ \MMM' \models A)
		\end{array}
		$
		\\
		Proof by induction on the structure of $A$, using Lemma \ref{lem:extensions_give_same_semantics_for_type_contexts} (i.e. $\GAMMA_0 \subseteq \GAMMA \ \MIMPLIES \ \SEM{\GAMMA_0}{\MMM}\equiv \SEM{\GAMMA_0}{\MMM'}$) and Lemma \ref{lem:semantics_expressions_equal_under_model_extensions}  (i.e. $\SEM{e}{\MMM} \equiv \SEM{e}{\MMM'}$)
		\begin{enumerate}
			\item Base Formulas:
			\begin{itemize}
				\item[$A\equiv \TRUTH, \FALSITY$] Clearly hold.
				\item[$A\equiv \FRESH{x}{\GAMMA_0}$] 
				$\FORMULATYPES{\GAMMA}{\FRESH{x}{\GAMMA_0}}$ 
				implies 
				$\EXPRESSIONTYPES{\GAMMA}{x}{\NAME}$ and $\TCTYPES{\GAMMA}{\GAMMA_0}$.
				\\
				Lemma  \ref{lem:semantics_expressions_equal_under_model_extensions}  implies $\SEM{x}{\MMM}\equiv \SEM{x}{\MMM'}$.
				\\ 
				Lemma \ref{lem:extensions_give_same_semantics_for_type_contexts} implies $\SEM{\GAMMA_0}{\MMM} \equiv \SEM{\GAMMA_0}{\MMM'}$.
				\\
				Lemma \ref{lem:eval_under_extensions_are_equivalent} implies
				$\Mexists M. \ \LTCDERIVEDVALUE{M}{\GAMMA_0}{\MMM}{ \SEM{x}{\MMM}}
				\ \MIFF \
				\Mexists M. \ \LTCDERIVEDVALUE{M}{\GAMMA_0}{\MMM'}{ \SEM{x}{\MMM'}}$
				
			\end{itemize}
			\ \\
			\item Core Inductive Cases:
			\begin{itemize}
				\item[$A\equiv \neg C_1$] Holds from IH on $C_1$ as $\MMM \models C_1 \ \MIFF \ \MMM' \models C_1$ 
				\\
				hence $\MMM \not\models C_1 \ \MIFF \ \MMM' \not\models C_1$,
				hence $\MMM \models \neg C_1 \ \MIFF \ \MMM' \models \neg C_1$.
				
				\item[$A\equiv C_1 \PAND C_2$] 
				Holds from IH on $C_1$ and $C_2$, as $\MMM \models C_i \ \MIFF \ \MMM' \models C_i$ implies
				\\
				$\MMM \models C_1 \PAND C_2 \ \MIFF \ \MMM \models C_1 \MAND \MMM \models C_2 \ \MIFF \ \MMM' \models C_1 \MAND \MMM' \models C_2 \ \MIFF \ \MMM' \models C_1 \PAND C_2$.
				
				\item[$A\equiv C_1 \POR C_2$] Derivable given $A \POR B \equiv \neg (\neg A \PAND \neg B)$
				\item[$A\equiv C_1 \PIMPLIES C_2$] Derivable given $A \PIMPLIES B \equiv \neg (A \PAND \neg B)$
				\item[$A\equiv \ONEEVAL{u}{e}{m}{C_1}$] 
				Lemma \ref{lem:eval_under_extensions_are_equivalent} implies $\LTCDERIVEDVALUE{ue}{\GAMMA}{\MMM}{V_m} \ \MIFF \ \LTCDERIVEDVALUE{ue}{\GAMMA}{\MMM'}{V_m}$.
				\\
				Lemma \ref{lem:Gamma_derived_terms_maintain_extension_when_added} implies $ \MMM \EXTSTAR \MMM' \MAND \LTCDERIVEDVALUE{ue}{\GAMMA}{\MMM}{V_m} \MIMPLIES \MMM \cdot m:V_m \EXTSTAR \MMM' \cdot m:V_m $.
				\\
				Induction on $C_1$ implies $\MMM \cdot m:V_m \models C_1 \ \MIFF \ \MMM' \cdot m:V_m \models C_1$.
				\\
				Hence
				$
				\begin{array}[t]{l}
					\MMM \models \ONEEVAL{u}{e}{m}{C_1} 
					\\ \MIFF \ 
					\LTCDERIVEDVALUE{ue}{\GAMMA}{\MMM}{V_m} \ \MAND \ \MMM \cdot m : V_m \models C_1
					\\ \MIFF \ 
					\LTCDERIVEDVALUE{ue}{\GAMMA}{\MMM'}{V_m} \ \MAND \ \MMM' \cdot m : V_m \models C_1 
					\\ \MIFF \ 
					\MMM' \models \ONEEVAL{u}{e}{m}{C_1} 
				\end{array}
				$.

				\item[$A\equiv \FORALL{x}{\GAMMA_0} C_1$] 
				Lemma \ref{lem:extensions_give_same_semantics_for_type_contexts} implies $\SEM{\GAMMA_0}{\MMM} \equiv \SEM{\GAMMA_0}{\MMM'}$
				\\
				Lemma \ref{lem:eval_under_extensions_are_equivalent} implies
				$\Mexists M. \ \LTCDERIVEDVALUE{M}{\GAMMA_0}{\MMM}{ \SEM{x}{\MMM}}
				\ \MIFF \
				\Mexists M. \ \LTCDERIVEDVALUE{M}{\GAMMA_0}{\MMM'}{ \SEM{x}{\MMM'}}$
				\\
				Lemma \ref{lem:Gamma_derived_terms_maintain_extension_when_added} implies $ \MMM \EXTSTAR \MMM' \MAND \LTCDERIVEDVALUE{M_x}{\GAMMA_0}{\MMM}{V_x} \MIMPLIES \MMM \cdot x:V_x \EXTSTAR \MMM' \cdot x:V_x $ 
				\\
				Induction on $C_1$ implies $\MMM \cdot x:V_x \models C_1 \ \MIFF \ \MMM' \cdot x:V_x \models C_1$
				\\
				hence 
				$
				\begin{array}[t]{lll}
					\MMM \models \FORALL{x}{\GAMMA_0} C_1 
					& \MIFF \ 
					\Mforall M.
					\LTCDERIVEDVALUE{M}{\GAMMA_0}{\MMM}{W}
					\
					\MIMPLIES \ \MMM \cdot x:W \models C_1
					&
					\text{Sem. $\FORALL{x}{\GAMMA_0}$ }
					\\ 
					&
					\MIFF \ 
					\Mforall M.
					\LTCDERIVEDVALUE{M}{\GAMMA_0}{\MMM'}{W}
					\
					\MIMPLIES \ \MMM \cdot x:W \models C_1
					&
					\begin{array}{l}
						\SEM{\GAMMA_0}{\MMM} \equiv \SEM{\GAMMA_0}{\MMM'} 
						\\
						\text{Lemma \ref{lem:eval_under_extensions_are_equivalent}}
					\end{array}
					\\ 
					&
					\MIFF \ 
					\Mforall M.
					\LTCDERIVEDVALUE{M}{\GAMMA_0}{\MMM'}{W}
					\
					\MIMPLIES \ \MMM' \cdot x:W \models C_1
					&
					\text{IH on $C_1$}
					\\ 
					& 
					\MIFF \ 
					\MMM' \models \FORALL{x}{\GAMMA_0} C_1
					& \text{Sem. $\FORALL{x}{\GAMMA_0}$ }
				\end{array}
				$
				
				\item[$A\equiv \EXISTS{x}{\GAMMA_0} C_1$] 
				Derivable from $\EXISTS{x}{\GAMMA_0} C_1 \equiv \neg \FORALL{x}{\GAMMA_0} \neg C_1$
			\end{itemize}
			\ \\
			\item $\FAD{\TCV}$-Inductive Cases:
			\begin{itemize}	
				\item[$A\equiv \FAD{\TCV} C_1 $] holds by IH on $C_1$, Knowing that $C_1$ is $\TCV$-free then:
				\\
				$\begin{array}{lll}
					\MMM \models \FAD{\TCV} C_1 
					& 
					\MIFF \ \Mforall \MMM_1^{\GAMMA_1}. \MMM \EXTSTAR \MMM_1 \MIMPLIES \MMM_1 \cdot \TCV:\GAMMA_1 \models C_1
					&
					\text{Sem. $\FAD{\TCV}$ }
					\\
					& 
					\MIMPLIES \ \Mforall \MMM_1^{\GAMMA_1}. \MMM \EXTSTAR \MMM' \EXTSTAR \MMM_1 \MIMPLIES \MMM_1 \cdot \TCV:\GAMMA_1   \models C_1
					&
					\text{Subset $\Mforall$}
					\\
					& 
					\MIMPLIES \ \Mforall \MMM_1^{\GAMMA_1}. \MMM'\EXTSTAR \MMM_1 \MIMPLIES \MMM_1 \cdot \TCV:\GAMMA_1   \models C_1
					&
					\MMM \EXTSTAR \MMM' \text{ clearly holds}
					\\
					&
					\MIMPLIES \ \MMM' \models \FAD{\TCV} C_1 
					&
					\text{Sem. $\FAD{\TCV}$ }
					\\\\
					\MMM' \models \FAD{\TCV} C_1 
					& 
					\MIFF \ \Mforall \MMM_1^{\GAMMA_1}. \MMM' \EXTSTAR \MMM_1 \MIMPLIES \MMM_1 \cdot \TCV:\GAMMA_1 \models C_1
					&
					\text{Sem. $\FAD{\TCV}$ }
					\\
					& 
					\MIMPLIES \ \Mforall \MMM_1^{\GAMMA_1}. \MMM' \EXTSTAR \MMM_1 \MIMPLIES \MMM_1  \models C_1
					&
					\text{Lemma \ref{lem:model_and_model_plus_TCV_models_equivalently_TCV-free_formula}  $C_1^{-\TCV}$ }
					\\
					& 
					\MIMPLIES \ \MMM'  \models C_1
					&
					\Mforall \MMM_1 \rightarrow \MMM'
					\\
					& 
					\MIMPLIES \ \MMM  \models C_1
					&
					\text{IH}
					\\
					& 
					\MIMPLIES \ \Mforall \MMM_2^{\GAMMA_2}. \MMM \EXTSTAR \MMM_2 \ \MIMPLIES \ \MMM_2  \models C_1
					&
					\text{IH}
					\\
					& 
					\MIMPLIES \ \Mforall \MMM_2^{\GAMMA_2}. \MMM \EXTSTAR \MMM_2 \ \MIMPLIES \ \MMM_2 \cdot \TCV:\GAMMA_2  \models C_1
					&
					\text{Lemma \ref{lem:model_and_model_plus_TCV_models_equivalently_TCV-free_formula}  $C_1^{-\TCV}$ }
					\\
					&
					\MIMPLIES \ \MMM \models \FAD{\TCV} C_1 
					&
					\text{Sem. $\FAD{\TCV}$ }
				\end{array}$
				
				\item[$A\equiv \FAD{\TCV} \FORALL{x}{\TCV} C_1 $] holds by IH on $C_1^{-\TCV}$, 
				Lemma \ref{lem:EXTINDEP-CONSTR-FADForall}
			\end{itemize}
			\ \\
			\item Two specific cases:
			\begin{itemize}
				
				\item[$\FAD{\TCV} \ONEEVAL{f}{()}{b}{\FRESH{b}{\TCV}}$]
				Lemma \ref{lem:GS_formula_is_EXTINDEP}
				
				\item[$\FAD{\TCV} \FORALL{x}{\TCV} \ONEEVAL{f}{x}{b}{\FRESH{b}{\TCV \PLUSV x}}$]
				Lemma \ref{lem:GS_formula_Plus_FORALL_is_EXTINDEP}
			\end{itemize}
		\end{enumerate}
	\end{proof}
\end{lemma}

%% file: appendix/syntactic_EXTINDEP.tex
A subset of all possible $\EXTINDEP$ formulae are defined as $\SYNEXTINDEP$ below and are proven as $\EXTINDEP$ in the lemma that follows.

\begin{definition}[A Syntactic classification of \EXTINDEP \ formulae, written \SYNEXTINDEP ]
	\label{def:syntactic-EXTINDEP}
	Define \SYNEXTINDEP \ inductively as follows:
	\begin{enumerate}
		\item  $\TRUTH$,  $\FALSITY$, $e=e'$ and $\FRESH{x}{\GAMMA'}$ are all \SYNEXTINDEP.
		
		\item 
		If $C_1$ and $C_2$ are  \EXTINDEP \ then 
		\\
		$\neg C_1$, $C_1 \PAND C_2$, $C_1 \POR C_2$, $C_1 \PIMPLIES C_2$,  $\ONEEVAL{u}{e}{m^{\alpha}}{C_1}$, $\EXISTS{x}{\GAMMA'} C_1$, $\FORALL{x}{\GAMMA'} C_1$  are all \SYNEXTINDEP.
		
		\item  
		If $C_1$ is \EXTINDEP \ and contains no reference to $\TCV$
		then $ \FAD{\TCV}  C_1$  
		and $\FAD{\TCV} \FORALL{x}{\TCV} C_1$ are \SYNEXTINDEP
		
		\item Specific cases: 
		$\FAD{\TCV} \ONEEVAL{f}{()}{b}{\FRESH{b}{\TCV}}$ 
		and $\FAD{\TCV} \FORALL{x}{\TCV} \ONEEVAL{f}{x}{b}{\FRESH{b}{\TCV \PLUSV x}}$ are \SYNEXTINDEP
	\end{enumerate}
	
	This is an incomplete characterisation of all \EXTINDEP \ formulae, but covers all cases required for the proofs.
\end{definition}

%% file: appendix/appendix_Thinness.tex
\section{Syntactic Thinness Construction Lemmas}
\label{appendix_thinness}
\label{lem:apndx_SYNThinness_implies_Thinness}
\input{appendix/syntactic_thinness}

We now define some lemmas that help prove that syntactic thinness implies thinness.

\begin{lemma}[\DONE Syntactically thin implies thin for Freshness]
	\label{lem:syntactic_thin_implies_thin_Fresh-typed}
	\[
	\FORMULATYPES{(\GAMMA \PLUSV y \PLUSG \GAMMA')\REMOVEVARIABLE y}{\FRESH{x}{\GAMMA_0}} \ \MAND \ \TCTYPES{\GAMMA}{\GAMMA_0} \ \MIMPLIES \FRESH{x}{\GAMMA_0} \THINWRT{y}
	\]
	\PROOFFINISHED
	{
		Proof:
		\begin{NDERIVATION}{1}
			\NLINE{\text{Assume: } \FORMULATYPES{(\GAMMA \PLUSV y \PLUSG \GAMMA')\REMOVEVARIABLE y}{\FRESH{x}{\GAMMA_0}} \ \MAND \ \TCTYPES{\GAMMA}{\GAMMA_0} }{}
			\NLINE{\MIMPLIES \ \FORMULATYPES{\GAMMA \PLUSG \GAMMA'}{\FRESH{x}{\GAMMA_0}} \ \MAND \ \TCTYPES{\GAMMA}{\GAMMA_0} }{}
			\NLINE{\text{Assume: $\MMM^{\GAMMA \PLUSV y \PLUSG \GAMMA'}$ such that $\MMM \models \FRESH{x}{\GAMMA_0}$ } }{}
			\NLINE{\MIMPLIES \  \neg \Mexists M_x. \LTCDERIVEDVALUE{M_x}{\GAMMA_0}{\MMM}{\SEM{x}{\MMM}} }{}
			\NLINE{\MIMPLIES \  \neg \Mexists M_x. \LTCDERIVEDVALUE{M_x}{\GAMMA_0}{\MMM}{\SEM{x}{\MMM \REMOVEVARIABLE y}} }{$\SEM{x}{\MMM} \equiv \SEM{x}{\MMM \REMOVEVARIABLE y}$}
			\NLINE{\MIMPLIES \  \neg \Mexists M_x.
				\begin{array}[t]{l}
					\TYPES{\SEM{\GAMMA_0}{\MMM}}{M_x}{\NAME}
					\
					\MAND \ \AN{M_x} = \emptyset
					\\
					\MIMPLIES \ (\AN{\MMM}, \ M_x \MMM) \CONV (\AN{\MMM}, G', \  \SEM{x}{\MMM\REMOVEVARIABLE y} ) 
				\end{array}
			}{Sem. $\LTCDERIVEDVALUE{}{}{}{}$}
			\NLINE{\MIMPLIES \  \neg \Mexists M_x.
				\begin{array}[t]{l}
					\TYPES{\SEM{\GAMMA_0}{\MMM\REMOVEVARIABLE y}}{M_x}{\NAME}
					\
					\MAND \ \AN{M_x} = \emptyset
					\\
					\MIMPLIES \ (\AN{\MMM}, \ M_x \MMM) \CONV (\AN{\MMM}, G', \  \SEM{x}{\MMM\REMOVEVARIABLE y} ) 
				\end{array}
			}{\RBOX{$\TCTYPES{\GAMMA}{\GAMMA_0}$ \\ $ \MIMPLIES \ \SEM{\GAMMA_0}{\MMM} \equiv \SEM{\GAMMA_0}{\MMM\REMOVEVARIABLE y}$}}
			\NPLINE{\MIMPLIES \  \neg \Mexists M_x.
				\begin{array}[t]{l}
					\TYPES{\SEM{\GAMMA_0}{\MMM\REMOVEVARIABLE y}}{M_x}{\NAME}
					\
					\MAND \ \AN{M_x} = \emptyset
					\\
					\MIMPLIES \ (\AN{\MMM}, \ M_x (\MMM\REMOVEVARIABLE y)) \CONV (\AN{\MMM}, G', \  \SEM{x}{\MMM\REMOVEVARIABLE y} ) 
				\end{array}
			}{5cm}{$\TCTYPES{\GAMMA}{\GAMMA_0} \ \MAND \ \TYPES{\SEM{\GAMMA_0}{\MMM\REMOVEVARIABLE y}}{M_x}{\NAME}$ \\ $ \MIMPLIES \ M \MMM \equiv M(\MMM\REMOVEVARIABLE y)$}
			\NLINE{\MIMPLIES \  \neg \Mexists M_x.
				\begin{array}[t]{l}
					\TYPES{\SEM{\GAMMA_0}{\MMM\REMOVEVARIABLE y}}{M_x}{\NAME}
					\
					\MAND \ \AN{M_x} = \emptyset
					\\
					\MIMPLIES \ (\AN{\MMM\REMOVEVARIABLE y} \cup \AN{\MMM(y)}, \ M_x (\MMM\REMOVEVARIABLE y)) \CONV (\AN{\MMM}, G', \  \SEM{x}{\MMM\REMOVEVARIABLE y} ) 
				\end{array}
			}{$\AN{\MMM} \equiv \AN{\MMM\REMOVEVARIABLE y} \cup \AN{\MMM(y)}$}
			\NLINE{\MIMPLIES \  \neg \Mexists M_x.
				\begin{array}[t]{l}
					\TYPES{\SEM{\GAMMA_0}{\MMM\REMOVEVARIABLE y}}{M_x}{\NAME}
					\
					\MAND \ \AN{M_x} = \emptyset
					\\
					\MIMPLIES \ (\AN{\MMM \REMOVEVARIABLE y}, \ M_x (\MMM\REMOVEVARIABLE y)) \CONV (\AN{\MMM \REMOVEVARIABLE y}, G', \  \SEM{x}{\MMM\REMOVEVARIABLE y} ) 
				\end{array}
			}{$y \notin \FV{M_x}$, Lemma \ref{lem:adding/remove_unused_names_maintains_evaluation}}
			\NLASTLINE{\MIMPLIES \ \MMM \REMOVEVARIABLE y \models \FRESH{y}{\GAMMA_0} }{Sem. $\FRESH{}{}$}
			
		\end{NDERIVATION}
	}
\end{lemma}

\begin{lemma}[\DONE Syntactically Thin formulae implies Thin]
	\label{lem:Syntactically_Thin_implies_Thin}
	Define a method to syntactically define thin formulae as follows:
	\begin{enumerate}
		\item \label{lem:syntactic_thinness_lemma_TYBASE}
		If $\FORMULATYPES{\GAMMA}{A}$ and $\EXPRESSIONTYPES{\GAMMA}{x}{\TYBASE}$ then $A \THINWRT{x}$
		\item \label{lem:syntactic_thinness_lemma_Core}
		The assertions $A \equiv \TRUTH, \FALSITY, e=e', \ e\neq e', \ \FRESH{y}{\GAMMA_1\REMOVETCVfrom}$ and are all free from $x$ ($x\notin \FV{A}$) then $A \THINWRT{x}$
		\item \label{lem:syntactic_thinness_lemma_Fresh-Sub-LTC}
		If $	\FORMULATYPES{(\GAMMA \PLUSV y \PLUSG \GAMMA')\REMOVEVARIABLE y}{\FRESH{x}{\GAMMA_0}} \ \MAND \ \TCTYPES{\GAMMA}{\GAMMA_0} $ then $\FRESH{x}{\GAMMA_0} \THINWRT{y}$
		\item \label{lem:syntactic_thinness_lemma_Fresh-Sub-LTC+Name}
		If $	\FORMULATYPES{(\GAMMA \PLUSV y:\NAME \PLUSG \GAMMA')\REMOVEVARIABLE y}{\FRESH{x}{\GAMMA_0 \PLUSV b}} \ \MAND \ \TCTYPES{\GAMMA}{\GAMMA_0} \ \MAND \ \EXPRESSIONTYPES{\GAMMA'}{b}{\NAME} $ then $\FRESH{x}{\GAMMA_0 \PLUSV b} \THINWRT{y}$
		\item \label{lem:syntactic_thinness_lemma_Core-Inductive}
		If $A_1 \THINWRT{x}$ and $A_2 \THINWRT{x}$ then  
		$A_1 \PAND A_2$, 
		$A_1 \POR A_2$, 
		$\ONEEVAL{u}{e}{m}{A_1}$, 
		$\FORALL{y}{\GAMMA_1} A_1$, 
		$\EXISTS{y^{\TYBASE}}{\GAMMA_1} A_1$, 
		$\EXISTS{y}{\GAMMA_1\REMOVETCVfrom} A_1$ are all $\THINWRT{x}$.
		\item \label{lem:syntactic_thinness_lemma_FAD}
		If  $A_1 \THINWRT{x}$ and $\TCV \notin \FTCV{A_1}$ then  $\FAD{\TCV} A_1 \THINWRT{x}$
		\item \label{lem:syntactic_thinness_lemma_FAD-FORALL}
		If $A_1 \THINWRT{x}$ then $\FAD{\TCV} \FORALL{y^{\alpha_y}}{\GAMMA \PLUSTC \TCV} A_1^{-\TCV} \THINWRT{x}$
	\end{enumerate}
\begin{proof}
	This proof assumes that $\MMM \REMOVEVARIABLE x$ (or $\MMM \REMOVEVARIABLE y$) is a well-constructed models.
	
	\begin{enumerate}
		\item If $\FORMULATYPES{\GAMMA \REMOVEVARIABLE x}{A}$ and $\EXPRESSIONTYPES{\GAMMA}{x}{\TYBASE}$ then $A \THINWRT{x}$:
		\\
		Given $\SEM{e}{\MMM} \equiv \SEM{e}{\MMM \cdot x:V_x^{\TYBASE}}$ and  $\SEM{\GAMMA_1}{\MMM} \equiv \SEM{\GAMMA_1}{\MMM \cdot x:V_x^{\TYBASE}}$ the proof holds easily.
		\\
		 Lemmas \ref{lem:Base_Values_Dont_Extend_Reach}, \ref{lem:Base_Values_can_be_added/removed_from_TCV:TC}, \ref{lem:Base_types_added/removed_maintain_extension}.
		
		\item The assertions $A \equiv \TRUTH, \FALSITY, e=e', \ e\neq e', \ \FRESH{y}{\GAMMA_1\REMOVETCVfrom}$ and are all free from $x$ ($x\notin \FV{A}$) then $A \THINWRT{x}$:
		\begin{itemize}
			\item[\DONE $\TRUTH, \FALSITY$] clearly hold.
			\item [\DONE $e=e'$] 
			\begin{NDERIVATION}{1}
				\NLINE{\MMM \models e=e'}{}
				\NLINE{\MIFF \ \SEM{e}{\MMM} \CONGCONTEXT{\alpha}{\AN{\MMM}} \SEM{e'}{\MMM}}{Sem. $=$}
				\NLINE{\MIFF \ \SEM{e}{\MMM\REMOVEVARIABLE x} \CONGCONTEXT{\alpha}{\AN{\MMM}} \SEM{e'}{\MMM\REMOVEVARIABLE x}
				}{$\EXPRESSIONTYPES{\GAMMA \REMOVEVARIABLE x}{e}{\alpha}  \MIMPLIES \SEM{e}{\MMM} \equiv \SEM{e}{\MMM \REMOVEVARIABLE x}$}
				\NLINE{ \MIFF \ \SEM{e}{\MMM\REMOVEVARIABLE x} \CONGCONTEXT{\alpha}{\AN{\MMM\REMOVEVARIABLE x}} \SEM{e'}{\MMM\REMOVEVARIABLE x}
				}{Lemma \ref{lem:adding/removing_names_to_congruence_makes_no_difference}}
				\NLASTLINE{\MIFF \ \MMM\REMOVEVARIABLE x \models e=e'}{}
			\end{NDERIVATION}
			\item[\DONE $e \neq e'$] See above.
			\\
			\item [\DONE $\FRESH{y}{\GAMMA_1\REMOVETCVfrom}$]
			\begin{NDERIVATION}{1}
				\NLINE{
					\MMM \models \FRESH{y}{\GAMMA_1\REMOVETCVfrom}
				}{}
				\NLINE{
					\MIFF \ \neg \Mexists  M_y. \LTCDERIVEDVALUE{M_y}{\GAMMA_1\REMOVETCVfrom}{\MMM}{\SEM{y}{\MMM}}
				}{Sem. $\FRESH{}{}$}
				\NLINE{
					\MIFF \ \neg \Mexists  M_y. \LTCDERIVEDVALUE{M_y}{\GAMMA_1\REMOVETCVfrom}{\MMM}{\SEM{y}{\MMM \REMOVEVARIABLE x}}
				}{$\EXPRESSIONTYPES{\GAMMA \REMOVEVARIABLE x}{y}{\NAME} \MIMPLIES \SEM{y}{\MMM} \equiv \SEM{y}{\MMM \REMOVEVARIABLE x}$}
				\NLINE{
					\MIFF \ \neg \Mexists  M_y. \LTCDERIVEDVALUE{M_y}{\GAMMA_1\REMOVETCVfrom}{\MMM}{\SEM{y}{\MMM \REMOVEVARIABLE x}}
				}{
				$
					\SEM{\GAMMA_1}{\MMM} \equiv \SEM{\GAMMA_1}{\MMM \REMOVEVARIABLE x} 
				$
				}
				\NLASTLINE{
					\MIMPLIES \ 
					\MMM\REMOVEVARIABLE x \models \FRESH{y}{\GAMMA_1\REMOVETCVfrom}
				}{}
			\end{NDERIVATION}
			
		\end{itemize}
		
		\item If $	\FORMULATYPES{(\GAMMA \PLUSV y \PLUSG \GAMMA')\REMOVEVARIABLE y}{\FRESH{x}{\GAMMA_0}} \ \MAND \ \TCTYPES{\GAMMA}{\GAMMA_0} $ then $\FRESH{x}{\GAMMA_0} \THINWRT{y}$:
		\\
		Lemma 
		\ref{lem:syntactic_thin_implies_thin_Fresh-typed}
		
		\item If $	\FORMULATYPES{(\GAMMA \PLUSV y:\NAME \PLUSG \GAMMA')\REMOVEVARIABLE y}{\FRESH{x}{\GAMMA_0 \PLUSV b}} \ \MAND \ \TCTYPES{\GAMMA}{\GAMMA_0} \ \MAND \ \EXPRESSIONTYPES{\GAMMA'}{b}{\NAME} $ then $\FRESH{x}{\GAMMA_0 \PLUSV b} \THINWRT{y}$
		\\
		Essentially relies on $b:\NAME$ must either be equal to $y$ in which case no issue, or different in which case no issue.
		\\
		\begin{NDERIVATION}{1}
			\NLINE{\MMM^{\GAMMA \PLUSV y \PLUSG \GAMMA'} \models \FRESH{x}{\GAMMA_0 \PLUSV b}}{}
			\NLINE{
			\MIFF \ \neg \Mexists  M_x. \LTCDERIVEDVALUE{M_x}{\GAMMA_0  \PLUSV b}{\MMM}{\SEM{x}{\MMM}}
			}{Sem. $\FRESH{}{}$}
			\NLINE{
			\MIFF \ \neg \Mexists  M_x. \LTCDERIVEDVALUE{M_x}{\GAMMA_0  \PLUSV b}{\MMM}{\SEM{x}{\MMM \REMOVEVARIABLE y}}
			}{$\EXPRESSIONTYPES{\GAMMA \REMOVEVARIABLE y}{x}{\NAME} \MIMPLIES \SEM{x}{\MMM} \equiv \SEM{x}{\MMM \REMOVEVARIABLE y}$}
			\NLINE{
				\MIFF \ \neg \Mexists  M_x. 
				\begin{array}[t]{l}
					\TYPES{\SEM{\GAMMA_0 \PLUSV b}{\MMM}}{M_x}{\NAME}
					\
					\MAND \ 
					\AN{M_x}= \emptyset
					\\
					\MAND \ 
					(\AN{\MMM}, \ M_x\MMM) \CONV (\AN{\MMM}, G, \ \SEM{x}{\MMM\REMOVEVARIABLE y})
				\end{array}
			}{Sem. $\LTCDERIVEDVALUE{}{}{}{}$}
			\NPLINE{
				\MIMPLIES \ \neg \Mexists  M_x. 
				\begin{array}[t]{l}
					\TYPES{\SEM{\GAMMA_0 \PLUSV b}{\MMM \REMOVEVARIABLE y}}{M_x}{\NAME}
					\
					\MAND \ 
					\AN{M_x}= \emptyset
					\\
					\MAND \ 
					(\AN{\MMM}, \ M_x(\MMM \REMOVEVARIABLE y)) \CONV (\AN{\MMM}, G, \ \SEM{x}{\MMM\REMOVEVARIABLE y})
				\end{array}
			}{4cm}{
				Subset $M_x$
				\\
				$\SEM{\GAMMA_0 \PLUSV b}{\MMM \REMOVEVARIABLE y} \subseteq \SEM{\GAMMA_0 \PLUSV b}{\MMM}$
			}
			\\
			\NLINE{\text{Case $\MMM(y) \in \AN{\MMM\REMOVEVARIABLE y}$ implies }  \ \AN{\MMM} \equiv \AN{\MMM \REMOVEVARIABLE y} }{}
			\NLINE{\MIMPLIES \ \neg \Mexists  M_x. 
				\begin{array}[t]{l}
					\TYPES{\SEM{\GAMMA_0 \PLUSV b}{\MMM \REMOVEVARIABLE y}}{M_x}{\NAME}
					\
					\MAND \ 
					\AN{M_x}= \emptyset
					\\
					\MAND \ 
					(\AN{\MMM\REMOVEVARIABLE y}, \ M_x(\MMM\REMOVEVARIABLE y)) \CONV (\AN{\MMM\REMOVEVARIABLE y}, G, \ \SEM{x}{\MMM\REMOVEVARIABLE y})
				\end{array}
			}{$\AN{\MMM} \equiv \AN{\MMM\REMOVEVARIABLE y}$}
			\\
			\NLINE{\text{Case $\MMM(y) \notin \AN{\MMM\REMOVEVARIABLE y}$ implies }  \ \MMM(y) \notin \AN{M_1 (\MMM \REMOVEVARIABLE y)} }{}
			\NLINE{
				\MIMPLIES \ \neg \Mexists  M_x. 
				\begin{array}[t]{l}
					\TYPES{\SEM{\GAMMA_0 \PLUSV b}{\MMM \REMOVEVARIABLE y}}{M_x}{\NAME}
					\
					\MAND \ 
					\AN{M_x}= \emptyset
					\\
					\MAND \ 
					(\AN{\MMM \REMOVEVARIABLE y}, \ M_x(\MMM \REMOVEVARIABLE y)) \CONV (\AN{\MMM \REMOVEVARIABLE y}, G, \ \SEM{x}{\MMM\REMOVEVARIABLE y})
				\end{array}
			}{
				Lemma \ref{lem:adding/remove_unused_names_maintains_evaluation}
			}
			\\
			\NLINE{
				\MIMPLIES \ \neg \Mexists  M_x. \LTCDERIVEDVALUE{M_x}{\GAMMA_0  \PLUSV b}{\MMM\REMOVEVARIABLE y}{\SEM{x}{\MMM \REMOVEVARIABLE y}}
			}{Sem. $\LTCDERIVEDVALUE{}{}{}{}$}
			\NLASTLINE{\MIMPLIES \ 
				\MMM\REMOVEVARIABLE y \models \FRESH{x}{\GAMMA_0 \PLUSV b}
			}{}
		\end{NDERIVATION}
		
		\item
		If $A_1 \THINWRT{x}$ and $A_2 \THINWRT{x}$ then  
		$A_1 \PAND A_2$, 
		$A_1 \POR A_2$, 
		$\ONEEVAL{u}{e}{m}{A_1}$, 
		$\FORALL{y}{\GAMMA_1} A_1$, 
		$\EXISTS{y^{\TYBASE}}{\GAMMA_1} A_1$, 
		$\EXISTS{y}{\GAMMA_1\REMOVETCVfrom} A_1$ are all $\THINWRT{x}$:
		\\
		\begin{itemize}
		
		\item[\DONE Elementary cases:] obvious cases are $A_1 \PAND A_2$, $A_1 \POR A_2$ where the proofs are omitted. 
		
		\item[\DONE $\ONEEVAL{e}{e'}{m}{A_1}$:] holds by IH on $A_1$ and  
			
			\begin{NDERIVATION}{1}
				\NLINE{\MMM \models \ONEEVAL{e}{e'}{m}{A_1}}{}
				\NLINE{\MIFF \ \LTCDERIVEDVALUE{ee'}{\GAMMA}{\MMM}{V_m} \ \MAND \ \MMM \cdot m:V_m \models A_1}{IH}
				\NPLINE{\MIMPLIES \ \LTCDERIVEDVALUE{ee'}{\GAMMA}{\MMM\REMOVEVARIABLE x}{V_m} \ \MAND \ \MMM  \REMOVEVARIABLE x \cdot m:V_m \models A_1
				}{5cm}{
					$\SEM{ee'}{\MMM} \equiv \SEM{ee'}{\MMM\REMOVEVARIABLE x}$,
					\
					Lemma \ref{lem:adding/remove_unused_names_maintains_evaluation} 
					\\ 
					($\AN{V_m} \cap \AN{V_x}  \subseteq \AN{\MMM\REMOVEVARIABLE x}$)
				}
				\NLASTLINE{\MIMPLIES \ 
					\MMM\REMOVEVARIABLE x \models \ONEEVAL{e}{e'}{m}{A_1}}{}
			\end{NDERIVATION}
			
		\item[\DONE $\FORALL{y}{\GAMMA_1} A_1$:] holds assuming $A_1 \THINWRT{x} $ then
		\begin{NDERIVATION}{1}
			\NLINE{\MMM \models \FORALL{y}{\GAMMA_1} A_1 }{}
			\NLINE{\MIFF \ \Mforall M_y. \LTCDERIVEDVALUE{M_y}{\GAMMA_1}{\MMM}{V_y} \ \MIMPLIES \ \MMM \cdot y:V_y \models A_1}{Sem. $\FORALL{y}{\GAMMA_1} A_1$}
			\NLINE{\MIMPLIES \ \Mforall M_y. 
				\begin{array}[t]{l}
					\AN{M_y}=\emptyset
					\ \MAND \  
					\TYPES{\SEM{\GAMMA_1}{\MMM}}{M_y}{\alpha} 
					\\
					\MAND \ (\AN{\MMM}, M_y\MMM) \CONV (\AN{\MMM},G', \ V_y) 
				\end{array}
				\ \MIMPLIES \ 
				\MMM  \cdot y:V_y \models A_1
			}{Sem. $\LTCDERIVEDVALUE{}{}{}{}$}
			\NLINE{\MIMPLIES \ \Mforall M_y. 
				\begin{array}[t]{l}
					\AN{M_y}=\emptyset
					\ \MAND \  
					\TYPES{\SEM{\GAMMA_1}{\MMM\REMOVEVARIABLE x}}{M_y}{\alpha} 
					\\
					\MAND \ (\AN{\MMM}, M_y\MMM) \CONV (\AN{\MMM},G', \ V_y) 
				\end{array}
				\ \MIMPLIES \ 
				\MMM  \cdot y:V_y \models A_1
			}{$\SEM{\GAMMA_1}{\MMM\REMOVEVARIABLE x} \subseteq \SEM{\GAMMA_1}{\MMM}$}
			\NPLINE{\MIMPLIES \ \Mforall M_y. 
				\begin{array}[t]{l}
					\AN{M_y}=\emptyset
					\ \MAND \  
					\TYPES{\SEM{\GAMMA_1}{\MMM \REMOVEVARIABLE x}}{M_y}{\alpha} 
					\\
					\MAND \ (\AN{\MMM \REMOVEVARIABLE x}, M_y\MMM \REMOVEVARIABLE x) \CONV (\AN{\MMM \REMOVEVARIABLE x},G', \ V_y) 
				\end{array}
				\ \MIMPLIES \ 
				\MMM  \cdot y:V_y \models A_1
			}{4cm}{$x \notin \FV{\SEM{\GAMMA_1}{\MMM\REMOVEVARIABLE x}}$, \\  Lemma \ref{lem:adding/remove_unused_names_maintains_evaluation} ($\AN{V_y} \cap \AN{V_x}  \subseteq \AN{\MMM\REMOVEVARIABLE x}$)}
			\NLINE{\MIMPLIES \ \Mforall M_y. \LTCDERIVEDVALUE{M_y}{\GAMMA_1}{\MMM\REMOVEVARIABLE x}{V_y} \ \MIMPLIES \ (\MMM \cdot y:V_y)  \models A_1
			}{Sem. $\LTCDERIVEDVALUE{}{}{}{}$}
			\NLINE{\MIMPLIES \ \Mforall M_y. \LTCDERIVEDVALUE{M_y}{\GAMMA_1}{\MMM\REMOVEVARIABLE x}{V_y} \ \MIMPLIES \ (\MMM \cdot y:V_y) \REMOVEVARIABLE x \models A_1}{IH}
			\NLINE{\MIMPLIES \ \Mforall M_y. \LTCDERIVEDVALUE{M_y}{\GAMMA_1}{\MMM\REMOVEVARIABLE x}{V_y} \ \MIMPLIES \ \MMM \REMOVEVARIABLE x \cdot y:V_y \models A_1}{Def $\MMM \REMOVEVARIABLE x$}
			\NLASTLINE{
				\MIMPLIES \ 
				\MMM\REMOVEVARIABLE x \models \FORALL{y}{\GAMMA_1} A_1 
			}{Sem. $\FORALL{y}{\GAMMA_1} A_1$}
		\end{NDERIVATION}
		\item[\DONE $\EXISTS{y^{\TYBASE}}{\GAMMA_1} A_1$:] holds as follows:
		\begin{NDERIVATION}{1}
			\NLINE{\MMM \models \EXISTS{y^{\TYBASE}}{\GAMMA_1} A_1 }{}
			\NLINE{\MIFF \ \Mexists M_y. \LTCDERIVEDVALUE{M_y}{\GAMMA_1}{\MMM}{V_y} \ \MAND \ \MMM \cdot y:V_y \models A_1}{Sem. $\EXISTS{y}{\GAMMA_1} A_1$}
			\NLINE{\MIMPLIES \ \Mexists M_y. 
				\begin{array}[t]{l}
					\AN{M_y}=\emptyset
					\ \MAND \  
					\TYPES{\SEM{\GAMMA_1}{\MMM}}{M_y}{\alpha} 
					\\
					\MAND \ (\AN{\MMM}, M_y\MMM) \CONV (\AN{\MMM},G', \ V_y) 
				\end{array}
				\ \MAND \ 
				\MMM \cdot y:V_y \models A_1
			}{Sem. $\LTCDERIVEDVALUE{}{}{}{}$}
		\NPLINE{\MIMPLIES \ \Mexists M_y. 
			\begin{array}[t]{l}
				\AN{M_y}=\emptyset
				\ \MAND \  
				\TYPES{\SEM{\GAMMA_1}{\MMM\REMOVEVARIABLE x}}{M_y}{\alpha} 
				\\
				\MAND \ (\AN{\MMM}, M_y\MMM) \CONV (\AN{\MMM},G', \ V_y) 
			\end{array}
			\ \MAND \ 
			\MMM \cdot y:V_y \models A_1
			}{4.5cm}{$V_y^{\TYBASE}$ equally derivable from any TC,  Lemma \ref{lem:BaseValues_are_equally_derivable_from_any_LTC}} 
			\NPLINE{\MIMPLIES \ \Mexists M_y. 
				\begin{array}[t]{l}
					\AN{M_y}=\emptyset
					\ \MAND \  
					\TYPES{\SEM{\GAMMA_1}{\MMM \REMOVEVARIABLE x}}{M_y}{\alpha} 
					\\
					\MAND \ (\AN{\MMM \REMOVEVARIABLE x}, M_y\MMM \REMOVEVARIABLE x) \CONV (\AN{\MMM \REMOVEVARIABLE x},G', \ V_y) 
				\end{array}
				\ \MAND \ 
				\MMM \cdot y:V_y \models A_1
			}{4cm}{$\begin{array}[t]{r}
				x \notin \FV{\SEM{\GAMMA_1}{\MMM\REMOVEVARIABLE x}}, \\ \text{ Lemma \ref{lem:adding/remove_unused_names_maintains_evaluation}} \\ (\AN{V_y} \cap \AN{V_x}  \subseteq \AN{\MMM\REMOVEVARIABLE x})
			\end{array}$}
			\NLINE{\MIMPLIES \ \Mexists M_y. \LTCDERIVEDVALUE{M_y}{\GAMMA_1}{\MMM\REMOVEVARIABLE x}{V_y} \ \MAND \ \MMM \cdot y:V_y\models A_1
			}{Sem. $\LTCDERIVEDVALUE{}{}{}{}$}
			
			\NLINE{\MIMPLIES \ \Mexists M_y. \LTCDERIVEDVALUE{M_y}{\GAMMA_1}{\MMM}{V_y} \ \MAND \ (\MMM \cdot y:V_y) \REMOVEVARIABLE x \models A_1}{IH}
			\NLINE{\MIMPLIES \ \Mexists M_y. \LTCDERIVEDVALUE{M_y}{\GAMMA_1}{\MMM}{V_y} \ \MAND \ \MMM \REMOVEVARIABLE x \cdot y:V_y \models A_1}{Def $\MMM \REMOVEVARIABLE x$}
			\NLASTLINE{\MIMPLIES \ 
				\MMM\REMOVEVARIABLE x \models \EXISTS{y}{\GAMMA_1} A_1 
			}{Sem. $\EXISTS{y}{\GAMMA_1} A_1$}
		\end{NDERIVATION}
			
		\item[\DONE $\EXISTS{y}{\GAMMA_1\REMOVETCVfrom} A_1$:] holds as follows:
		\begin{NDERIVATION}{1}
			\NLINE{\MMM \models \EXISTS{y}{\GAMMA_1\REMOVETCVfrom} A_1 }{}
			\NLINE{\MIFF \ \Mexists M_y. \LTCDERIVEDVALUE{M_y}{\GAMMA_1\REMOVETCVfrom}{\MMM}{V_y} \ \MAND \ \MMM \cdot y:V_y \models A_1
			}{Sem. $\EXISTS{y}{\GAMMA_1\REMOVETCVfrom} A_1$}
			\NLINE{\MIMPLIES \ \Mexists M_y. 
				\begin{array}[t]{l}
					\AN{M_y}=\emptyset
					\ \MAND \  
					\TYPES{\SEM{\GAMMA_1\REMOVETCVfrom}{\MMM}}{M_y}{\alpha} 
					\\
					\MAND \ (\AN{\MMM}, M_y\MMM) \CONV (\AN{\MMM},G', \ V_y) 
				\end{array}
				\ \MAND \ 
				\MMM \cdot y:V_y \models A_1
			}{Sem. $\LTCDERIVEDVALUE{}{}{}{}$}
			\NLINE{ \MIMPLIES \ \Mexists M_y. 
				\begin{array}[t]{l}
					\AN{M_y}=\emptyset
					\ \MAND \  
					\TYPES{\SEM{\GAMMA_1\REMOVETCVfrom}{\MMM\REMOVEVARIABLE x}}{M_y}{\alpha} 
					\\
					\MAND \ (\AN{\MMM}, M_y\MMM) \CONV (\AN{\MMM},G', \ V_y) 
				\end{array}
				\ \MAND \ 
				\MMM \cdot y:V_y \models A_1
			}{$\SEM{\GAMMA_0}{\MMM\REMOVEVARIABLE x} \subseteq \SEM{\GAMMA_0}{\MMM}$}
			\NLINE{\MIMPLIES \ \Mexists M_y. 
				\begin{array}[t]{l}
					\AN{M_y}=\emptyset
					\ \MAND \  
					\TYPES{\SEM{\GAMMA_1\REMOVETCVfrom}{\MMM \REMOVEVARIABLE x}}{M_y}{\alpha} 
					\\
					\MAND \ (\AN{\MMM \REMOVEVARIABLE x}, M_y\MMM \REMOVEVARIABLE x) \CONV (\AN{\MMM \REMOVEVARIABLE x},G', \ V_y) 
				\end{array}
				\ \MAND \ 
				\MMM \cdot y:V_y \models A_1
				\hspace{-1cm}
			}{$
				\begin{array}[t]{r}
					x \notin \FV{\SEM{\GAMMA_1}{\MMM\REMOVEVARIABLE x}}, \\ \text{ Lemma \ref{lem:adding/remove_unused_names_maintains_evaluation}} \\ (\AN{V_y} \cap \AN{V_x}  \subseteq \AN{\MMM\REMOVEVARIABLE x})
				\end{array}$
			}
			\NLINE{\MIMPLIES \ \Mexists M_y. \LTCDERIVEDVALUE{M_y}{\GAMMA_1\REMOVETCVfrom}{\MMM}{V_y} \ \MAND \ (\MMM \cdot y:V_y) \REMOVEVARIABLE x \models A_1}{IH}
			\NLINE{\MIMPLIES \ \Mexists M_y. \LTCDERIVEDVALUE{M_y}{\GAMMA_1\REMOVETCVfrom}{\MMM}{V_y} \ \MAND \ \MMM \REMOVEVARIABLE x \cdot y:V_y \models A_1
			}{Def $\MMM \REMOVEVARIABLE x$}
			\NLINE{
				\MIMPLIES \ \Mexists M_y. \LTCDERIVEDVALUE{M_y}{\GAMMA_1\REMOVETCVfrom}{\MMM\REMOVEVARIABLE x}{V_y} \ \MAND \ (\MMM \cdot y:V_y) \REMOVEVARIABLE x \models A_1
			}{Sem. $\LTCDERIVEDVALUE{}{}{}{}$}
			\NLASTLINE{
				\MIMPLIES \ 
				\MMM\REMOVEVARIABLE x \models \EXISTS{y}{\GAMMA_1\REMOVETCVfrom} A_1 
			}{Sem. $\EXISTS{y}{\GAMMA_1\REMOVETCVfrom} A_1$}
		\end{NDERIVATION}
		
	\end{itemize}

		\item[\ref{def:syntactic_thinness_def_FAD}] If  $A_1 \THINWRT{x}$ and $\TCV \notin \FTCV{A_1}$ then  $\FAD{\TCV} A_1 \THINWRT{x}$:
		\\
		Holds by IH on $A_1$ as follows
		\begin{NDERIVATION}{1}
			\NLINE{\MMM \models \FAD{\TCV} A_1 }{}
			\NLINE{\MIFF \ 
				\Mforall \MMM'^{\GAMMA'}. \ \MMM \EXTSTAR \MMM' \ \MIMPLIES \ \MMM' \cdot \TCV:\GAMMA' \models A_1
			}{Sem. $\FAD{\TCV}$}
			\NLINE{ \MIFF \ 
				\Mforall \MMM'^{\GAMMA'}. \ \MMM \EXTSTAR \MMM' \ \MIMPLIES \ \MMM' \models A_1^{-\TCV}
			}{Lemma \ref{lem:model_and_model_plus_TCV_models_equivalently_TCV-free_formula}}
			\NLINE{ \MIFF \ 
				\Mforall \MMM'^{\GAMMA'}. \ \MMM \EXTSTAR \MMM' \ \MIMPLIES \ \MMM'\REMOVEVARIABLE x \models A_1^{-\TCV}
			}{IH}
			\NPLINE{ \MIMPLIES \ 
				\Mforall \MMM'\REMOVEVARIABLE x ^{\GAMMA'\REMOVEVARIABLE x}. \ \MMM \REMOVEVARIABLE x \EXTSTAR \MMM' \REMOVEVARIABLE x \ \MIMPLIES \ \MMM'\REMOVEVARIABLE x \models A_1^{-\TCV}
			}{4cm}{Subset of $\Mforall \MMM'^{\GAMMA'}$ \\ as $\SEM{\GAMMA \REMOVEVARIABLE x}{\MMM \REMOVEVARIABLE x} \subseteq \SEM{\GAMMA}{\MMM}$ \\ and $\MMM \REMOVEVARIABLE x \subseteq \MMM$}
			\NLINE{ \MIMPLIES \ 
				\Mforall \MMM'' \ \! ^{\GAMMA''}. \ \MMM \REMOVEVARIABLE x \EXTSTAR \MMM'' \ \MIMPLIES \ \MMM'' \models A_1^{-\TCV}
			}{Rename $\MMM'' \equiv \MMM' \REMOVEVARIABLE x$}
			\NLINE{ \MIMPLIES \ 
				\Mforall \MMM'' \ \! ^{\GAMMA''}. \ \MMM \REMOVEVARIABLE x \EXTSTAR \MMM'' \ \MIMPLIES \ \MMM'' \cdot \TCV: \GAMMA'' \models A_1^{-\TCV}
			}{Lemma \ref{lem:model_and_model_plus_TCV_models_equivalently_TCV-free_formula}}
			\NLASTLINE{\MIMPLIES \ 
				\MMM\REMOVEVARIABLE x \models \FAD{\TCV} A_1^{-\TCV}
			}{Sem. $\FAD{\TCV}$}
			
		\end{NDERIVATION}
		
		\item [\ref{def:syntactic_thinness_def_FAD-FORALL}]
		If $A_1 \THINWRT{x}$ then $\FAD{\TCV} \FORALL{y^{\alpha_y}}{\GAMMA \PLUSTC \TCV} A_1^{-\TCV} \THINWRT{x}$
		\\
			{\small
				Holds by IH on $A_1$ as follows:
				\begin{NDERIVATION}{1}
					\NLINE{\text{Assume: } A \THINWRT{x}}{}
					\NLINE{\text{i.e. assume } \Mforall \GAMMA, \TCV, y^{\alpha_y}. \ \FORMULATYPES{\GAMMA \PLUSTC \TCV \PLUSV y: \alpha_y \REMOVEVARIABLE x}{A} \ \MIMPLIES \ \Mforall \MMM_{dy}^{\GAMMA \PLUSTC \TCV \PLUSV y: \alpha_y}. \ \MMM_{dy} \models A \ \MIMPLIES \ \MMM_{dy} \REMOVEVARIABLE x \models A}{}
					\NLINE{\text{Assume: $\GAMMA$, $\MMM^{\GAMMA}$ s.t. $\FORMULATYPES{\GAMMA \REMOVEVARIABLE x}{\FAD{\TCV} \FORALL{y}{\GAMMA \PLUSTC \TCV} A}$ and $\MMM \models \FAD{\TCV} \FORALL{y}{\GAMMA \PLUSTC \TCV} A$ }}{}
					\NPLINE{\MIMPLIES \ 
						\begin{array}[t]{l}
							\FORMULATYPES{\GAMMA \REMOVEVARIABLE x \PLUSTC \TCV \PLUSV y:\alpha_y}{ A}
							\\ \MAND \ 
							\Mforall \MMM_d^{\GAMMA_d}. \MMM \EXTSTAR \MMM_d 
							\ \MIMPLIES \
							\Mforall M_y^{\alpha_y}. \LTCDERIVEDVALUE{M_y}{\GAMMA \PLUSTC \TCV}{\MMM_d \cdot \TCV:\GAMMA_d}{V_y} 
							\ \MIMPLIES \
							\MMM_d  \cdot \TCV:\GAMMA_d \cdot y:V_y \models A
						\end{array}
					}{2.5cm}{Typing rules,\\ Sem. $\FAD{}$, $\FORALL{}{}$}
					\NLINE{\MIMPLIES \ 
						\Mforall \MMM_d^{\GAMMA_d}. \MMM \EXTSTAR \MMM_d 
						\ \MIMPLIES \
						\Mforall M_y^{\alpha_y}. \LTCDERIVEDVALUE{M_y}{\GAMMA \PLUSTC \TCV}{\MMM_d \cdot \TCV:\GAMMA_d}{V_y} 
						\ \MIMPLIES \
						\MMM_d \cdot y:V_y \models A
					}{$A^{-\TCV}$, Lemma \ref{lem:model_and_model_plus_TCV_models_equivalently_TCV-free_formula}}
					\NLINE{\MIMPLIES \ 
						\Mforall \MMM_d^{\GAMMA_d}. \MMM \EXTSTAR \MMM_d 
						\ \MIMPLIES \
						\Mforall M_y^{\alpha_y}. \LTCDERIVEDVALUE{M_y}{\GAMMA \PLUSTC \TCV}{\MMM_d \cdot \TCV:\GAMMA_d}{V_y} 
						\ \MIMPLIES \
						(\MMM_d \cdot y:V_y)\REMOVEVARIABLE x \models A
					}{Assumption, 2}
					\NPLINE{\MIMPLIES \ 
						\Mforall \MMM_d^{\GAMMA_d}. \MMM \EXTSTAR \MMM_d 
						\ \MIMPLIES \
						\Mforall M_y^{\alpha_y}. \LTCDERIVEDVALUE{M_y}{\GAMMA \PLUSTC \TCV}{\MMM_d \cdot \TCV:\GAMMA_d}{V_y} 
						\ \MIMPLIES \
						(\MMM_d  \cdot \TCV:\GAMMA_d \cdot y:V_y)\REMOVEVARIABLE x \models A
						\hspace{-1cm}
					}{4cm}{$A^{-\TCV}$, Lemma \ref{lem:model_and_model_plus_TCV_models_equivalently_TCV-free_formula}
						\\
						(Ignoring other TCV's)
					}
					\NLINE{\MIMPLIES \ 
						\Mforall \MMM_d^{\GAMMA_d}.  \
						\MMM \EXTSTAR \MMM_d 
						\ \MIMPLIES \
						\Mforall M_y^{\alpha_y}. 
						\begin{array}[t]{l} 
							\TYPES{\SEM{\GAMMA\PLUSTC \TCV}{\MMM_d \cdot \TCV:\GAMMA_d}}{M_y}{\alpha_y}
							\
							\MAND \ \AN{M_y}=\emptyset
							\\
							\MAND \ (\AN{\MMM_d \cdot \TCV:\GAMMA_d}, \ M_y(\MMM_d \cdot \TCV:\GAMMA_d)) \CONV (\AN{\MMM_d \cdot \TCV:\GAMMA_d}, G', \ V_y)
							\\ \MIMPLIES \
							(\MMM_d \cdot \TCV:\GAMMA_d \cdot y:V_y)\REMOVEVARIABLE x \models A
						\end{array}
					}{Sem. $\LTCDERIVEDVALUE{}{}{}{}$}
					\NPLINE{\MIMPLIES \ 
						\MAND \
						\Mforall \MMM_d^{\GAMMA_d}.  \
						\MMM \EXTSTAR \MMM_d 
						\ \MIMPLIES \
						\Mforall M_y^{\alpha_y}. 
						\begin{array}[t]{l} 
							\MMM_{dd} \equiv \MMM_d \cdot \TCV:\GAMMA_d
							\\
							\MAND \
							\TYPES{\SEM{(\GAMMA\PLUSTC \TCV) \REMOVEVARIABLE x}{\MMM_{dd} \REMOVEVARIABLE x}}{M_y}{\alpha_y}
							\
							\MAND \ \AN{M_y}=\emptyset
							\\
							\MAND \ (\AN{\MMM_{dd}}, \ M_y\MMM_{dd}) \CONV (\AN{\MMM_{dd}}, G', \ V_y)
							\\ \MIMPLIES \
							(\MMM_{dd} \cdot y:V_y)\REMOVEVARIABLE x \models A
						\end{array}
						\hspace{-1cm}
					}{5cm}{
						Subset of Type Context
						\\
						$\SEM{\GAMMA\PLUSTC \TCV}{\MMM_{dd}} \REMOVEVARIABLE x \subseteq \SEM{\GAMMA\PLUSTC \TCV}{\MMM_{dd}} $
						\\
						$\SEM{\GAMMA_0}{\MMM_0}\REMOVEVARIABLE x \equiv \SEM{\GAMMA_0 \REMOVEVARIABLE x}{\MMM_0 \REMOVEVARIABLE x}$
					}
					\NPLINE{\MIMPLIES \ 
						\Mforall \MMM_d^{\GAMMA_d}. 
						\MMM \EXTSTAR \MMM_d 
						\ \MIMPLIES \
						\Mforall M_y^{\alpha_y}. 
						\begin{array}[t]{l} 
							\MMM_{dd-x} \equiv (\MMM_d \cdot \TCV:\GAMMA_d) \REMOVEVARIABLE x
							\\
							\MAND \
							\TYPES{\SEM{(\GAMMA\PLUSTC \TCV) \REMOVEVARIABLE x}{\MMM_{dd-x}}}{M_y}{\alpha_y}
							\
							\MAND \ \AN{M_y}=\emptyset
							\\
							\MAND \ (\AN{\MMM_d \cdot \TCV:\GAMMA_d}, \ M_y\MMM_{dd-x}) \CONV (\AN{\MMM_d \cdot \TCV:\GAMMA_d}, G', \ V_y)
							\\ \MIMPLIES \
							\MMM_{dd-x} \cdot y:V_y \models A
						\end{array}
						\hspace{-2cm}
					}{5cm}{
						$x \notin \FV{M} \MIMPLIES M\MMM \equiv M(\MMM \REMOVEVARIABLE x) $
					}
					\NPLINE{\MIMPLIES \ 
						\Mforall \MMM_d^{\GAMMA_d}. 
						\MMM \EXTSTAR \MMM_d 
						\ \MIMPLIES \
						\Mforall M_y^{\alpha_y}. 
						\begin{array}[t]{l} 
							\MMM_{dd-x} \equiv (\MMM_d \cdot \TCV:\GAMMA_d) \REMOVEVARIABLE x
							\\
							\MAND \
							\TYPES{\SEM{(\GAMMA\PLUSTC \TCV) \REMOVEVARIABLE x}{\MMM_{dd-x}}}{M_y}{\alpha_y}
							\
							\MAND \ \AN{M_y}=\emptyset
							\\
							\MAND \ (\AN{\MMM_{dd-x}}, \ M_y\MMM_{dd-x}) \CONV (\AN{\MMM_{dd-x}}, G', \ V_y)
							\\ \MIMPLIES \
							\MMM_{dd-x} \cdot y:V_y \models A
						\end{array}
					}{4cm}{
						Lemma \ref{lem:adding/remove_unused_names_maintains_evaluation}
					}
					\NLINE{\MIMPLIES \ 
						\Mforall \MMM_d^{\GAMMA_d}. 
						\MMM \EXTSTAR \MMM_d 
						\ \MIMPLIES \
						(\MMM_d \cdot \TCV:\GAMMA_d) \REMOVEVARIABLE x \models \FORALL{y}{\GAMMA \PLUSV \TCV} A
					}{
						Sem. $\FORALL{}{}$
					}
					\NPLINE{\MIMPLIES \ 
						\Mforall \MMM_{d-x}^{\GAMMA_{d-x}}. \
						\MMM \EXTSTAR \MMM_{d-x}  \cdot x:V_x
						\ \MIMPLIES \
						(\MMM_{d-x}  \cdot x:V_x \cdot \TCV:\GAMMA_d) \REMOVEVARIABLE x \models \FORALL{y}{\GAMMA \PLUSV \TCV} A
					}{4cm}{
						$\MMM_d \equiv \MMM_{d-x} \cdot x:V_x$
					}
					\NLINE{\MIMPLIES \ 
						\Mforall \MMM_{d-x}^{\GAMMA_{d-x}}. \ 
						\MMM \REMOVEVARIABLE x \EXTSTAR \MMM_{d-x}
						\ \MIMPLIES \
						\MMM_{d-x} \cdot \TCV:\GAMMA_{d-x} \models \FORALL{y}{\GAMMA \PLUSV \TCV} A
					}{
						Subset of $\Mforall \MMM_{d-x}$
					}
					\NLINE{\MIMPLIES \ \MMM\REMOVEVARIABLE x \models \FAD{\TCV} \FORALL{y}{\GAMMA \PLUSV \TCV} A
					}{
						Sem. $\FAD{\TCV}$
					}
					\NLINE{\text{Hence: } 
						\begin{array}[t]{l}
							\Mforall \GAMMA. \ \FORMULATYPES{\GAMMA \REMOVEVARIABLE x}{\FAD{\TCV} \FORALL{y}{\GAMMA \PLUSTC \TCV} A} 
							\\ \MIMPLIES \ \Mforall \MMM^{\GAMMA}. \ \MMM \models \FAD{\TCV} \FORALL{y}{\GAMMA \PLUSTC \TCV} A \ \MIMPLIES \ \MMM \REMOVEVARIABLE x \models \FAD{\TCV} \FORALL{y}{\GAMMA \PLUSTC \TCV} A
						\end{array}
					}{lines 3-15}
					\NLASTLINE{\text{Hence: } \FAD{\TCV} \FORALL{y}{\GAMMA \PLUSTC \TCV} A \THINWRT{x}}{}
				\end{NDERIVATION}
			}
	\end{enumerate}
\end{proof}

\end{lemma}

%% file: appendix/syntactic_thinness.tex
\begin{definition}
	\label{def:Syntactically_Thin_formulae}
	Syntactically define thin formulae as follows:
	\begin{enumerate}
		\item \label{def:syntactic_thinness_def_TYBASE}
		If $\FORMULATYPES{\GAMMA}{A}$ and $\EXPRESSIONTYPES{\GAMMA}{x}{\TYBASE}$ then $A \THINWRT{x}$
		\\
		Base types can always be removed if they do not occur in the assertion.
		
		\item \label{def:syntactic_thinness_def_Core}
		For $A \equiv \TRUTH, \FALSITY, e=e', \ e\neq e', \ \FRESH{y}{\GAMMA_1\REMOVETCVfrom}$ and $x\notin \FV{A}$ then $A \THINWRT{x}$
		\\
		The base assertions are $\THINWRT{x}$ under basic assumptions.
		
		\item \label{def:syntactic_thinness_def_Fresh-Sub-ATC}
		If $	\FORMULATYPES{(\GAMMA \PLUSV x \PLUSG \GAMMA')\REMOVEVARIABLE x}{\FRESH{y}{\GAMMA_0}} \ \MAND \ \TCTYPES{\GAMMA}{\GAMMA_0} $ then $\FRESH{y}{\GAMMA_0} \THINWRT{x}$
		\\
		If $\GAMMA \PLUSV x$ and $\TCTYPES{\GAMMA}{\GAMMA_0}$ then $x$ is not included in $\GAMMA_0$ and cannot be included by any TCV occurring in $\GAMMA_0$.
		
		\item \label{def:syntactic_thinness_def_Fresh-Sub-ATC+Name}
		If $	\FORMULATYPES{(\GAMMA \PLUSV x:\NAME \PLUSG \GAMMA')\REMOVEVARIABLE x}{\FRESH{y}{\GAMMA_0 \PLUSV b}} \ \MAND \ \TCTYPES{\GAMMA}{\GAMMA_0} \ \MAND \ \EXPRESSIONTYPES{\GAMMA'}{b}{\NAME} $ then $\FRESH{y}{\GAMMA_0 \PLUSV b} \THINWRT{x}$
		\\
		Similar to \ref{def:syntactic_thinness_def_Fresh-Sub-ATC}, but two possibilities: $b=x \POR b \neq x $ both of which hold.
		
		\item \label{def:syntactic_thinness_def_Core-Inductive}
		If $A_1 \THINWRT{x}$ and $A_2 \THINWRT{x}$ then  
		$A_1 \PAND A_2$, 
		$A_1 \POR A_2$, 
		$\ONEEVAL{u}{e}{m}{A_1}$, 
		$\FORALL{y}{\GAMMA_1} A_1$, 
		$\EXISTS{y^{\TYBASE}}{\GAMMA_1} A_1$, 
		$\EXISTS{y}{\GAMMA_1\REMOVETCVfrom} A_1$ are all $\THINWRT{x}$.
		\\
		Constructing thin assertions with thin assertions holds in most cases.
		
		\item \label{def:syntactic_thinness_def_FAD}
		If  $A \THINWRT{x}$ and $\TCV \notin \FTCV{A}$ then  $\FAD{\TCV} A \THINWRT{x}$
		\\
		If $\TCV$ is unused in $A$ then thinness is maintained.
		
		\item \label{def:syntactic_thinness_def_FAD-FORALL}
		If $A \THINWRT{x}$ and $\TCV \notin \FTCV{A}$ then $\FAD{\TCV} \FORALL{y^{\alpha_y}}{\GAMMA \PLUSTC \TCV} A \THINWRT{x}$
		\\
		If $\TCV$ is only used to derive $y$, then thinness is maintained as it reduces the potentially derived values.
	\end{enumerate}
\end{definition}

%% file: appendix/conservativity.tex
\newcommand{\NUCLOGIC}{$\nu$-logic}

\section{Conservativity}
\label{appendix_conservativitiy}
We sketch the  proof that our logic is an extension of the STLC logic as defined in \cite{HY04PPDP}, and do not explain any of the trivial details.
We define the logic introduced in this paper as the \NUCLOGIC\ for brevity.
We ignore the other possible extensions of $\INT$ and operations on $\INT$, $\BOOL$,...

\PARAGRAPH{The language for STLC}

\begin{GRAMMAR}
	\begin{array}{lclllllll}
		\alpha
		&\ ::= \ &
		\UNIT \VERTICAL \BOOL \VERTICAL \alpha \FS \alpha \VERTICAL \alpha \times \alpha
		\\
		\Gamma
		&::=&
        \emptyset \VERTICAL \Gamma , x:\alpha
                \\[1mm]
		V
		&\ ::=\ &
                x \VERTICAL c \VERTICAL \lambda x.M \VERTICAL \PAIR{V}{V}
		\\[0.5mm]
		M
		& ::= &
		V \VERTICAL MM  \VERTICAL \LET{x}{M}{M} \VERTICAL M = M
		\\
		&& \VERTICAL \IFTHENELSE{M}{M}{M} \VERTICAL \PAIR{M}{M} \VERTICAL \PROJ{}{i}{M}
	\end{array}
\end{GRAMMAR}

\PARAGRAPH{The logic for STLC}

	\begin{GRAMMAR}
		\alpha
		&::=&
   		\UNIT \VERTICAL \BOOL \VERTICAL \alpha \FS \alpha \VERTICAL \alpha \times \alpha
		\\
		e
		&::=&
		\LOGIC{c} \VERTICAL x \VERTICAL \pi_i(e) \VERTICAL \PAIR{e}{e}
		\\
		A
		&::=&
		e = e
		\VERTICAL
		\neg A
		\VERTICAL
		A \AND A
		\VERTICAL
		\EVALFORMULA{e}{e}{x^{\alpha}}{A}
		\VERTICAL 
		\forall x^{\alpha}. A
	\end{GRAMMAR}
\PARAGRAPH{Axioms for STLC}
The axioms are:
\begin{itemize}
	\item STLC terms can only be reasoned about through types which are $\NAME$-free, for which we introduce two axioms:
	\[
	\begin{array}{lrclr}
	(ext) & \quad \EXT{e_1, e_2} & \quad  \PIFF \quad & e_1 = e_2
	\\
	(u5) & \quad  \FORALL{x^{\alpha}}{\emptyset} A & \quad \PIFF \ & \FORALL{x^{\alpha}}{\GAMMA} A & \qquad \alpha \text{ is $\NAME$-free}
	\end{array}
	\]
	Axiom $(ext)$ is similar to that of \cite{HY04PPDP}, but extended to all $\NAME$-free types terms.
	Axiom $(u5)$ ensures names in a $\NAME$-free terms can be swapped for any fresh name.
	\\
	
	\item The axioms for predicate logic, which are equivalent in the \NUCLOGIC.
	
	\item The axioms for evaluation formulae ($\ONEEVAL{}{}{}{}$) which are all lifted directly from the STLC-logic \cite{BergerM:prologfhgrtmp, HY04PPDP} but with the added assumption that all terms terminate and new definition of $\FORALL{x^{\alpha}}{\GAMMA}$ which behaves identically to $\forall x.$ when $\alpha$ is $\NAME$-free.
	
	\item The axioms for $\forall x. A$ which are those of FOL, are represented by their translation to the \NUCLOGIC\ and axioms for $\FORALL{x}{\GAMMA}$.
	
	\item  $(ext)$ this axiom holds in the STLC-logic because functions must always produce the same result if applied twice. This still holds in \NUCLOGIC\ as seen in Sec.~\ref{sec:axioms}, but only for $\NAME$-free types.
	A simple example of why this fails for other types is $\GENSYM : \UNIT \FS \NAME$ which clearly produces different names each time it is applied.
	
\end{itemize}

\PARAGRAPH{Rules for STLC} Can be seen in Fig.~\ref{fig:rules_for_STLC}.
\begin{FIGURE}
  \begin{RULES}
	\ZEROPREMISERULENAMEDRIGHT
        {
		\ASSERT{A\SUBST{x}{m}}{x}{m}{A}
        }{[Var]}
		\quad
	\ZEROPREMISERULENAMEDRIGHT
        {
		\ASSERT{A\SUBST{\LOGIC{c}}{m}}{\PROGRAM{c}}{m}{A}
        }{[Const]}
	\quad
     \TWOPREMISERULENAMEDRIGHT
     {
        	\ASSERT{A}{M}{m}{B}
     }
     {
        	\ASSERT{B}{N}{n}{C\SUBST{\EQA{m}{n}}{u}}
     }
    {
        	\ASSERT{A}{M = N}{u}{C}
     }{[Eq]}
		\\\\
		\ONEPREMISERULENAMEDRIGHT
		{
			\ASSERT{A^{\MINUS x} \AND B}{M}{m}{C}
		}
		{
			\ASSERT
			{A}
			{\lambda x^{\alpha}. M}{u}
			{				\forall x^{\alpha}. (B \IMPLIES \ONEEVAL{u}{x}{m}{C})}
		}{[Lam]}
		\quad
		\TWOPREMISERULENAMEDRIGHT
		{
			\ASSERT{A}{M}{m}{B}
		}
		{
			\ASSERT{B}{N}{n}{\ONEEVAL{m}{n}{u}{C}}
		}
		{
			\ASSERT{A}{MN}{u}{C}
		}{[App]}
		\\\\
		\FIVEPREMISERULENAMEDRIGHT
		{
			\ASSERT{A}{M}{m}{B}
		}
		{
			\ASSERT{B\SUBST{b_i}{m}}{N_i}{u}{C}
		}
		{
			b_1 = \TRUE
		}
		{
			b_2 = \FALSE
		}
		{
			i = 1, 2
		}
		{
			\ASSERT{A}{\IFTHENELSE{M}{N_1}{N_2}}{u}{C}
		}{[If]}
		\\\\
		\TWOPREMISERULENAMEDRIGHT
		{
                  \ASSERT{A}{M}{m}{B}
		}
		{
		  \ASSERT{B}{N}{n}{C \SUBST{\PAIR{m}{n}}{u}}
		}
		{
                  \ASSERT{A}{\PAIR{M}{N}}{u}{C}
		}{[Pair]}
		\quad
		\ONEPREMISERULENAMEDRIGHT
		{
                  \ASSERT{A}{M}{m}{B\SUBST{\pi_i(m)}{u}}
		}
		{
			\ASSERT{A}{\pi_i(M)}{u}{B}
		}{[Proj($i$)]}
\end{RULES}
\caption{Rules for the STLC}
\label{fig:rules_for_STLC}
\end{FIGURE}

\begin{definition}
The logic for the \NUC \ is a \EMPH{conservative extension} of the logic	for
the simply typed $\lambda$-calculus, provided:
\begin{itemize}

\item If $A$ is a formula in the $\lambda$-logic derivable from the
  axioms for the STLC above, then the translation of $A$ to the $\nu$-logic is also derivable from the axioms in the
  logic of the \NUC.

\item If $\ASSERT{A}{M}{m}{B}$ is a triple in the logic of the STLC, derivable
  from the rules for the STLC then $\ASSERT{A}{M}{m}{B}$  is also derivable from
  the rules in the logic of the \NUC.

\end{itemize}
\end{definition}
\newcommand{\TRANSLATELAMNU}[1]{\TRANSLATE{#1}_{\lambda \FS \nu}}
\begin{definition}
	Define a translation of STLC assertions (and triples) into \NUC \ assertions (and triples)as follows:
	\begin{equation}
		\begin{array}{rcl}
			\TRANSLATELAMNU{e=e'} &=& e=e'
			\\
			\TRANSLATELAMNU{\neg A} &=& \neg \TRANSLATELAMNU{A}
			\\
			\TRANSLATELAMNU{A \PAND B} &=& \TRANSLATELAMNU{A} \PAND \TRANSLATELAMNU{B}
			\\
			\TRANSLATELAMNU{\ONEEVAL{e}{e'}{m}{A}} &=& \ONEEVAL{e}{e'}{m}{\TRANSLATELAMNU{A} }
			\\
			\TRANSLATELAMNU{\forall x^{\alpha}. A} &=& \FORALL{x^{\alpha}}{\emptyset} \TRANSLATELAMNU{A} 
			\\
			\\
			\TRANSLATELAMNU{\ASSERT{A}{M}{m}{B}} &=&  \ASSERT{\TRANSLATELAMNU{A}}{M}{m}{\TRANSLATELAMNU{B}}
		\end{array}
	\end{equation}
\end{definition}

\begin{theorem}
The logic for the \NUC\ is a conservative extension of the logic for
the simply typed $\lambda$-calculus above.
\end{theorem}
\begin{proof}
Most \textbf{axioms} are simple extensions of the STLC logic, and thus no proof is required.
Note that for all $ A_{\lambda}. \TRANSLATELAMNU{A}-\EXTINDEP$ as $\TRANSLATELAMNU{A}$ is TCV free. 
Axioms that should hold easily are: $(eq1)$, $(u3)$, $(u4)$, $(e1-3)$,  as these are copied from the STLC logic.
For $(u1)$ an adaption $(u5)$ for all types which are name free is required to instantiate to variables derived  from $\emptyset$.

See Sec \ref{ax_ext_soundness} for soundness of $(ext)$ for the limited case of name free types in the LTC and formulae.
	
All \textbf{rules} for \NUC \ are an extension of the rules for the STLC, with a few exceptions:
\begin{itemize}
	\item Substitution: given all LTCs in the logic are reduced to $\emptyset$ $(u5)$ then substitution in all rules becomes standard as seen in the STLC as we never need to perform $\GAMMA \LSUBST{e}{x}$.
	
	\item \RULENAME{[Lam]} introduces $\FAD{\TCV}$ however this can just be instantiated.
	It also introduces $\FORALL{x}{\TCV}$ which can be reduced to $\FORALL{x}{\emptyset}$ via $(u2)$.
	
	\item $\THINWRT{\VEC{x}}$: Given names are not present then thinness simply implies non-existence of such a variable, i.e. $x \notin \FV{A}$ or $A^{-x}$
	
	\item $\EXTINDEP$ holds for all assertions that are free from $\TCV$ hence holds for all STLC logical assertions.

\end{itemize}
\end{proof}